\documentclass[opre,nonblindrev]{informs4}

\OneAndAHalfSpacedXI 

\usepackage{endnotes}
\let\footnote=\endnote
\usepackage{algorithm,wasysym}
\usepackage{algorithmicx,algpseudocode}
\usepackage{xcolor}
\usepackage{booktabs,multirow}
\usepackage{amsmath,amssymb,amsfonts}
\usepackage{natbib}
\bibpunct[, ]{(}{)}{,}{a}{}{,}%
\def\bibfont{\small}%
\usepackage{enumitem}
\usepackage{lmodern}
\usepackage{rotating}
\usepackage{fancyvrb}
\usepackage{xr}
\usepackage{bbm}
\usepackage{float}
\usepackage{placeins}
\usepackage{smile}
\usepackage{multibib}
\usepackage{mathrsfs}
\newcites{EC}{E-Companion References}

\def\theARTICLETOP{} %
\def\theLRHFirstLine{}
\def\theLRHSecondLine{}
\def\theRRHFirstLine{}
\def\theRRHSecondLine{}

\TheoremsNumberedThrough     
\ECRepeatTheorems
\JOURNAL{Operations Research}

\EquationsNumberedThrough    

\MANUSCRIPTNO{}

\begin{document}
	
\JOURNAL{Operations Research}
\MANUSCRIPTNO{}
\RUNAUTHOR{xxx et al.}
	\RUNTITLE{Expected Shortfall Factor Models}
	\TITLE{Expected Shortfall Factor Models: Common Tail Losses and Expected Returns}
	
	\ARTICLEAUTHORS{%
	\AUTHOR{Yujie Hou,\textsuperscript{a}  Xinbing Kong,\textsuperscript{b} Yalin Wang,\textsuperscript{c} Bin Wu,\textsuperscript{d}\textsuperscript{*}} 
		
	\AFF{$^{a}$University of Chinese Academy of Sciences; $^{b}$Southeast University; $^{c}$Shandong University; $^{d}$University of Science and Technology of China}
		
		\AFF{{$^{*}$Corresponding author}}
		\AFF{{\bf Contact:}
			houyujie26@mails.ucas.ac.cn (YH), xinbingkong@126.com (XK), wangyalin@mail.sdu.edu.cn (YW), bin.w@ustc.edu.cn (BW)}
	}

	\ABSTRACT{We develop an expected shortfall factor model (ESFM) to estimate and price common variation in the severity of lower-tail losses in large panels of asset returns. Mean factor models describe common variation in average returns, while quantile factor models describe common movements in tail thresholds. ESFM instead captures common variation in the average severity of losses below those thresholds. The model combines observed risk exposures with latent common factors. We estimate ESFM using an orthogonalized two-step procedure under which first-stage quantile estimation error has no first-order effect on the ES coefficient estimates. We establish nonasymptotic error bounds for the ES coefficients, a finite-sample Gaussian approximation, and consistent selection of the number of latent factors. Applied to a large panel of equities, ESFM uncovers common factors that react sharply to market stress and contain information not captured by mean and quantile factors. Average returns increase across portfolios sorted on ESFM exposure; high-minus-low portfolios earn annualized returns of 8.0\%--11.7\% and Fama--French five-factor alphas of 10.3\%--15.0\%. These spreads remain positive and statistically significant after conditioning separately and jointly on mean- and quantile-factor exposures. Tail-by-tail spanning tests show that ESFM factors retain significant alphas after controlling for standard traded factors and the corresponding mean and quantile factors. Adding ESFM to these benchmark factor sets increases the maximum attainable Sharpe ratio. These findings identify common loss severity as a distinct and priced dimension of downside risk.}
    
	\maketitle

\section{Introduction}\label{sect:introduction}

Panel factor models for asset returns are commonly formulated for conditional means or conditional quantiles. Mean factor models describe common variation in average returns \citep{connor1986performance,bai2009panel}, whereas quantile factor models estimate common movements in specified conditional quantiles \citep{ando2020quantile,chen2021quantile}. A conditional mean does not describe the lower tail, and a conditional quantile gives only the tail cutoff, not the average return below it. Information within the tail is important because the occurrence and magnitude of extreme losses are distinct features of risk \citep{dierkes2024measuring}. For example, two stocks may have the same conditional 10\% quantile but substantially different average returns among their worst 10\% of outcomes. If these within-tail outcomes move together across stocks, they constitute a source of common downside risk that cannot be recovered from conditional means or quantiles alone. We study whether this common variation in loss severity can be estimated from a large panel of returns and whether stocks' exposures to it are related to expected returns.

The relation between downside risk and the cross-section of stock returns has received considerable attention. Existing studies examine covariation in declining markets through downside beta \citep{ang2006downside}, exposures to common tail risk \citep{kelly2014tail}, sensitivity during extreme market downturns \citep{vanOordt2016systematic}, and lower-tail dependence between individual stocks and the market \citep{chabiyo2018crash}. Other work studies firm-level crash probabilities \citep{jang2019probability}, option-implied bear-market risk \citep{lu2019bear}, return continuation following left-tail events \citep{atilgan2020lefttail}, or exposure to extreme realizations of several systematic factors \citep{chabiyo2022multivariate}. Using a measure of bear-market risk, \citet{massacci2026factor} develop a conditional latent factor model in which an estimated threshold separates good and bad market states and the factor structure may differ across states. More recently, \citet{barunik2026common} use quantile factor analysis to identify common movements in firm-level return quantiles and document a premium for exposure to the lower-tail factor. The resulting measures describe different features of adverse market conditions and need not convey the same information \citep{dierkes2024measuring}. None of these approaches, however, models common variation in the average return below asset-specific conditional quantiles. Such a model is needed to separate observed downside exposures from latent shocks that deepen losses across many assets and to determine whether exposures to those shocks are priced.

We propose an expected shortfall factor model (ESFM) for large panels of asset returns to estimate common variation in the severity of lower-tail outcomes. The central challenge is to separate latent shocks to tail-loss severity from variation associated with observed risk exposures. Doing so requires measuring cross-sectional comovement in tail-loss severity across a high-dimensional panel, even though the relevant tail-loss outcomes are not directly observed and must be constructed using estimated asset-specific thresholds. We address this generated-response problem through an orthogonalized transformation that ensures, under our conditions, that threshold-estimation error enters the second stage only through higher-order terms. ESFM accordingly decomposes conditional ES into a component associated with observed risk exposures and a latent common component with asset-specific loadings. Both the coefficients and the factor structure may depend on the tail probability level $\tau$. The model can therefore recover common movements in tail severity without specifying the underlying shocks in advance.

Expected shortfall (ES) provides the distributional object required for this analysis. At a lower-tail probability $\tau$, ES is the conditional average return given that the return lies below its conditional $\tau$-quantile. Unlike the quantile itself, ES depends on the outcomes below the cutoff and therefore changes as tail losses become more or less severe. ES is a coherent risk measure with an axiomatic foundation for portfolio risk assessment and plays an important role in financial risk management and regulation \citep{artzner1999coherent,Acerbi2002on,Rockafellar2014Superquantile,wang2021axiomatic}. Regression methods for ES allow the tail expectation to depend on observed predictors \citep{dimitriadis2019joint,Patton2019Dynamic,Barendse2020Efficiently,he2023robust,Zhang2025High}. These methods are developed mainly for a single response series or applied separately across assets. They therefore do not use cross-sectional comovement to distinguish asset-specific variation in tail severity from latent shocks that affect many assets at the same time. Without this distinction, a source of common downside risk may be missed even when each asset's ES is estimated accurately in isolation.


\subsection{Estimation and Theory}\label{subsec:intro_method}

Estimation is complicated by two features of the model. First, ES is not elicitable by itself, so its estimation requires the associated conditional quantile \citep{fissler2016higher}. Second, the panel used to estimate the ES factor structure is not directly observed: the tail-loss outcome for each asset depends on its unknown conditional quantile. We address these issues with a two-step procedure. The first step estimates the conditional quantile for each asset. The second step uses these estimates to construct an orthogonalized tail-loss outcome whose conditional expectation equals ES. We then jointly estimate the heterogeneous ES coefficients, factor loadings, and common factors from the resulting panel.

This is not a standard latent-factor problem with an observed dependent variable. The outcome is generated using estimated, asset-specific tail thresholds, and the first-step error enters both coefficient and factor estimation. We derive the estimator for this setting and propose an information criterion for selecting the number of ES factors. The criterion accounts for the generated tail-loss outcome rather than treating it as an observed return panel.

The construction of the second-step outcome is orthogonal to local perturbations in the first-step quantile. Consequently, quantile estimation error has no first-order effect on the ES coefficient estimator. Our main theoretical results are nonasymptotic. We establish simultaneous high-probability bounds for the first-stage quantile estimators and finite-sample error bounds for the asset-specific ES coefficients. The ES coefficient bound separates the oracle ES-regression error, of order $\sqrt{(p+1)/T}$, the contribution from first-step quantile estimation, of order $(p+\log N)/T$, and factor-estimation uncertainty, of order $N^{-1/2}+T^{-1/2}$. Thus, ordinary regression uncertainty decreases with the time-series dimension, the first-step error is higher order because of orthogonality, and estimation of the latent factor structure uses both dimensions of the panel. We further establish a finite-sample Gaussian approximation for standardized linear contrasts of the ES coefficients, with an explicit error bound in $N$, $T$, and the covariate dimension, and prove consistency of the proposed factor-number selector. Complementary asymptotic results are provided in the appendix. Across a range of Monte Carlo designs, ESFM produces smaller coefficient estimation errors than an ES regression that omits the latent factor structure.

\subsection{Common Downside Risk and the Cross-Section of Stock Returns}\label{subsec:intro_empirical}

Recent work on factor estimation with long panels suggests that extending the estimation window can help recover weak latent signals that are obscured in shorter samples \citep{anderson2026long}. We apply ESFM to daily returns on a large panel of liquid Chinese equities drawn from the CSI 300 universe from January 2000 through December 2023. This period includes the 2015 stock market crash, the 2018 trade tensions, and the COVID-19 episode. We extract latent ES factors and compare them with conventional mean- and quantile-factor benchmarks. The mean benchmark follows the approximate-factor approach to asset pricing \citep{connor1986performance,connor1988risk,lettau2020factors}. The quantile benchmark estimates common movements in tail thresholds and is closely related to the pricing evidence in \citet{barunik2026common}. We apply the same estimation and portfolio procedures across the three models, allowing us to assess the information in loss severity separately from that in average returns and tail thresholds.

The estimated ES factors display pronounced movements during periods of market stress, particularly around the 2015 crash. The quantile factors capture movements in the lower part of the return distribution, although their response to the most severe loss episodes is weaker. Generalized correlations between the ESFM factor space and the corresponding mean- and quantile-factor spaces are also relatively low. Together, these findings show that common movements in the severity of tail losses differ empirically from common movements in average returns and tail thresholds.

Exposures to the ES factors are also related to stock returns. Each month, we estimate stock-specific exposures using the preceding 60 months and form portfolios that are held for one month. Across $\tau\in\{0.10,0.20,0.30\}$, high-minus-low portfolios formed on ESFM exposures earn annualized average returns between 8.0\% and 11.7\%. Their Fama--French five-factor alphas range from 10.3\% to 15.0\%, with Newey--West $t$-statistics above 3.2. Average returns also increase monotonically from the lowest- to the highest-exposure ESFM quintile in all six combinations of tail level and observable-factor weighting. The return spread is therefore not generated solely by the unusual performance of one endpoint portfolio. Portfolios formed on Mean and QFM exposures produce smaller and less stable spreads. In two-pass regressions, the estimated price of ES factor risk remains positive when Mean and QFM exposures enter jointly and when the Fama--French factors are included.

The return relation also survives more demanding tests of incremental pricing content. At the stock level, dependent bivariate sorts produce annualized conditional high-minus-low spreads of approximately 5.6\%--9.9\% after conditioning separately on QFM and Mean exposures, jointly on both exposures, or on idiosyncratic volatility. The corresponding 95\% confidence intervals lie above zero in every specification. At the factor level, spanning regressions yield annualized ESFM alphas of 6.6\%--9.4\% after controlling jointly for MKT, SMB, HML, RMW, CMA, MOM, and the corresponding Mean and QFM factors; the associated Newey--West $t$-statistics range from 2.48 to 3.39. Adding ESFM raises the maximum Sharpe ratio from 1.52--1.81 for the benchmark factor sets to 1.78--2.05. Joint tests across the three tail levels are more sensitive to whether the observed factors are value or equal weighted, so the spanning evidence is most direct on a tail-by-tail basis. The stock-level and factor-level results therefore indicate that ESFM contributes to both the cross-section of expected returns and the investment opportunity set beyond nearby measures of downside risk.

The comparison with the quantile model is especially informative. Both models use information from the lower part of the conditional return distribution, but only ESFM incorporates the magnitude of losses below the estimated quantile. Their different pricing performance therefore indicates that movements in tail thresholds do not fully describe the severity of tail losses. This result is also distinct from measures based on a stock's comovement with extreme losses in the market or other prespecified factors \citep{vanOordt2016systematic,chabiyo2018crash,chabiyo2022multivariate}. Our evidence instead concerns stocks' exposures to latent common shocks that change loss severity below stock-specific tail thresholds.

\subsection{Relation to the Literature}\label{subsec:intro_literature}

{The econometric contribution is most closely related to the ES regression literature for univariate responses.} Existing work develops joint quantile--ES regressions, dynamic specifications, and robust or high-dimensional estimators based on observed predictors \citep{dimitriadis2019joint,Patton2019Dynamic,Barendse2020Efficiently,he2023robust,Zhang2025High}. We introduce a latent common component to a heterogeneous ES regression. Both the tail-loss outcome and the common factor structure must be estimated because the outcome used in the second step is constructed from first-stage quantile estimates. For this generated-outcome setting, we develop an orthogonalized two-step estimator, an information criterion for factor-number selection, nonasymptotic error bounds, and a finite-sample Gaussian approximation. The cross-section is used to recover common shocks to loss severity that asset-by-asset ES regressions leave unmodeled.\footnote{This panel structure is also related to likelihood-based factor models for non-Gaussian outcomes, such as \citet{kong2025high}, but the target and estimation problem differ: our dependent variable is an orthogonalized tail outcome rather than a directly observed binary response.}

The paper also contributes to models of distributional comovement in financial panels. Classical approximate factor models recover low-dimensional common variation in mean returns \citep{connor1986performance,connor1988risk,bai2009panel}, while characteristic-based factor models link time-varying loadings to firm characteristics \citep{kelly2019characteristics}. Related methods use observed proxies to improve latent-factor estimation \citep{fan2021augmented,wan2024mining}. Quantile factor models instead recover common variation at specified conditional quantiles \citep{ando2020quantile,chen2021quantile,barunik2021measurement,belloni2023highdimensional,yang2024assetpricing}. In asset pricing, \citet{barunik2026common} document that exposure to a common idiosyncratic lower-tail quantile factor is priced. ESFM models a different distributional object: the average return among outcomes below an asset-specific conditional quantile. It can therefore capture common changes in the magnitude of tail losses even when the corresponding quantile changes little. This object is not determined by a conditional mean or by a quantile at one probability level.

ESFM is also related to the asset-pricing literature on downside and tail risk. Prior studies examine covariation with market declines, exposure to aggregate tail risk, systematic tail risk, and lower-tail dependence with the market \citep{ang2006downside,kelly2014tail,vanOordt2016systematic,chabiyo2018crash}. Other work studies hybrid tail covariance, firm-level crash probabilities, bear-market exposure, return continuation following left-tail events, exposure to extreme realizations of multiple systematic factors, {and stock co-jump dependence \citep{bali2014hybrid,jang2019probability,lu2019bear,atilgan2020lefttail,pelger2020understanding,chabiyo2022multivariate,ding2024stock}.} In particular, \citet{chabiyo2022multivariate} define adverse states using extreme realizations of a prespecified set of systematic factors and measure whether an individual stock simultaneously enters its own tail. Aggregate and option-based approaches examine tail-risk premia, tail-risk dynamics, and ex ante concerns about rare losses \citep{bollerslev2011tails,massacci2017tail,gao2019tail}. {Most directly, \citet{massacci2026factor} allow the latent factor structure of returns to differ between good and bad market states defined by bear-market risk.} These approaches characterize downside risk through market states, tail-event probabilities or dependence, or aggregate tail-risk measures. ESFM instead recovers latent common variation in the average severity of losses below stock-specific conditional thresholds. These objects capture different features of adverse return distributions and need not contain the same information \citep{dierkes2024measuring}.

Finally, our pricing analysis relates to the broader latent-factor asset-pricing literature. \citet{herskovic2016common} show that idiosyncratic volatility has a strong common component with cross-sectional pricing implications. Other methods incorporate expected returns, observed characteristics, or pricing restrictions into latent-factor extraction \citep{lettau2020estimating,lettau2020factors,gu2021autoencoder}, {while related work interprets or shrinks high-dimensional factor representations and addresses omitted factors \citep{kozak2018interpreting,kozak2020shrinking,giglio2021asset,he2023shrinking,bryzgalova2023asset}}. ESFM extracts factors for a different statistical target: conditional tail expectations. Cross-sectional pricing errors and future portfolio returns do not enter its estimation criterion. The subsequent pricing tests therefore assess an implication of the estimated tail factor structure rather than impose that implication during factor extraction.

\paragraph{Organization.} Section~\ref{sec:Methodologies} introduces ESFM, the two-step estimator, and the factor-number criterion. Section~\ref{sec:Theoretical Results} develops the nonasymptotic theory and establishes the consistency of the factor-number selector. Section~\ref{sec:Simulation Study} reports the simulation evidence. Section~\ref{sec:Empirical Application} studies common downside risk and expected returns in the large panel equity market. {All asymptotic results are provided in the Supplemental Material}.

\paragraph{Notation.} For a vector or matrix $A$, $A^\top$ denotes its transpose, and $\|A\|=\left(\operatorname{tr}(A^\top A)\right)^{1/2}$ denotes the Euclidean norm when $A$ is a vector and the Frobenius norm when $A$ is a matrix. For a matrix $A$, its operator norm is defined as $\|A\|_{\rm op}=\sqrt{\lambda_{\max}(A^\top A)}.$ For a positive definite matrix $\Sigma$ and a conformable vector $v$, define the weighted norm $\|v\|_{\Sigma} = \sqrt{v^\top \Sigma v}$. For a symmetric matrix $A$, $\lambda_{\min}(A)$ and $\lambda_{\max}(A)$ denote its smallest and largest eigenvalues, respectively; $A\succ0$ and $A\succeq0$ indicate that $A$ is positive definite and positive semidefinite, respectively. For a full-column-rank matrix $A$, let $P_A=A(A^\top A)^{-1}A^\top$ denote the orthogonal projection matrix onto its column space. $I_k$ denotes the $k\times k$ identity matrix, and $\otimes$ denotes the Kronecker product. The symbols $\stackrel{p}{\longrightarrow}$, $\stackrel{d}{\longrightarrow}$, and $\stackrel{a.s.}{\longrightarrow}$ denote convergence in probability, in distribution, and almost surely, respectively. For a deterministic positive sequence $a_{NT}$, $X_{NT}=o_p(a_{NT})$ means that $X_{NT}/a_{NT}\stackrel{p}{\longrightarrow}0$, whereas $X_{NT}=O_p(a_{NT})$ means that $X_{NT}/a_{NT}$ is stochastically bounded. For any positive integer $m$, let $[m]=\{1,\ldots,m\}$. Finally, $\mathbbm{1}(\cdot)$ denotes the indicator function.

\section{Methodologies}\label{sec:Methodologies}
\subsection{Preliminaries}

Let $Y$ be a real-valued random variable with $E|Y|<\infty$, and let $F_Y$ denote its cumulative distribution function. For a given quantile level $\tau\in(0,1)$, the quantile and ES  are defined as 
\[
Q_{\tau}(Y)=\inf\{y\in\mathbb{R}:F_Y(y)\geq\tau\},\qquad \mathrm{ES}_{\tau}(Y)=\frac{1}{\tau}\int_0^\tau Q_u(Y)\,du.
\]
If the distribution of $Y$ is continuous at $Q_\tau(Y)$, then
\[
\mathrm{ES}_{\tau}(Y)=E\!\left[Y\mid Y\leq Q_\tau(Y)\right].
\]

The quantile and ES characterize different aspects of tail risk. The quantile determines the threshold below which a fraction $\tau$ of the distribution lies, but it does not account for the magnitude of outcomes beyond that threshold. In contrast, ES averages the outcomes within the lower tail and therefore captures the severity of losses once the quantile threshold has been crossed. Consequently, two random variables may have the same $\tau$-quantile while exhibiting substantially different ES values if their losses below the quantile differ in magnitude \citep{Acerbi2002on,Rockafellar2002conditional,Tasche2002expected}.

An important advantage of ES is its coherence as a risk measure. Under the lower-tail return convention adopted in this paper, ES satisfies the following properties. For random variables $Y_1$ and $Y_2$, a constant $c\in\mathbb{R}$, and $a\geq0$:

\begin{enumerate}
    \item \textit{Monotonicity}. If $Y_1\leq Y_2$ almost surely, then $\mathrm{ES}_\tau(Y_1)\leq\mathrm{ES}_\tau(Y_2).$    

    \item \textit{Translation equivariance}. $\mathrm{ES}_\tau(Y+c)=\mathrm{ES}_\tau(Y)+c$.

    \item \textit{Positive homogeneity}. $\mathrm{ES}_\tau(aY)=a\,\mathrm{ES}_\tau(Y)$.

    \item \textit{Superadditivity}. $\mathrm{ES}_\tau(Y_1+Y_2)\geq\mathrm{ES}_\tau(Y_1)+\mathrm{ES}_\tau(Y_2)$.
\end{enumerate}

These properties make ES suitable for measuring downside risk because it accounts for the magnitude of outcomes within a tail region of probability $\tau$. Although ES is not elicitable on its own, \cite{fissler2016higher} show that the quantile and ES are jointly elicitable, which provides a foundation for their joint estimation.

\subsection{Model Setup}\label{subsec:Model Setup} 

Suppose we observe, for each of $N$ units over $T$ periods, a response variable $Y_{it}$ along with a $(p+1)$-dimensional vector $X_{it}=(1,X_{it,1},\ldots,X_{it,p})^\top$. For a given quantile level $\tau\in (0,1)$, 
we characterize the  conditional $\tau$-quantile of $Y_{it}$ given $X_{it}$ using a standard linear quantile regression specification $Q_{\tau}(Y_{it}|X_{it})=X_{it}^\top\alpha_{i,\tau}^0$. 
Our primary interest is in the corresponding conditional ES. To capture cross-sectional information, we introduce the Expected Shortfall Factor Model (ESFM), which incorporates both observed covariate effects and latent common factors. Specifically, we assume that 
	\begin{equation}\label{eq:ES models}
		\mathrm{ES}_{\tau}(Y_{it}|X_{it},f_{t,\tau}^0)=X_{it}^\top\beta_{i,\tau}^0+\lambda_{i,\tau}^{0\top} f_{t,\tau}^0,\quad i\in[N],\quad t\in[T].
	\end{equation}
Here, $\alpha_{i,\tau}^0=(\alpha_{i,0,\tau}^{0},\alpha_{i,1,\tau}^0,\ldots,\alpha_{i,p,\tau}^0)^\top$ and $\beta_{i,\tau}^0=(\beta_{i,0,\tau}^0,\beta_{i,1,\tau}^0,\ldots,\beta_{i,p,\tau}^0)^\top$ denote the true quantile‐ and ES‐regression coefficients, respectively. $f_{t,\tau}^0$ is an $r_{\tau}^0$-dimensional vector  of unobservable factors, where $r_{\tau}^0$ denotes the true number of factors; $\lambda_{i,\tau}^{0}$ is the corresponding loading vector. Moreover, we include a leading constant equal to 1 in each $X_{it}$, so that the first elements of $\alpha_{i,\tau}^0$ and $\beta_{i,\tau}^0$ represent the intercept terms  in the quantile and ES models, respectively. Allowing $\alpha_{i,\tau}^0$, $\beta_{i,\tau}^0$, $f_{t,\tau}^0$, $\lambda_{i,\tau}^0$, and $r_{\tau}^0$ to vary with $\tau$ enables our ES factor model  to flexibly capture quantile-specific heterogeneity and tail-specific dynamics.\footnote{The specification nests a panel ES regression without latent factors when $r_\tau^0=0$ \citep{dimitriadis2019joint,he2023robust,Zhang2025High} and a pure ES factor model when $\beta_{i,\tau}^0=0$. It also permits distinct covariate vectors in the quantile and ES equations; we use a common vector $X_{it}$ throughout to simplify notation.}
	
We treat $f_{t,\tau}^0$ and $\lambda_{i,\tau}^0$ as  unknown parameters.\footnote{This latent factor structure is commonly referred to as the ``interactive effects'' model in the mean‐regression literature (e.g., \citealt{bai2009panel}) and has also been extended to the quantile regression framework (e.g., \citealt{ando2020quantile}).} By capturing common shocks through $f_{t,\tau}^0$ and unit‐specific loadings $\lambda_{i,\tau}^0$, it provides a flexible way to model cross‐sectional dependence in tail expectations. This factor structure parsimoniously captures residual tail dependence beyond observed covariate effects while allowing units to respond heterogeneously to common tail shocks.

\subsection{Two-Stage Estimation of the ES Factor Model}\label{subsec:Two-Stage Estimation}
Our objective is to estimate the  parameters
$B_\tau^0=(\beta_{1,\tau}^0,\ldots,\beta_{N,\tau}^0)^\top$, $\Lambda_\tau^0=(\lambda_{1,\tau}^0,\ldots,\lambda_{N,\tau}^0)^\top$ and $F_\tau^0=(f_{1,\tau}^0,\ldots,f_{T,\tau}^0)^\top$. A well-known result in the literature on factor models is that the factors $\{f_{t,\tau}^0\}$ and loadings $\{\lambda_{i,\tau}^0\}$ cannot be
separately identified without imposing normalizations; see \cite{bai2002determining}. Following \cite{bai2013principal}, we impose the following normalization
conditions without loss of generality:
	{
		\begin{align}\label{eq:constraints for factor}
			F_\tau^{0\top} F_\tau^0/T=I_{r_\tau^0},\quad\Lambda_\tau^{0\top}\Lambda_\tau^0/N \rightarrow \Sigma_{\Lambda_\tau^0} \succ 0 .
	\end{align}}
Let $\theta_\tau^0=(f_{1,\tau}^{0\top},\ldots,f_{T,\tau}^{0\top},\lambda_{1,\tau}^{0\top},\ldots,\lambda_{N,\tau}^{0\top})^\top$ denote the vector of true parameters. Define the admissible parameter space 
	\[
	\begin{aligned}
		\Theta^{r_\tau^0}=\{\theta_\tau\in\mathbb{R}^{(N+T) r_{\tau}^0}:
		&\ \lambda_{i,\tau},f_{t,\tau}\in\mathbb{R}^{r_{\tau}^0} \text{ for all } i, t, \\
		&\ \{f_{t,\tau}\} \text{ and } \{\lambda_{i,\tau}\} \text{ satisfy the normalization in \eqref{eq:constraints for factor}} \bigr\}.
	\end{aligned}
	\]
	
Since ES depends on the quantile but not vice versa, the quantile parameters $\alpha_{i,\tau}$ naturally act as nuisance parameters when ES is the primary object of interest. Inspired by the Neyman-orthogonality idea, which constructs score functions that are locally insensitive to first-order perturbations in nuisance parameters \citep{neyman1979c,chernozhukov2018double,Barendse2020Efficiently}, we adopt the two-stage strategy of \citet{Barendse2020Efficiently,he2023robust,Zhang2025High} to avoid non-convex optimization. Specifically, estimation of the ESFM in \eqref{eq:ES models} proceeds in two stages:
	\begin{itemize}
		\item \textbf{Stage 1: Quantile Regression}
        
{Since the quantile specification is only used to obtain the individual tail thresholds required for the subsequent ES estimation, we estimate the conditional quantile parameters separately for each unit using standard quantile regression: \footnote{This unit-by-unit time-series formulation is sufficient for our purpose and avoids imposing additional restrictions on the cross-sectional dependence structure at the quantile level.}}	
        
		\begin{align}\label{eq:step 1 QR}
			\widehat{\alpha}_{i,\tau}=\argmin_{\alpha_{i,\tau}\in\mathbb{R}^{p+1}}\frac{1}{T}\sum_{t=1}^T\rho_{\tau}(Y_{it}-X_{it}^\top\alpha_{i,\tau}),\quad\text{for each $i$},
		\end{align}
		and let $\widehat{A}_\tau=(\widehat{\alpha}_{1,\tau},\ldots,\widehat{\alpha}_{N,\tau})^\top$, where $\rho_{\tau}(u)=(\tau-\mathbbm{1}(u<0))u$ is the check loss function (\citealt{koenker1978regression}).
		
		\item \textbf{Stage 2: Orthogonalized ES Regression}
		
		Define 
		\begin{align*}
			S_0(\alpha_{i,\tau},\beta_{i,\tau},f_{t,\tau},\lambda_{i,\tau};Y_{it},X_{it})=\tau \lambda_{i,\tau}^\top f_{t,\tau}+\tau X_{it}^\top(\beta_{i,\tau}-\alpha_{i,\tau})-\left(Y_{it}-X_{it}^\top\alpha_{i,\tau}\right) \mathbbm{1}\left(Y_{it} \leq X_{it}^\top\alpha_{i,\tau}\right).
		\end{align*}
		Given $\widehat{A}_\tau$ obtained from the first step, the ES parameters are estimated as
		\begin{align}\label{eq:step 2 ES}
			\begin{aligned}
				(\widehat{B}_\tau,\widehat{\theta}_\tau)=&\argmin_{B_\tau\in\mathbb{R}^{N\times (p+1)},\theta_\tau\in\Theta^{r_\tau^0}}\frac{1}{NT}\sum_{i=1}^N\sum_{t=1}^TS_0^2(\widehat{\alpha}_{i,\tau},\beta_{i,\tau},f_{t,\tau},\lambda_{i,\tau};Y_{it},X_{it})\\
				=&\argmin_{B_\tau\in\mathbb{R}^{N\times (p+1)},\theta_\tau\in\Theta^{r_\tau^0}}\frac{1}{NT}\sum_{i=1}^N\sum_{t=1}^T\left[Z_{it}(\widehat{\alpha}_{i,\tau})-\tau(X_{it}^\top\beta_{i,\tau}+\lambda_{i,\tau}^\top f_{t,\tau})\right]^2.
			\end{aligned}
		\end{align}
		Here, for each $i\in[N]$ and $t\in[T]$, we define
		\begin{align}\label{eq:Z}
			Z_{it}(\alpha_{i,\tau})=\left(Y_{it}-X_{it}^\top\alpha_{i,\tau}\right) \mathbbm{1}\left(Y_{it} \leq X_{it}^\top\alpha_{i,\tau}\right)+\tau X_{it}^\top\alpha_{i,\tau}.
		\end{align}

	\end{itemize}
	The rationale for the estimation procedure in \eqref{eq:step 2 ES} lies in the following considerations. 
	Define 
    
    \begin{equation*}
		\begin{aligned}
			\psi_0(\alpha_{i,\tau}, \beta_{i,\tau}, \lambda_{i,\tau}; X_{it},f_{t,\tau}) =E\left(Z_{it}(\alpha_{i,\tau}) \mid X_{it},f_{t,\tau}\right) -\tau \left(X_{it}^\top \beta_{i,\tau}+\lambda_{i,\tau}^\top f_{t,\tau}\right),
		\end{aligned}
	\end{equation*}
	which satisfies $\psi_0(\alpha_{i,\tau}^0, \beta_{i,\tau}^0, \lambda_{i,\tau}^0; X_{it},f_{t,\tau}^0) =0$. Let $F_{Y_{it}|X_{it}}$ be the conditional distribution function of $Y_{it}$ given $X_{it}$. Provided that $F_{Y_{it}|X_{it}}$ is continuously 
	differentiable, it can be shown that for any $\beta_{i,\tau}$,
	\begin{align*}
		\partial_{\alpha_{i,\tau}} \psi_0(\alpha_{i,\tau}, \beta_{i,\tau}, \lambda_{i,\tau}; X_{it},f_{t,\tau})\mid_{\alpha_{i,\tau}=\alpha_{i,\tau}^0}  =\left\{\tau-F_{Y_{it} \mid X_{it}}\left(X_{it}^\top \alpha_{i,\tau}^0\right)\right\} X_{it}=0,
	\end{align*}
	thereby satisfying the Neyman orthogonality condition.
	Therefore, the quantile regression estimation error is first-order negligible in the ES regression estimation and thus does not affect the asymptotic distribution of the ES estimators.
	
Because the latent factors and loadings enter the second-step criterion
multiplicatively, the estimator
$\widehat{\theta}_{\tau}$ does not admit an analytical
closed-form solution. We therefore adopt an iterative procedure to estimate $\theta_\tau$ and $B_\tau$. 
For a given $F_{\tau}$ satisfying $F_{\tau}^{\top}F_{\tau}/T=I_{r_\tau^0}$,
let
$$
M_{F_{\tau}}=I_T-F_{\tau}(F_{\tau}^{\top}F_{\tau})^{-1}F_{\tau}^{\top}=I_T-\frac{1}{T}F_{\tau}F_{\tau}^{\top}.
$$
For each $i\in[N]$, define $X_i=(X_{i1},\ldots,X_{iT})^\top$ and $Z_i^\ast(\widehat{\alpha}_{i,\tau})=(Z^\ast_{i1}(\widehat{\alpha}_{i,\tau}),\ldots,Z^\ast_{iT}(\widehat{\alpha}_{i,\tau}))^\top$, where $Z_{it}^\ast(\widehat{\alpha}_{i,\tau})=\tau^{-1}Z_{it}(\widehat{\alpha}_{i,\tau})$.
Given $\widehat{\alpha}_{i,\tau}$ and $F_\tau$, the least squares estimator for $\beta_{i,\tau}$ is $\widehat{\beta}_{i,\tau}=\left(X_i^\top M_{F_\tau}X_i\right)^{-1}\left(X_i^\top M_{F_\tau}Z_i^\ast(\widehat{\alpha}_{i,\tau})\right)$. The corresponding residual matrix is
	\begin{align*}
		\widehat{W}_\tau=(\widehat{W}_{it,\tau})_{N\times T}, \quad \widehat{W}_{it,\tau}=Z_{it}^\ast(\widehat{\alpha}_{i,\tau})-X_{it}^\top\widehat{\beta}_{i,\tau}.
	\end{align*}
We then obtain updated factor estimates by applying classical principal components analysis (PCA) to $\widehat{W}_\tau^\top \widehat{W}_\tau/(TN)$. Let $\widehat{F}_\tau$ be $\sqrt{T}$ times the eigenvectors corresponding to the largest $r_\tau^0$ eigenvalues of  $\widehat{W}_\tau^\top \widehat{W}_\tau/(TN)$, and the  corresponding loading matrix be $\widehat{\Lambda}_\tau=\widehat{W}_\tau\widehat{F}_\tau/T$. The updates are repeated until convergence.\footnote{For the theoretical analysis, $(\widehat{B}_\tau,\widehat{F}_\tau)$ denotes a global minimizer of~\eqref{eq:step 2 ES}. Algorithm~\ref{alg:ES_factor} provides a computational procedure for obtaining a numerical solution to this problem.} The two-stage estimation procedure is summarized in Algorithm \ref{alg:ES_factor}.\footnote{We initialize the iteration by regressing $Z_{it}^{\ast}(\widehat{\alpha}_{i,\tau})$ on $X_{it}$ separately for each unit, forming the residual matrix $\widehat{W}_{\tau}^{(0)}$, and setting the initial $\widehat{F}_{\tau}$ equal to $\sqrt{T}$ times the eigenvectors associated with the largest $r$ (candidate number of factors) eigenvalues of $\widehat{W}_{\tau}^{(0)\top}\widehat{W}_{\tau}^{(0)}/(TN)$.}
	
	\begin{algorithm}[H]
		\caption{Two-Stage Estimation of the ES Factor Model}
		\label{alg:ES_factor}
		\begin{algorithmic}[1]
			
			\Require Quantile level $\tau$, candidate number of factors $r$, 
			initial factor estimator $\widehat{F}_\tau$.
			
			\State Estimate $\widehat{A}_\tau$ via quantile regression in \eqref{eq:step 1 QR}.
			
			\State Compute $Z_{it}(\widehat{\alpha}_{i,\tau})$ as in \eqref{eq:Z}, and define
			\[
			Z_i^\ast(\widehat{\alpha}_{i,\tau})
			=
			\frac{1}{\tau}
			\left(
			Z_{i1}(\widehat{\alpha}_{i,\tau}),\ldots,
			Z_{iT}(\widehat{\alpha}_{i,\tau})
			\right)^\top.
			\]
			
			\State Update $\widehat{\beta}_{i,\tau}$ by
			\[
			\widehat{\beta}_{i,\tau}
			=
			\left(X_i^\top M_{\widehat{F}_\tau}X_i\right)^{-1}
			X_i^\top M_{\widehat{F}_\tau}
			Z_i^\ast(\widehat{\alpha}_{i,\tau}).
			\]
			
			\State Compute
			\[
			\widehat{W}_{it,\tau}
			=
			Z_{it}^\ast(\widehat{\alpha}_{i,\tau})
			-
			X_{it}^\top \widehat{\beta}_{i,\tau},
			\]
			form $\widehat{W}_\tau=(\widehat{W}_{it,\tau})_{N\times T}$,
			and update $\widehat{F}_\tau$ as $\sqrt{T}$ times the eigenvectors associated with the largest $r$ eigenvalues of
			$\widehat{W}_\tau^\top \widehat{W}_\tau/(TN)$.
			
			\State Repeat Steps 3--4 until convergence. 
			Compute the loading matrix
			$
			\widehat{\Lambda}_\tau
			=
			\widehat{W}_\tau \widehat{F}_\tau/T.
			$
			
			\Ensure Final estimators 
			$\widehat{A}_\tau$, $\widehat{B}_\tau$, and $\widehat{\theta}_\tau$.
			
		\end{algorithmic}
	\end{algorithm}
	

	
\subsection{Determining the Number of Factors}\label{subsec:determine}
In practice, the number of latent factors $r_{\tau}^0$ in model \eqref{eq:ES models} is unknown and must be determined before implementing the two-stage ES estimation procedure. Although several methods exist for selecting the number of factors in panel mean models (e.g., \citealp{bai2002determining,ahn2013eigenvalue,kong2017number}), they are designed for models where the response variable is directly observed and follows a mean regression structure. Within the ES factor model framework, the target variable
$Z_{it}^{\ast}(\alpha_{i,\tau}^{0})$ is unobservable and is replaced by its feasible counterpart $Z_{it}^{\ast}(\widehat{\alpha}_{i,\tau})$, constructed using the first-stage quantile estimator. As a result, conventional factor-number selection methods are not directly applicable, as their theoretical validity relies on assumptions that may not hold in our framework.
	
In this paper, we propose a modified Information Criterion (IC) tailored to the  ES factor model. Let $\widehat{V}_\tau(r)$ denote the mean squared residual 
from the ES model, defined as
	\begin{equation}\label{eq:Vr}
		\widehat{V}_\tau(r)=\frac{1}{NT}\sum_{i=1}^{N}\sum_{t=1}^{T}
		\left[Z_{it}^{\ast}(\widehat{\alpha}_{i,\tau})-X_{it}^{\top}\widehat{\beta}_{i,\tau}(r)-\widehat{\lambda}_{i,\tau}(r)^{\top}\widehat{f}_{t,\tau}(r)\right]^{2},
	\end{equation}
where $\widehat{\beta}_{i,\tau}(r)$, 
$\widehat{\lambda}_{i,\tau}(r)$, 
and $\widehat{f}_{t,\tau}(r)$ 
are the estimators obtained from Algorithm \ref{alg:ES_factor} for a given number of factors $r$. Our information criterion is defined as
	\begin{equation}\label{eq:IC}
		\mathrm{IC}_{\tau}(r)=\log\widehat{V}_\tau(r)+r\cdot q(N,T),
	\end{equation}
where $q(N,T)$ is a penalty function that depends on the panel dimensions. 
Following \cite{bai2002determining},
we specify
$q(N,T)=\log\left((NT)/(N+T)\right)\left((N+T)/(NT)\right).$ The estimated number of factors is  selected as $\widehat{r}_{\tau}= \underset{0\le r\le R_{\max}}{\arg\min}\; \mathrm{IC}_{\tau}(r),$
where $R_{\max}$ is a prescribed
maximum number of factors.

\section{Theoretical Results}\label{sec:Theoretical Results}
In this section, we develop nonasymptotic theory for the proposed two-stage ES factor estimator. We first establish high-probability error bounds for the first-stage quantile estimator and the second-stage ES coefficient estimator in Section \ref{sec:3.1}. Section \ref{sec:3.2} establishes a finite-sample Gaussian approximation for standardized linear combinations of the estimation errors in the unit-specific ES coefficient. Section \ref{sec:3.3} establishes the consistency of the proposed information criterion for selecting the number of latent factors.

\subsection{Estimation Error Bounds}\label{sec:3.1}
We begin with the first-stage panel quantile regressions, which provide the generated responses used in the second-stage ES estimation. The following result provides simultaneous control of the first-stage estimation errors over all cross-sectional units.

Let $\Sigma_i=E(X_{it}X_{it}^{\top})$ denote the population second-moment matrix of the covariates for unit $i\in[N]$, and define the quantile regression error as $\epsilon_{it}=Y_{it}-X_{it}^{\top}\alpha_{i,\tau}^0$. We impose the following conditions.
\begin{assumption}
\label{assum:Regularity}
The following conditions hold uniformly over $i\in[N]$ and $t\in[T]$.
\begin{enumerate}[label=(\roman*)]
\item The conditional quantile restriction $P(\epsilon_{it}\leq0\mid X_{it})=\tau$ holds almost surely. Moreover, the conditional distribution of $\epsilon_{it}$ given $X_{it}$ admits a density $f_{\epsilon_{it}\mid X_{it}}(\cdot\mid X_{it})$, and there exist constants $\underline f,L_0>0$ such that, almost surely,
$f_{\epsilon_{it}\mid X_{it}}(0\mid X_{it})\geq\underline f$ and $\left|f_{\epsilon_{it}\mid X_{it}}(x\mid X_{it})-f_{\epsilon_{it}\mid X_{it}}(0\mid X_{it})\right|\leq L_0|x|$, where $x\in\mathbb R$.
\item There exist constants $0<\kappa_{\min}\leq\kappa_{\max}<\infty$ such that $\kappa_{\min}I_{p+1}\preceq\Sigma_i\preceq
\kappa_{\max}I_{p+1}$.
\item Let $W_{it}=\Sigma_i^{-1/2}X_{it}$, then $\ E(W_{it}W_{it}^{\top})=I_{p+1}$. There exists a constant $\upsilon_1\geq1$ such that, for every $u\in\mathbb S^{p}$ and $x\geq0$, $P\left(|u^{\top}W_{it}|\geq\upsilon_1x\right)\leq 2\exp(-x^2/2)$, where $\mathbb S^{p}$ denotes the unit sphere in $\mathbb R^{p+1}$.
\item For every $i\in[N]$, the process $\{(Y_{it},X_{it}^{\top},f_{t,\tau}^{0\top})^{\top}:t\in\mathbb Z\}$ is strictly stationary. Let $\mathscr F_{i,a}^{b}=\sigma((Y_{it},X_{it}^{\top},f_{t,\tau}^{0\top})^{\top}:a\leq t\leq b)$ and define its strong-mixing coefficients by
$$
\alpha_i(h)=\sup_{s\in\mathbb Z}\sup_{\substack{A\in\mathscr F_{i,-\infty}^{s}\\
B\in\mathscr F_{i,s+h}^{\infty}}}\left|P(A\cap B)-P(A)P(B)\right|.
$$
There exist constants $c_\alpha,a_\alpha>0$ such that $\sup_{i\in[N]}\alpha_i(h)\leq c_\alpha\exp(-a_\alpha h),\ h\geq1$.
\end{enumerate}
\end{assumption}

Assumption \ref{assum:Regularity} collects standard and relatively mild regularity conditions commonly used in quantile regression. Part (i) ensures correct specification of the conditional quantile and provides sufficient local curvature for identification. Part (ii) imposes uniform eigenvalue bounds on the covariate second-moment matrices, ensuring a stable and nondegenerate design. Part (iii) imposes a standard sub-Gaussian tail condition on the standardized covariates to control stochastic fluctuations. Parts (i)--(iii) are uniform panel analogues of Conditions 1 and 2 in \cite{he2023robust}. Part (iv) imposes strict stationarity and geometric $\alpha$-mixing, a commonly used weak-dependence condition that permits serial dependence within each unit as long as it decays sufficiently rapidly.

\begin{proposition}
\label{pro:quantile}
Let $s_{N,\delta}=p+2+\log\left({2N}/{\delta}\right)$ for every
$\delta\in(0,1)$. Under Assumption \ref{assum:Regularity}, there exist constants
$C_1,C_2>0$, depending only on $(\upsilon_1,c_\alpha,a_\alpha)$, such that, if $T\geq C_2\left\{s_{N,\delta}^{2}+L_0^2\underline f^{-4}s_{N,\delta}\right\}$, then, with probability at least $1-\delta$,
$$
\max_{i\in[N]}\|\widehat{\alpha}_{i,\tau}-\alpha_{i,\tau}^0\|_{\Sigma_i}\leq C_1\underline f^{-1}\sqrt{\frac{s_{N,\delta}}{T}}.
$$
\end{proposition}
Proposition \ref{pro:quantile} provides a simultaneous finite-sample guarantee for the first-stage estimators. With probability at least $1-\delta$, their leading error rate is $\sqrt{\left(p+\log(N/\delta)\right)/{T}}$. For a single cross-sectional unit, this bound reduces to $\sqrt{\left(p+\log(1/\delta)\right)/T}$, which is of the same order as the finite-sample rate established under independent observations in Proposition 3 of \cite{he2023robust}. It does, however, require the moderate-deviation condition $T\gtrsim\{p+\log(N/\delta)\}^2$, while the additional $\log N$ term
reflects the cost of simultaneous control over all $N$ units.

\begin{remark}
The conditions in Assumption \ref{assum:Regularity} are satisfied by a broad class of stationary location-scale quantile models with nondegenerate sub-Gaussian covariates, a conditional density that is positive and Lipschitz around the target quantile, and geometrically $\alpha$-mixing dynamics. These conditions allow heavy-tailed quantile errors, since no moment condition is imposed on $\epsilon_{it}$, and they do not require cross-sectional independence. Hence, when the sample-size condition in Proposition \ref{pro:quantile} holds, the first-stage quantile estimators are uniformly accurate over all units. Moreover, the Neyman-orthogonal construction of the second-stage generated response makes the ES estimator locally insensitive to perturbations in the first-stage quantile estimates. Consequently, first-stage estimation errors have no first-order effect on the ES estimation and enter only through higher-order terms. This substantially reduces the sensitivity of the proposed two-stage procedure to first-stage estimation uncertainty.
\end{remark}

We now turn to the second-stage ES coefficient estimator, which involves both an estimated quantile threshold and an unobserved factor component. Recall the generated response $Z_{it}(\cdot)$ defined in \eqref{eq:Z}, and let $\mathcal C_{it}=\sigma(X_{it},f_{t,\tau}^{0})$. At the true quantile coefficient, $Z_{it}(\alpha_{i,\tau}^{0})=\tau X_{it}^{\top}\beta_{i,\tau}^{0}+\tau\lambda_{i,\tau}^{0\top}f_{t,\tau}^{0}+\xi_{it,\tau}$, $E(\xi_{it,\tau}\mid\mathcal C_{it})=0$, where $\xi_{it,\tau}$ is the oracle ES regression error. In addition to Assumption \ref{assum:Regularity}, the following conditions control the factor-conditional quantile behavior and the tails of the variables entering the second-stage analysis.

\begin{assumption}\label{assum:ES-score}
The following conditions hold uniformly over $i\in[N]$ and $t\in[T]$.
\begin{enumerate}[label=(\roman*)]
\item Let $F_{it}(u\mid\mathcal C_{it})=P(\epsilon_{it}\leq u\mid\mathcal C_{it})$. The factor-conditional quantile restriction $F_{it}(0\mid\mathcal C_{it})=\tau$ holds almost surely. Moreover, there exists a constant $\overline f<\infty$ such that, for every $u\in\mathbb R$, $|F_{it}(u\mid\mathcal C_{it})-F_{it}(0\mid\mathcal C_{it})|\leq\overline f|u|$, almost surely.
\item For some fixed $\eta>0$ and constants $K_{\xi},K_{WF}<\infty$,
$E\left(|\xi_{it,\tau}|^{4+\eta}\mid\mathcal C_{it}\right)\leq K_{\xi}^{4+\eta}$ almost surely. In addition, for every $u\in\mathbb R^{p+1}$, $v\in\mathbb R^{r_\tau^0}$ satisfying $\|u\|^2+\|v\|^2=1$, and every $x\geq0$, $P\left(|u^{\top}W_{it}+v^{\top}f_{t,\tau}^{0}|\geq K_{WF}x\right)\leq2\exp(-x^2/2)$.
\end{enumerate}
\end{assumption}

Assumption \ref{assum:ES-score} imposes standard conditions for the second-stage ES analysis. Part (i) extends the conditional quantile restriction to include the latent factor and imposes a Lipschitz condition on the conditional distribution. Together, these conditions ensure that the bias induced by first-stage quantile estimation is quadratic, reflecting Neyman orthogonality. Part (ii) imposes conventional moment and sub-Gaussian tail conditions on the oracle ES error, covariates, and factors. These conditions control the stochastic fluctuations in the second stage, while serial dependence is accommodated through Assumption \ref{assum:Regularity}(iv).

To separate the oracle second-stage error from the perturbation introduced by first-stage quantile estimation, for $\gamma\in\mathbb R^{p+1}$, define $R_{it}(\gamma)=Z_{it}\left(\alpha_{i,\tau}^{0}+\Sigma_i^{-1/2}\gamma\right)-Z_{it}(\alpha_{i,\tau}^{0})$ and write $\mu_{it}(\gamma)=E\{R_{it}(\gamma)\mid\mathcal C_{it}\}$, $R_{it}^{c}(\gamma)=R_{it}(\gamma)-\mu_{it}(\gamma)$. For $\Gamma=(\gamma_1,\ldots,\gamma_N)$, let
$$
\mathcal R^{c}(\Gamma)=\left(R_{it}^{c}(\gamma_i)\right)_{t\leq T,\,i\leq N},\ \mathcal G(\rho)=\left\{\Gamma:
\max_{i\leq N}\|\gamma_i\|\leq\rho\right\},
$$
and collect the oracle errors in $\Xi_\tau=(\xi_{it,\tau})_{t\leq T,\,i\leq N}$.
\begin{assumption}
\label{assum:factor-spectral}
The following conditions hold with constants independent of $(N,T,p)$.
\begin{enumerate}[label=(\roman*)]
\item If $r_\tau^0>0$, the true factors satisfy the exact normalization $T^{-1}F_\tau^{0\top}F_\tau^0=I_{r_\tau^0}$. There exist constants
$0<c_\lambda\leq C_\lambda<\infty$ and $K_\lambda<\infty$ such that $c_\lambda I_{r_\tau^0}\preceq N^{-1}\Lambda_\tau^{0\top}\Lambda_\tau^0\preceq C_\lambda I_{r_\tau^0}$, and $\max_{i\leq N}\|\lambda_{i,\tau}^0\|
\leq K_\lambda$.
\item Irrespective of whether $r_\tau^0$ is zero, there exist constants $\overline\rho,K_\Xi,K_R>0$ such that $\left\{E\|\Xi_\tau\|_{\mathrm{op}}^4\right\}^{1/4}\leq K_\Xi(\sqrt N+\sqrt T)$ and, for every $0<\rho\leq\overline\rho$,
$\left[E\{\sup_{\Gamma\in\mathcal G(\rho)}\|\mathcal R^c(\Gamma)\|_{\mathrm{op}}^4\}\right]^{1/4}\leq K_R\rho\left\{\sqrt{N(p+1)}+\sqrt T\right\}$.
\end{enumerate}
\end{assumption}
Assumption \ref{assum:factor-spectral}(i) is the standard strong-factor condition. It ensures that each factor has a nonnegligible cross-sectional contribution, while the bound on individual loadings prevents the factor strength from being driven by only a few units. Part (ii) controls the spectral size of the oracle and generated-response errors. It allows cross-sectional dependence but rules out an additional pervasive error component that could be mistaken for a latent factor.

We next impose a joint identification condition tailored to the joint global minimization problem. Write $W_i=(W_{i1},\ldots,W_{iT})^\top$. For $C=(c_1^\top,\ldots,c_N^\top)^\top\in\mathbb R^{N\times(p+1)}$, let $\mathcal W(C)=(W_1c_1,\ldots,W_Nc_N)$, and let $\mathbb L_q$ denote the collection of $T\times N$ matrices with rank at most $q$. Moreover, for every $F\in\mathbb R^{T\times r_\tau^0}$ satisfying $T^{-1}F^\top F=I_{r_\tau^0}$, define $d_\tau(F)=\|P_F-P_{F_\tau^0}\|_{\mathrm{op}}$ and
$$
\mathcal Q_\tau^0(C,F)=\inf_{\ell_i\in\mathbb R^{r_\tau^0},\,i\in[N]}\frac{1}{NT}\sum_{i=1}^N\left\|F_\tau^0\lambda_{i,\tau}^0-W_ic_i-F\ell_i\right\|^2.
$$
\begin{assumption}
\label{assum:local-identification}
There exist constants $\rho_{F,0}\in(0,1)$, $\kappa_X,\kappa_F,K_{\mathrm{dec}}>0$, independent of $(N,T,p)$, such that the following conditions hold.
\begin{enumerate}[label=(\roman*)]
\item If $r_\tau^0=0$, then, for every $i\in[N]$, $\lambda_{\min}\left(T^{-1}W_i^\top W_i\right)\geq\kappa_X$.

\item If $r_\tau^0>0$, the following two conditions hold. First, for each $q\in\{r_\tau^0,2r_\tau^0\}$, every $C\in\mathbb R^{N\times(p+1)}$, and every $L\in\mathbb L_q$, there exist $\widetilde C=(\widetilde c_1^\top,\ldots,\widetilde c_N^\top)^\top$ and $\widetilde L\in\mathbb L_q$ such that $\mathcal W(C)+L=\mathcal W(\widetilde C)+\widetilde L$ and $N^{-1}\sum_{i=1}^N\|\widetilde c_i\|^2+(NT)^{-1}\|\widetilde L\|^2\leq K_{\mathrm{dec}}(NT)^{-1}\|\mathcal W(C)+L\|^2$. Second, for every $C\in\mathbb R^{N\times(p+1)}$ and every $F\in\mathbb R^{T\times r_\tau^0}$ satisfying $T^{-1}F^\top F=I_{r_\tau^0}$, $\mathcal Q_\tau^0(C,F)\geq\kappa_F\min\left\{d_\tau(F)^2,\rho_{F,0}^2\right\}+\kappa_X\mathbbm{1}\left\{d_\tau(F)\leq\rho_{F,0}\right\}N^{-1}\sum_{i=1}^N\|c_i\|^2$.
\end{enumerate}
\end{assumption}
Assumption \ref{assum:local-identification} separates global factor-space identification from local coefficient identification. When $r_\tau^0=0$, it reduces to the standard nonsingularity condition for the regression design. When $r_\tau^0>0$, the stable-decomposition condition prevents the regression and low-rank components from becoming arbitrarily large while nearly cancelling each other in the fitted signal. The second condition imposes a global separation requirement on the factor space, while requiring quadratic identification of the regression coefficients only when the candidate factor space lies in a neighborhood of the true factor space. Thus, a factor space outside this neighborhood incurs a nonnegligible signal loss even after the regression coefficients and loadings have been optimally adjusted. Once the estimated factor space has been localized, the regression coefficients and the factor space are jointly identified at a quadratic rate.

Let $\iota_\tau=\mathbbm{1}\{r_\tau^0>0\}$, $\rho_{N,T,\delta}^{Q}=C_1\underline f^{-1}\sqrt{{s_{N,\delta/8}}/{T}}$, where $s_{N,\delta}$ and $C_1$ are defined in Proposition \ref{pro:quantile}. The following theorem establishes a unified high-probability finite-sample bound for the unit-specific ES coefficient estimator, covering both cases with and without latent factors.
\begin{theorem}
\label{theo:nonasymES}
Under Assumptions \ref{assum:Regularity}--\ref{assum:local-identification}, there exist constants $C,c>0$, depending only on the fixed constants in these assumptions and on $\tau$, such that the following holds for every fixed $i\in[N]$ and $\delta\in(0,1)$. If $T\geq C\left\{s_{N,\delta/8}^{\,2}\log^4(2T)+L_0^2\underline f^{-4}s_{N,\delta/8}\right\}$, $\rho_{N,T,\delta}^{Q}\leq\min(\overline\rho,1)$, and $\mathfrak R_{N,T,p,\delta}\leq c$, then the following bound holds with probability at least $1-\delta$,
\begin{equation}
\label{eq:beta-only-main-bound}
\tau\|\widehat\beta_{i,\tau}-\beta_{i,\tau}^0\|_{\Sigma_i}\leq C\left(1+\iota_\tau\delta^{-1/(4+\eta)}\right)\mathfrak R_{N,T,p,\delta},
\end{equation}
where $\mathfrak R_{N,T,p,\delta}=\delta^{-1/4}\left(1+\rho_{N,T,\delta}^{Q}\right)\left\{\sqrt{({p+1})/{T}}+\iota_\tau\left(N^{-1/2}+T^{-1/2}\right)\right\}+\overline f\left(\rho_{N,T,\delta}^{Q}\right)^2$ is the overall second-stage error scale.
\end{theorem}
Theorem \ref{theo:nonasymES} provides a finite-sample bound for a fixed unit-specific ES coefficient. The bound separates the oracle regression error from the two additional sources of uncertainty introduced by the proposed procedure. Within $\mathfrak R_{N,T,p,\delta}$, the term $\sqrt{(p+1)/T}$ is the oracle ES-regression error, the terms involving $\rho_{N,T,\delta}^{Q}$ capture the effect of first-stage quantile estimation, and $N^{-1/2}+T^{-1/2}$ captures factor-estimation uncertainty. Because the first-stage radius enters only through products with other estimation errors or through its square, its contribution is of second order, reflecting the Neyman orthogonality of the ES transformation. When $r_\tau^0=0$, the factor-estimation terms vanish, and the result reduces to the two-stage ES regression setting without latent factors.

\begin{corollary}
\label{cor:ES-beta-rate}
Under Assumptions \ref{assum:Regularity}--\ref{assum:local-identification}, the following
results hold for every fixed $i\in[N]$.
\begin{enumerate}[label=(\roman*)]
\item If $N,T\to\infty$ and $\{p+1+\log(2N)\}\log^2(2T)=o(\sqrt T)$,
then $\|\widehat\beta_{i,\tau}-\beta_{i,\tau}^0\|_{\Sigma_i}=o_p(1)$.
\item If $p$ is fixed and $N,T\to\infty$ with $N/T\to\kappa\in(0,\infty)$, then
$\sqrt T\|\widehat\beta_{i,\tau}-\beta_{i,\tau}^0\|_{\Sigma_i}=O_p(1)$.
\end{enumerate}
\end{corollary}
Corollary \ref{cor:ES-beta-rate} summarizes the large-sample behavior of the unit-specific ES coefficient estimator under two growth regimes. Part (i) establishes consistency while allowing $N$ and $T$ to grow at different rates and $p$ to increase subject to the stated condition. Part (ii) considers the stronger setting in which $p$ is fixed and $N\asymp T$, and shows that the estimator attains the conventional $\sqrt T$ rate. In both parts, the first-stage quantile estimation error enters only through higher-order remainder terms because of the orthogonality of the generated ES response.

\subsection{Finite-sample Gaussian Approximation}\label{sec:3.2}

We next establish a finite-sample Gaussian approximation for a fixed unit-specific ES coefficient. The approximation is formulated in terms of an oracle influence score that accounts for the cross-unit feedback induced by estimation of the latent factor space. We first introduce the population quantities needed to define this score.

Under the normalization $T^{-1}F_\tau^{0\top}F_\tau^0=I_{r_\tau^0}$, define the factor-residualized covariate $X_{jt}^{\circ}=X_{jt}-E(X_{jt}f_{t,\tau}^{0\top})f_{t,\tau}^0$. For $j,k\in[N]$, let
$a_{jk,\tau}^0=N\lambda_{j,\tau}^{0\top}(\Lambda_\tau^{0\top}\Lambda_\tau^0)^{-1}\lambda_{k,\tau}^0$, $\overline\xi_{jt,\tau}=\xi_{jt,\tau}-N^{-1}\sum_{\ell=1}^Na_{j\ell,\tau}^0\xi_{\ell t,\tau}$. Let $\mathcal J_\tau^0$ be the $N(p+1)\times N(p+1)$ block matrix whose $(j,k)$ block is $[\mathcal J_\tau^0]_{jk}=\mathbbm{1}\{j=k\}E(X_{jt}^{\circ}X_{jt}^{\circ\top})-N^{-1}{a_{jk,\tau}^0}E(X_{jt}^{\circ}X_{kt}^{\circ\top})$, and set $\mathcal G_\tau^0=(\mathcal J_\tau^0)^{-1}$, with $(i,j)$ block $\mathcal G_{ij,\tau}^0$. The oracle influence score for unit $i$ is $\psi_{it,\tau}^0=\sum_{j=1}^N\mathcal G_{ij,\tau}^0X_{jt}^{\circ}\overline\xi_{jt,\tau}$. When $r_\tau^0=0$, set $X_{jt}^{\circ}=X_{jt}$ and $a_{jk,\tau}^0=0$. The preceding definitions then reduce automatically to the usual oracle regression influence score without latent factors.

\begin{assumption}
\label{assum:Gaussian}
The following conditions hold with constants that do not depend on $(N,T,p)$.
\begin{enumerate}[label=(\roman*)]
\item The eigenvalues of $E(X_{jt}^{\circ}X_{jt}^{\circ\top})$ are bounded above and away from zero uniformly in $j$. Moreover, for constants $c_{\mathcal J},K_{\mathcal G}>0$, $\sigma_{\min}(\mathcal J_\tau^0)\geq c_{\mathcal J}$, $\max_{i\leq N}\sum_{j=1}^N\|\mathcal G_{ij,\tau}^0\|_{\mathrm{op}}\leq K_{\mathcal G}$.
\item Define $\mathcal V_{jk,t}=\left(f_{t,\tau}^{0\top},X_{jt}^{\top},X_{kt}^{\top},\overline\xi_{jt,\tau},\overline\xi_{kt,\tau}\right)^{\top}$, $\mathcal X_t=\sigma(f_{t,\tau}^0,X_{1t},\ldots,X_{Nt})$, and let $\mathscr X=\sigma(\mathcal X_t:t\in\mathbb Z)$. Uniformly in $N,j,k$, the processes $\{\mathcal V_{jk,t}:t\in\mathbb Z\}$ are strictly stationary and geometrically strong mixing, with mixing coefficients bounded by $c_{\mathcal V}e^{-a_{\mathcal V}h}$. For every fixed $i$, $\{\psi_{it,\tau}^0:t\in\mathbb Z\}$ satisfies the same conditions. Conditional on $\mathscr X$, the innovation processes $\{(\epsilon_{it},\xi_{it,\tau}):t\in\mathbb Z\}_{i=1}^N$ can be partitioned into mutually independent clusters of size at most $K_{\mathrm{cl}}$. Moreover, uniformly in $i,t$, $E(\xi_{it,\tau}\mid\mathscr X)=0$, $E\left(|\xi_{it,\tau}|^{4+\eta}\mid\mathcal X_t\right)\leq K_{\xi,G}^{4+\eta}\quad\text{a.s.}$
\item Define $\Omega_{i,\tau,T}=T^{-1}\sum_{t=1}^T\sum_{s=1}^T\operatorname{Cov}(\psi_{it,\tau}^0,\psi_{is,\tau}^0)$. For every fixed target unit $i$, there exist constants $0<c_\Omega\leq C_\Omega<\infty$ such that $c_\Omega I_{p+1}\preceq\Omega_{i,\tau,T}\preceq C_\Omega I_{p+1}$.
\end{enumerate}
\end{assumption}
Assumption \ref{assum:Gaussian} collects the additional conditions needed for the finite-sample Gaussian approximation. Part (i) ensures that the factor-residualized design and the population feedback system are well conditioned and bounds the propagation of estimation errors across units. Part (ii) provides the dependence and moment conditions required for the linear representation and the Berry--Esseen bound. The bounded-cluster condition permits within-cluster dependence while preserving the $N^{-1/2}$ order of the relevant cross-sectional averages. Part (iii) ensures that the finite-sample long-run variance is uniformly bounded and nondegenerate, making the standardization valid.
\begin{theorem}
\label{theo:Gaussian}
Under Assumptions \ref{assum:Regularity}--\ref{assum:Gaussian}, there exist constants $c,C>0$, depending only on the fixed constants in these assumptions and on $\tau$, such that, if ${\{p+1+\log(2NT)\}^{2}\log^4(2T)}/{T}\leq c$, $\mathfrak b_{N,T,p}\leq c$, then, for every fixed $i\in[N]$,
\begin{equation}
\label{eq:GA-main}
\sup_{\substack{u\in\mathbb S^p\\ z\in\mathbb R}}\left|P\left\{\frac{\tau\sqrt T\,u^{\top}(\widehat\beta_{i,\tau}-\beta_{i,\tau}^0)}{\sqrt{u^{\top}\Omega_{i,\tau,T}u}}\leq z\right\}-\Phi(z)\right|\leq C\left\{\frac{\log^2(2T)}{\sqrt T}+\mathfrak b_{N,T,p}^{2/3}\right\},
\end{equation}
where $\mathfrak b_{N,T,p}=\{{p+1+\log(2NT)}\}/{\sqrt T}+\iota_\tau\left[\sqrt{\{p+1+\log(2NT)\}/{N}}+{\sqrt T}/{N}\right]$ is the linearization error scale.
\end{theorem}
Theorem \ref{theo:Gaussian} provides a finite-sample Gaussian approximation for every linear projection of a fixed unit-specific ES coefficient. The first term in \eqref{eq:GA-main} is the Berry--Esseen error for the serially dependent oracle influence score. Within $\mathfrak b_{N,T,p}$, the first term collects the generated-response and sample-moment remainders, while the terms multiplied by $\iota_\tau$ arise from estimating the factor space. Thus, the factor-estimation contribution disappears when $r_\tau^0=0$. If $N\asymp T$, the approximation error converges to zero whenever $\{p+1+\log(2NT)\}^{2}\log^4(2T)/T\to0$.

\subsection{Factor Number Selection}\label{sec:3.3}
We finally study selection of the number of ES factors. The preceding results treat $r_\tau^0$ as given, whereas here it is fixed but unknown.

To express the residual criterion compactly, let $\widehat{\mathcal Z}_\tau^\ast$ denote the $T\times N$ matrix whose $i$th column is $Z_i^\ast(\widehat\alpha_{i,\tau})$. For $B=(b_1^\top,\ldots,b_N^\top)^\top\in\mathbb R^{N\times(p+1)}$, let $\mathcal X(B)$ denote the $T\times N$ matrix whose $i$th column is $X_i b_i$. Finally, let $\mathbb L_r$ be the collection of $T\times N$ matrices with rank at most $r$. The minimized residual criterion underlying \eqref{eq:IC} can then be written as 
$$
\widehat V_\tau(r)=\inf_{\substack{B\in\mathbb R^{N\times(p+1)}\\L\in\mathbb L_r}}\frac1{NT}\left\|\widehat{\mathcal Z}_\tau^\ast-\mathcal X(B)-L\right\|^2.
$$

Write
$L_\tau^0=F_\tau^0\Lambda_\tau^{0\top}$, $\mathcal S_\tau^0=\mathcal X(B_\tau^0)+L_\tau^0$ for the true common component and the true noiseless ES signal, respectively. For each candidate rank $r$, define
$$
\Delta_{\tau,NT}(r)=\inf_{\substack{B\in\mathbb R^{N\times(p+1)}\\L\in\mathbb L_r}}\frac1{NT}\left\|\mathcal S_\tau^0-\mathcal X(B)-L\right\|^2.
$$
Thus, $\Delta_{\tau,NT}(r)$ measures the signal lost by restricting the common component to have rank at most $r$, after the unit-specific regression coefficients have been optimally adjusted. Since $L_\tau^0\in\mathbb L_{r_\tau^0}$, $\Delta_{\tau,NT}(r)=0$ for every $r\geq r_\tau^0$.

\begin{assumption}
\label{assum:factor-selection}
The following conditions hold.
\begin{enumerate}[label=(\roman*)]
\item The upper bound $R_{\max}$ is fixed and satisfies $r_\tau^0\leq R_{\max}$. If $r_\tau^0>0$, there exists a constant
$\Delta_\tau>0$ such that $P\left\{\min_{0\leq r<r_\tau^0}\Delta_{\tau,NT}(r)\geq\Delta_\tau\right\}\rightarrow1$.
\item For some $0<\sigma_\tau^2<\infty$, $(NT)^{-1}\|\tau^{-1}\Xi_\tau\|^2\xrightarrow{p}\sigma_\tau^2$.

\item There exists a constant $K_{\mathrm{dec}}<\infty$ such that, on events whose probabilities approach one, the following property holds. For every $C=(c_1^\top,\ldots,c_N^\top)^\top\in\mathbb R^{N\times(p+1)}$ and $L\in\mathbb L_{2R_{\max}}$, there exist $\widetilde C=(\widetilde c_1^\top,\ldots,\widetilde c_N^\top)^\top$ and $\widetilde L\in\mathbb L_{2R_{\max}}$ such that $\mathcal X(C)+L=\mathcal X(\widetilde C)+\widetilde L$ and
$N^{-1}\sum_{i=1}^N\|\widetilde c_i\|_{\Sigma_i}^2+({NT})^{-1}\|\widetilde L\|^2\leq (NT)^{-1}K_{\mathrm{dec}}\|\mathcal X(C)+L\|^2$.
\end{enumerate}
\end{assumption}
Assumption \ref{assum:factor-selection} provides the additional conditions needed for consistent factor-number selection. Part (i) ensures that using too few factors produces a nonvanishing increase in the residual loss. Part (ii) keeps the residual criterion nondegenerate. Part (iii) ensures that the regression and factor components remain sufficiently distinguishable. Together, these conditions allow the information criterion to separate the true factor number from both underfitted and overfitted alternatives.
\begin{theorem}
\label{thm:selection}
Under Assumptions \ref{assum:Regularity}--\ref{assum:factor-spectral} and \ref{assum:factor-selection}, suppose that $N,T\rightarrow\infty$, $\varrho_{N,T,p}\to0$, and the sample-size condition in Proposition \ref{pro:quantile} holds eventually for every fixed confidence level. If $q(N,T)>0$, $q(N,T)\to0$, and ${\omega_{N,T,p}}/{q(N,T)}\to0$, then
$P(\widehat r_\tau=r_\tau^0)\rightarrow1$, where $\omega_{N,T,p}=(p+1)/{T}+N^{-1}+\varrho_{N,T,p}^2
\left\{({p+1})/{T}+N^{-1}\right\}+\varrho_{N,T,p}^4$ and $\varrho_{N,T,p}=\sqrt{\{{p+1+\log(2N)}\}/{T}}$. 
\end{theorem}
Theorem \ref{thm:selection} shows that the proposed information criterion consistently selects the number of ES factors. Underfitting produces the nonvanishing signal loss in Assumption \ref{assum:factor-selection}(i), whereas the reduction in the residual criterion produced by additional factors is at most $O_p(\omega_{N,T,p})$. The penalty must therefore converge to zero while dominating $\omega_{N,T,p}$. In particular, if $p$ is fixed and $N/T\rightarrow\kappa\in(0,\infty)$, then $\omega_{N,T,p}=O(T^{-1})$, and the penalty specified in Subsection \ref{subsec:determine}, $q(N,T)=(N+T)/(NT)\log\left((NT)/({N+T})\right)$, satisfies both requirements.

\section{Simulation Study}\label{sec:Simulation Study}
	
	We conduct a Monte Carlo study to assess the finite-sample performance of the proposed ESFM under a variety of economically and statistically relevant environments. The design is organized to isolate (i) heavy-tailed innovations, (ii) cross-sectional heterogeneity in slope coefficients, and (iii) latent factors that primarily affect downside risk through time-varying tail dispersion. Throughout, we compare ESFM to a baseline ES regression that ignores latent factors, with particular emphasis on whether incorporating tail factors improves the accuracy of slope estimation and factor-space recovery.
	\subsection{Simulation Design}
	
	For $i=1,\ldots,N$ and $t=1,\ldots,T$, we generate a covariate vector $X_{it}\in\mathbb{R}^{p+1}$ that includes an intercept and $p=3$ observed regressors. We consider tail probability levels $\tau\in\{0.01,0.05,0.10\}$. Panel dimensions vary over $N\in\{100,200,300\}$ and $T\in\{100,200,300\}$. For each $(N,T,\tau)$ configuration and each scenario described below, we run $MC=100$ replications with independent random seeds. In estimation, we set the number of latent ES factors to match the data-generating value ($r_\tau^0=2$) when reporting the main results. For details on the simulation experiment, see the appendix.
	
	For each simulated dataset, we estimate: (i) an ES regression that ignores latent factors (equivalently, ESFM with $r=0$), and (ii) the proposed ESFM with $r=r_\tau^0$ factors. Performance is summarized by (a) mean squared errors of the slope coefficients $\beta_i$ (averaged over $i$), (b) the root mean squared error of the estimated factor space, (c) ES estimators (included in the appendix), and (d) number of factors (included in the appendix). 
	
	We consider seven designs that vary the strength of tail-factor dependence, slope heterogeneity, covariate endogeneity, volatility-driven comovement, jump risk, and tail asymmetry. The complete data-generating processes are described in the appendix.
	
\subsection{Simulation Results}\label{Simulation Results}
	
Table~\ref{tab:Estimation error of beta} reports the mean squared errors (MSEs) of the estimated covariate coefficients $\beta_i$ under the expected shortfall regression (ESR) and the proposed ESFM. Results are reported for seven data-generating scenarios, three cross-sectional dimensions ($N=100,200,300$), three time dimensions ($T=100,200,300$), and three tail probability levels $\tau\in\{0.10,0.05,0.01\}$.
	
\begin{table}
	\footnotesize
	\TABLE
	{Estimation error of $\beta_\tau$ in expected shortfall (ES) models.
		\label{tab:Estimation error of beta}}
	{\footnotesize
		\setlength{\tabcolsep}{2.5pt}
		\renewcommand{\arraystretch}{0.85}
		\begin{tabular*}{\textwidth}
			{@{\extracolsep{\fill}}lccccccccccccccccccc}
			\toprule
			& \multicolumn{2}{c}{Scenario 1}
			& \multicolumn{2}{c}{Scenario 2}
			& \multicolumn{2}{c}{Scenario 3}
			& \multicolumn{2}{c}{Scenario 4}
			& \multicolumn{2}{c}{Scenario 5}
			& \multicolumn{2}{c}{Scenario 6}
			& \multicolumn{2}{c}{Scenario 7} \\
			& ESR & ESFM
			& ESR & ESFM
			& ESR & ESFM
			& ESR & ESFM
			& ESR & ESFM
			& ESR & ESFM
			& ESR & ESFM \\
			\midrule
			\multicolumn{15}{c}{Panel A: $\tau = 0.10$} \\
			\midrule
			\multicolumn{15}{c}{$N$ = 100} \\
			\midrule
			$T=100$ & 0.4130 & 0.3683 & 0.4179 & 0.3722 & 0.4143 & 0.3740 & 0.4146 & 0.3685 & 0.4146 & 0.3706 & 0.4166 & 0.3705 & 0.4147 & 0.3716 \\
			$T=200$ & 0.3559 & 0.3232 & 0.3536 & 0.3221 & 0.3557 & 0.3243 & 0.3545 & 0.3228 & 0.3525 & 0.3223 & 0.3523 & 0.3194 & 0.3556 & 0.3226 \\
			$T=300$ & 0.3306 & 0.3057 & 0.3311 & 0.3065 & 0.3286 & 0.3032 & 0.3296 & 0.3098 & 0.3297 & 0.3028 & 0.3297 & 0.3052 & 0.3285 & 0.3028 \\
			\midrule
			\multicolumn{15}{c}{$N$ = 200} \\
			\midrule
			$T=100$ & 0.4179 & 0.3662 & 0.4141 & 0.3643 & 0.4143 & 0.3631 & 0.4175 & 0.3664 & 0.4142 & 0.3631 & 0.4115 & 0.3616 & 0.4148 & 0.3639 \\
			$T=200$ & 0.3551 & 0.3223 & 0.3532 & 0.3193 & 0.3530 & 0.3191 & 0.3520 & 0.3185 & 0.3554 & 0.3206 & 0.3529 & 0.3184 & 0.3547 & 0.3211 \\
			$T=300$ & 0.3295 & 0.3060 & 0.3299 & 0.3071 & 0.3273 & 0.3030 & 0.3307 & 0.3056 & 0.3294 & 0.3064 & 0.3284 & 0.3067 & 0.3296 & 0.3075 \\
			\midrule
			\multicolumn{15}{c}{$N$ = 300} \\
			\midrule
			$T=100$ & 0.4148 & 0.3640 & 0.4173 & 0.3650 & 0.4151 & 0.3645 & 0.4170 & 0.3626 & 0.4170 & 0.3634 & 0.4162 & 0.3635 & 0.4174 & 0.3641 \\
			$T=200$ & 0.3533 & 0.3216 & 0.3530 & 0.3202 & 0.3529 & 0.3191 & 0.3550 & 0.3232 & 0.3529 & 0.3190 & 0.3537 & 0.3206 & 0.3541 & 0.3212 \\
			$T=300$ & 0.3292 & 0.3144 & 0.3297 & 0.3099 & 0.3288 & 0.3123 & 0.3289 & 0.3089 & 0.3296 & 0.3128 & 0.3287 & 0.3111 & 0.3285 & 0.3122 \\
			\midrule
			\multicolumn{15}{c}{Panel B: $\tau = 0.05$} \\
			\midrule
			\multicolumn{15}{c}{$N$ = 100} \\
			\midrule
			$T=100$ & 0.4809 & 0.3603 & 0.4826 & 0.3622 & 0.4840 & 0.3632 & 0.4855 & 0.3604 & 0.4813 & 0.3612 & 0.4836 & 0.3596 & 0.4826 & 0.3621 \\
			$T=200$ & 0.4160 & 0.3277 & 0.4131 & 0.3227 & 0.4133 & 0.3249 & 0.4129 & 0.3242 & 0.4090 & 0.3228 & 0.4134 & 0.3265 & 0.4163 & 0.3276 \\
			$T=300$ & 0.3806 & 0.3087 & 0.3803 & 0.3088 & 0.3801 & 0.3087 & 0.3779 & 0.3061 & 0.3816 & 0.3068 & 0.3810 & 0.3084 & 0.3805 & 0.3071 \\
			\midrule
			\multicolumn{15}{c}{$N$ = 200} \\
			\midrule
			$T=100$ & 0.4828 & 0.3551 & 0.4817 & 0.3553 & 0.4855 & 0.3550 & 0.4865 & 0.3585 & 0.4784 & 0.3515 & 0.4827 & 0.3536 & 0.4848 & 0.3568 \\
			$T=200$ & 0.4149 & 0.3227 & 0.4106 & 0.3202 & 0.4112 & 0.3193 & 0.4130 & 0.3197 & 0.4151 & 0.3212 & 0.4137 & 0.3207 & 0.4134 & 0.3208 \\
			$T=300$ & 0.3791 & 0.3053 & 0.3789 & 0.3050 & 0.3768 & 0.3028 & 0.3805 & 0.3058 & 0.3797 & 0.3049 & 0.3788 & 0.3050 & 0.3800 & 0.3057 \\
			\midrule
			\multicolumn{15}{c}{$N$ = 300} \\
			\midrule
			$T=100$ & 0.4813 & 0.3536 & 0.4833 & 0.3528 & 0.4806 & 0.3517 & 0.4871 & 0.3523 & 0.4844 & 0.3526 & 0.4833 & 0.3535 & 0.4838 & 0.3516 \\
			$T=200$ & 0.4137 & 0.3203 & 0.4127 & 0.3189 & 0.4133 & 0.3203 & 0.4152 & 0.3207 & 0.4136 & 0.3190 & 0.4126 & 0.3188 & 0.4134 & 0.3200 \\
			$T=300$ & 0.3786 & 0.3050 & 0.3797 & 0.3044 & 0.3795 & 0.3041 & 0.3791 & 0.3037 & 0.3795 & 0.3046 & 0.3793 & 0.3041 & 0.3786 & 0.3050 \\
			\midrule
			\multicolumn{15}{c}{Panel C: $\tau = 0.01$} \\
			\midrule
			\multicolumn{15}{c}{$N$ = 100} \\
			\midrule
			$T=100$ & 0.4048 & 0.3789 & 0.4010 & 0.3747 & 0.4041 & 0.3797 & 0.4054 & 0.3800 & 0.4094 & 0.3825 & 0.4045 & 0.3792 & 0.4070 & 0.3804 \\
			$T=200$ & 0.4895 & 0.3404 & 0.5002 & 0.3399 & 0.4952 & 0.3389 & 0.4958 & 0.3418 & 0.4904 & 0.3351 & 0.4954 & 0.3400 & 0.4994 & 0.3391 \\
			$T=300$ & 0.5224 & 0.3134 & 0.5212 & 0.3130 & 0.5210 & 0.3124 & 0.5211 & 0.3124 & 0.5359 & 0.3128 & 0.5215 & 0.3106 & 0.5212 & 0.3128 \\
			\midrule
			\multicolumn{15}{c}{$N$ = 200} \\
			\midrule
			$T=100$ & 0.4061 & 0.3803 & 0.4088 & 0.3809 & 0.4078 & 0.3808 & 0.4080 & 0.3819 & 0.4049 & 0.3774 & 0.4066 & 0.3798 & 0.4104 & 0.3820 \\
			$T=200$ & 0.4932 & 0.3385 & 0.4986 & 0.3365 & 0.4880 & 0.3370 & 0.4982 & 0.3347 & 0.4985 & 0.3351 & 0.5002 & 0.3386 & 0.4976 & 0.3387 \\
			$T=300$ & 0.5257 & 0.3100 & 0.5113 & 0.3108 & 0.5175 & 0.3123 & 0.5234 & 0.3097 & 0.5236 & 0.3117 & 0.5202 & 0.3115 & 0.5189 & 0.3114 \\
			\midrule
			\multicolumn{15}{c}{$N$ = 300} \\
			\midrule
			$T=100$ & 0.4083 & 0.3797 & 0.4074 & 0.3793 & 0.4059 & 0.3786 & 0.4059 & 0.3785 & 0.4099 & 0.3807 & 0.4050 & 0.3772 & 0.4048 & 0.3779 \\
			$T=200$ & 0.4998 & 0.3346 & 0.4957 & 0.3365 & 0.4960 & 0.3364 & 0.4989 & 0.3343 & 0.4956 & 0.3352 & 0.4974 & 0.3351 & 0.4944 & 0.3371 \\
			$T=300$ & 0.5236 & 0.3122 & 0.5186 & 0.3091 & 0.5231 & 0.3111 & 0.5222 & 0.3102 & 0.5180 & 0.3117 & 0.5228 & 0.3104 & 0.5188 & 0.3102 \\
			\bottomrule
		\end{tabular*}%
	}
	{This table reports the mean squared error (MSE) of the estimated covariate coefficients $\beta_\tau$ across Monte Carlo replications. Results are presented for seven data-generating scenarios, comparing the benchmark expected shortfall regression (ESR) and the proposed expected shortfall factor model (ESFM).}
\end{table}%

ESFM has a lower MSE than ESR in every design reported in Table~\ref{tab:Estimation error of beta}. The largest differences
occur at $\tau=0.01$ when $T=200$ or $300$: the MSE of ESFM declines as $T$ increases, whereas that of ESR rises. ESFM also
retains its advantage in the designs with stronger tail dependence, endogenous covariates, common volatility, and jumps.

Table~\ref{tab:factor_space_error} reports the root mean squared error (RMSE) of the estimated ES factor space. For each Monte Carlo replication, factor-space error is measured by the squared Frobenius distance between the estimated and true projection matrices. The reported RMSE is obtained by
averaging these squared distances across replications and taking the square root. Here, $P_F=F(F^\top F)^{-1}F^\top$ denotes the projection matrix onto the factor space. Results are reported for all seven scenarios and the same set of sample sizes and tail probabilities as in Table~\ref{tab:Estimation error of beta}.

\begin{table}
\footnotesize
\TABLE
{Root mean squared error of the estimated factor space.
\label{tab:factor_space_error}}
{
\setlength{\tabcolsep}{2pt}
\renewcommand{\arraystretch}{0.8}
\begin{tabular*}{0.95\textwidth}{@{\extracolsep{\fill}}lccccccc}
\toprule
& Scenario 1 & Scenario 2 & Scenario 3 & Scenario 4
& Scenario 5 & Scenario 6 & Scenario 7 \\
\midrule
\multicolumn{8}{c}{Panel A: $\tau = 0.10$} \\
\midrule
\multicolumn{8}{c}{$N=100$} \\
\midrule
$T=100$ & 0.4402 & 0.4403 & 0.4400 & 0.4403 & 0.4395 & 0.4407 & 0.4405 \\
$T=200$ & 0.3705 & 0.3704 & 0.3699 & 0.3704 & 0.3696 & 0.3708 & 0.3709 \\
$T=300$ & 0.3339 & 0.3336 & 0.3338 & 0.3320 & 0.3342 & 0.3335 & 0.3344 \\
\midrule
\multicolumn{8}{c}{$N=200$} \\
\midrule
$T=100$ & 0.4383 & 0.4389 & 0.4384 & 0.4383 & 0.4395 & 0.4393 & 0.4390 \\
$T=200$ & 0.3655 & 0.3669 & 0.3672 & 0.3667 & 0.3674 & 0.3678 & 0.3666 \\
$T=300$ & 0.3276 & 0.3273 & 0.3288 & 0.3289 & 0.3270 & 0.3274 & 0.3263 \\
\midrule
\multicolumn{8}{c}{$N=300$} \\
\midrule
$T=100$ & 0.4359 & 0.4361 & 0.4354 & 0.4379 & 0.4361 & 0.4363 & 0.4367 \\
$T=200$ & 0.3615 & 0.3630 & 0.3650 & 0.3606 & 0.3641 & 0.3632 & 0.3632 \\
$T=300$ & 0.3132 & 0.3184 & 0.3154 & 0.3191 & 0.3156 & 0.3159 & 0.3165 \\
\midrule
\multicolumn{8}{c}{Panel B: $\tau = 0.05$} \\
\midrule
\multicolumn{8}{c}{$N=100$} \\
\midrule
$T=100$ & 0.1962 & 0.1960 & 0.1962 & 0.1962 & 0.1960 & 0.1967 & 0.1967 \\
$T=200$ & 0.1393 & 0.1396 & 0.1393 & 0.1396 & 0.1393 & 0.1393 & 0.1393 \\
$T=300$ & 0.1140 & 0.1140 & 0.1140 & 0.1140 & 0.1140 & 0.1140 & 0.1145 \\
\midrule
\multicolumn{8}{c}{$N=200$} \\
\midrule
$T=100$ & 0.1960 & 0.1957 & 0.1960 & 0.1952 & 0.1967 & 0.1962 & 0.1957 \\
$T=200$ & 0.1389 & 0.1393 & 0.1393 & 0.1389 & 0.1393 & 0.1389 & 0.1393 \\
$T=300$ & 0.1136 & 0.1136 & 0.1140 & 0.1136 & 0.1140 & 0.1136 & 0.1136 \\
\midrule
\multicolumn{8}{c}{$N=300$} \\
\midrule
$T=100$ & 0.1954 & 0.1960 & 0.1952 & 0.1962 & 0.1960 & 0.1954 & 0.1957 \\
$T=200$ & 0.1389 & 0.1393 & 0.1393 & 0.1389 & 0.1389 & 0.1389 & 0.1389 \\
$T=300$ & 0.1131 & 0.1131 & 0.1136 & 0.1136 & 0.1131 & 0.1131 & 0.1131 \\
\midrule
\multicolumn{8}{c}{Panel C: $\tau = 0.01$} \\
\midrule
\multicolumn{8}{c}{$N=100$} \\
\midrule
$T=100$ & 0.1934 & 0.1931 & 0.1939 & 0.1934 & 0.1934 & 0.1936 & 0.1936 \\
$T=200$ & 0.1371 & 0.1364 & 0.1360 & 0.1367 & 0.1356 & 0.1364 & 0.1360 \\
$T=300$ & 0.1118 & 0.1118 & 0.1122 & 0.1122 & 0.1114 & 0.1114 & 0.1118 \\
\midrule
\multicolumn{8}{c}{$N=200$} \\
\midrule
$T=100$ & 0.1936 & 0.1944 & 0.1939 & 0.1929 & 0.1947 & 0.1939 & 0.1931 \\
$T=200$ & 0.1375 & 0.1382 & 0.1375 & 0.1375 & 0.1382 & 0.1371 & 0.1371 \\
$T=300$ & 0.1131 & 0.1136 & 0.1127 & 0.1136 & 0.1127 & 0.1131 & 0.1131 \\
\midrule
\multicolumn{8}{c}{$N=300$} \\
\midrule
$T=100$ & 0.1934 & 0.1939 & 0.1929 & 0.1934 & 0.1939 & 0.1939 & 0.1942 \\
$T=200$ & 0.1382 & 0.1382 & 0.1378 & 0.1382 & 0.1386 & 0.1382 & 0.1382 \\
$T=300$ & 0.1136 & 0.1136 & 0.1136 & 0.1136 & 0.1131 & 0.1127 & 0.1131 \\
\bottomrule
\end{tabular*}%
}
{This table reports the root mean squared error (RMSE) of the
estimated factor space associated with the expected shortfall factors.
Smaller values indicate more accurate recovery of the true factor space.}
\end{table}%

The results indicate that ESFM recovers the latent factor space with good accuracy. For a given $N$, the RMSE declines steadily as $T$ increases, indicating more accurate factor-space recovery. For instance, when $N=100$ and $\tau=0.10$, the RMSE decreases from about 0.44 at $T=100$ to 0.33 at $T=300$, with similar patterns for larger $N$. Increasing $N$ also generally improves accuracy at $\tau=0.10$, whereas the differences across $N$ are small at $\tau=0.05$ and $0.01$. The RMSEs are lower at $\tau=0.05$ and $0.01$ than at $\tau=0.10$, with little difference between the two more extreme tail levels. Error levels remain similar across scenarios. Overall, these findings show that ESFM reliably recovers the underlying tail-risk factor space in finite samples.
    
\section{Empirical Application}\label{sec:Empirical Application}
\subsection{Data}\label{subsec:Data}
    
Our empirical analysis is conducted using a large cross-section of Chinese equities. We focus on the constituents of the CSI 300 index over the period from January 2000 to December 2023.\footnote{The sample spans several major market episodes, including the 2015 Chinese stock market crash, the 2018 trade tensions, and the COVID-19 episode.} The resulting panel consists of approximately 250 equities with relatively stable time-series coverage. Daily closing prices are obtained from the Wind database.
	
For asset pricing analysis, we augment the dataset with standard risk factors. Specifically, we employ the Fama–French five-factor returns—including the market, size, value, profitability, and investment factors—along with the risk-free rate, obtained from publicly available data sources. Both equal-weighted (EW) and value-weighted (VW) Fama–French five-factor returns are considered. These factors enter our ESFM framework as observable covariates, capturing systematic components of asset returns.
	
	
\subsection{Estimation Results}\label{subsec:Estimation Results}

We compare ESFM with two benchmarks that separate loss severity from more familiar sources of common return variation. The mean factor model is estimated as an interactive-effects specification \citep{bai2009panel}. Its asset-pricing interpretation follows the approximate-factor tradition, in which latent common components of returns are used to explain expected returns \citep{connor1986performance,connor1988risk}; \citet{lettau2020factors} further show that economically useful latent factors should fit both return covariation and the cross-section of average returns. The quantile factor model is estimated following \citet{ando2020quantile} and captures common movements in conditional tail thresholds. Its role as an asset-pricing benchmark is motivated most directly by \citet{barunik2026common}, who document a distinct premium for exposure to a common lower-tail factor in idiosyncratic return quantiles. We condition all three models on the same Fama--French five factors and use identical samples and information sets. The comparison therefore asks whether common variation in within-tail loss severity contains pricing information beyond common variation in average returns or tail thresholds.


For each model, we report the cross-sectional averages of the coefficients on the observable factors and the estimated latent
factor paths. The main text uses value-weighted Fama--French factors and reports the factor paths at $\tau=0.10$; equal-weighted results and results for other tail levels are reported in the appendix.

	
\begin{figure}
\FIGURE
{\includegraphics[scale=0.45]{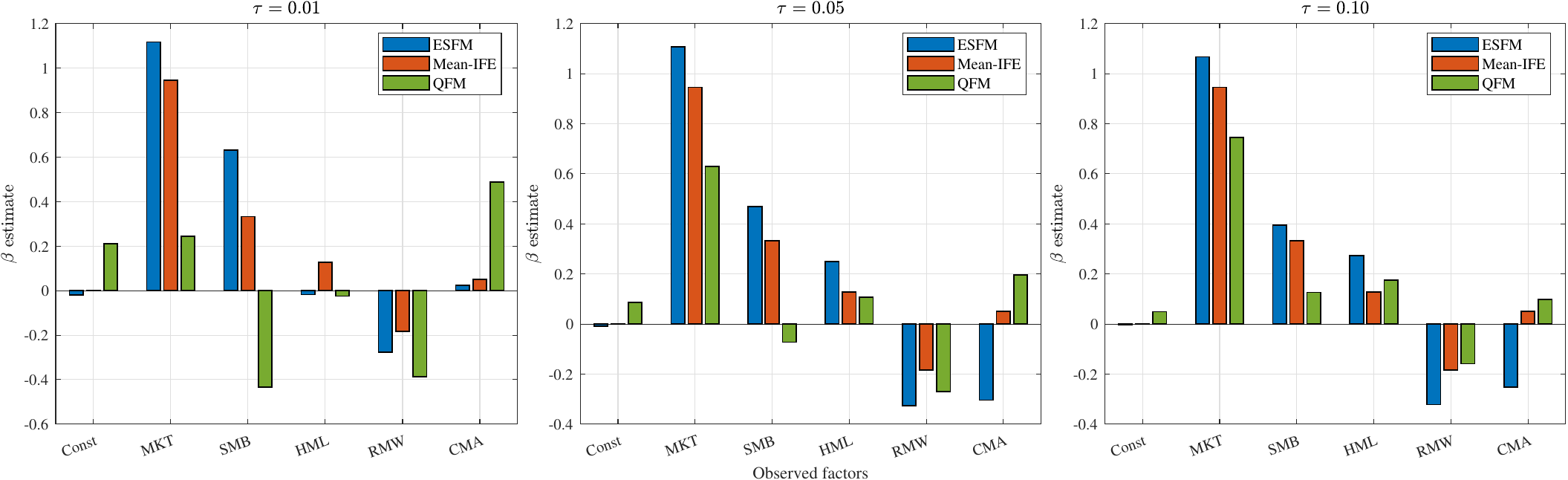}}
{Estimated coefficients ($\beta$) on observable risk factors. \label{fig:beta_VW}}
{This figure reports the estimated coefficients on observable factors (MKT, SMB, HML, RMW, CMA, and a constant) from three models: ESFM, the mean factor model (Mean-IFE), and the quantile factor model (QFM). The estimates are obtained using VW observable factor returns, and results are shown for three tail levels ($\tau = 0.01, 0.05, 0.10$). Each panel corresponds to a different $\tau$, and bars represent the average coefficients across assets.}
\end{figure}

Figure~\ref{fig:beta_VW} reports the cross-sectional averages of the estimated coefficients on the observable factors. The coefficient profiles differ across the mean, quantile, and ES specifications, particularly for the market factor. Because the three models characterize different conditional functionals, the magnitudes are not directly comparable as measures of model fit or explanatory importance. We therefore use the figure to summarize how the observable component varies across specifications, rather than to rank their performance.
	
	\begin{figure}
    \FIGURE
{\includegraphics[scale=0.59]{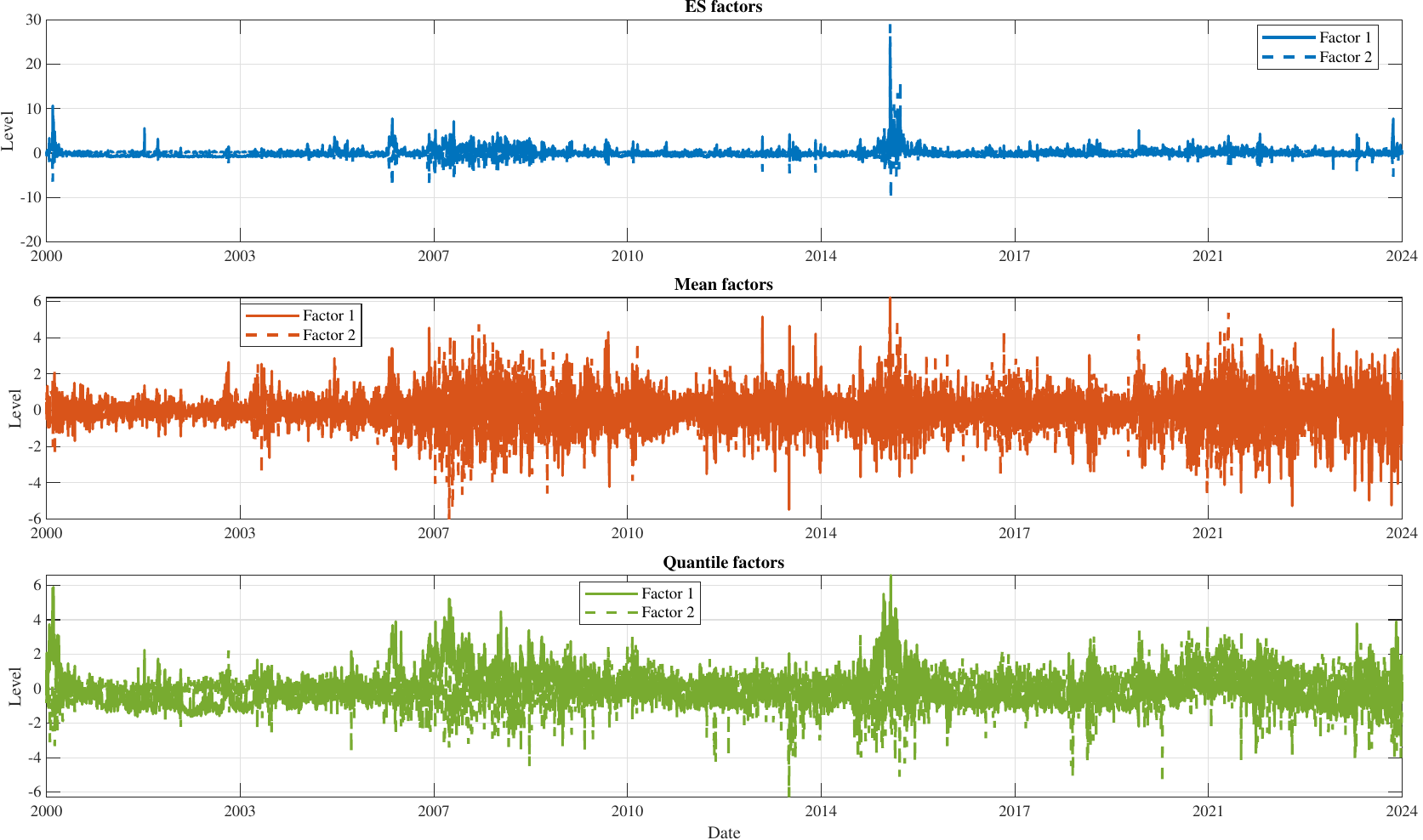}}
{Estimated factor paths across models ($r=2,\tau=0.10$). \label{fig:factor_VW_0p10}}
{This figure plots the estimated latent factor paths from three models: ESFM (top panel), the mean factor model (middle panel), and the quantile factor model (bottom panel). Each model extracts two factors ($r = 2$). The covariates in these models include the value-weighted Fama--French five factors. The quantile and ES factors are constructed using tail information at level $\tau = 0.10$.}
\end{figure}%
    
Figure~\ref{fig:factor_VW_0p10} plots the estimated factor paths for ESFM and the mean and quantile benchmarks. The mean factors vary comparatively smoothly, whereas the quantile factors exhibit more short-run fluctuations. The ESFM factors display their largest movements during market stress, particularly around the 2015 market crash. These differences reflect their respective targets: average returns, conditional tail thresholds, and the severity of returns below those thresholds.
	
Finally, following the quantile factor literature (e.g., \citealt{ando2020quantile}), we examine the correlations across factors estimated from different models.\footnote{The first and second generalized correlations are the ordered canonical correlations between the estimated factor spaces. They lie between zero and one, with larger values indicating greater overlap between the two
spaces.} The purpose of this exercise is to assess whether these models capture similar underlying sources of systematic risk or uncover distinct dimensions. This comparison provides a structural benchmark for the subsequent asset pricing analysis. The results are presented in Table \ref{tab:GC}.
	
\begin{table}
	\TABLE
	{Generalized correlations of estimated factors across models.
		\label{tab:GC}}
	{
		\begin{tabular}{lcccccccc}
			\toprule
			& \multicolumn{2}{c}{$\tau=0.01$}
			& & \multicolumn{2}{c}{$\tau=0.05$}
			& & \multicolumn{2}{c}{$\tau=0.10$} \\
			\cmidrule{2-3}\cmidrule{5-6}\cmidrule{8-9}
			& 1.GC & 2.GC
			& & 1.GC & 2.GC
			& & 1.GC & 2.GC \\
			\midrule
			\multicolumn{9}{c}{Panel A: based on the value--weighted Fama–French five factors} \\
			\midrule
			ESFM vs Mean & 0.1903 & 0.0632 & & 0.3026 & 0.0802 & & 0.3670 & 0.0819 \\
			ESFM vs QFM  & 0.1991 & 0.0025 & & 0.4756 & 0.006  & & 0.6187 & 0.0401 \\
			Mean vs QFM  & 0.1023 & 0.0401 & & 0.5918 & 0.0443 & & 0.6853 & 0.0377 \\
			\midrule
			\multicolumn{9}{c}{Panel B: based on the equal-weighted Fama–French five factors} \\
			\midrule
			ESFM vs Mean & 0.2863 & 0.0342 & & 0.3781 & 0.1084 & & 0.4229 & 0.1575 \\
			ESFM vs QFM  & 0.1904 & 0.0033 & & 0.4417 & 0.0029 & & 0.5462 & 0.0338 \\
			Mean vs QFM  & 0.1304 & 0.0391 & & 0.6277 & 0.1862 & & 0.7545 & 0.2471 \\
			\bottomrule
		\end{tabular}%
	}
	{This table reports the generalized correlations (GC) between estimated latent factors across the ESFM, mean factor model (Mean), and quantile factor model (QFM). Results are presented for three tail levels ($\tau = 0.01, 0.05, 0.10$).}
\end{table}%

Lower generalized correlation indicates weaker overlap in factor space and hence more model-specific information, whereas higher generalized correlation suggests that the corresponding factors reflect more similar underlying risk components. In our results, the ES factors are generally the least correlated with the other factors, especially at higher tail levels, implying that they contain additional information not captured by mean- or quantile-based specifications. By contrast, the mean and quantile factors are the most strongly related, particularly when $\tau$ is relatively high. This pattern changes when $\tau$ is very small, suggesting that the dependence structure across models becomes materially different in the far lower tail.
	
\subsection{Asset Pricing}\label{subsec:Asset Pricing}

We next ask whether exposures to the common components extracted by ESFM, the mean factor model, and the quantile factor model are priced in the cross-section. At the end of each month, we estimate factor exposures using the preceding 60 months, sort stocks into five portfolios, and hold the portfolios for one month. Portfolio returns are equal-weighted, and the high-minus-low spread measures the return associated with moving from the lowest to the highest exposure. We apply the identical procedure to all three models. This comparison is economically demanding: latent mean factors have long been used to explain expected returns \citep{connor1986performance,connor1988risk,lettau2020factors}, while \citet{barunik2026common} show that common lower-tail quantile exposure itself commands a distinct premium.\footnote{Related evidence from high-frequency panels shows that latent market-friction factors can also carry pricing information \citep{kong2026staleness}.} ESFM can therefore add pricing information only if within-tail loss severity matters beyond the common movements in average returns and tail thresholds captured by these benchmarks. We report results for $\tau\in\{0.10,0.20,0.30\}$, which balance the economic relevance of adverse states against the sampling instability of very low quantiles.
    
\subsubsection{Excess Returns}\label{subsubsec:Excess Returns}
To examine whether exposures to the estimated factors are priced, each month we sort stocks into five portfolios using exposures estimated over the preceding 60 months. The portfolios are equally weighted and held for one month, and the H--L return is the return difference between the highest- and lowest-exposure portfolios. We report results for $\tau\in\{0.10,0.20,0.30\}$, centered on $\tau=0.20$, to capture economically relevant downside states while retaining sufficient tail observations for estimation.\footnote{\citet{giglio2016systemic} and \citet{delikouras2019single} also study downside risk in moderately adverse states, with disappointment corresponding to approximately the worst 20\% of outcomes.}
	
\begin{figure}
\FIGURE
{\includegraphics[scale=0.46]{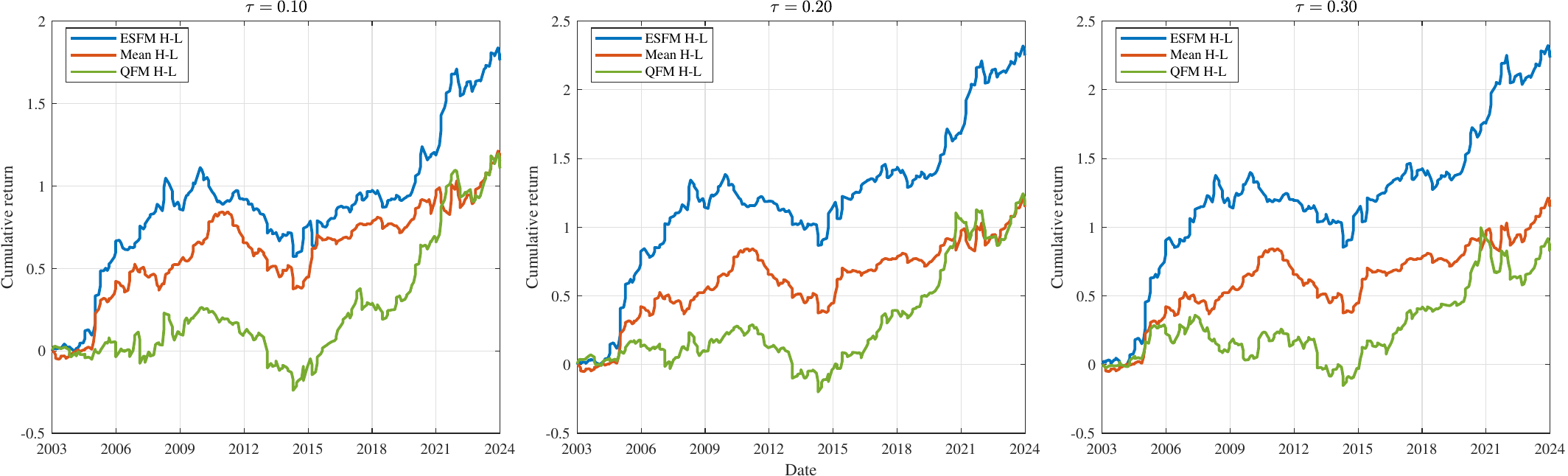}}
{Cumulative high-minus-low returns across models. \label{fig:cumu_returns_VW}}
{This figure plots cumulative returns of high-minus-low (H--L) portfolios formed on factor exposures estimated from three models: ESFM, the mean factor model (Mean), and the quantile factor model (QFM). Stocks are sorted into five portfolios based on estimated exposures, and the H--L portfolio is constructed as the difference between the highest- and lowest-exposure portfolios. Portfolio returns are computed using equal-weighted (EW) returns, while the underlying Fama--French five factors are constructed using value-weighted (VW) returns. Each panel corresponds to a different tail level ($\tau=0.10,0.20,0.30$).}
\end{figure}%

Figure~\ref{fig:cumu_returns_VW} reports the cumulative returns of the H--L portfolios across the three models. The ESFM portfolio accumulates a larger return over the sample than the portfolios based on the mean factor model and the quantile factor model at each value of $\tau$, although its performance varies over time. The Mean and QFM portfolios display weaker and less stable cumulative performance, with periods of stagnation and reversals. The differences are more visible around several market stress episodes, suggesting that exposure to expected shortfall risk is associated with economically meaningful return variation. This graphical evidence is consistent with the average return and alpha estimates reported below.

\begin{table}
\footnotesize
\TABLE
{Portfolio returns and alphas sorted on factor exposures across models. \label{tab:Alphas}}
{
\small
\setlength{\tabcolsep}{2.5pt}
\renewcommand{\arraystretch}{0.95}
\begin{tabular*}{\textwidth}{@{\extracolsep{\fill}}rrccccrccccrcccc}
\toprule
&       & \multicolumn{4}{c}{$\tau=0.10$} &       & \multicolumn{4}{c}{$\tau=0.20$} &       & \multicolumn{4}{c}{$\tau=0.30$} \\
\cmidrule{3-6}\cmidrule{8-11}\cmidrule{13-16}
No.G   &       & Average & CAPM  & FF3   & FF5   &       & Average & CAPM  & FF3   & FF5   &       & Average & CAPM  & FF3   & FF5 \\
\midrule
\multicolumn{1}{l}{5} & \multicolumn{1}{l}{ESFM} & 8.04\% & 9.01\% & 9.40\% & 10.31\% &       & 10.27\% & 11.26\% & 11.75\% & 12.25\% &       & 10.20\% & 11.28\% & 11.81\% & 12.34\% \\
&       & (2.62) & (3.09) & (3.38) & (3.33) &       & (3.08) & (3.49) & (3.84) & (3.61) &       & (2.96) & (3.41) & (3.78) & (3.48) \\
\multicolumn{1}{l}{5} & \multicolumn{1}{l}{Mean} & 5.29\% & 5.20\% & 5.33\% & 5.18\% &       & 5.29\% & 5.20\% & 5.33\% & 5.18\% &       & 5.29\% & 5.20\% & 5.33\% & 5.18\% \\
&       & (2.18) & (2.14) & (2.19) & (2.18) &       & (2.19) & (2.14) & (2.19) & (2.18) &       & (2.19) & (2.14) & (2.19) & (2.18) \\
\multicolumn{1}{l}{5} & \multicolumn{1}{l}{QFM} & 5.05\% & 5.92\% & 5.71\% & 5.73\% &       & 5.29\% & 5.98\% & 6.13\% & 8.38\% &       & 3.77\% & 3.87\% & 4.43\% & 5.25\% \\
&       & (1.80) & (2.13) & (1.97) & (1.87) &       & (1.87) & (2.07) & (2.13) & (2.51) &       & (1.22) & (1.24) & (1.58) & (1.48) \\
\multicolumn{1}{l}{10} & \multicolumn{1}{l}{ESFM} & 9.34\% & 10.55\% & 11.02\% & 13.09\% &       & 10.78\% & 11.97\% & 12.54\% & 14.04\% &       & 11.68\% & 13.01\% & 13.69\% & 14.96\% \\
&       & (2.48) & (2.99) & (3.28) & (3.33) &       & (2.73) & (3.13) & (3.46) & (3.28) &       & (2.84) & (3.27) & (3.69) & (3.39) \\
\multicolumn{1}{l}{10} & \multicolumn{1}{l}{Mean} & 8.27\% & 8.23\% & 8.52\% & 8.14\% &       & 8.27\% & 8.23\% & 8.52\% & 8.14\% &       & 8.27\% & 8.23\% & 8.52\% & 8.14\% \\
&       & (2.54) & (2.52) & (2.60) & (2.59) &       & (2.54) & (2.52) & (2.60) & (2.59) &       & (2.54) & (2.52) & (2.60) & (2.59) \\
\multicolumn{1}{l}{10} & \multicolumn{1}{l}{QFM} & 7.76\% & 8.74\% & 8.42\% & 8.26\% &       & 6.88\% & 7.79\% & 8.01\% & 10.56\% &       & 4.28\% & 4.49\% & 5.10\% & 7.28\% \\
&       & (2.23) & (2.53) & (2.33) & (2.11) &       & (1.90) & (2.09) & (2.18) & (2.47) &       & (1.19) & (1.24) & (1.54) & (1.92) \\
\bottomrule
\end{tabular*}
}
{The table reports annualized portfolio returns and alphas (in \%) sorted on factor exposures estimated from different models (ESFM, Mean, and QFM) across three tail levels ($\tau=0.10,0.20,0.30$). We report estimated intercept (alphas) from regressing the returns on various sets of asset pricing factors: market (CAPM), three factors of \cite{fama1993common} (FF3), and five factors of \cite{fama2015five}. Stocks are sorted into No.G $=5$ or $10$ portfolios based on estimated exposures, and portfolio returns are computed using equal-weighted (EW) scheme. The underlying Fama–French three and five factors are constructed using value-weighted (VW) returns. The Average column reports the annualized mean H--L return. The CAPM, FF3, and FF5 columns report annualized intercepts from the corresponding time-series regressions. Newey--West $t$-statistics (six lags) are reported in parentheses. The ``High–Low'' spread captures the return difference between the highest- and lowest-exposure portfolios (the average is placed in column ``Average'').}
\end{table}%

Table~\ref{tab:Alphas} complements this evidence by reporting the average H--L portfolio returns and the corresponding pricing alphas. Consistent with the graphical evidence, the ESFM portfolios deliver the largest return spreads within each sorting scheme, both economically and statistically. To assess whether these spreads can be explained by standard risk factors, we regress the H--L portfolio returns on CAPM, FF3, and FF5. The resulting alphas remain sizable and statistically significant across specifications. A statistically significant alpha indicates that the benchmark factors do not fully account for the average return on the ESFM spread. The consistently positive alphas therefore suggest that ESFM exposures contain incremental pricing information beyond conventional factors.
	

    Portfolios formed on mean- and quantile-factor exposures also earn positive spreads in several specifications, consistent with prior evidence that latent mean factors and lower-tail quantile risk can be priced \citep{lettau2020factors,barunik2026common}. Their spreads, however, are smaller and less stable than those generated by ESFM, and neither benchmark subsumes the ESFM premium in the joint two-pass regressions. The comparison should therefore not be read as evidence that mean or quantile risk is generally unpriced. Rather, holding the return panel, observed factors, estimation window, and pricing design fixed, the results show that exposure to common loss severity contributes an additional priced component that is not captured by common variation in average returns or tail thresholds.

The H--L return uses only the endpoint portfolios. To assess whether it reflects a broader cross-sectional pattern, we report average returns for all five ESFM-exposure quintiles. Portfolios are formed using the baseline procedure, with P1 and P5 containing stocks with the lowest and highest exposures, respectively.

    \begin{figure}
\FIGURE
{\includegraphics[scale=0.44]{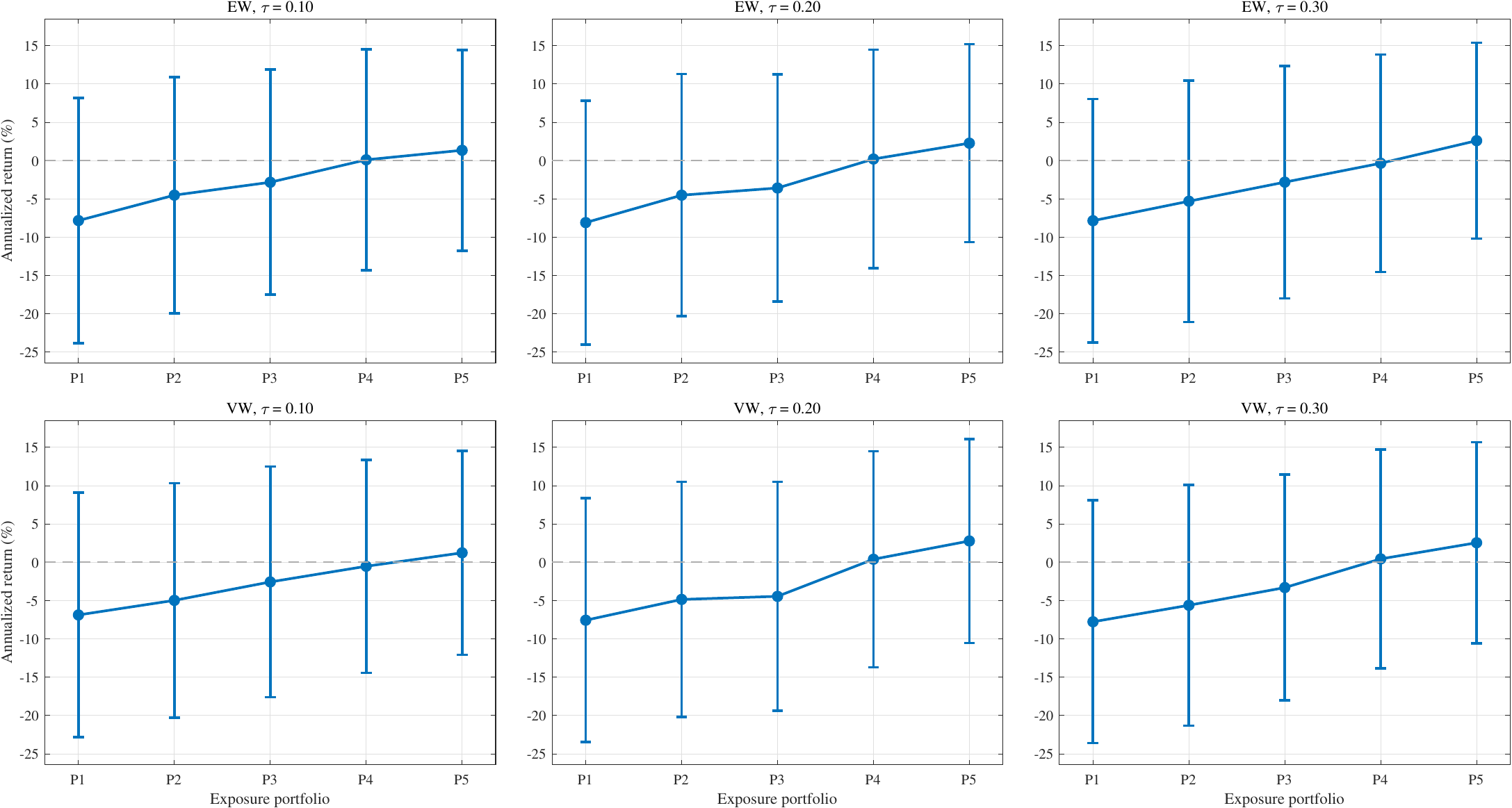}}
{Portfolio returns across ESFM exposure quintiles.
\label{fig:esfm_quintile_returns}}
{This figure reports annualized average returns for quintile portfolios sorted on estimated ESFM exposure. At each portfolio-formation date, stocks are assigned to five portfolios, with P1 containing stocks with the lowest exposure and P5 those with the highest exposure. The top and bottom rows use latent factors estimated with equal-weighted and value-weighted versions of the observable factors, respectively. The columns correspond to $\tau=0.10$, $0.20$, and $0.30$. Sorting variables are estimated using information available at portfolio formation. Portfolios are equally weighted and held for one month. Returns are reported in percent per annum, and the vertical bars denote 95\% confidence intervals based on Newey--West standard errors with six lags.}
\end{figure}

Figure~\ref{fig:esfm_quintile_returns} shows a pronounced return gradient across the ESFM exposure portfolios. Average returns increase monotonically from P1 to P5 in all six panels. Low-exposure portfolios earn annualized returns of approximately $-8\%$, whereas the returns of the highest-exposure portfolios are positive and generally between $1\%$ and $3\%$. The pattern is similar across the three tail levels and under both observable-factor specifications. Although the confidence intervals for individual portfolio means are wide, the monotonicity of the point estimates and the significant endpoint spreads reported in Table~\ref{tab:Alphas} indicate that the H--L premium is not generated solely by an isolated extreme portfolio.\footnote{The corresponding decile sorts, reported in Figure~\ref{fig:ec_esfm_decile_returns}, produce a similar low-to-high return pattern.}

\subsubsection{Fama-MacBeth Regressions}\label{subsubsec:Fama-MacBeth Regressions}
	
We complement the portfolio-sorting evidence with two-pass Fama–MacBeth regressions following \cite{fama1973risk}. This approach allows us to assess whether the pricing effects associated with ES-based exposures persist after controlling for other sources of risk in a unified cross-sectional framework. In the first step, we estimate factor loadings by regressing individual stock returns on the factor-mimicking portfolio returns constructed from ESFM, Mean, and QFM specifications. In the second step, we run cross-sectional regressions of next-period returns on these estimated loadings to obtain the corresponding prices of risk. We consider specifications with individual factors, jointly estimated factors, and augmentations that include the Fama–French five factors.
	
	\begin{table}
\footnotesize
\TABLE
{Fama--MacBeth regressions for factor risk premia. \label{tab:FamaMacBeth}}
{
\setlength{\tabcolsep}{2.5pt}
\renewcommand{\arraystretch}{0.95}
\begin{tabular*}{\textwidth}{@{\extracolsep{\fill}}lcccccccccc}
\toprule
& \multicolumn{1}{l}{ESFM} & \multicolumn{1}{l}{Mean} & \multicolumn{1}{l}{QFM} & \multicolumn{1}{l}{Joint} & \multicolumn{1}{l}{MKT} & \multicolumn{1}{l}{SMB} & \multicolumn{1}{l}{HML} & \multicolumn{1}{l}{RMW} & \multicolumn{1}{l}{CMA} & \multicolumn{1}{l}{FF5} \\
\midrule
\multicolumn{11}{c}{Panel A: $\tau=0.10$} \\
\midrule
ESFM  & 1.3647  &       &       & 0.9168  & 0.9857  & 0.9196  & 0.9438  & 1.1717  & 1.1575  & 0.9540  \\
Mean  &       & 0.5329  &       & 0.7383  &       &       &       &       &       &  \\
QFM   &       &       & 1.2761  & 0.8392  &       &       &       &       &       &  \\
MKT   &       &       &       &       & -3.7597  &       &       &       &       & -4.1021  \\
SMB   &       &       &       &       &       & 2.4672  &       &       &       & 1.9328  \\
HML   &       &       &       &       &       &       & -1.3910  &       &       & -1.4622  \\
RMW   &       &       &       &       &       &       &       & -1.3169  &       & -0.3998  \\
CMA   &       &       &       &       &       &       &       &       & 1.2188  & 0.4814  \\
Intercept & -4.9495  & -6.5308  & -6.0740  & -6.9415  & -3.5136  & -8.2988  & -6.8780  & -8.4304  & -7.6352  & -4.0660  \\
$R^2_{\rm adj}$ & 1.6926  & 1.2676  & 1.6981  & 3.9438  & 2.9411  & 4.1018  & 3.4559  & 3.7234  & 3.4245  & 7.5766  \\
\midrule
\multicolumn{11}{c}{Panel B: $\tau=0.20$} \\
\midrule
ESFM  & 1.5763  &       &       & 1.5114  & 1.3221  & 1.0540  & 1.2372  & 1.3216  & 1.1529  & 1.1411  \\
Mean  &       & 0.5329  &       & 0.8828  &       &       &       &       &       &  \\
QFM   &       &       & 0.9383  & 0.7903  &       &       &       &       &       &  \\
MKT   &       &       &       &       & -3.5932  &       &       &       &       & -4.1551  \\
SMB   &       &       &       &       &       & 2.5241  &       &       &       & 1.9809  \\
HML   &       &       &       &       &       &       & -1.3113  &       &       & -1.5855  \\
RMW   &       &       &       &       &       &       &       & -1.4412  &       & -0.4665  \\
CMA   &       &       &       &       &       &       &       &       & 1.2467  & 0.4924  \\
Intercept & -4.8764  & -6.5308  & -6.4467  & -5.1027  & -3.4591  & -8.4723  & -6.8684  & -8.6122  & -8.0020  & -4.0429  \\
$R^2_{\rm adj}$ & 1.5654  & 1.2676  & 1.5754  & 4.1850  & 2.9794  & 4.0088  & 3.4352  & 3.6315  & 3.2959  & 7.6445  \\
\midrule
\multicolumn{11}{c}{Panel C: $\tau=0.30$} \\
\midrule
ESFM  & 1.6291  &       &       & 1.4751  & 1.4691  & 0.9541  & 1.3822  & 1.2611  & 1.1664  & 1.0543  \\
Mean  &       & 0.5329  &       & 0.9036  &       &       &       &       &       &  \\
QFM   &       &       & -0.7204  & -0.5199  &       &       &       &       &       &  \\
MKT   &       &       &       &       & -3.5169  &       &       &       &       & -4.1316  \\
SMB   &       &       &       &       &       & 2.5074  &       &       &       & 1.8856  \\
HML   &       &       &       &       &       &       & -1.3335  &       &       & -1.5727  \\
RMW   &       &       &       &       &       &       &       & -1.4713  &       & 0.4522  \\
CMA   &       &       &       &       &       &       &       &       & 1.2274  & 0.4904  \\
Intercept & -4.7123  & -6.5308  & -7.0182  & -6.5134  & -3.5416  & -8.3365  & -6.6911  & -8.4056  & -7.8404  & -4.0318  \\
$R^2_{\rm adj}$ & 1.6064  & 1.2676  & 1.5518  & 4.3523  & 3.0341  & 4.0294  & 3.4728  & 3.6972  & 3.3020  & 7.7286  \\
\bottomrule
\end{tabular*}
}
{The table reports estimated prices of risk from two-pass Fama--MacBeth regressions. In the first stage, factor loadings are obtained from time-series regressions of individual stock returns on factor-mimicking portfolio returns constructed from ESFM, Mean, and QFM specifications. In the second stage, cross-sectional regressions are performed each month to estimate factor risk premia. We report specifications with individual factors, jointly estimated factors, and augmentations of the ESFM factor with Fama--French factors constructed using value-weighted returns. Columns labeled MKT, SMB, HML, RMW, and CMA add the corresponding factor to the ESFM factor, and the FF5 column adds all five Fama--French factors. The reported coefficients are time-series averages of monthly estimates. Adjusted $R^2$ is the time-series average of the monthly cross-sectional adjusted $R^2$ and is reported in percent. The coefficients are multiplied by 100 for presentation.}
\end{table}%

Table~\ref{tab:FamaMacBeth} reports a positive price of ESFM risk in every specification. When the Mean and QFM exposures enter jointly, the ESFM estimate ranges from $0.9168$ to $1.5114$ across the three tail levels. It remains between $0.9540$ and $1.1411$ when all five Fama--French factors are included. The QFM estimates are less stable and become negative at $\tau=0.30$. Thus, the positive relation between ESFM exposure and expected returns remains after controlling for the benchmark exposures considered in the table.

\subsubsection{Incremental Pricing Content}\label{subsubsec:Incremental Pricing Content}

The preceding results establish an unconditional relation between ESFM exposure and subsequent returns. This relation may nevertheless reflect the correlation of ESFM exposure with common movements in average returns, tail thresholds, or idiosyncratic volatility. To separate these effects, we implement dependent bivariate portfolio sorts in the spirit of \cite{barunik2026common}. Each month, stocks are first assigned to groups based on QFM exposure, Mean exposure, or idiosyncratic volatility. We also construct joint control cells using QFM and Mean exposures. Within each control cell,
stocks are then sorted into quintiles by ESFM exposure. The corresponding portfolios are obtained by averaging across control cells, producing variation in ESFM exposure while keeping the conditioning variables approximately constant.

\begin{figure}
\FIGURE
{\includegraphics[scale=0.45]{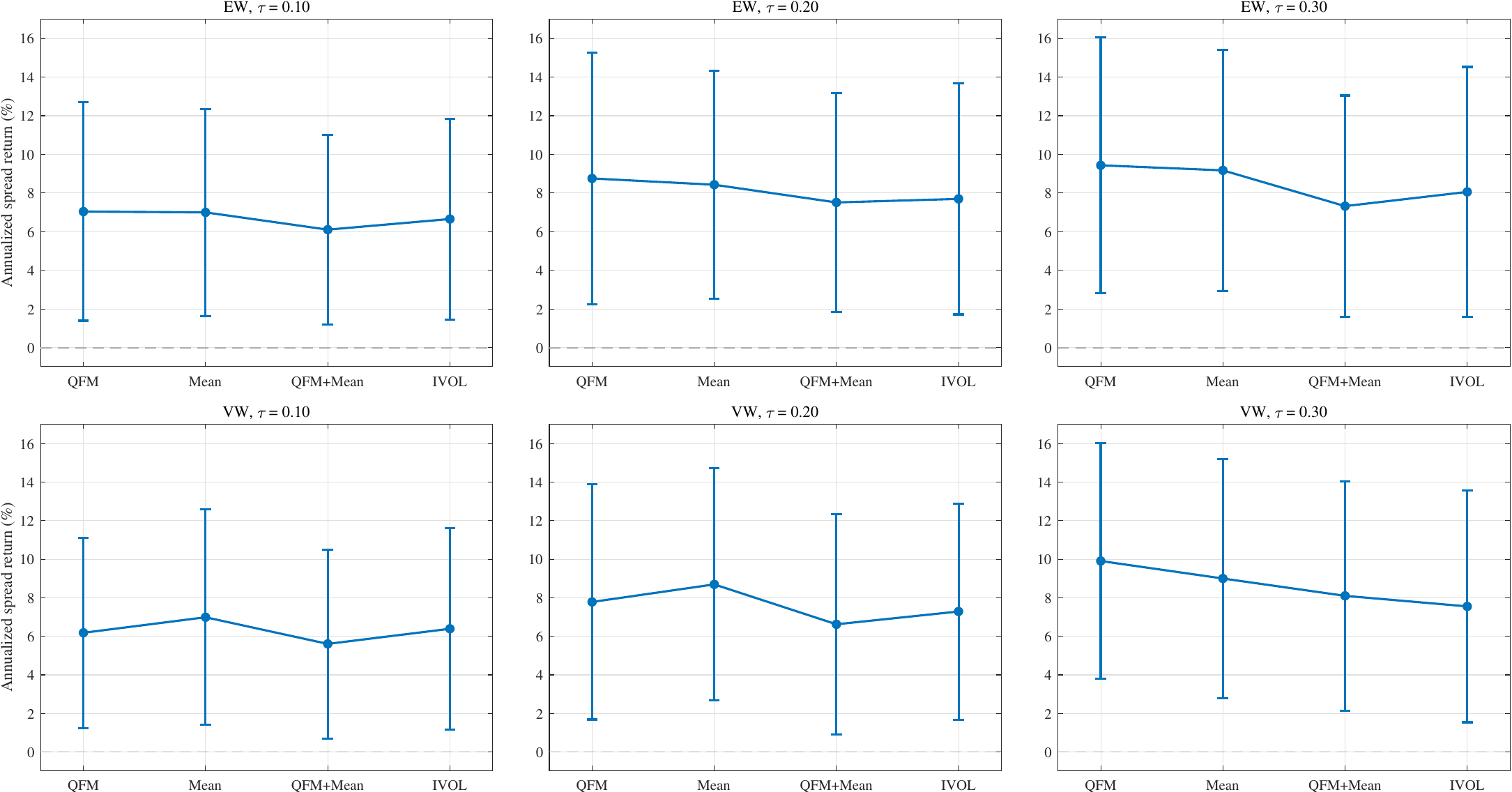}}
{Conditional return spreads across ESFM exposure portfolios.
\label{fig:conditional_esfm_returns}}
{This figure reports annualized high-minus-low returns from conditional portfolio sorts on estimated ESFM exposure. Each month, stocks are first assigned to five groups based on QFM exposure, Mean exposure, or idiosyncratic volatility. For the joint QFM--Mean specification, stocks are assigned to $3\times3$ control cells. Within each control cell, stocks are sorted into five portfolios by ESFM exposure, and the resulting high-minus-low returns are averaged across cells. The top and bottom rows use latent factors estimated with equal-weighted and value-weighted versions of the observable factors, respectively. The columns correspond to $\tau=0.10$, $0.20$, and $0.30$. Sorting variables are estimated using information available at portfolio formation. Portfolios are equally weighted and held for one month. Returns are reported in percent per annum, and the vertical bars denote 95\% confidence intervals based on Newey--West standard errors with six lags.}
\end{figure}

Figure~\ref{fig:conditional_esfm_returns} shows positive conditional H--L returns across all four control specifications. The annualized spreads range from approximately $5.6\%$ to $9.9\%$, and their confidence intervals remain above zero in every panel. The result holds when conditioning separately on QFM and Mean exposures, when conditioning jointly on both exposures, and when controlling for idiosyncratic volatility. It is also stable across the three tail levels and the two observable-factor specifications. Thus, high-ESFM-exposure stocks continue to earn higher subsequent returns than low-exposure stocks even among firms with similar mean, quantile, or volatility exposures. The persistence of the spread indicates that the ESFM premium is not driven by a mechanical association with movements in average returns, tail thresholds, or conventional idiosyncratic risk.

Conditional sorts assess whether ESFM exposure contains distinct information at the stock level. We next examine the complementary factor-level question of whether the traded ESFM factor can be replicated by existing factors. Following the factor-spanning and mean--variance efficiency logic emphasized by \cite{barillas2017which} and \cite{feng2026selecting}, we regress the ESFM H--L factor on MKT, SMB, HML, RMW, CMA, MOM, and the corresponding Mean and QFM factors. A nonzero intercept indicates that the benchmark factors do not span the average ESFM factor return. We also compare the maximum Sharpe ratio of the benchmark factor set with that obtained after adding ESFM. The latter comparison measures whether ESFM expands the attainable mean--variance frontier rather than merely repackaging existing factor exposures.

\begin{table}
\footnotesize
\TABLE
{ESFM factor spanning tests. \label{tab:spanning_esfm}}
{
\setlength{\tabcolsep}{4pt}
\renewcommand{\arraystretch}{0.95}
\begin{tabular*}{\textwidth}{@{\extracolsep{\fill}}lcccccc}
\toprule
& \multicolumn{3}{c}{Value-weighted factors} & \multicolumn{3}{c}{Equal-weighted factors} \\
\cmidrule{2-4}\cmidrule{5-7}
& $\tau=0.10$ & $\tau=0.20$ & $\tau=0.30$
& $\tau=0.10$ & $\tau=0.20$ & $\tau=0.30$ \\
\midrule
$\alpha$ (\% p.a.) 
& 6.60 & 8.65 & 9.37 & 7.35 & 7.91 & 8.66 \\
& (2.80) & (2.48) & (2.70) & (3.39) & (3.09) & (3.14) \\
Mean factor 
& 0.336 & 0.467 & 0.546 & 0.305 & 0.389 & 0.451 \\
& (2.94) & (4.09) & (4.60) & (3.32) & (3.98) & (4.44) \\
QFM factor 
& 0.361 & 0.161 & 0.085 & 0.306 & 0.193 & 0.040 \\
& (3.94) & (1.70) & (0.90) & (3.58) & (2.20) & (0.52) \\
$R^2$ (\%) 
& 48.85 & 42.59 & 45.52 & 45.78 & 46.96 & 49.72 \\
Max SR, benchmark 
& 1.56 & 1.64 & 1.52 & 1.78 & 1.81 & 1.78 \\
Max SR, with ESFM 
& 1.78 & 1.90 & 1.82 & 2.02 & 2.05 & 2.04 \\
$\Delta$SR$^2$ 
& 0.723 & 0.921 & 1.010 & 0.892 & 0.898 & 0.990 \\
Months 
& 262 & 262 & 262 & 262 & 262 & 262 \\
\midrule
Joint $\alpha$ test: $\chi^2(3)$ 
& \multicolumn{3}{c}{4.50} & \multicolumn{3}{c}{7.49} \\
$p$-value 
& \multicolumn{3}{c}{0.213} & \multicolumn{3}{c}{0.058} \\
Months 
& \multicolumn{3}{c}{188} & \multicolumn{3}{c}{188} \\
\bottomrule
\end{tabular*}
}
{The table reports spanning regressions for ESFM high-minus-low factors at three tail levels. The ESFM factor is regressed on MKT, SMB, HML, RMW, CMA, MOM, and the corresponding Mean and QFM factors. Columns labeled value-weighted and equal-weighted refer to the observable Fama--French factors used in estimation. Alphas are annualized and reported in percent. Newey--West $t$-statistics are shown in parentheses. Max SR denotes the maximum Sharpe ratio spanned by the benchmark factors, with and without the ESFM factor. The joint test is a Wald test that the three ESFM alphas are jointly zero.}
\end{table}%

Table~\ref{tab:spanning_esfm} shows that the benchmark factors do not absorb the average ESFM factor return in the individual spanning regressions. Annualized alphas range from $6.60\%$ to $9.37\%$ under the value-weighted observable-factor specification and from $7.35\%$ to $8.66\%$ under the equal-weighted observable-factor specification. The associated Newey--West $t$-statistics range from 2.48 to 3.39. Thus, at each tail level, the ESFM factor earns a positive and statistically significant return after controlling jointly for standard traded factors and the corresponding Mean and QFM factors.

The estimated factor loadings also clarify the relation among the three latent specifications. The ESFM factor loads positively and significantly on the Mean factor throughout. Its loading on the QFM factor is strongest at $\tau=0.10$ and declines toward zero as $\tau$ increases. This pattern is consistent with ESFM and QFM sharing information about the lower tail while remaining distinct: QFM tracks movements in the tail threshold, whereas ESFM also incorporates the severity of outcomes beyond that threshold. Importantly, the ESFM alpha remains economically large as its loading on QFM weakens.

Adding ESFM also improves the attainable investment opportunity set. Across the six specifications, the maximum Sharpe ratio increases from 1.52--1.81 for the benchmark factors to 1.78--2.05 after ESFM is included. The joint-alpha evidence is more qualified. The hypothesis that the three ESFM alphas are jointly zero is not rejected under the value-weighted observable-factor specification ($p=0.213$), whereas the equal-weighted specification provides marginal evidence against the joint null ($p=0.058$). The results therefore provide strong tail-by-tail evidence that ESFM contributes to the factor span and improves the mean--variance frontier, although joint significance across the three tail levels depends on the observable-factor specification.

	\section{Conclusion}\label{sec:Conclusion}

This paper develops an expected shortfall factor model (ESFM) to estimate common variation in the severity of lower-tail losses and examine whether exposures to this variation are priced. ESFM extracts latent common factors from average returns below asset-specific conditional quantiles. This object differs from the average-return comovement captured by mean factor models and the movements in tail thresholds captured by quantile factor models. The model combines observed risk exposures with latent common factors and allows coefficients and factor loadings to vary across assets. Estimation proceeds through an orthogonalized two-step procedure that prevents errors in the estimated quantiles from affecting the ES coefficient estimates to first order. We establish high-probability error bounds for the quantile and ES coefficient estimators, a finite-sample Gaussian approximation, and consistency of the factor-number selector. Complementary asymptotic results are provided in the appendix. Across a range of Monte Carlo designs, ESFM estimates the ES coefficients more accurately than separate asset-level ES regressions when common tail shocks are present.

Applied to a large panel of equities from 2000 to 2023, ESFM recovers common factors that react sharply to market stress and contain information not captured by mean and quantile factors. Portfolios sorted on ESFM exposures generate annualized high-minus-low returns of 8.0\%--11.7\% and Fama--French five-factor alphas of 10.3\%--15.0\%. These return spreads remain positive after conditioning separately and jointly on mean- and quantile-factor exposures and after conditioning on idiosyncratic volatility. Spanning tests further show that ESFM factors retain significant alphas after controlling for standard traded factors and the corresponding mean and quantile factors. Adding ESFM to these benchmark factor sets increases the maximum attainable Sharpe ratio. Because neither expected returns nor pricing errors enter factor extraction, these pricing results are not imposed by construction. Together, the findings identify common loss severity as a distinct and priced dimension of downside risk.

The framework can be extended in several directions. An ES connectedness model could extend quantile-based financial networks \citep{ando2022connectedness} by measuring spillovers in the severity of outcomes beyond a tail cutoff. A matrix ES factor model could preserve two-way structures in returns indexed by dimensions such as assets and markets. Allowing the number, loadings, or dynamics of the ES factors to change at unknown dates would provide a way to detect structural breaks in common downside risk. {Alternatively, following the state-varying factor framework of \citet{pelger2022state}, one could allow the ES loadings to evolve with observable state variables, enabling the model to capture systematic shifts in tail-risk exposures over time.} Finally, combining ESFM with regularization or other dimension-reduction methods could support ES prediction when the number of observed predictors is large.

	\ACKNOWLEDGMENT{All authors contributed equally and are listed in alphabetical order.}
    \theendnotes
	
\bibliographystyle{informs2014} 
\bibliography{ref.bib} 

\ECSwitch


\ECHead{Technical Proofs and Additional Results}

This appendix contains technical proofs and additional results for the paper.
\begin{itemize}
    \item Appendix \ref{sec:Proofs} provides the proofs of nonasymptotic results in the main text.
    \item Appendix \ref{sec:Additional Results} provides additional asymptotic results.
    \item Appendix \ref{sec:proof Additional Simulation Results} provides the proofs of additional asymptotic results.
    \item Appendix \ref{sec:Additional Simulation Results} provides the additional simulation results.
    \item Appendix \ref{sec:Additional Empirical Results} provides additional empirical results.
\end{itemize}

\section{Proofs of Nonasymptotic Results}\label{sec:Proofs}
For notational simplicity, we suppress the dependency on $\tau$ such that $\alpha_{i,\tau}=\alpha_i$, $\beta_{i,\tau}=\beta_i$, $f_{t,\tau}=f_t$, $\lambda_{i,\tau}=\lambda_i$, $B_\tau=B$, $F_\tau=F$, $\Lambda_\tau=\Lambda$, $r^0_\tau=r^0$, etc throughout the proof. We denote the estimated parameters as $\widehat{\alpha}_i$, $\widehat{\beta}_i$, $\widehat{f}_t$, $\widehat{\lambda}_i$, $\widehat{B}$, $\widehat{F}$, $\widehat{\Lambda}$ and $\widehat{r}$, etc. Also, we denote the true parameters as $\alpha_i^0$, $\beta_i^0$, $f_t^0$, $\lambda_i^0$, $B^0$, $F^0$, $\Lambda^0$ and $r^0$, etc.

The symbol $C$ denotes a positive constant whose value may change from line to line and which depends only on the fixed constants in the assumptions and on $\tau$. All statements about the second-stage estimator refer to an arbitrary joint global minimizer of \eqref{eq:step 2 ES}. For $B=(b_1^\top,\ldots,b_N^\top)^\top$, write $\mathcal X(B)=(X_1b_1,\ldots,X_Nb_N)$, and use $\langle A,B\rangle=\operatorname{tr}(A^\top B)$ for the Frobenius
inner product. The notation $\|A\|_\ast$ denotes the nuclear norm of $A$.

\subsection{Proof of Proposition \ref{pro:quantile}}
\proof{Proof of Proposition \ref{pro:quantile}.}
Fix $i\in[N]$ and define $\gamma=\Sigma_i^{1/2}(\alpha-\alpha_{i}^0)$ and $\widehat\gamma_i=\Sigma_i^{1/2}
(\widehat\alpha_{i}-\alpha_{i}^0)$. By construction, $\|\widehat\alpha_{i}-\alpha_{i}^0\|_{\Sigma_i}=
\|\widehat\gamma_i\|$. 

We first establish local curvature of the population quantile objective. Let $\psi(u)=\tau-\mathbbm{1}\{u\leq0\}$. Knight's identity gives
$$
\rho_\tau(\epsilon-v)-\rho_\tau(\epsilon)=-v\psi(\epsilon)+\int_0^v\left\{\mathbbm{1}(\epsilon\leq s)-\mathbbm{1}(\epsilon\leq0)\right\}\,ds.
$$
Setting $v=W_{it}^{\top}\gamma$ and using $E\{\psi(\epsilon_{it})\mid X_{it}\}=0$, we obtain
\begin{equation}\label{eq:tagA.5}
D_i(\gamma)=E\left[\int_0^{W_{it}^{\top}\gamma}\left\{F_{\epsilon_{it}\mid X_{it}}(s\mid X_{it})-F_{\epsilon_{it}\mid X_{it}}(0\mid X_{it})\right\}\,ds\right].
\end{equation}
For every $s\in\mathbb R$, the density condition in Assumption \ref{assum:Regularity}(i) implies $F_{\epsilon_{it}\mid X_{it}}(s\mid X_{it})-F_{\epsilon_{it}\mid X_{it}}(0\mid X_{it})=f_{\epsilon_{it}\mid X_{it}}(0\mid X_{it})s
+r_{it}(s)$,
where $|r_{it}(s)|\leq{L_0}s^2/2$. Therefore, for every $v\in\mathbb R$,
\begin{equation}\label{eq:tagA.6}
\int_0^v\left\{F_{\epsilon_{it}\mid X_{it}}(s\mid X_{it})-F_{\epsilon_{it}\mid X_{it}}(0\mid X_{it})\right\}\,ds\geq
\frac{\underline f}{2}v^2-\frac{L_0}{6}|v|^3.
\end{equation}
Combining \eqref{eq:tagA.5} and \eqref{eq:tagA.6} yields
$$
D_i(\gamma)\geq\frac{\underline f}{2}E(W_{it}^{\top}\gamma)^2-\frac{L_0}{6}E|W_{it}^{\top}\gamma|^3.
$$
Since $E(W_{it}W_{it}^{\top})=I_{p+1}$, we have $E(W_{it}^{\top}\gamma)^2=\|\gamma\|^2$. Moreover, the sub-Gaussian tail condition implies $\sup_{u\in\mathbb S^p}E|u^{\top}W_{it}|^3\leq C\upsilon_1^3$. It follows that
\begin{equation}\label{eq:tagA.7}
D_i(\gamma)\geq\frac{\underline f}{2}\|\gamma\|^2-C L_0\upsilon_1^3\|\gamma\|^3.
\end{equation}
Set $x_\delta=\log\left(\frac{2eN}{\delta}\right)=1+\log\left(\frac{2N}{\delta}\right)$. Then $p+1+x_\delta=p+2+\log\left({2N}/{\delta}\right)=s_{N,\delta}$. Define $a_T=\sqrt{\frac{s_{N,\delta}}{T}}$. Let $C_{\mathrm{ep}}>0$ denote the constant in the simplified bound in Lemma \ref{lem:mixing-empirical-process}, and define the deterministic radius
\begin{equation}\label{eq:tagA.8}
\rho_T=8C_{\mathrm{ep}}\upsilon_1\underline f^{-1}a_T.
\end{equation}

Under the sample-size condition $T\geq C s_{N,\delta}^2$, Lemma \ref{lem:mixing-empirical-process}, applied with
$x=x_\delta$ and $\rho=\rho_T$, gives
$$
P\left\{\sup_{\|\gamma\|\leq\rho_T}|(\widehat D_i-D_i)(\gamma)|>C_{\mathrm{ep}}\upsilon_1\rho_Ta_T\right\}\leq 2e^{-x_\delta}.
$$
A union bound over $i\in[N]$, together with $2Ne^{-x_\delta}={\delta}/{e}\leq\delta$, shows that, with probability at least $1-\delta$, simultaneously for all $i\in[N]$,
\begin{equation}\label{eq:tagA.9}
\sup_{\|\gamma\|\leq\rho_T}|(\widehat D_i-D_i)(\gamma)|\leq C_{\mathrm{ep}}\upsilon_1\rho_Ta_T.
\end{equation}
The remaining sample-size condition $T\geq C L_0^2\underline f^{-4}s_{N,\delta}$, where the dependence on $\upsilon_1$ is absorbed into $C$, ensures that $C L_0\upsilon_1^3\rho_T\leq{\underline f}/{4}$. It follows from \eqref{eq:tagA.7} that, whenever $\|\gamma\|=\rho_T$,
\begin{equation}\label{eq:tagA.10}
D_i(\gamma)\geq\frac{\underline f}{4}\rho_T^2.
\end{equation}
Combining \eqref{eq:tagA.9} and \eqref{eq:tagA.10}, we obtain
\begin{equation}\nonumber
\widehat D_i(\gamma)\geq D_i(\gamma)-|(\widehat D_i-D_i)(\gamma)|\geq\frac{\underline f}{4}\rho_T^2-C_{\mathrm{ep}}\upsilon_1\rho_Ta_T.
\end{equation}
By the definition of $\rho_T$ in \eqref{eq:tagA.8}, ${\underline f}\rho_T^2/{4}-C_{\mathrm{ep}}\upsilon_1\rho_Ta_T=8C_{\mathrm{ep}}^2\upsilon_1^2\underline f^{-1}a_T^2>0$. Therefore,
\begin{equation}\label{eq:tagA.11}
\widehat D_i(\gamma)>0\text{ for every }\gamma\text{ such that }\|\gamma\|=\rho_T.
\end{equation}

On the other hand, because $\widehat\alpha_{i}$ minimizes the empirical quantile objective, $\widehat D_i(\widehat\gamma_i)\leq\widehat D_i(0)=0$. Since $\widehat D_i$ is convex, \eqref{eq:tagA.11} implies $\|\widehat\gamma_i\|\leq\rho_T$. Indeed, suppose instead that $\|\widehat\gamma_i\|>\rho_T$ and define $\widetilde\gamma_i={\rho_T}\widehat\gamma_i/{\|\widehat\gamma_i\|}$. Convexity then gives
\begin{equation}\nonumber
\widehat D_i(\widetilde\gamma_i)\leq\frac{\rho_T}{\|\widehat\gamma_i\|}\widehat D_i(\widehat\gamma_i)+\left(1-
\frac{\rho_T}{\|\widehat\gamma_i\|}\right)\widehat D_i(0)\leq0,
\end{equation}
which contradicts \eqref{eq:tagA.11}, because $\|\widetilde\gamma_i\|=\rho_T$.

Consequently, with probability at least $1-\delta$,
\begin{equation}\nonumber
\max_{i\in[N]}\|\widehat\alpha_{i}-\alpha_{i}^0\|_{\Sigma_i}=\max_{i\in[N]}\|\widehat\gamma_i\|\leq 8C_{\mathrm{ep}}\upsilon_1\underline f^{-1}\sqrt{\frac{s_{N,\delta}}{T}}.
\end{equation}
Setting $C_1=8C_{\mathrm{ep}}\upsilon_1$ and absorbing the remaining fixed constants into $C_2$ completes the proof.
\Halmos
\endproof
\begin{lemma}\label{lem:mixing-empirical-process}
For $a\in\mathbb R^{p+1}$, define the unit-specific empirical quantile objective $\widehat Q_i(a)=T^{-1}\sum_{t=1}^T\rho_\tau(Y_{it}-X_{it}^{\top}a)$. For every $i\in[N]$ and $\gamma\in\mathbb R^{p+1}$, define
$\widehat D_i(\gamma)=\widehat Q_i\left(\alpha_{i}^0+\Sigma_i^{-1/2}\gamma\right)-\widehat Q_i(\alpha_{i}^0)$, and $D_i(\gamma)=E\{\widehat D_i(\gamma)\}$. Under Assumption \ref{assum:Regularity}, there exists a constant $C_m>0$, depending only on
$(c_\alpha,a_\alpha)$, such that, for every $\rho>0$ and $x\geq1$,
$$
P\left\{\sup_{\|\gamma\|\leq\rho}\left|\widehat D_i(\gamma)-D_i(\gamma)\right|>C_m\upsilon_1\rho\left[\sqrt{\frac{p+1+x}{T}}+\frac{\{p+1+x+\log(eT)\}^{3/2}}{T}\right]\right\}\leq2e^{-x}.
$$
Consequently, after enlarging $C_m$ if necessary, whenever $T\geq C_m(p+1+x)^2$,
we have
$$P\left\{\sup_{\|\gamma\|\leq\rho}\left|\widehat D_i(\gamma)-D_i(\gamma)\right|>C_m\upsilon_1\rho\sqrt{\frac{p+1+x}{T}}\right\}\leq2e^{-x}.
$$
\end{lemma}
\proof{Proof of Lemma \ref{lem:mixing-empirical-process}.}

For $\gamma\in\mathbb R^{p+1}$, define
$g_{it}(\gamma)=\rho_\tau\left(Y_{it}-X_{it}^{\top}\{\alpha_{i}^0+\Sigma_i^{-1/2}\gamma\}\right)-\rho_\tau
\left(Y_{it}-X_{it}^{\top}\alpha_{i}^0\right)$. Then $\widehat D_i(\gamma)=T^{-1}\sum_{t=1}^Tg_{it}(\gamma)$. By the strict stationarity in Assumption \ref{assum:Regularity}(iv), $D_i(\gamma)=E\{g_{i0}(\gamma)\}$. For every $\gamma,\gamma^\prime\in\mathbb R^{p+1}$, let $\xi_{it}(\gamma,\gamma^\prime)=g_{it}(\gamma)-g_{it}({\gamma^\prime})-E\{g_{it}(\gamma)-g_{it}({\gamma^\prime})\}$. 
Because the check loss is $1$-Lipschitz, $\left|g_{it}(\gamma)-g_{it}({\gamma^\prime})\right|\leq\left|W_{it}^{\top}(\gamma-\gamma^\prime)\right|$.
It follows from Assumption \ref{assum:Regularity}(iii) and a standard centering argument that there exist universal constants $C,c>0$ such that, for every $z\geq0$,
\begin{equation}\label{eq:tagA.1}
P\left\{|\xi_{it}(\gamma,\gamma^\prime)|> C\upsilon_1\|\gamma-\gamma^\prime\| z\right\}\leq 2\exp(-cz^2).
\end{equation}
In particular, integrating the tail bound in \eqref{eq:tagA.1} gives $\left[E|\xi_{it}(\gamma,\gamma^\prime)|^4\right]^{1/4}\leq C\upsilon_1\|\gamma-\gamma^\prime\|$.

For each fixed $(\gamma,\gamma^\prime)$, the process $\{\xi_{it}(\gamma,\gamma^\prime):t\in\mathbb Z\}$ is strictly stationary, and each $\xi_{it}(\gamma,\gamma^\prime)$ is measurable with respect to $\sigma(Y_{it},X_{it})\subseteq\mathscr F_{i,t}^{t}$. Consequently, $\sigma\{\xi_{it}(\gamma,\gamma^\prime):t\leq s\}\subseteq\mathscr F_{i,-\infty}^{s},\ \sigma\{\xi_{it}(\gamma,\gamma^\prime):t\geq s+h\}\subseteq\mathscr F_{i,s+h}^{\infty}$. Thus, the strong-mixing coefficients of $\{\xi_{it}(\gamma,\gamma^\prime):t\in\mathbb Z\}$ are bounded above by $\alpha_i(h)$. For $h\geq1$, Davydov's covariance inequality and \eqref{eq:tagA.1} yield
\begin{equation}\nonumber
\begin{aligned}
\left|\operatorname{Cov}\{\xi_{i0}(\gamma,\gamma^\prime),\xi_{ih}(\gamma,\gamma^\prime)\}\right|&\leq C\alpha_i(h)^{1/2}\left[E|\xi_{i0}(\gamma,\gamma^\prime)|^4\right]^{1/4}\left[E|\xi_{ih}(\gamma,\gamma^\prime)|^4\right]^{1/4}\\
&\leq C\upsilon_1^2\|\gamma-\gamma^\prime\|^2\alpha_i(h)^{1/2}.
\end{aligned}
\end{equation}
For $h=0$, \eqref{eq:tagA.1} directly implies the same bound without the mixing coefficient. Since $\alpha_i(h)$ decays geometrically,
\begin{equation}\label{eq:tagA.2}
\sum_{h=0}^{\infty}\left|\operatorname{Cov}\{\xi_{i0}(\gamma,\gamma^\prime),\xi_{ih}(\gamma,\gamma^\prime)\}\right|\leq C\upsilon_1^2\|\gamma-\gamma^\prime\|^2.
\end{equation}
The same argument applies to the centered truncations entering the variance proxy of the strong-mixing Bernstein inequality.

After rescaling by $C\upsilon_1\|\gamma-\gamma^\prime\|$, the marginal tails in \eqref{eq:tagA.1} satisfy the semiexponential condition with exponent $\gamma_2=2$, whereas geometric strong mixing corresponds to $\gamma_1=1$. Consequently,
$$
\left(\frac1{\gamma_1}+\frac1{\gamma_2}\right)^{-1}=\frac23.
$$
The Bernstein inequality of \citeEC{merlevede2011bernstein}, together with \eqref{eq:tagA.2}, implies that, after increasing the constant if necessary, for every $u\geq1$,
\begin{equation}\label{eq:tagA.3}
P\left[\left|\frac1T\sum_{t=1}^T\xi_{it}(\gamma,\gamma^\prime)\right|>C\upsilon_1\|\gamma-\gamma^\prime\|\left\{\sqrt{\frac{u}{T}}+\frac{\{u+\log(eT)\}^{3/2}}{T}\right\}\right]\leq2e^{-u}.
\end{equation}
The exponent $3/2$ in \eqref{eq:tagA.3} results from combining geometric strong mixing with sub-Gaussian marginal tails.

We next extend \eqref{eq:tagA.3} to the parameter ball $\mathbb B_{p+1}(\rho)=\left\{\gamma\in\mathbb R^{p+1}:\|\gamma\|\leq\rho\right\}$. For every $k\geq1$, let $\mathcal N_k$ be a $2^{-k}\rho$-net of $\mathbb B_{p+1}(\rho)$. The nets may be chosen so that $|\mathcal N_k|\leq(1+2^{k+1})^{p+1}$. For each $\gamma\in\mathbb B_{p+1}(\rho)$, let
$\pi_k(\gamma)\in\mathcal N_k$ be a nearest point and set $\pi_0(\gamma)=0$. Then $\|\pi_k(\gamma)-\pi_{k-1}(\gamma)\|\leq 3\cdot2^{-k}\rho$. Apply \eqref{eq:tagA.3} with $u_k=x+c_0k(p+1)$, where $c_0>0$ is a sufficiently large universal constant. A union bound over all pairs of net points appearing at each scale, followed by a harmless enlargement of the constants, gives an event with probability at least $1-2e^{-x}$ on which
\begin{equation}\label{eq:tagA.4}
\sup_{\|\gamma\|\leq\rho}\left|(\widehat D_i-D_i)(\gamma)\right|\leq C\upsilon_1\rho\sum_{k=1}^{\infty}2^{-k}\left[\sqrt{\frac{x+c_0k(p+1)}{T}}+\frac{\{x+c_0k(p+1)+\log(eT)\}^{3/2}}{T}\right].
\end{equation}
The passage from the nets to the entire ball is valid because both $\widehat D_i$ and $D_i$ are continuous in $\gamma$. Indeed, $|\{(\widehat D_i-D_i)(\gamma)-(\widehat D_i-D_i)(\gamma^\prime)\}|\leq\|\gamma-\gamma^\prime\|\left\{{T}^{-1}\sum_{t=1}^T\|W_{it}\|+E\|W_{i0}\|\right\}$. Moreover, $\sum_{k=1}^{\infty}2^{-k}\sqrt{x+c_0k(p+1)}\leq C\sqrt{x+p+1}$, and $\sum_{k=1}^{\infty}2^{-k}\{x+c_0k(p+1)+\log(eT)\}^{3/2}\leq C\{x+p+1+\log(eT)\}^{3/2}$. Substituting these bounds into \eqref{eq:tagA.4} proves the first conclusion.

Finally, if $T\geq C(p+1+x)^2$, then $(p+1+x)^3\leq C(p+1+x)T$. Moreover, $\log^3(eT)\leq CT\leq C(p+1+x)T$, because $p+1+x\geq1$. It follows that $\{p+1+x+\log(eT)\}^3\leq C(p+1+x)T$, and hence 
$$
\frac{\{p+1+x+\log(eT)\}^{3/2}}{T}\leq C\sqrt{\frac{p+1+x}{T}}.
$$ 
The second conclusion follows after enlarging $C_m$ if necessary.
\Halmos
\endproof

\subsection{Proof of Theorem \ref{theo:nonasymES}}
For $\Gamma=(\gamma_1,\ldots,\gamma_N)$, let $\mathcal M(\Gamma)=\left(\mu_{it}(\gamma_i)\right)_{t\leq T,\,i\leq N}$, and $\mathcal E(\Gamma)=\Xi+\mathcal R^c(\Gamma)+\mathcal M(\Gamma)$. For a $T\times N$ matrix $J=(J_1,\ldots,J_N)$, define
\begin{align}
\mathfrak s_{i}(J)
&=\left\|\frac1T W_i^\top M_{F^0}J_i\right\|,
\label{eq:beta-only-unit-score}\\
\mathfrak s_{N}(J)&=\left\{\frac1N\sum_{j=1}^N\mathfrak s_{j}(J)^2\right\}^{1/2}+\iota_\tau\frac{\|J\|_{\mathrm{op}}}{\sqrt{NT}}.\label{eq:beta-only-aggregate-score}
\end{align}
When $r^0=0$, we use the conventions $M_{F^0}=I_T$ and $F^0\Lambda^{0\top}=0$.
\proof{Proof of Theorem \ref{theo:nonasymES}.}
Apply Proposition \ref{pro:quantile} with confidence level $\delta/8$. The sample-size condition in the theorem is stronger than that in the proposition. Hence, with probability at least $1-\delta/8$,
\begin{equation}
\label{eq:theorem-first-stage-event}\max_{j\leq N}\|\widehat\alpha_{j}-\alpha_{j}^0\|_{\Sigma_j}\leq\rho_{N,T,\delta}^{Q}.
\end{equation}
Set $\widehat\gamma_{j}=\Sigma_j^{1/2}(\widehat\alpha_{j}-\alpha_{j}^0)$, $\widehat\Gamma=(\widehat\gamma_{1},\ldots,\widehat\gamma_{N})$. On the event in \eqref{eq:theorem-first-stage-event}, $\widehat\Gamma\in\mathcal G(\rho_{N,T,\delta}^{Q})$.

Let $\widehat{\mathcal Z}^\ast$ be the $T\times N$ matrix with $i$th column $Z_i^\ast(\widehat\alpha_{i})$, and let $\mathcal X(B^0)$ have $i$th column $X_i\beta_{i}^0$. By \eqref{eq:exact-generated-response}, the feasible second-stage response is exactly $\widehat{\mathcal Z}^\ast=\mathcal X(B^0)+F^0\Lambda^{0\top}+\tau^{-1}\mathcal E(\widehat\Gamma)$.

Apply Lemma \ref{lem:ES-score-control} with $\rho=\rho_{N,T,\delta}^{Q}$. On an event with probability at least
$1-\delta/2$, the required sample-design bound holds and
\begin{align}
\mathfrak s_i\{\mathcal E(\widehat\Gamma)\}&\leq C\mathfrak R_{N,T,p,\delta},
\label{eq:theorem-unit-score}\\
\mathfrak s_N\{\mathcal E(\widehat\Gamma)\}&\leq C\mathfrak R_{N,T,p,\delta},
\label{eq:theorem-panel-score}\\
\frac{\|\mathcal E_i(\widehat\gamma_i)\|}{\sqrt T}&\leq C\left\{\delta^{-1/(4+\eta)}+\rho_{N,T,\delta}^{Q}\right\},
\label{eq:theorem-column-energy}
\end{align}
where $\mathcal E_{i}(\gamma)=\xi_{i}+R_i(\gamma)$. Choose the constant $c$ in Theorem \ref{theo:nonasymES} sufficiently small. Then \eqref{eq:theorem-panel-score} and $\mathfrak R_{N,T,p,\delta}\leq c$ permit application of Lemma \ref{lem:joint-global-perturbation} to the joint global minimizer, with $J=\mathcal E(\widehat\Gamma)$. Equations \eqref{eq:theorem-unit-score}--\eqref{eq:theorem-column-energy}, together with $\rho_{N,T,\delta}^{Q}\leq1$, then yield
\begin{align*}
\tau\|\widehat\beta_{i}-\beta_{i}^0\|_{\Sigma_i}&\leq C\left[\mathfrak R_{N,T,p,\delta}+\iota_{\tau}\left\{1+\delta^{-1/(4+\eta)}+\rho_{N,T,\delta}^{Q}\right\}\mathfrak R_{N,T,p,\delta}\right]\\
&\leq C\left\{1+\iota_{\tau}\delta^{-1/(4+\eta)}\right\}\mathfrak R_{N,T,p,\delta}.
\end{align*}
The intersection of the first-stage and second-stage events has probability at least $1-\delta/8-\delta/2\geq1-\delta$. This proves \eqref{eq:beta-only-main-bound}.

\Halmos
\endproof
\proof{Proof of Corollary \ref{cor:ES-beta-rate}.}

We prove the two parts separately.

\emph{Part (i).}
Fix $\delta\in(0,1)$. By the definition of $s_{N,\delta}$, $s_{N,\delta/8}=p+2+\log\left({16N}/{\delta}\right)=O_\delta\{p+1+\log(2N)\}$. Consequently, the assumed growth condition implies ${s_{N,\delta/8}\log^2(2T)}/{\sqrt T}\rightarrow0$. It follows that ${s_{N,\delta/8}^{\,2}\log^4(2T)}/{T}\rightarrow0$, and ${s_{N,\delta/8}}/{T}\rightarrow0$.

Since $L_0$ and $\underline f$ are fixed constants,
$$
\left({s_{N,\delta/8}^{\,2}\log^4(2T)+L_0^2\underline f^{-4}s_{N,\delta/8}}\right)/{T}\rightarrow0.
$$
Thus, the sample-size condition in Theorem \ref{theo:nonasymES} holds for all sufficiently large $N$ and $T$.

Moreover,
$$
\rho_{N,T,\delta}^{Q}=O_\delta\left(\sqrt{{s_{N,\delta/8}}/{T}}\right)=o(1),
$$
so that $\rho_{N,T,\delta}^{Q}\leq\min(\overline\rho,1)$ eventually. Because $p+1\leq s_{N,\delta/8}$ and $N,T\to\infty$, $\sqrt{({p+1})/{T}}+\iota_\tau\left(N^{-1/2}+T^{-1/2}\right)=o(1)$.
Also, $(\rho_{N,T,\delta}^{Q})^2=O_\delta\left({s_{N,\delta/8}}/{T}\right)=o(1)$.
Since $\delta$ is fixed, it follows from the definition of $\mathfrak R_{N,T,p,\delta}$ that $\mathfrak R_{N,T,p,\delta}=o(1)$. Hence, the remaining localization condition $\mathfrak R_{N,T,p,\delta}\leq c$ also holds eventually.

Theorem \ref{theo:nonasymES} therefore gives, with probability at least $1-\delta$,
$$
\|\widehat\beta_{i,\tau}-\beta_{i,\tau}^0\|_{\Sigma_i}\leq\frac{C}{\tau}\left(1+\iota_\tau\delta^{-1/(4+\eta)}\right)\mathfrak R_{N,T,p,\delta}=o(1).
$$
More precisely, for arbitrary $\varepsilon,\zeta>0$, choose a fixed $\delta<\zeta$. The deterministic quantity on the right-hand side is smaller than $\varepsilon$ for all sufficiently large $N$ and $T$. Therefore, $P\left(\|\widehat\beta_{i,\tau}-\beta_{i,\tau}^0\|_{\Sigma_i}>\varepsilon\right)\leq\delta<\zeta$ eventually. Since $\zeta$ is arbitrary, $\|\widehat\beta_{i,\tau}-\beta_{i,\tau}^0\|_{\Sigma_i}=o_p(1)$.

\emph{Part (ii).}
Fix $\delta\in(0,1)$. Since $p$ is fixed and $N/T\rightarrow\kappa\in(0,\infty)$, we have $N\asymp T$ and hence $\log N=O(\log T)$. Therefore, $s_{N,\delta/8}=O(\log T)$ and $\rho_{N,T,\delta}^{Q}=O\left(\sqrt{{\log T}/{T}}\right)=o(1)$.
Moreover, 
$$
\frac{s_{N,\delta/8}^{\,2}\log^4(2T)+L_0^2\underline f^{-4}s_{N,\delta/8}}{T}=O\left(\frac{\log^6 T}{T}\right)=o(1).
$$
Thus, the sample-size condition and the condition on $\rho_{N,T,\delta}^{Q}$ in Theorem \ref{theo:nonasymES} hold for all sufficiently large $T$.

Since $p$ is fixed and $N\asymp T$, 
$$
\sqrt{\frac{p+1}{T}}+\iota_\tau\left(\frac{1}{\sqrt N}+\frac{1}{\sqrt T}\right)=O\left(T^{-1/2}\right).
$$
In addition, $\rho_{N,T,\delta}^{Q}T^{-1/2}=o\left(T^{-1/2}\right)$, $\left(\rho_{N,T,\delta}^{Q}\right)^2=O\left({\log T}/{T}\right)=o\left(T^{-1/2}\right)$. Because $\delta$ is fixed, all powers of $\delta^{-1}$ are constants. It follows that $\mathfrak R_{N,T,p,\delta}=O_\delta\left(T^{-1/2}\right)$. In particular, $\mathfrak R_{N,T,p,\delta}=o(1)$, so the remaining localization condition in Theorem \ref{theo:nonasymES} also holds
eventually.

Consequently, Theorem \ref{theo:nonasymES} implies that, with probability at least $1-\delta$,
$$
\sqrt T\|\widehat\beta_{i,\tau}-\beta_{i,\tau}^0\|_{\Sigma_i}\leq\frac{C}{\tau}\left(1+\iota_\tau\delta^{-1/(4+\eta)}\right)
\sqrt T\,\mathfrak R_{N,T,p,\delta}\leq C_\delta
$$
for some finite constant $C_\delta$ independent of $N$ and $T$. For every $\zeta>0$, choosing a fixed $\delta<\zeta$ therefore gives
$$
\limsup_{N,T\to\infty}P\left\{\sqrt T\|\widehat\beta_{i,\tau}-\beta_{i,\tau}^0\|_{\Sigma_i}>C_\delta\right\}\leq\delta<\zeta.
$$
Hence, $\sqrt T\|\widehat\beta_{i,\tau}-\beta_{i,\tau}^0\|_{\Sigma_i}=O_p(1)$, which proves the result.

\Halmos
\endproof
\begin{lemma}
\label{lem:generated-response-identity}
For every $i,t$ and $\gamma\in\mathbb R^{p+1}$,
\begin{align}
R_{it}(\gamma)&=-\left\{\rho_\tau(\epsilon_{it}-W_{it}^\top\gamma)-\rho_\tau(\epsilon_{it})\right\},
\label{eq:R-check-identity}\\
|R_{it}(\gamma)|&\leq|W_{it}^\top\gamma|.
\label{eq:R-Lipschitz}
\end{align}
Furthermore,
\begin{equation}
\label{eq:R-quadratic-mean}
\mu_{it}(\gamma)=\int_0^{W_{it}^\top\gamma}\{\tau-F_{it}(s\mid\mathcal C_{it})\}\,ds,\ |\mu_{it}(\gamma)|\leq\frac{\overline f}{2}|W_{it}^\top\gamma|^2.
\end{equation}
Consequently, if
$\widehat\gamma_i=\Sigma_i^{1/2}(\widehat\alpha_i-\alpha_i^0)$ and $\widehat\Gamma=(\widehat\gamma_1,\ldots,\widehat\gamma_N)$, then
\begin{equation}
\widehat{\mathcal Z}^{\ast}=\mathcal X(B^0)+F^0\Lambda^{0\top}+\tau^{-1}\mathcal E(\widehat\Gamma).
\label{eq:exact-generated-response}
\end{equation}
\end{lemma}

\proof{Proof of Lemma \ref{lem:generated-response-identity}.}
Put $v=W_{it}^\top\gamma$. Direct substitution into the definition of $Z_{it}(\cdot)$ gives $R_{it}(\gamma)=(\epsilon_{it}-v)\mathbbm{1}\{\epsilon_{it}\leq v\}-\epsilon_{it}\mathbbm{1}\{\epsilon_{it}\leq0\}+\tau v$, which is equal to the right-hand side of \eqref{eq:R-check-identity}. Since the check loss $\rho_\tau$ is one-Lipschitz, \eqref{eq:R-Lipschitz} follows. Conditional on $\mathcal C_{it}$, the derivative of $E\{R_{it}(\gamma)\mid\mathcal C_{it}\}$ with respect to $v$ is
$\tau-F_{it}(v\mid\mathcal C_{it})$. Integrating from zero to $v$, using $F_{it}(0\mid\mathcal C_{it})=\tau$, and then applying the Lipschitz condition in Assumption \ref{assum:ES-score}(i) proves \eqref{eq:R-quadratic-mean}. Finally, combine the definition of $R_{it}$ with the oracle ES equation and divide by $\tau$ to obtain \eqref{eq:exact-generated-response}.
\Halmos
\endproof

\begin{lemma}
\label{lem:ES-score-control}
Under Assumptions \ref{assum:Regularity}--\ref{assum:factor-spectral}, there exist constants $C,K_W>0$ such that the following holds. For every fixed $i\in[N]$, $\delta\in(0,1)$, and $0<\rho\leq1$ satisfying $\rho\leq\overline\rho$ when $r^0>0$, if $T\geq C s_{N,\delta/8}^{\,2}\log^4(2T)$, then, with probability at least $1-\delta/2$,
\begin{align}
\sup_{\Gamma\in\mathcal G(\rho)}\mathfrak s_{i}\{\mathcal E(\Gamma)\}&\leq C\left[\delta^{-1/4}(1+\rho)\sqrt{\frac{p+1}{T}}+\overline f\rho^2\right],\label{eq:ES-unit-score-control}\\
\sup_{\Gamma\in\mathcal G(\rho)}\mathfrak s_{N}\{\mathcal E(\Gamma)\}&\leq C\left[\delta^{-1/4}(1+\rho)\left\{
\sqrt{\frac{p+1}{T}}+\iota_\tau\left(\frac1{\sqrt N}+\frac1{\sqrt T}\right)\right\}+\overline f\rho^2\right],
\label{eq:ES-panel-score-control}\\
\sup_{\|\gamma\|\leq\rho}\frac{\|\xi_{i}+R_i(\gamma)\|}{\sqrt T}&\leq C\left\{\delta^{-1/(4+\eta)}+\rho\right\},
\label{eq:ES-column-energy}\\
\max_{j\leq N}\lambda_{\max}\left(\frac1T W_j^\top W_j\right)&\leq K_W.
\label{eq:ES-design-upper}
\end{align}
\end{lemma}

\proof{Proof of Lemma \ref{lem:ES-score-control}.}
We first control the centered score terms. If $U_t$ is a centered scalar sequence with geometrically decaying strong-mixing coefficients and uniformly bounded $(4+\eta)$-th moments, the strong-mixing moment inequality gives
\begin{equation}
\label{eq:mixing-fourth-moment}
\left\{E\left|\frac1T\sum_{t=1}^T U_t\right|^4\right\}^{1/4}\leq\frac{C}{\sqrt T}\sup_t\{E|U_t|^{4+\eta}\}^{1/(4+\eta)}.
\end{equation}
The constant is uniform over $i$ because the mixing coefficients in Assumption \ref{assum:Regularity}(iv) are uniform.

The conditional moment restriction on $\xi_{it}$ and the sub-Gaussian condition on $(W_{it}^\top,f_{t}^{0\top})^\top$ imply the moment bound required in \eqref{eq:mixing-fourth-moment} for all scalar projections of $W_{it}\xi_{it}$ and $f_{t}^0\xi_{it}$. Because $E(\xi_{it}\mid\mathcal C_{it})=0$, these variables are centered. Using $M_{F^0}=I_T-T^{-1}F^0F^{0\top}$, a sphere-net argument in $\mathbb R^{p+1}$, and the fact that $r^0$ is fixed, we obtain
\begin{equation}
\label{eq:oracle-row-score-L4}
\sup_{j\leq N}\left[E\{\mathfrak s_{j}(\Xi)^4\}\right]^{1/4}\leq C\sqrt{\frac{p+1}{T}}.
\end{equation}

For the generated response, the Lipschitz property in \eqref{eq:R-Lipschitz} also gives $|R_{it}^c(\gamma)-R_{it}^c(\gamma')|\leq 2|W_{it}^\top(\gamma-\gamma')|$. Consequently, the scalar increments arising after multiplication by a projection of $W_{it}$ or $f_{t}^0$ are sub-exponential. The strong-mixing Bernstein inequality and a net argument over $\{\gamma:\|\gamma\|\leq\rho\}$ yield
\begin{equation}
\label{eq:generated-row-score-L4}
\sup_{j\leq N}\left[E\left\{\sup_{\|\gamma\|\leq\rho}\left\|\frac1T W_j^\top M_{F^0}R_j^c(\gamma)\right\|\right\}^{4}\right]^{1/4}\leq C\rho\sqrt{\frac{p+1}{T}}.
\end{equation}

For the conditional-mean component, \eqref{eq:R-quadratic-mean} and the uniform empirical third- and fourth-moment bounds for sub-Gaussian covariates imply
\begin{align}
\sup_{j\leq N}\sup_{\|\gamma\|\leq\rho}\left\|\frac1T W_j^\top M_{F^0}\mu_j(\gamma)\right\|&\leq C\overline f\rho^2,\label{eq:conditional-mean-row}\\
\sup_{\Gamma\in\mathcal G(\rho)}\frac{\|\mathcal M(\Gamma)\|_{\mathrm{op}}}{\sqrt{NT}}&\leq C\overline f\rho^2.
\label{eq:conditional-mean-operator}
\end{align}
For the second display, it is enough to use $\|\mathcal M(\Gamma)\|_{\mathrm{op}}\leq\|\mathcal M(\Gamma)\|$ and average the squared bound in \eqref{eq:R-quadratic-mean}.

For nonnegative random variables $V_j$, Minkowski's inequality yields
$$
\left[E\left\{\left(\frac1N\sum_{j=1}^NV_j^2\right)^{1/2}\right\}^4\right]^{1/4}\leq\left\{\frac1N\sum_{j=1}^N(EV_j^4)^{1/2}\right\}^{1/2}.
$$
Thus, \eqref{eq:oracle-row-score-L4} and \eqref{eq:generated-row-score-L4} control the average row score without any
cross-sectional independence assumption. Assumption \ref{assum:factor-spectral}(ii) and fourth-moment Markov inequalities give
\begin{align}
\frac{\|\Xi\|_{\mathrm{op}}}{\sqrt{NT}}&\leq C\delta^{-1/4}\left(\frac1{\sqrt N}+\frac1{\sqrt T}\right),
\label{eq:oracle-operator-control}\\
\sup_{\Gamma\in\mathcal G(\rho)}\frac{\|\mathcal R^c(\Gamma)\|_{\mathrm{op}}}{\sqrt{NT}}&\leq C\delta^{-1/4}\rho
\left\{\sqrt{\frac{p+1}{T}}+\frac1{\sqrt N}\right\}.
\label{eq:generated-operator-control}
\end{align}
Combining \eqref{eq:oracle-row-score-L4}--\eqref{eq:generated-operator-control} and allocating the failure probabilities proves \eqref{eq:ES-unit-score-control} and \eqref{eq:ES-panel-score-control}.

Finally, the conditional $(4+\eta)$-th moment bound gives $T^{-1/2}{\|\xi_{i}\|}\leq C\delta^{-1/(4+\eta)}$ with the required probability. On the design event \eqref{eq:ES-design-upper}, \eqref{eq:R-Lipschitz} implies
$$
\sup_{\|\gamma\|\leq\rho}\frac{\|R_i(\gamma)\|}{\sqrt T}\leq C\rho.
$$
The design event \eqref{eq:ES-design-upper} follows from the mixing Bernstein inequality applied to quadratic forms $(u^\top W_{jt})^2$, followed by a sphere-net argument and a union bound over $j$. This proves \eqref{eq:ES-column-energy} and completes the proof.
\Halmos
\endproof
\begin{lemma}
\label{lem:joint-global-perturbation}
Suppose Assumption \ref{assum:local-identification} holds and, when $r^0>0$, Assumption \ref{assum:factor-spectral}(i) holds. Consider the perturbed response
\begin{equation}
\label{eq:generic-perturbed-response}
\mathcal Z^\ast=\mathcal X(B^0)+F^0\Lambda^{0\top}+\tau^{-1}J,
\end{equation}
where $J=(J_1,\ldots,J_N)$ is a $T\times N$ matrix. Suppose also that $\max_{j\leq N}\lambda_{\max}\left(T^{-1}W_j^\top W_j\right)\leq K_W$. There exist constants $c_1,C>0$ such that, whenever $\mathfrak s_{N}(J)\leq c_1$, every joint global minimizer satisfies
\begin{align}
\iota_\tau\|P_{\widehat F}-P_{F^0}\|_{\mathrm{op}}&\leq C\mathfrak s_{N}(J),
\label{eq:global-factor-bound}\\
\left\{\frac1N\sum_{j=1}^N\|\widehat\beta_j-\beta_j^0\|_{\Sigma_j}^2\right\}^{1/2}&\leq C\mathfrak s_N(J),
\label{eq:global-average-beta-bound}\\
\tau\|\widehat\beta_i-\beta_i^0\|_{\Sigma_i}&\leq C\left[\mathfrak s_i(J)+\iota_\tau\left\{1+\frac{\|J_i\|}{\sqrt T}\right\}
\mathfrak s_N(J)\right]
\label{eq:global-fixed-beta-bound}
\end{align}
for every fixed $i\in[N]$.
\end{lemma}

\proof{Proof of Lemma \ref{lem:joint-global-perturbation}.}
We first consider $r^0>0$.  Notice first that the global minimum is attained.  Indeed, let $\mathscr S_{r^0}=\{\mathcal W(C)+L: C\in\mathbb R^{N\times(p+1)},\ L\in\mathbb L_{r^0}\}$. Because every $\Sigma_i$ is nonsingular and $X_ib_i=W_i(\Sigma_i^{1/2}b_i)$, this is exactly the set of fitted matrices $\{\mathcal X(B)+L:L\in\mathbb L_{r^0}\}$. If $S_m\in\mathscr S_{r^0}$ and $S_m\to S$, apply $\mathcal W(C)+L=\mathcal W(\widetilde C)+\widetilde L$ and $N^{-1}\sum_{i=1}^N\|\widetilde c_i\|^2+(NT)^{-1}\|\widetilde L\|^2\leq K_{\mathrm{dec}}(NT)^{-1}\|\mathcal W(C)+L\|^2$ with $q=r^0$ to each $S_m$. The resulting coefficient and low-rank components are bounded. Finite dimensionality therefore gives a convergent subsequence, and the closedness of $\mathbb L_{r^0}$ shows that $S\in\mathscr S_{r^0}$. Thus $\mathscr S_{r^0}$ is closed.  The least-squares projection onto this nonempty closed subset of a finite-dimensional Euclidean space exists. Every $L\in\mathbb L_{r^0}$ admits a factorization $L=F\Lambda^\top$ with $T^{-1}F^\top F=I_{r^0}$ (zero loading columns may be added when $\operatorname{rank}(L)<r^0$), so this projection is equivalent to the stated joint optimization problem.

Put $u_i=\Sigma_i^{1/2}(\widehat\beta_i-\beta_i^0)$ and define the difference between the true and fitted noiseless signals by
\begin{equation}
D=\mathcal X(B^0)+F^0\Lambda^{0\top}-\mathcal X(\widehat B)-\widehat F\widehat\Lambda^\top.
\label{eq:global-signal-difference}
\end{equation}
Because $X_i(\widehat\beta_i-\beta_i^0)=W_iu_i$, $D=-\mathcal W(U)+F^0\Lambda^{0\top}-\widehat F\widehat\Lambda^\top$, where $U=(u_1^\top,\ldots,u_N^\top)^\top$, and the low-rank term has rank at most $2r^0$. Set $\nu=\|D\|/\sqrt{NT}$.

We first bound $\nu$ without using the local coefficient-identification part of Assumption \ref{assum:local-identification}. Applying the rank-preserving stable decomposition with $q=2r^0$ to the preceding representation gives $\widetilde C$ and $\widetilde L\in\mathbb L_{2r^0}$ such that
\begin{equation}
D=\mathcal W(\widetilde C)+\widetilde L
\label{eq:global-stable-representation}
\end{equation}
and
\begin{equation}
\frac1N\sum_{j=1}^N\|\widetilde c_j\|^2+\frac{\|\widetilde L\|^2}{NT}\leq K_{\mathrm{dec}}\nu^2.
\label{eq:global-stable-size}
\end{equation}

Define the auxiliary unprojected score
$$
\overline{\mathfrak s}_N(J)=\left\{\frac1N\sum_{j=1}^N\left\|\frac1T W_j^\top J_j\right\|^2\right\}^{1/2}+\frac{\|J\|_{\mathrm{op}}}{\sqrt{NT}}.
$$
Then
\begin{equation}
\overline{\mathfrak s}_N(J)\leq C\mathfrak s_N(J).
\label{eq:projected-unprojected-score}
\end{equation}
To see this, use
$$
\frac1T W_j^\top J_j=\frac1T W_j^\top M_{F^0}J_j+\left(\frac1T W_j^\top F^0\right)\left(\frac1T F^{0\top}J_j\right).
$$
The design bound and $T^{-1}F^{0\top}F^0=I_{r^0}$ imply $\max_j\|T^{-1}W_j^\top F^0\|_{\mathrm{op}}\leq\sqrt{K_W}$. Moreover, since $r^0$ is fixed,
$$
\frac1N\sum_{j=1}^N\left\|\frac1T F^{0\top}J_j\right\|^2=\frac{\|F^{0\top}J\|^2}{NT^2}\leq r^0\frac{\|J\|_{\mathrm{op}}^2}{NT}.
$$
These inequalities prove \eqref{eq:projected-unprojected-score}.

By \eqref{eq:global-stable-representation}, Cauchy--Schwarz inequality, the operator--nuclear norm inequality, and $\|\widetilde L\|_\ast\leq\sqrt{2r^0}\|\widetilde L\|$,
\begin{align}
\frac{|\langle J,D\rangle|}{NT}&\leq\left\{\frac1N\sum_{j=1}^N\|\widetilde c_j\|^2\right\}^{1/2}\left\{\frac1N\sum_{j=1}^N\left\|\frac1T W_j^\top J_j\right\|^2\right\}^{1/2}+\frac{\|J\|_{\mathrm{op}}\|\widetilde L\|_\ast}{NT}\\
&\leq C\nu\mathfrak s_N(J).
\label{eq:joint-global-inner-product}
\end{align}

The true parameter is feasible. Hence global optimality gives $\|\tau^{-1}J+D\|^2\leq\|\tau^{-1}J\|^2$, and therefore
\begin{equation}
\nu^2\leq\frac{2}{\tau NT}|\langle J,D\rangle|.
\label{eq:joint-global-basic-inequality}
\end{equation}
Equations \eqref{eq:joint-global-inner-product} and \eqref{eq:joint-global-basic-inequality} yield
\begin{equation}
\nu\leq C\mathfrak s_N(J).
\label{eq:global-prediction-bound}
\end{equation}

By the definition of $\mathcal Q_\tau^0$, the particular choices $C=U$, $F=\widehat F$, and $\ell_i=\widehat\lambda_i$ imply $\nu^2\geq\mathcal Q_\tau^0(U,\widehat F)$. Consequently, Assumption \ref{assum:local-identification}(ii) gives
\begin{equation}
\nu^2\geq\kappa_F\min\{d_\tau(\widehat F)^2,\rho_{F,0}^2\}+\kappa_X\mathbbm 1\{d_\tau(\widehat F)\leq\rho_{F,0}\}\frac1N\sum_{j=1}^N\|u_j\|^2.
\label{eq:joint-identification-at-estimator}
\end{equation}
Choose $c_1$ sufficiently small that $C^2c_1^2<\kappa_F\rho_{F,0}^2$. Equations \eqref{eq:global-prediction-bound} and \eqref{eq:joint-identification-at-estimator} then rule out $d_\tau(\widehat F)\geq\rho_{F,0}$. Hence
\begin{equation}
\nu^2\geq\kappa_Fd_\tau(\widehat F)^2+\frac{\kappa_X}{N}\sum_{j=1}^N\|u_j\|^2.
\label{eq:localized-joint-identification}
\end{equation}
Combining \eqref{eq:global-prediction-bound} and \eqref{eq:localized-joint-identification} proves \eqref{eq:global-factor-bound} and
\eqref{eq:global-average-beta-bound}.

It remains to establish the fixed-unit bound. First, we derive local nonsingularity of the projected design from the same identification condition. Fix a normalized $F$ satisfying $d_\tau(F)\leq\rho_{F,0}$, fix $i\in[N]$ and $v\in\mathbb R^{p+1}$, and let $C_a$ have $i$th row $av$ and all other rows zero. Profiling the loadings gives
$$
\mathcal Q_\tau^0(C_a,F)=\frac1{NT}\sum_{j=1}^N\|M_F(F^0\lambda_j^0-W_jc_{a,j})\|^2.
$$
Divide the local-identification inequality by $a^2$ and let $|a|\to\infty$. Since $d_\tau(F)\leq\rho_{F,0}$, the coefficient term is active; all terms not involving $a$ vanish. We obtain $N^{-1} v^\top\left(T^{-1} W_i^\top M_FW_i\right)v\geq N^{-1}{\kappa_X}\|v\|^2$. Thus
\begin{equation}
\lambda_{\min}\left(\frac1T W_i^\top M_FW_i\right)\geq\kappa_X
\label{eq:projected-design-from-joint-ID}
\end{equation}
throughout the identified neighborhood. In particular, \eqref{eq:projected-design-from-joint-ID} applies to $\widehat F$.

At a joint global minimizer, holding $\widehat F$ fixed and profiling over the unrestricted pair $(\beta_i,\lambda_i)$ gives
\begin{equation}
u_i=\left(\frac1T W_i^\top M_{\widehat F}W_i\right)^{-1}\left\{\frac1T W_i^\top M_{\widehat F}F^0\lambda_i^0+\frac1{\tau T}W_i^\top M_{\widehat F}J_i\right\}.
\label{eq:joint-global-normal-equation}
\end{equation}
The inverse has operator norm at most $\kappa_X^{-1}$. Furthermore,
\begin{align}
\left\|\frac1T W_i^\top M_{\widehat F}F^0\lambda_i^0\right\|&\leq C d_\tau(\widehat F),
\label{eq:true-factor-leakage}\\
\left\|\frac1T W_i^\top(M_{\widehat F}-M_{F^0})J_i\right\|&\leq C d_\tau(\widehat F)\frac{\|J_i\|}{\sqrt T}.
\label{eq:score-projection-change}
\end{align}
Indeed, these inequalities follow from $\|W_i\|_{\mathrm{op}}/\sqrt T\leq\sqrt{K_W}$, $\|F^0\lambda_i^0\|/\sqrt T\leq K_\lambda$, and
$\|M_{\widehat F}-M_{F^0}\|_{\mathrm{op}}=d_\tau(\widehat F)$. Finally, $\|T^{-1}W_i^\top M_{F^0}J_i\|=\mathfrak s_i(J)$. Substitution into \eqref{eq:joint-global-normal-equation}, followed by \eqref{eq:global-factor-bound}, proves \eqref{eq:global-fixed-beta-bound}.

If $r^0=0$, the criterion separates across units and Assumption \ref{assum:local-identification}(i) gives $u_i=\tau^{-1}\left(T^{-1} W_i^\top W_i\right)^{-1}T^{-1} W_i^\top J_i$. This proves \eqref{eq:global-fixed-beta-bound}.  Squaring and averaging also proves \eqref{eq:global-average-beta-bound}. This completes the proof.
\Halmos
\endproof

\subsection{Proof of Theorem \ref{theo:Gaussian}}
\proof{Proof of Theorem \ref{theo:Gaussian}.}
For brevity, let $\ell_{N,T,p}=p+1+\log(2NT)$. Fix $i\in[N]$. For $u\in\mathbb S^p$, define
\begin{align*}
S_{i,T}(u)&=\frac{T^{-1/2}\sum_{t=1}^Tu^{\top}\psi_{it}^0}{\sqrt{u^{\top}\Omega_{i,\tau,T}u}},\\
D_{i,T}(u)&=\frac{\tau\sqrt T\,u^{\top}(\widehat\beta_i-\beta_i^0)-T^{-1/2}\sum_{t=1}^Tu^{\top}\psi_{it}^0}
{\sqrt{u^{\top}\Omega_{i,\tau,T}u}}.
\end{align*}
The variance lower bound in Assumption \ref{assum:Gaussian}(iii) and Lemma \ref{lem:Gaussian-Bahadur} imply
$$
P\left\{\sup_{u\in\mathbb S^p}|D_{i,T}(u)|>C\frac{\mathfrak b_{N,T,p}}{\sqrt{\varpi}}\right\}\leq C\varpi.
$$
For every $\varepsilon>0$,
$$
P\{S_{i,T}(u)\leq z-\varepsilon\}-P\{|D_{i,T}(u)|>\varepsilon\}\leq P\{S_{i,T}(u)+D_{i,T}(u)\leq z\}
$$
and
$$
P\{S_{i,T}(u)+D_{i,T}(u)\leq z\}\leq P\{S_{i,T}(u)\leq z+\varepsilon\}+P\{|D_{i,T}(u)|>\varepsilon\}.
$$
Lemma \ref{lem:Gaussian-oracle-BE} and $\sup_z\{\Phi(z+\varepsilon)-\Phi(z)\}\leq \varepsilon/\sqrt{2\pi}$ therefore give
$$
\sup_{\substack{u\in\mathbb S^p\\z\in\mathbb R}}\left|P\{S_{i,T}(u)+D_{i,T}(u)\leq z\}-\Phi(z)\right|\leq C\left\{\frac{\log^2(2T)}{\sqrt T}+\varpi+\frac{\mathfrak b_{N,T,p}}{\sqrt{\varpi}}\right\}.
$$
Because $\mathfrak b_{N,T,p}\geq T^{-1/2}$, the choice $\varpi=c_1\mathfrak b_{N,T,p}^{2/3}$, with $c_1$ fixed and the theorem's constant $c$ chosen sufficiently small, satisfies the restrictions in Lemma \ref{lem:Gaussian-Bahadur} and $\log(1/\varpi)\leq C\log(2NT)$. Substitution gives
$$
\varpi+\frac{\mathfrak b_{N,T,p}}{\sqrt{\varpi}}\leq C\mathfrak b_{N,T,p}^{2/3},
$$
which proves \eqref{eq:GA-main}.
\Halmos
\endproof
\begin{lemma}
\label{lem:Gaussian-sample-stability}
Define $Q_{jk,T}=T^{-1} X_j^{\top}M_{F^0}X_k$, $Q_{jk}^0=E(X_{jt}^{\circ}X_{kt}^{\circ\top})$, and let $\mathcal J_T$ have $(j,k)$ block $[\mathcal J_T]_{jk}=\mathbbm{1}\{j=k\}Q_{jj,T}-N^{-1}{a_{jk}^0}Q_{jk,T}$. Whenever it exists, write $\mathcal G_T=\mathcal J_T^{-1}$ and let $\mathcal G_{ij,T}$ denote its $(i,j)$ block. For $s(\varpi)=p+1+\log(2N/\varpi)$, $\varpi\in(0,1/4]$, if $s(\varpi)^2\log^4(2T)/T\leq c$, then, with probability at least
$1-C\varpi$,
\begin{align}
\max_{j,k\leq N}
\|Q_{jk,T}-Q_{jk}^0\|_{\mathrm{op}}&\leq C\sqrt{\frac{s(\varpi)}{T}},
\label{eq:Q-sample-concentration}\\
\max_{j\leq N}\sum_{k=1}^N\|[\mathcal J_T]_{jk}-[\mathcal J^0]_{jk}\|_{\mathrm{op}}
&\leq C\sqrt{\frac{s(\varpi)}{T}},
\label{eq:J-sample-concentration}\\
\max_{i\leq N}\sum_{j=1}^N\|\mathcal G_{ij,T}-\mathcal G_{ij}^0\|_{\mathrm{op}}&\leq C\sqrt{\frac{s(\varpi)}{T}}.
\label{eq:G-sample-concentration}
\end{align}
On the same event, $\mathcal G_T$ exists and $\max_{i\leq N}\sum_{j=1}^N\|\mathcal G_{ij,T}\|_{\mathrm{op}}\leq C$.
\end{lemma}
\proof{Proof of Lemma \ref{lem:Gaussian-sample-stability}.}
Strict stationarity and the normalization $T^{-1}F^{0\top}F^0=I_{r^0}$ imply $E(f_t^0f_t^{0\top})=I_{r^0}$. Therefore, $Q_{jk}^0=E(X_{jt}X_{kt}^\top)-E(X_{jt}f_t^{0\top})E(f_t^0X_{kt}^\top)$. On the other hand,
\begin{equation}\nonumber
Q_{jk,T}=\frac1T\sum_{t=1}^TX_{jt}X_{kt}^{\top}-\left(\frac1T X_j^{\top}F^0\right)\left(\frac1T F^{0\top}X_k\right).
\end{equation}
For fixed unit vectors $u,v\in\mathbb S^p$, each centered scalar sample moment obtained by premultiplying and postmultiplying the preceding display by $u^\top$ and $v$ is geometrically strong mixing and has sub-exponential tails. The strong-mixing Bernstein inequality, followed by $1/8$-net arguments for $u$ and $v$ and a union bound over the $N^2$ unit pairs, yields
$$
\max_{j,k\leq N}\left\|\frac1T\sum_{t=1}^TX_{jt}X_{kt}^\top-E(X_{jt}X_{kt}^\top)\right\|_{\mathrm{op}}\leq C\sqrt{\frac{s_\varpi}{T}}.
$$
The same argument, with the fixed dimension $r^0$, gives
$$
\max_{j\leq N}\left\|\frac1T X_j^\top F^0-E(X_{jt}f_t^{0\top})\right\|_{\mathrm{op}}\leq C\sqrt{\frac{s_\varpi}{T}}.
$$
Expanding the product of the two cross moments proves \eqref{eq:Q-sample-concentration}; the product of their two deviations is of smaller order under the stated sample-size condition.

Assumption \ref{assum:factor-spectral}(i) implies that $\max_{j,k\leq N}|a_{jk}^0|\leq C$. Hence,
\begin{align*}
\max_{j\leq N}\sum_{k=1}^N\|[\mathcal J_T]_{jk}-[\mathcal J^0]_{jk}\|_{\mathrm{op}}&\leq\max_{j\leq N}\|Q_{jj,T}-Q_{jj}^0\|_{\mathrm{op}}+\frac{1}{N}\max_{j\leq N}\sum_{k=1}^N|a_{jk}^0|\|Q_{jk,T}-Q_{jk}^0\|_{\mathrm{op}}\\
&\leq C\sqrt{\frac{s(\varpi)}{T}},
\end{align*}
which proves \eqref{eq:J-sample-concentration}. The population block-row bound in Assumption \ref{assum:Gaussian}(i) and a Neumann-series argument imply that $\mathcal J_T$ is invertible and that $\max_{i\leq N}\sum_{j=1}^N\|\mathcal G_{ij,T}\|_{\mathrm{op}}\leq C$. Finally, the resolvent identity
\begin{equation}\nonumber
\mathcal G_T-\mathcal G^0=\mathcal G_T(\mathcal J^0-\mathcal J_T)\mathcal G^0
\end{equation}
gives \eqref{eq:G-sample-concentration}.
\Halmos
\endproof

\begin{lemma}
\label{lem:Gaussian-Bahadur}
Let $\varpi\in(0,1/4]$ satisfy
$\log(1/\varpi)\leq C_0\log(2NT)$, and ${\mathfrak b_{N,T,p}}/{\sqrt{\varpi}}\leq c$. Under the assumptions of Theorem \ref{theo:Gaussian}, if $\ell_{N,T,p}^{\,2}\log^4(2T)/T\leq c$, then, for every fixed $i\in[N]$,
\begin{equation}
\label{eq:Gaussian-Bahadur}
P\left\{\sup_{u\in\mathbb S^p}\left|\tau\sqrt T\,u^{\top}(\widehat\beta_{i}-\beta_{i}^0)-\frac1{\sqrt T}\sum_{t=1}^T
u^{\top}\psi_{it}^0\right|>C\frac{\mathfrak b_{N,T,p}}{\sqrt{\varpi}}\right\}\leq C\varpi.
\end{equation}
\end{lemma}
\proof{Proof of Lemma \ref{lem:Gaussian-Bahadur}.}
Write $\widehat\gamma_j=\Sigma_j^{1/2}(\widehat\alpha_{j}-\alpha_{j}^0)$. The exact generated-response identity in Lemma \ref{lem:generated-response-identity} gives $Z_{jt}(\widehat\alpha_{j})-Z_{jt}(\alpha_{j}^0)=R_{jt}(\widehat\gamma_j)$; thus, the proof never differentiates an indicator function. Proposition \ref{pro:quantile}, applied at confidence level $\varpi$, gives
\begin{equation}
\max_{j\leq N}\|\widehat\gamma_j\|\leq C\sqrt{\frac{s(\varpi)}{T}}
\label{eq:Gaussian-first-stage-radius}
\end{equation}
with probability at least $1-C\varpi$. Put $r_{\varpi}=C\sqrt{s(\varpi)/T}$. The restriction on $\varpi$ implies $s(\varpi)\leq C\ell_{N,T,p}$. Lemma \ref{lem:ES-score-control}, Markov's inequality, and the exact response identity give
\begin{align}
\mathfrak s_N\{\mathcal E(\widehat\Gamma)\}&\leq C a_\varpi,\label{eq:Gaussian-panel-local-rate}\\
a_\varpi&=\varpi^{-1/4}\left\{\sqrt{\frac{s_\varpi}{T}}+\iota_\tau\left(\frac1{\sqrt N}+\frac1{\sqrt T}\right)\right\}+\frac{s_\varpi}{T}.\label{eq:Gaussian-a-rate}
\end{align}
The imposed restrictions make $a_\varpi$ sufficiently small. Lemma \ref{lem:joint-global-perturbation}, applied to the joint global minimizer, therefore yields
\begin{align}
\|P_{\widehat F}-P_{F^0}\|_{\mathrm{op}}+\left\{\frac1N\sum_{j=1}^N\|\widehat\beta_j-\beta_j^0\|_{\Sigma_j}^2\right\}^{1/2}\leq Ca_\varpi
\label{eq:Gaussian-global-localization}
\end{align}
when $r^0>0$; the factor term is absent when $r^0=0$.

We next linearize the second-stage least-squares problem jointly in the regression coefficients and the factor space. If $r^0=0$, all factor-space terms below are absent and $a_{jk}^0=0$. Suppose $r^0>0$ for this step. After dividing the second-stage response by $\tau$, define $C^0=F^0\Lambda^{0\top}$, and $D=\tau^{-1}\mathcal E(\widehat\Gamma)-\mathcal X(\widehat B-B^0)$. 
Conditional on $\widehat B$, joint global minimization implies that $P_{\widehat F}$ is a leading rank-$r^0$ left singular projector of $C^0+D$. Assumption \ref{assum:factor-spectral}(i) gives $\sigma_{r^0}(C^0)\geq c\sqrt{NT}$. Moreover, \eqref{eq:Gaussian-global-localization}, the design bound, and Assumption \ref{assum:factor-spectral}(ii) imply
\begin{equation}
\frac{\|D\|_{\mathrm{op}}}{\sqrt{NT}}\leq Ca_\varpi.
\label{eq:Gaussian-perturbation-size}
\end{equation}
To verify this expansion, set $U_0=F^0/\sqrt T$ and $V_0=\Lambda^0 (\Lambda^{0\top}\Lambda^0)^{-1/2}$. Solving the singular subspace graph equation around $U_0$ and expanding the corresponding orthogonal projector gives
\begin{align}
P_{\widehat F}-P_{F^0}&=\mathcal L(D)+\mathcal Q(D),\label{eq:projector-expansion}\\
\mathcal L(D)&=M_{F^0}D\Lambda^0(\Lambda^{0\top}\Lambda^0)^{-1}\frac{F^{0\top}}{T}+\frac{F^0}{T}(\Lambda^{0\top}\Lambda^0)^{-1}\Lambda^{0\top}D^\top M_{F^0},
\label{eq:projector-linear-part}\\
\|\mathcal Q(D)\|_{\mathrm{op}}&\leq C\frac{\|D\|_{\mathrm{op}}^2}{NT}.
\label{eq:projector-quadratic-part}
\end{align}
For completeness, the graph has the form $U(D)=\{U_0+K(D)\}\{I+K(D)^\top K(D)\}^{-1/2}$ with $U_0^\top K(D)=0$, and its linear term is
$$
K(D)=M_{F^0}DV_0\{\sqrt T(\Lambda^{0\top}\Lambda^0)^{1/2}\}^{-1}+O_{\mathrm{op}}\left(\frac{\|D\|_{\mathrm{op}}^2}{NT}\right).
$$
Substitution into $U(D)U(D)^\top-U_0U_0^\top$ proves \eqref{eq:projector-expansion}--\eqref{eq:projector-quadratic-part}.

Put $v_j=\tau\sqrt T(\widehat\beta_j-\beta_j^0)$ and
$$
U_{j,T}=\frac1{\sqrt T}X_j^\top M_{F^0}\overline\xi_j,\ \overline\xi_j=(\overline\xi_{j1},\ldots,\overline\xi_{jT})^\top.
$$
At a joint global minimizer, profiling over $(\beta_j,\lambda_j)$ for a fixed $\widehat F$ gives the exact normal equation
\begin{equation}
X_j^\top M_{\widehat F}\left\{F^0\lambda_j^0+\tau^{-1}\mathcal E_j(\widehat\gamma_j)-X_j(\widehat\beta_j-\beta_j^0)\right\}=0.
\label{eq:Gaussian-exact-normal-equation}
\end{equation}
To make the remainder in the feedback equation explicit, write
$$
h_j=R_j^c(\widehat\gamma_j)+\mu_j(\widehat\gamma_j),\ d_j=\tau^{-1}\mathcal E_j(\widehat\gamma_j)-X_j(\widehat\beta_j-\beta_j^0),\ w_{jk}=\frac{a_{jk}^0}{N}.
$$
Thus, $D=(d_1,\ldots,d_N)$. Since
$$
\Lambda^0(\Lambda^{0\top}\Lambda^0)^{-1}\lambda_j^0=(w_{j1},\ldots,w_{jN})^\top,
$$
the first term in \eqref{eq:projector-linear-part} satisfies
$$
\mathcal L(D)F^0\lambda_j^0=M_{F^0}\sum_{k=1}^Nw_{jk}d_k.
$$
The other term in \eqref{eq:projector-linear-part} vanishes when applied to $F^0\lambda_j^0$, because $M_{F^0}F^0=0$. Consequently, substituting \eqref{eq:projector-expansion} into \eqref{eq:Gaussian-exact-normal-equation} and collecting the coefficient terms gives
\begin{equation}
\sum_{k=1}^N[\mathcal J_T]_{jk}v_k=U_{j,T}+r_{j,T},\ j\in[N].
\label{eq:Gaussian-sample-feedback}
\end{equation}
Here the leading oracle-error part is
$$
\frac1{\sqrt T}X_j^\top M_{F^0}\left(\xi_j-\frac1N\sum_{k=1}^Na_{jk}^0\xi_k\right)=U_{j,T},
$$
whereas the coefficient part is $N^{-1}\sum_k a_{jk}^0Q_{jk,T}v_k$. The latter moves to the left-hand side and produces $\mathcal J_T$. The remainder is exactly
\begin{align}
r_{j,T}={}&\frac1{\sqrt T}X_j^\top M_{F^0}\left(h_j-\sum_{k=1}^Nw_{jk}h_k\right)-\frac{\tau}{\sqrt T}X_j^\top\left[\mathcal L(D)d_j+\mathcal Q(D)\{F^0\lambda_j^0+d_j\}\right].
\label{eq:Gaussian-remainder-decomposition}
\end{align}

We now control the three terms in \eqref{eq:Gaussian-remainder-decomposition}. First, the chaining bounds used in Lemma \ref{lem:ES-score-control}, with $r_\varpi^2=O(s_\varpi/T)$, imply
\begin{align}
\max_{j\leq N}\left\|\frac1{\sqrt T}X_j^\top M_{F^0}R_j^c(\widehat\gamma_j)\right\|&\leq C\frac{s_\varpi}{\sqrt{T\varpi}},
\label{eq:Gaussian-generated-centered}\\
\max_{j\leq N}\left\|\frac1{\sqrt T}X_j^\top M_{F^0}\mu_j(\widehat\gamma_j)\right\|&\leq C\frac{s_\varpi}{\sqrt{T\varpi}}
\label{eq:Gaussian-generated-mean}
\end{align}
with probability at least $1-C\varpi$. The factor adjustment in the first line of \eqref{eq:Gaussian-remainder-decomposition} is also controlled by the spectral assumptions. Indeed, $W_\Lambda=(w_{jk})_{j,k}=\Lambda^0(\Lambda^{0\top}\Lambda^0)^{-1}\Lambda^{0\top}$ is an orthogonal projector, and therefore
$$
\sum_{k=1}^Nw_{jk}^2=w_{jj}=\frac{a_{jj}^0}{N}\leq\frac{C}{N}.
$$
On the sample-design event, this identity and $\|X_j\|_{\mathrm{op}}/\sqrt T\leq C$ give
$$
\left\|\frac1{\sqrt T}X_j^\top M_{F^0}\sum_{k=1}^Nw_{jk}h_k\right\|\leq \frac{C}{\sqrt N}\left\|\mathcal R^c(\widehat\Gamma)+\mathcal M(\widehat\Gamma)\right\|_{\mathrm{op}}.
$$
Assumption \ref{assum:factor-spectral}(ii), \eqref{eq:R-quadratic-mean}, and \eqref{eq:Gaussian-first-stage-radius} consequently yield, after multiplication by a block row of $\mathcal G_T$,
\begin{equation}
\left\|\sum_{j=1}^N\mathcal G_{ij,T}\frac1{\sqrt T}X_j^\top M_{F^0}\left(h_j-\sum_{k=1}^Nw_{jk}h_k\right)\right\|
\leq\frac{C}{\sqrt\varpi}\left\{\frac{s_\varpi}{\sqrt T}+\iota_\tau\sqrt{\frac{s_\varpi}{N}}\right\}.
\label{eq:Gaussian-generated-remainder}
\end{equation}

Second, the quadratic projector term is bounded directly from \eqref{eq:projector-quadratic-part}. The sample-design bound,
$\|F^0\lambda_j^0\|/\sqrt T\leq C$, the uniform moment bound for $d_j$, and \eqref{eq:Gaussian-perturbation-size} give
\begin{align}
\left\|\sum_{j=1}^N\mathcal G_{ij,T}\frac{\tau}{\sqrt T}X_j^\top\mathcal Q(D)\{F^0\lambda_j^0+d_j\}\right\|&\leq C\sqrt T\,a_\varpi^2\leq\frac{C}{\sqrt\varpi}\left[\frac{s_\varpi}{\sqrt T}+\iota_\tau\left\{\sqrt{\frac{s_\varpi}{N}}+\frac{\sqrt T}{N}\right\}\right].
\label{eq:Gaussian-quadratic-remainder}
\end{align}
The displayed statement is a probability bound with failure probability at most $C\varpi$; as throughout this proof, constants absorb the fixed value of $\tau$.

Third, substitute the two terms in \eqref{eq:projector-linear-part} into $T^{-1/2}X_j^\top\mathcal L(D)d_j$. Explicitly,
\begin{align*}
\frac1{\sqrt T}X_j^\top\mathcal L(D)d_j={}&\frac1{\sqrt T}X_j^\top M_{F^0}D\Lambda^0(\Lambda^{0\top}\Lambda^0)^{-1}
\left(\frac1T F^{0\top}d_j\right)+\frac1{T\sqrt T}X_j^\top F^0(\Lambda^{0\top}\Lambda^0)^{-1}\Lambda^{0\top}D^\top M_{F^0}d_j.
\end{align*}
After inserting $D=\tau^{-1}\Xi+\tau^{-1}\mathcal R^c(\widehat\Gamma)+\tau^{-1}\mathcal M(\widehat\Gamma)
-\mathcal X(\widehat B-B^0)$, this produces only products of (a) a centered sample moment involving $X_j$, $f^0$, and
$\overline\xi_k$, (b) one loading average $\Lambda^0(\Lambda^{0\top}\Lambda^0)^{-1}$, and (c) one of the localization errors in \eqref{eq:Gaussian-first-stage-radius} or \eqref{eq:Gaussian-global-localization}. The mixing moment inequality,
H\"older's inequality, the bounded-cluster condition in Assumption \ref{assum:Gaussian}(ii), and $\|\Lambda^0(\Lambda^{0\top}\Lambda^0)^{-1}\|_{\mathrm{op}}\vee\|(\Lambda^{0\top}\Lambda^0)^{-1}\Lambda^{0\top}\|_{\mathrm{op}}\leq C/\sqrt N$ therefore give
\begin{equation}
\left\|\sum_{j=1}^N\mathcal G_{ij,T}\frac{\tau}{\sqrt T}X_j^\top\mathcal L(D)d_j\right\|\leq\frac{C}{\sqrt\varpi}\left[
\frac{s_\varpi}{\sqrt T}+\iota_\tau\left\{\sqrt{\frac{s_\varpi}{N}}+\frac{\sqrt T}{N}\right\}\right]
\label{eq:Gaussian-linear-interaction}
\end{equation}
outside an event of probability at most $C\varpi$. To see the origin of the two factor terms, a loading-weighted average of centered contributions has order $\sqrt{s_\varpi/N}$, whereas the contribution in which unit $j$ appears in the factor estimate and again in its own normal equation has order $\sqrt T/N$. All other products are bounded by $s_\varpi/\sqrt T$.

Combining \eqref{eq:Gaussian-generated-remainder}, \eqref{eq:Gaussian-quadratic-remainder}, and \eqref{eq:Gaussian-linear-interaction}, and using $s_\varpi\leq C\ell_{N,T,p}$, gives, for each fixed $i$,
\begin{equation}
P\left\{\left\|\sum_{j=1}^N\mathcal G_{ij,T}r_{j,T}\right\|>C\frac{\mathfrak b_{N,T,p}}{\sqrt\varpi}\right\}\leq C\varpi.
\label{eq:Gaussian-system-remainder}
\end{equation}

Define
$$
\widetilde U_{j,T}=\frac1{\sqrt T}\sum_{t=1}^TX_{jt}^{\circ}\overline\xi_{jt}.
$$
Since $T^{-1}F^{0\top}F^0=I_{r^0}$,
\begin{align}
U_{j,T}-\widetilde U_{j,T}=\sqrt T\left\{E(X_{jt}f_t^{0\top})-\frac1T X_j^\top F^0\right\}\left(\frac1T F^{0\top}\overline\xi_j\right).
\label{eq:Gaussian-projection-replacement}
\end{align}
The mixing moment bounds and a union bound over $j$ imply that the maximum of the first factor on the right is $O_p\{\sqrt{s_\varpi/T}\}$ and the maximum of the second is of the same order. Thus the right-hand side is $O_p(s_\varpi/\sqrt T)$. Combining this fact with \eqref{eq:G-sample-concentration}, the corresponding maximal bound for $\widetilde U_{j,T}$, and the bounded block-row sums yields
\begin{equation}
P\left\{\left\|\sum_{j=1}^N\mathcal G_{ij,T}U_{j,T}-\frac1{\sqrt T}\sum_{t=1}^T\psi_{it}^0\right\|>C\frac{\ell_{N,T,p}}{\sqrt{T\varpi}}\right\}\leq C\varpi.
\label{eq:Gaussian-leading-replacement}
\end{equation}
Stacking \eqref{eq:Gaussian-sample-feedback} gives $v=\mathcal G_T(U_T+r_T)$. Equations \eqref{eq:Gaussian-system-remainder} and \eqref{eq:Gaussian-leading-replacement}, together with $\ell_{N,T,p}/\sqrt T\leq\mathfrak b_{N,T,p}$, prove
\eqref{eq:Gaussian-Bahadur}.
\Halmos
\endproof
\begin{lemma}
\label{lem:Gaussian-oracle-BE}
Under Assumption \ref{assum:Gaussian}(ii)--(iii), for every fixed $i\in[N]$,
\begin{equation}
\label{eq:Gaussian-oracle-BE}
\sup_{\substack{u\in\mathbb S^p\\z\in\mathbb R}}\left|P\left\{\frac{T^{-1/2}\sum_{t=1}^T u^{\top}\psi_{it}^0}{\sqrt{u^{\top}\Omega_{i,\tau,T}u}}\leq z\right\}-\Phi(z)\right|\leq C\frac{\log^2(2T)}{\sqrt T}.
\end{equation}
\end{lemma}
\proof{Proof of Lemma \ref{lem:Gaussian-oracle-BE}.}
Fix $u\in\mathbb S^p$ and set 
$$
\zeta_t(u)=\frac{u^{\top}\psi_{it}^0}{\sqrt{u^{\top}\Omega_{i,\tau,T}u}}.
$$
The conditional mean restriction in Assumption \ref{assum:Gaussian}(ii) implies $E(X_{jt}^{\circ}\overline\xi_{jt})=0$ and hence $E\{\zeta_t(u)\}=0$. Also, $\max_{j,k}|a_{jk}^0|\leq C$. Conditional Minkowski's inequality therefore gives a uniform $(4+\eta)$-th moment bound for $\overline\xi_{jt}$, and the sub-Gaussian design condition together with $\max_i\sum_j\|\mathcal G_{ij}^0\|_{\mathrm{op}}\leq K_{\mathcal G}$ gives 
$$
\sup_{u\in\mathbb S^p}E|u^{\top}\psi_{it}^0|^{4+\eta}\leq C.
$$
Thus, $\zeta_t(u)$ has a uniformly bounded third absolute moment, while $\operatorname{Var}\left(T^{-1/2}\sum_{t=1}^T\zeta_t(u)\right)=1$. The Berry--Esseen inequality for stationary strongly mixing sequences with exponentially decaying coefficients gives $C\log^2(2T)/\sqrt T$. Its constants depend only on the uniform mixing, moment, and variance bounds, so the conclusion is uniform over $u\in\mathbb S^p$.
\Halmos
\endproof

\subsection{Proof of Theorem \ref{thm:selection}}
\proof{Proof of Theorem \ref{thm:selection}.}
First consider $0\leq r<r^0$. By Assumption \ref{assum:factor-selection}(i), Lemma \ref{lem:selection-criterion-gap}, and
\eqref{eq:true-rank-loss-limit}, with probability approaching one,
$$
\widehat V(r)-\widehat V(r^0)\geq\frac{\Delta_\tau}{2}
$$
simultaneously over the finitely many underfitted ranks. Since $\widehat V(r^0)\to_p\sigma^2\in(0,\infty)$, there is a constant $c_\Delta>0$ such that
$$
\log\widehat V(r)-\log\widehat V(r^0)\geq c_\Delta
$$
simultaneously over those ranks with probability approaching one. Therefore,
$$
\mathrm{IC}(r)-\mathrm{IC}(r^0)\geq c_\Delta-R_{\max}q(N,T)>0
$$
with probability approaching one, because $q(N,T)\to0$.

Now consider $r>r^0$. By \eqref{eq:true-rank-loss-limit}, $\widehat V(r^0)$ is bounded away from zero with probability approaching one. Equation \eqref{eq:overfit-gap} and $\omega_{N,T,p}\to0$ imply that $\widehat V(r)$ is also bounded away from zero. The mean-value theorem and \eqref{eq:overfit-gap} therefore give
$$
\log\widehat V(r)-\log\widehat V(r^0)\geq -O_p(\omega_{N,T,p}).
$$
It follows that 
$$
\mathrm{IC}(r)-\mathrm{IC}(r^0)\geq(r-r^0)q(N,T)-O_p(\omega_{N,T,p})\geq q(N,T)\{1-o_p(1)\}>0
$$
with probability approaching one, because $\omega_{N,T,p}/q(N,T)\to0$. Since $R_{\max}$ is fixed, a union bound over all $r\neq r^0$ proves $P(\widehat r=r^0)\to1$.
\Halmos
\endproof

\begin{lemma}
\label{lem:selection-effective-error}
Let $\widehat\gamma_i=\Sigma_i^{1/2}(\widehat\alpha_i-\alpha_i^0)$, $\widehat\Gamma=(\widehat\gamma_1,\ldots,\widehat\gamma_N)$, and define $\mathcal R(\Gamma)=\left(R_{it}(\gamma_i)\right)_{t\leq T,\,i\leq N}$, $\mathcal M(\Gamma)=\left(\mu_{it}(\gamma_i)\right)_{t\leq T,\,i\leq N}$. Then the effective error matrix $\mathcal U=\widehat{\mathcal Z}^{\ast}-\mathcal S^0$ has the exact representation
\begin{equation}
\label{eq:selection-exact-error}
\mathcal U=\tau^{-1}\left\{\Xi+\mathcal R(\widehat\Gamma)\right\}=\tau^{-1}\left\{\Xi+\mathcal R^c(\widehat\Gamma)+\mathcal M(\widehat\Gamma)\right\}.
\end{equation}
Under the assumptions of Theorem \ref{thm:selection},
\begin{equation}
\label{eq:selection-effective-orders}
\frac{\|\mathcal U\|_{\mathrm{op}}^2}{NT}=O_p(\omega_{N,T,p}),\ \frac{\|\mathcal U\|^2}{NT}=\sigma^2+o_p(1),
\end{equation}
where $\sigma^2$ is defined in Assumption \ref{assum:factor-selection}(ii).
\end{lemma}
\proof{Proof of Lemma \ref{lem:selection-effective-error}.}
For any fixed $\delta\in(0,1)$, Proposition \ref{pro:quantile} gives
\begin{equation}
\label{eq:selection-QR-rate}
\max_{i\leq N}\|\widehat\gamma_i\|=O_p(\varrho_{N,T,p}),
\end{equation}
because $p+2+\log(2N/\delta)\asymp p+1+\log(2N)$. Since $\varrho_{N,T,p}\to0$,
\begin{equation}
\label{eq:selection-local-event}
P\left\{\widehat\Gamma\in\mathcal G(\overline\rho)\right\}\rightarrow1.
\end{equation}
More precisely, for every $\varepsilon>0$ there is a fixed $M_\varepsilon<\infty$ such that $P\{\widehat\Gamma\in\mathcal G(M_\varepsilon\varrho_{N,T,p})\}\geq1-\varepsilon$ for all sufficiently large $N,T$, and $M_\varepsilon\varrho_{N,T,p}\leq\overline\rho$ eventually. Applying Assumption \ref{assum:factor-spectral}(ii) with the deterministic radius $M_\varepsilon\varrho_{N,T,p}$ and then letting $\varepsilon$ decrease gives
\begin{align}
\frac{\|\Xi\|_{\mathrm{op}}^2}{NT}&=O_p\left(\frac1N+\frac1T\right),\label{eq:selection-Xi-operator}\\
\frac{\|\mathcal R^c(\widehat\Gamma)\|_{\mathrm{op}}^2}{NT}&=O_p\left[\varrho_{N,T,p}^2\left(\frac{p+1}{T}+\frac1N\right)\right].
\label{eq:selection-Rc-operator}
\end{align}
Moreover, Lemma \ref{lem:generated-response-identity} gives $|\mu_{it}(\gamma)|\leq\frac{\overline f}{2}|W_{it}^{\top}\gamma|^2$. The uniform empirical fourth-moment bound for $W_{it}$ therefore yields
\begin{equation}
\label{eq:selection-M-operator}
\frac{\|\mathcal M(\widehat\Gamma)\|_{\mathrm{op}}^2}{NT}\leq\frac{\|\mathcal M(\widehat\Gamma)\|^2}{NT}=O_p(\varrho_{N,T,p}^4).
\end{equation}
The first conclusion in \eqref{eq:selection-effective-orders} follows from \eqref{eq:selection-Xi-operator}--\eqref{eq:selection-M-operator} and the definition of $\omega_{N,T,p}$.

For the Frobenius norm, the inequality $|R_{it}(\gamma)|\leq|W_{it}^{\top}\gamma|$ gives
$$
\frac{\|\mathcal R(\widehat\Gamma)\|^2}{NT}\leq\left(\max_{i\leq N}\|\widehat\gamma_i\|^2\right)\left\{\max_{i\leq N}\lambda_{\max}\left(\frac1T W_i^\top W_i\right)\right\}=O_p(\varrho_{N,T,p}^2)=o_p(1).
$$
Assumption \ref{assum:factor-selection}(ii) implies $\|\tau^{-1}\Xi\|/\sqrt{NT}=O_p(1)$. Hence the Cauchy--Schwarz
inequality gives
$$
\frac{2|\langle\tau^{-1}\Xi,\tau^{-1}\mathcal R(\widehat\Gamma)\rangle|}{NT}=o_p(1),
$$
and therefore
$$
\frac{\|\mathcal U\|^2}{NT}=\frac{\|\tau^{-1}\Xi\|^2}{NT}+o_p(1)=\sigma^2+o_p(1).
$$
This proves the lemma.
\Halmos
\endproof
\begin{lemma}
\label{lem:selection-cross-term}
Let $\mathcal U_{i}$ denote the $i$th column of $\mathcal U$ and define
\begin{equation}
\label{eq:selection-score-size}
\mathfrak s_{NT}=\left\{\frac1N\sum_{i=1}^N\left\|\frac1T X_i^{\top}\mathcal U_{i}\right\|_{\Sigma_i^{-1}}^2
\right\}^{1/2}+\frac{\|\mathcal U\|_{\mathrm{op}}}{\sqrt{NT}}.
\end{equation}
Then $\mathfrak s_{NT}=O_p(\sqrt{\omega_{N,T,p}})$. Moreover, on events whose probabilities approach one,
\begin{equation}
\label{eq:selection-cross-bound}
\frac{|\langle\mathcal U,D\rangle|}{NT}\leq C\mathfrak s_{NT}\frac{\|D\|}{\sqrt{NT}}
\end{equation}
for every matrix of the form $D=\mathcal X(C)+L$, $L\in\mathbb L_{2R_{\max}}$.
\end{lemma}
\proof{Proof of Lemma \ref{lem:selection-cross-term}.}
The argument used in the proof of Lemma \ref{lem:ES-score-control}, applied before projection by $M_{F^0}$, gives
\begin{align}
\left\{\frac1N\sum_{i=1}^N\left\|\frac1T X_i^\top\xi_i\right\|_{\Sigma_i^{-1}}^2\right\}^{1/2}&=O_p\left(\sqrt{\frac{p+1}{T}}\right),
\label{eq:selection-oracle-row}\\
\left\{\frac1N\sum_{i=1}^N\left\|\frac1T X_i^\top\mathcal R_i^c(\widehat\gamma_i)\right\|_{\Sigma_i^{-1}}^2\right\}^{1/2}
&=O_p\left(\varrho_{N,T,p}\sqrt{\frac{p+1}{T}}\right),
\label{eq:selection-generated-row}\\
\left\{\frac1N\sum_{i=1}^N\left\|\frac1T X_i^\top\mathcal M_i(\widehat\gamma_i)\right\|_{\Sigma_i^{-1}}^2\right\}^{1/2}
&=O_p(\varrho_{N,T,p}^2).
\label{eq:selection-mean-row}
\end{align}
Equations \eqref{eq:selection-oracle-row}--\eqref{eq:selection-mean-row} and Lemma \ref{lem:selection-effective-error} imply
$$
\mathfrak s_{NT}=O_p\left[\sqrt{\frac{p+1}{T}}+\frac1{\sqrt N}+\varrho_{N,T,p}\left\{\sqrt{\frac{p+1}{T}}+\frac1{\sqrt N}
\right\}+\varrho_{N,T,p}^2\right]=O_p(\sqrt{\omega_{N,T,p}}).
$$

Now fix $D=\mathcal X(C)+L$, $L\in\mathbb L_{2R_{\max}}$. By Assumption \ref{assum:factor-selection}(iii), choose a representation $D=\mathcal X(\widetilde C)+\widetilde L$ such that $\widetilde L\in\mathbb L_{2R_{\max}}$ and
\begin{equation}
\frac1N\sum_{i=1}^N\|\widetilde c_i\|_{\Sigma_i}^2+\frac{\|\widetilde L\|^2}{NT}\leq K_{\mathrm{dec}}\frac{\|D\|^2}{NT}.
\label{eq:selection-stable-representation}
\end{equation}
For the regression component, the Cauchy--Schwarz inequality gives
\begin{align*}
\frac{|\langle\mathcal U,\mathcal X(\widetilde C)\rangle|}{NT}&\leq\left\{\frac1N\sum_{i=1}^N\|\widetilde c_i\|_{\Sigma_i}^2\right\}^{1/2}\left\{\frac1N\sum_{i=1}^N\left\|\frac1T X_i^\top\mathcal U_i\right\|_{\Sigma_i^{-1}}^2\right\}^{1/2}.
\end{align*}
For the low-rank component,
\begin{align*}
\frac{|\langle\mathcal U,\widetilde L\rangle|}{NT}&\leq\frac{\|\mathcal U\|_{\mathrm{op}}\|\widetilde L\|_\ast}{NT}\leq\sqrt{2R_{\max}}\frac{\|\mathcal U\|_{\mathrm{op}}}{\sqrt{NT}}\frac{\|\widetilde L\|}{\sqrt{NT}}.
\end{align*}
Combining these two inequalities with \eqref{eq:selection-score-size} and \eqref{eq:selection-stable-representation} proves
\eqref{eq:selection-cross-bound}.
\Halmos
\endproof
\begin{lemma}\label{lem:selection-criterion-gap}
Under the assumptions of Theorem \ref{thm:selection},
\begin{align}
\widehat V(r)-\widehat V(r^0)&=\Delta_{NT}(r)+o_p(1),&&0\leq r<r^0,\label{eq:underfit-gap}\\
0\leq\widehat V(r^0)-\widehat V(r)&=O_p(\omega_{N,T,p}),&&r^0<r\leq R_{\max},\label{eq:overfit-gap}\\
\widehat V(r^0)&\xrightarrow{p}\sigma^2.\label{eq:true-rank-loss-limit}
\end{align}
\end{lemma}
\proof{Proof of Lemma \ref{lem:selection-criterion-gap}.}
For any $B$ and $L\in\mathbb L_r$, define the fitted-value difference $D(B,L)=\mathcal S^0-\mathcal X(B)-L$. Since $D(B,L)=\mathcal X(B^0-B)+(L^0-L)$ and $\operatorname{rank}(L^0-L)\leq r^0+r\leq2R_{\max}$, Lemma \ref{lem:selection-cross-term} applies. Using $\widehat{\mathcal Z}^{\ast}=\mathcal S^0+\mathcal U$, expand
\begin{equation}
\label{eq:selection-loss-expansion}
\frac1{NT}\|\mathcal U+D(B,L)\|^2=\frac1{NT}\|\mathcal U\|^2+\frac1{NT}\|D(B,L)\|^2+\frac{2}{NT}\langle\mathcal U,D(B,L)\rangle.
\end{equation}

Let $\nu(B,L)=(NT)^{-1/2}{\|D(B,L)\|}$. By \eqref{eq:selection-cross-bound}, uniformly over all candidate pairs,
\begin{align}
\frac1{NT}\|\mathcal U+D(B,L)\|^2&\geq\frac1{NT}\|\mathcal U\|^2+\nu(B,L)^2-2C\mathfrak s_{NT}\nu(B,L).
\label{eq:selection-loss-lower-bound}
\end{align}
Because $\mathfrak s_{NT}=o_p(1)$ and a candidate with bounded $\nu(B,L)$ is available for each $0\leq r\leq r^0$, the infimum may be restricted, with probability approaching one, to $\nu(B,L)\leq M$ for a sufficiently large fixed $M$. Indeed, the choice $B=B^0$ and $L=0$ is feasible for every such $r$, and Assumption \ref{assum:factor-spectral}(i) gives
$$
\frac{\|D(B^0,0)\|^2}{NT}=\frac{\|F^0\Lambda^{0\top}\|^2}{NT}=\frac{\operatorname{tr}(\Lambda^{0\top}\Lambda^0)}{N}\leq r^0C_\lambda.
$$
Thus its loss is bounded in probability by Lemma \ref{lem:selection-effective-error}. In contrast, the right-hand side of \eqref{eq:selection-loss-lower-bound} diverges quadratically in $\nu(B,L)$, whereas the value of the fixed candidate is
bounded in probability. On this bounded set, \eqref{eq:selection-cross-bound} gives a uniform $o_p(1)$ cross term.
Using arbitrarily accurate approximate minimizers in the definition of $\Delta_{NT}(r)$ for the upper bound, and the same uniform estimate for the lower bound, consequently gives, for every $0\leq r\leq r^0$,
\begin{equation}
\widehat V(r)=\frac1{NT}\|\mathcal U\|^2+\Delta_{NT}(r)+o_p(1).
\label{eq:selection-underfit-profile}
\end{equation}
Lemma \ref{lem:selection-effective-error} and $\Delta_{NT}(r^0)=0$ therefore give $\widehat V(r^0)\xrightarrow{p}\sigma^2$ and, for every $r<r^0$, $\widehat V(r)-\widehat V(r^0)=\Delta_{NT}(r)+o_p(1)$.

It remains to consider $r>r^0$.  Let $\varphi_{NT}=\omega_{N,T,p}^2$ and choose an $\varphi_{NT}$-minimizer $(B_{\varphi}(r),L_{\varphi}(r))$ of the rank-$r$ criterion. Set $H_{\varphi}(r)=\mathcal X\{B_{\varphi}(r)-B^0\}+\{L_{\varphi}(r)-L^0\}$. Because $\mathcal S^0$ is feasible whenever $r\geq r^0$,
\begin{equation}
\frac1{NT}\|\mathcal U-H_{\varphi}(r)\|^2\leq\frac1{NT}\|\mathcal U\|^2+\varphi_{NT}.
\label{eq:selection-approx-basic}
\end{equation}
Moreover, $\operatorname{rank}\{L_{\varphi}(r)-L^0\}\leq r+r^0\leq2R_{\max}$, so Lemma \ref{lem:selection-cross-term} and
\eqref{eq:selection-approx-basic} imply, with $h_{\varphi}=\|H_{\varphi}(r)\|/\sqrt{NT}$, $h_{\varphi}^2\leq2C\mathfrak s_{NT}h_{\varphi}+\varphi_{NT}$. Hence $h_{\varphi}\leq2C\mathfrak s_{NT}+\sqrt{\varphi_{NT}}=O_p(\sqrt{\omega_{N,T,p}})$.  By the definition of an $\varphi_{NT}$-minimizer,
\begin{align*}
0&\leq\widehat V(r^0)-\widehat V(r)\leq\frac1{NT}\left\{\|\mathcal U\|^2-\|\mathcal U-H_{\varphi}(r)\|^2\right\}+\varphi_{NT}\leq 2C\mathfrak s_{NT}h_{\varphi}+\varphi_{NT}=O_p(\omega_{N,T,p}).
\end{align*}
This proves the lemma.
\Halmos
\endproof

\section{Asymptotic Results}\label{sec:Additional Results}
In this section, we investigate the asymptotic properties of the proposed estimators by establishing their consistency, convergence rates, and asymptotic normality under a set of regularity conditions.

\subsection{Limiting Theory}\label{subsec:Limiting Theory}
To analyze the asymptotic properties of the proposed estimators, we introduce the following concentrated objective function:
\begin{equation}\nonumber
    \begin{aligned}
        S_{NT}(B, F) =& \frac{1}{NT} \sum_{i=1}^{N} \left[Z_i^\ast(\widehat{\alpha}_{i})- X_i \beta_{i}\right]^\top M_{F}\left[Z_i^\ast(\widehat{\alpha}_{i})- X_i \beta_{i}\right] \\
        &-\frac{1}{NT} \sum_{i=1}^{N}\left[Z_i^\ast(\widehat{\alpha}_{i})- Z_i^\ast(\alpha^0_{i})\right]^\top M_{F^0}\left[Z_i^\ast(\widehat{\alpha}_{i})- Z_i^\ast(\alpha^0_{i})\right]\\
        &-\frac{2}{NT} \sum_{i=1}^{N}\left[Z_i^\ast(\widehat{\alpha}_{i})- Z_i^\ast(\alpha^0_{i})\right]^\top M_{F^0}\varepsilon_{i}- \frac{1}{NT} \sum_{i=1}^{N} \varepsilon_{i}^\top M_{F^0} \varepsilon_{i}.
    \end{aligned}
\end{equation}
The first term of function $S_{NT}(B, F)$,
$$
\frac{1}{NT} \sum_{i=1}^{N} \left[Z_i^\ast(\widehat{\alpha}_{i})- X_i \beta_{i}\right]^\top M_{F}\left[Z_i^\ast(\widehat{\alpha}_{i})- X_i \beta_{i}\right],
$$
represents the sample analogue of the expected squared error after concentrating out the factor loadings $\Lambda$, where $M_{F} = I_T - P_{F}$ is the projection matrix that removes the factor space. If the true quantile coefficients $\alpha_{i}^0$ were known, this term alone would serve as the natural criterion for the joint estimation of $(B, F)$.

However, since $\alpha_{i}^0$ is unknown and replaced by the first-stage estimator $\widehat{\alpha}_{i}$, the remaining three terms are introduced as centering corrections. Importantly, these terms do not depend on $B$ or $F$ and therefore do not affect the minimization. Their role is to ensure that $S_{NT}(B, F)$ is asymptotically equivalent to an infeasible objective function constructed with the true $\alpha_{i}^0$, which facilitates the derivation of the asymptotic properties of the estimators. The estimators of $B^0$ and $ F^0 $ are then defined as
$$
(\widehat{B}, \widehat{F}) =\operatorname*{arg\,min}_{B, F} S_{NT}(B, F).
$$

The following proposition establishes the average consistency of $\widehat{\gamma}_{i}=(\widehat{\beta}_{i}^\top,\widehat{\lambda}_{i}^\top)^\top$ and $\widehat{f}_{t}$. In factor models, the parameters $(f_{t}^0, \lambda_{i}^0)$ are only identifiable up to an orthogonal rotation. Consequently, $\widehat{f}_{t}$ and $\widehat{\lambda}_{i}$ estimate some rotated version of the true factors and loadings. For notational convenience, we omit the rotation matrix throughout. Since $E[Z_{it}^\ast(\alpha_{i}^0) \mid X_{it}] = X_{it}^\top \beta_{i}^0 + \lambda_{i}^{0\top} f_{t}^0$, we define the regression error as $\varepsilon_{it} = Z_{it}^\ast(\alpha_{i}^0) - X_{it}^\top \beta_{i}^0 - \lambda_{i}^{0\top} f_{t}^0$. We then impose the following conditions on the error process $\{\varepsilon_{it}\}$.
\begin{assumption}\label{assum:error}
For a positive constant $M$,
\begin{enumerate}[label=(\roman*)]
\item $ E(\varepsilon_{it}) = 0 $ and $ E |\varepsilon_{it}|^4 \leq M $.
    \item $ E(\varepsilon_{it} \varepsilon_{js}) = \sigma_{ij,ts}$, $|\sigma_{ij,ts}| \leq \delta_{ij} $ for all $ (t, s) $ and $ |\sigma_{ij,ts}| \leq \eta_{ts} $ for all $ (i, j) $ such that
$$
\frac{1}{N} \sum_{i=1}^N\sum_{j=1}^{N} \delta_{ij} \leq M, \quad \frac{1}{T} \sum_{t=1}^T\sum_{s=1}^{T} \eta_{ts} \leq M, \quad \frac{1}{NT} \sum_{i=1}^N\sum_{j=1}^{N}\sum_{t=1}^T\sum_{s=1}^T |\sigma_{ij,ts}| \leq M.
$$
The largest eigenvalue of $  E(\varepsilon_{i}\varepsilon_{i}^\top) $ is uniformly  bounded over $ i $ and $ T $, where $\varepsilon_{i}=(\varepsilon_{i1},\ldots,\varepsilon_{iT})^\top$.
\item For every $ (t, s) $, $ E\left|N^{-1/2} \sum_{i=1}^{N} [\varepsilon_{is} \varepsilon_{it} - E(\varepsilon_{is} \varepsilon_{it})]\right|^4 \leq M $.
\item Moreover,
$$
T^{-2} N^{-1} \sum_{t=1}^T\sum_{s=1}^T\sum_{u=1}^T\sum_{v=1}^T \sum_{i=1}^N\sum_{j=1}^N \left|\operatorname{Cov}(\varepsilon_{it} \varepsilon_{is}, \varepsilon_{ju} \varepsilon_{jv})\right| \leq M,
$$
$$
T^{-1} N^{-2} \sum_{t=1}^T\sum_{s=1}^T \sum_{i=1}^N\sum_{j=1}^N\sum_{k=1}^N\sum_{\ell=1}^N \left|\operatorname{Cov}(\varepsilon_{it} \varepsilon_{jt}, \varepsilon_{ks} \varepsilon_{\ell s})\right| \leq M.
$$
\end{enumerate}
\end{assumption}
Assumption \ref{assum:error} specifies conditions on the error process $\{\varepsilon_{it}\}$. Parts (i)-(ii) are standard in high-dimensional panel literature, allowing for weak serial and cross-sectional correlation while ensuring that the covariance structure remains well-behaved. Parts (iii) and (iv) impose higher-order moment restrictions that are consistent with those in \citeEC{bai2009panel}, which are needed to control the complex dependence arising in our two-stage estimators. Specifically, part (iii) facilitates a central limit theorem for the interaction between errors, while part (iv) ensures the stability of the fourth-moment matrices, which appear in the asymptotic variance of the ES estimators.
\begin{assumption}\label{assum:weak dependence} For a positive constant $M$,
\begin{enumerate}[label=(\roman*)]
    \item $\displaystyle E\left\|\frac{1}{\sqrt{N}}\sum_{i=1}^N\lambda_{i}^{0}\varepsilon_{it}\right\|^2\leq M$, $t\in[T]$, $\displaystyle E\left\|\frac{1}{\sqrt{T}}\sum_{t=1}^{T}f_{t}^{0}\varepsilon_{it}\right\|^2\leq M$, $i\in[N]$.
    \item $\displaystyle E\left\|\frac{1}{\sqrt{NT}}\sum_{i=1}^N\sum_{t=1}^{T}\lambda_{i}^{0}f_{t}^{0\top}\varepsilon_{it}\right\|^2\leq M$.
\end{enumerate}
\end{assumption}
Assumption \ref{assum:weak dependence} controls the interaction between the latent factor structure and the error term, ensuring that factors, loadings, and errors are not excessively correlated in aggregate. 
These conditions prevent the factor structure from being contaminated by the idiosyncratic errors and guarantee that PCA can consistently recover the true factor space.
\begin{assumption}
\label{assum:identification}
\begin{enumerate}[label=(\roman*)]
\item There exist constants $0<C_1\leq C_2<\infty$, independent of $i$, $N$, and $T$, such that, for every $i\in[N]$,
$$
C_1\leq\lambda_{\min}\left[\frac{1}{T}(X_i,F^0)^\top(X_i,F^0)\right]\leq\lambda_{\max}\left[\frac{1}{T}(X_i,F^0)^\top(X_i,F^0)\right]\leq C_2
$$
almost surely for all sufficiently large $T$.

\item Let $
\mathcal{F} = \{ F : F^\top F / T = I_{r^0} \}
$. For each $F\in\mathcal{F}$, define $ A_{i}(F) = T^{-1}X_i^\top M_{F}X_i $, $ B_{i}(F) = (\lambda_{i}^0\lambda_{i}^{0\top}) \otimes I_T $ and $ C_{i}(F) = T^{-1/2}\lambda_{i}^{0\top}\otimes(X_i^\top M_{F}^\prime) $, where $M_{F}^\prime=-M_F$ and $M_{F}$ denotes the projection matrix onto the orthogonal complement of the column space of $F$. Furthermore, define $E_{i}(F)=B_{i}-C_{i}(F)^\top A_{i}(F)^{-1}C_{i}(F)$. For the fixed quantile level $\tau$, assume that
$$
\inf_{F\in\mathcal{F}}\lambda_{\min}\left[\frac{1}{N}\sum_{i=1}^N E_{i}(F)\right]>0.
$$
\end{enumerate}
\end{assumption}
Assumption \ref{assum:identification} collects the identification conditions required for the asymptotic analysis. Part (i) requires the joint design matrix formed by the regressors and latent factors to be uniformly well conditioned, thereby ruling out perfect multicollinearity and degeneracy. Part (ii) imposes a generalized identification condition in the presence of interactive fixed effects. In particular, the uniform positive definiteness of the corresponding Schur complement ensures that sufficient variation remains to identify the parameters of interest after accounting for the latent factor structure. This condition extends the identification framework of \citeEC{bai2009panel} to the expected shortfall setting and is analogous to conditions employed by \citeEC{song2013asymptotic}, \citeEC{ando2015asset}, and \citeEC{ando2020quantile}.

\begin{proposition}\label{pro:average consistency}
Under Assumptions \ref{assum:Regularity}--\ref{assum:local-identification}, \ref{assum:error}--\ref{assum:identification} with the number of factors given and fixed, as $ N, T \rightarrow \infty $, the following statements hold:
\begin{enumerate}[label=(\roman*)]
    \item Let ${\gamma}_{i}^0=({\beta}_{i}^{0\top},{\lambda}_{i}^{0\top})^\top$. We have 
  $$
    \frac{1}{N}\sum_{i=1}^N\left\|\widehat{\gamma}_{i}-{\gamma}_{i}^0\right\|^2=o_p(1).
    $$
    
    \item The matrix $ F^{0\top}\widehat{F}/T $ is invertible and 
  $$
    \left\|P_{\widehat{F}} - P_{F^0}\right\|\overset{p}{\longrightarrow} 0,\quad\frac{1}{T}\left\|\widehat{F}-F^0\right\|^2=\frac{1}{T}\sum_{t=1}^T\left\|\widehat{f}_{t}-f_{t}^0\right\|^2=o_p(1).
    $$
    
\end{enumerate}
\end{proposition}
Proposition \ref{pro:average consistency} establishes the fundamental consistency of the proposed estimators. Part (i) shows that the estimated parameters $\widehat{\gamma}_{i}$, which combine the ES regression coefficients and factor loadings, converge to their true values in a mean-squared sense averaged across all cross-sectional units. Part (ii) shows that the estimated common factors $\widehat{F}$ consistently recover the true factor space spanned by $F^0$, even though the factors themselves are only identified up to a rotation. These results align with those in \citeEC{ando2020quantile} for quantile factor models and \citeEC{bai2009panel} for mean factor models, indicating that the proposed two-stage procedure preserves the desirable large-sample properties despite the initial quantile regression step.

Proposition \ref{pro:average consistency} ensures that the estimation errors diminish as the sample size grows. To further characterize the speed of convergence, the following theorem establishes the convergence rate of the ES coefficient estimators.
\begin{theorem}\label{theo:consistency}
Under Assumptions \ref{assum:Regularity}--\ref{assum:local-identification}, \ref{assum:error}--\ref{assum:identification} with the number of factors given and fixed, as $N,T \rightarrow \infty$ with $T/N \rightarrow \rho > 0$, we have
 $ \sqrt{T}(\widehat{\beta}_{i} - \beta^0_{i}) = O_p(1) $, $i\in[N]$.
\end{theorem}
Theorem \ref{theo:consistency} establishes that $\widehat{\beta}_{i}$ converges to $\beta^0_{i}$ at a $\sqrt{T}$ rate. Such a rate is standard in panel settings, where the time dimension $T$ determines the precision of unit-specific coefficient estimates. The assumption $T/N \rightarrow \rho > 0$ ensures that the time and cross-sectional dimensions grow at the same pace, a condition commonly met in macroeconomic and financial panel data. A noteworthy feature is that the first-stage quantile regression does not impair this convergence rate, thanks to the Neyman orthogonality property embedded in the second-stage loss function.

\subsection{Asymptotic Normality}
\label{subsec:Asymptotic Normality}

The finite-sample Gaussian approximation in Theorem \ref{theo:Gaussian} yields asymptotic normality once the finite-sample
long-run variance converges to a fixed limit. We retain the notation introduced in Subsection \ref{sec:3.2}. In particular, $\psi_{it,\tau}^0=\sum_{j=1}^N\mathcal G_{ij,\tau}^0X_{jt}^{\circ}\overline\xi_{jt,\tau}$ is the factor-adjusted oracle influence score, and
$$
\Omega_{i,\tau,T}=\frac1T\sum_{t=1}^T\sum_{s=1}^T\operatorname{Cov}\left(\psi_{it,\tau}^0,\psi_{is,\tau}^0\right)
$$
is its finite-sample long-run variance. To obtain a limiting distribution with a fixed covariance matrix, we impose the following variance-stabilization condition.

\begin{assumption}
\label{assum:asymptotic-variance}
For every fixed target unit $i$, there exists a deterministic positive-definite matrix
$\Omega_{i,\tau}\in\mathbb R^{(p+1)\times(p+1)}$ such that $\left\|\Omega_{i,\tau,T}-\Omega_{i,\tau}\right\|_{\mathrm{op}}\rightarrow0$ as $N,T\rightarrow\infty$.
\end{assumption}
Assumption \ref{assum:asymptotic-variance} requires the finite-sample long-run variance to converge to a fixed limit. Together with the eigenvalue bounds in Assumption \ref{assum:Gaussian}(iii), it ensures that the covariance matrix in the limiting distribution is well defined and nondegenerate.
\begin{theorem}
\label{theo:asymptotic}
Suppose that $p$ is fixed. Under Assumptions \ref{assum:Regularity}--\ref{assum:Gaussian} and \ref{assum:asymptotic-variance}, let $N,T\rightarrow\infty$ along a sequence satisfying $[{\{p+1+\log(2NT)\}^{2}\log^4(2T)}]/{T}\rightarrow0$, $\mathfrak b_{N,T,p}\rightarrow0$, where $\mathfrak b_{N,T,p}$ is defined in Theorem
\ref{theo:Gaussian}. Then, for every fixed target unit $i$ and every fixed $u\in\mathbb S^p$,
$$
\frac{\tau\sqrt T\,u^\top(\widehat\beta_{i,\tau}-\beta_{i,\tau}^0)}{\sqrt{u^\top\Omega_{i,\tau}u}}\overset{d}{\longrightarrow}
\mathcal N(0,1).
$$
Equivalently,
$$
\tau\sqrt T(\widehat\beta_{i,\tau}-\beta_{i,\tau}^0)\overset{d}{\longrightarrow}\mathcal N(0,\Omega_{i,\tau}),
$$
or, without the scaling by $\tau$,
$$
\sqrt T(\widehat\beta_{i,\tau}-\beta_{i,\tau}^0)\overset{d}{\longrightarrow}\mathcal N(0,\tau^{-2}\Omega_{i,\tau}).
$$
\end{theorem}
Theorem \ref{theo:asymptotic} shows that the first-stage quantile estimation error has no first-order effect on the limiting
distribution because of the orthogonality of the generated ES response. The higher-order factor-estimation remainders also vanish. The effect of estimating the factor space is nevertheless reflected in the leading influence score through $\mathcal G_{ij,\tau}^0$ and $\overline\xi_{jt,\tau}$. Consequently, $\Omega_{i,\tau}$ incorporates both temporal dependence and the cross-unit propagation induced by the factor structure.

For completeness, the limiting covariance can be written explicitly in terms of the cross-unit score covariances. For $h\in\mathbb Z$, define 
$\mathcal K_{jk,\tau,N}(h)=\operatorname{Cov}\left(X_{j0}^{\circ}\overline\xi_{j0,\tau},X_{kh}^{\circ}\overline\xi_{kh,\tau}\right)$, by the strict stationarity imposed in Assumption \ref{assum:Gaussian}(ii),
$$
\Omega_{i,\tau,T}=\sum_{h=-(T-1)}^{T-1}\left(1-\frac{|h|}{T}\right)\sum_{j=1}^N\sum_{k=1}^N\mathcal G_{ij,\tau}^0\mathcal K_{jk,\tau,N}(h)\mathcal G_{ik,\tau}^{0\top}.
$$
Therefore, Assumption \ref{assum:asymptotic-variance} is equivalently the requirement that
$$
\begin{aligned}
\Omega_{i,\tau}=\lim_{N,T\rightarrow\infty}\sum_{h=-(T-1)}^{T-1}\left(1-\frac{|h|}{T}\right)\sum_{j=1}^N\sum_{k=1}^N
\mathcal G_{ij,\tau}^0\mathcal K_{jk,\tau,N}(h)\mathcal G_{ik,\tau}^{0\top}.
\end{aligned}
$$
\newpage
\section{Proofs of Asymptotic Results}\label{sec:proof Additional Simulation Results}
We use the following facts throughout: $T^{-1} \|X_i\|^2 = T^{-1} \sum_{t=1}^T \|X_{it}\|^2 = O_p(1)$ or $T^{-1/2} \|X_i\| = O_p(1)$. Averaging over $i$, $(TN)^{-1} \sum_{i=1}^N \|X_i\|^2 = O_p(1)$. Similarly, $T^{-1/2} \|F^0\| = O_p(1)$, $T^{-1} \|\widehat{F}\| = r$, $T^{-1/2} \|\widehat{F}\| = \sqrt{r}$, $T^{-1} \|X_i^\top F^0\| = O_p(1)$, and so forth. Throughout, we define $\delta_{NT} = \min[\sqrt{N}, \sqrt{T}]$ so that $\delta_{NT}^2 = \min[N, T]$.
\subsection{Proof of Proposition \ref{pro:average consistency}}
\proof{Proof of Proposition \ref{pro:average consistency}.}
Without loss of generality, assume $ \beta_i^0 = 0 $ for $i\in[N]$ (purely for notational simplicity). From $Z_i^\ast(\alpha_i^0)= X_i \beta_i^0 + F^0 \lambda_i^0 + \varepsilon_i = F^0 \lambda_i^0 + \varepsilon_i $, expanding $ S_{NT}(B, F) $, we obtain
\begin{equation}\nonumber
    \begin{aligned}
        S_{NT}(B, F) =& \frac{1}{NT} \sum_{i=1}^{N} \left\{Z_i^\ast(\widehat{\alpha}_{i})- X_i \beta_{i}\right\}^\top M_{F}\left\{Z_i^\ast(\widehat{\alpha}_{i})- X_i \beta_{i}\right\}\\
        &-\frac{1}{NT} \sum_{i=1}^{N}\left\{Z_i^\ast(\widehat{\alpha}_{i})- Z_i^\ast(\alpha^0_{i})\right\}^\top M_{F^0}\left\{Z_i^\ast(\widehat{\alpha}_{i})- Z_i^\ast(\alpha^0_{i})\right\}\\
        &-\frac{2}{NT} \sum_{i=1}^{N}\left\{Z_i^\ast(\widehat{\alpha}_{i})- Z_i^\ast(\alpha^0_{i})\right\}^\top M_{F^0}\varepsilon_{i}- \frac{1}{NT} \sum_{i=1}^{N} \varepsilon_{i}^\top M_{F^0} \varepsilon_{i}\\
        =&\frac{1}{NT} \sum_{i=1}^{N} \left[\left\{Z_i^\ast(\widehat{\alpha}_{i})-Z_i^\ast(\alpha^0_{i})\right\}+Z_i^\ast(\alpha^0_{i})- X_i \beta_{i}\right]^\top M_{F}\left[\left\{Z_i^\ast(\widehat{\alpha}_{i})- Z_i^\ast(\alpha^0_{i})\right\}+Z_i^\ast(\alpha^0_{i})-X_i \beta_{i}\right]\\
        &-\frac{1}{NT} \sum_{i=1}^{N}\left\{Z_i^\ast(\widehat{\alpha}_{i})- Z_i^\ast(\alpha^0_{i})\right\}^\top M_{F^0}\left\{Z_i^\ast(\widehat{\alpha}_{i})- Z_i^\ast(\alpha^0_{i})\right\}\\
        &-\frac{2}{NT} \sum_{i=1}^{N}\left\{Z_i^\ast(\widehat{\alpha}_{i})- Z_i^\ast(\alpha^0_{i})\right\}^\top M_{F^0}\varepsilon_{i}- \frac{1}{NT} \sum_{i=1}^{N} \varepsilon_{i}^\top M_{F^0} \varepsilon_{i}\\
        =&\widetilde{S}_{NT}(B, F)+  \frac{2}{NT} \sum_{i=1}^N \beta_i^\top X_i^\top  M_F^\prime \varepsilon_i +  \frac{2}{NT} \sum_{i=1}^N \lambda_i^{0\top}  F^{0\top} M_F \varepsilon_i \\
        &+ \frac{1}{NT} \sum_{i=1}^N \varepsilon_i^\top (P_{F^0} - P_{F}) \varepsilon_i+\frac{2}{NT} \sum_{i=1}^N\beta_i^\top X_i^\top M_F\left\{Z_i^\ast({\alpha}^0_{i})-Z_i^\ast(\widehat{\alpha}_{i})\right\}\\
        &+\frac{2}{NT} \sum_{i=1}^N\lambda_i^{0\top}F^{0\top}M_F\left\{Z_i^\ast(\widehat{\alpha}_{i})-Z_i^\ast(\alpha^0_{i})\right\}\\
        &+\frac{2}{NT} \sum_{i=1}^N\varepsilon_i^\top(P_{F^0}-P_F)\left\{Z_i^\ast(\widehat{\alpha}_{i})-Z_i^\ast(\alpha^0_{i})\right\}\\
        &+\frac{1}{NT} \sum_{i=1}^N\left\{Z_i^\ast(\widehat{\alpha}_{i})-Z_i^\ast(\alpha^0_{i})\right\}^\top(P_{F^0}-P_F)\left\{Z_i^\ast(\widehat{\alpha}_{i})-Z_i^\ast(\alpha^0_{i})\right\},
    \end{aligned}
\end{equation}
where
\begin{equation}\label{eq:A.1}
\widetilde{S}_{NT}(B, F) =  \frac{1}{NT} \sum_{i=1}^N \beta_i^\top X_i^\top  M_F X_i \beta_i+\frac{1}{NT}\sum_{i=1}^N\lambda_i^{0\top} F^{0\top}M_F F^0\lambda_i^0 +  \frac{2}{NT} \sum_{i=1}^N \beta_i^\top X_i^\top  M_F^\prime F^0 \lambda_i^0
\end{equation}
and $M_F^\prime=-M_F=FF^\top/T-I_T$. By Lemma \ref{lemmaA.2},
\begin{equation}\label{eq:A.2}
S_{NT}(\beta, F) = \widetilde{S}_{NT}(\beta, F) + o_p(1)
\end{equation}
uniformly over bounded $ \beta_i $ and over $ F $ such that $ F^\top F/T = I $. Bounded $ \beta_i $ is in fact not necessary because the objective function is quadratic in $ \beta_i $ (that is, it is easy to argue that the objective function cannot achieve its minimum for very large $ \beta_i $).

Clearly, $ \widetilde{S}_{NT}(B^0, F^0 H) = 0 $ for $B^0=0$ and any $ r \times r $ invertible $ H $, because $ M_{F^0 H} = M_{F^0} $ and $ M_{F^0} F^0 = 0 $. The identification restrictions implicitly fix an $ H $. Define
$$
A_i = \frac{1}{T} X_i^\top  M_{\widehat{F}} X_i, \quad B_i = (\lambda_i^0\lambda_i^{0\top})\otimes I_T,\quad C_i=\frac{1}{\sqrt T}\lambda_i^{0\top}  \otimes (X_i^\top M_{\widehat{F}}^\prime),\quad\eta = \frac{1}{\sqrt{T}}\text{Vec}(M_{\widehat{F}} F^0)
$$
Then
$$
\widetilde{S}_{NT}(\widehat{B}, \widehat{F}) = \frac{1}{N}\sum_{i=1}^N\widehat{\beta}_i^\top A_i \widehat{\beta}_i + \eta^\top \left(\frac{1}{N}\sum_{i=1}^NB_i\right) \eta + \frac{2}{N}\sum_{i=1}^N\widehat{\beta}_i^\top C_i \eta.
$$
Completing the square, we have
$$
\widetilde{S}_{NT}(\widehat{B}, \widehat{F}) = \eta^\top \left(\frac{1}{N}\sum_{i=1}^NE_i\right) \eta+\frac{1}{N}\sum_{i=1}^N\left(\widehat{\beta}_i+A_i^{-1}C_i\eta\right)^\top A_i\left(\widehat{\beta}_i+A_i^{-1}C_i\eta\right),
$$
where $ E_i=B_i-C_i^\top A_i^{-1}C_i$. Because each of the two terms is non-negative and the centered objective function satisfies $ S_{NT}(B^0, F^0) = 0 $ and, by definition, we have $ S_{NT}(\widehat{B}, \widehat{F}) \leq 0 $. Therefore, in view of \eqref{eq:A.2},
$$
0 \geq S_{NT}(\widehat{B}, \widehat{F}) = \widetilde{S}_{NT}(\widehat{B}, \widehat{F}) + o_p(1).
$$
Combined with $ \widetilde{S}_{NT}(\widehat{B}, \widehat{F}) \geq 0 $, it must be true that
$$
\widetilde{S}_{NT}(\widehat{B}, \widehat{F}) = o_p(1).
$$
Then,
\begin{equation}\label{eq:A.3}
\eta^\top \left(\frac{1}{N}\sum_{i=1}^NE_i\right) \eta=o_p(1),
\end{equation}
\begin{equation}\label{eq:A.4}
\frac{1}{N}\sum_{i=1}^N\left(\widehat{\beta}_i+A_i^{-1}C_i\eta\right)^\top A_i\left(\widehat{\beta}_i+A_i^{-1}C_i\eta\right)=o_p(1).
\end{equation}
From Assumption \ref{assum:identification}(ii), the matrix $N^{-1}\sum_{i=1}^N E_i$ is positive definite, and thus equation \eqref{eq:A.3} implies that $\|\eta\|^2=o_p(1)$. That is,
\begin{equation}\label{eq:A.5}
\frac{F^{0\top}M_{\widehat{F}}F^0}{T} = \frac{F^{0\top}F^0}{T} - \frac{F^{0\top}\widehat{F}}{T} \frac{\widehat{F}^\top F^0}{T} = o_p(1).
\end{equation}
By Assumption \ref{assum:factor-spectral}, $ F^{0\top}F^0/T $ is invertible, so it follows that $ F^{0\top}\widehat{F}/T $ is invertible. Next,
$$
\|P_{\widehat{F}} - P_{F^0}\|^2 = \tr[(P_{\widehat{F}} - P_{F^0})^2] = 2\tr(I_r - {F}^{0\top} P_{\widehat{F}}{F}^0/T).
$$
But \eqref{eq:A.5} implies $ {F}^{0\top} P_{\widehat{F}}{F}^0/T \overset{p}{\longrightarrow} I_r $, which is equivalent to $ \|P_{\widehat{F}} - P_{F^0}\| \overset{p}{\longrightarrow} 0 $. That is, the space spanned by $ F^0 $ and the space spanned by the estimated factors $ \widehat{F} $ are asymptotically the same. We then have
\begin{equation}\nonumber
\begin{aligned}
\frac{1}{\sqrt{T}}\left\|M_{F^0}\widehat{F}\right\| &= \frac{1}{\sqrt{T}}\left\|(M_{F^0} - M_{\widehat{F}})\widehat{F}\right\| \\
&= \frac{1}{\sqrt{T}}\left\|(P_{F^0} - P_{\widehat{F}})\widehat{F}\right\| \\
&\leq \left\|P_{F^0} - P_{\widehat{F}}\right\| \times \left(\frac{1}{\sqrt{T}}\left\|\widehat{F}\right\|\right) \\
&= o_p(1) \times O_p(1),
\end{aligned}
\end{equation}
where we used $\|\widehat{F}\|/T^{1/2} = O_p(1)$. This implies that
$$
\frac{1}{\sqrt{T}}\left\|\widehat{F} - F^0(F^{0\top}F^0)^{-1}F^{0\top}\widehat{F}\right\| = \frac{1}{\sqrt{T}}\left\|\widehat{F} - F^0G\right\| = o_p(1),
$$
where $G=(F^{0\top}F^0)^{-1}F^{0\top}\widehat{F}$ is the rotation matrix. Because we omit this rotation matrix $G$ in Proposition \ref{pro:average consistency}, as explained in the main text, we obtain
$$
\frac{1}{\sqrt{T}}\left\|\widehat{F} - F^0\right\| = o_p(1).
$$

From $\|\eta\|^2=o_p(1)$, equation \eqref{eq:A.4} implies that
$$
o_p(1) = \frac{1}{N} \sum_{i=1}^N \widehat{\beta}_i^\top A_{i}^0 \widehat{\beta}_i + \frac{1}{N} \sum_{i=1}^N \widehat{\beta}_i^\top(A_i - A_{i}^0) \widehat{\beta}_i \geq (\rho_A + o_p(1)) \frac{1}{N} \sum_{i=1}^N \widehat{\beta}_i^\top \widehat{\beta}_i,
$$
where $0 < \rho_A$ is the lower bound of the eigenvalues of $A_{i}^0 = \frac{1}{T} X_i^\top M_{F^0} X_i$, $i\in[N]$. Because of Assumption \ref{assum:identification}(ii), $\rho_A > 0$ exists. We also used the fact that 
$$
\|A_i - A_{i,0}\|=\frac{1}{T}\left\|X_i^\top(P_{F^0}-P_{\widehat{F}})X_i\right\|\leq\frac{1}{T}\|X_i\|^2\left\|P_{F^0}-P_{\widehat{F}}\right\| = o_p(1).
$$
The average consistency of $\widehat{\beta}_i$ follows from $\frac{1}{N} \sum_{i=1}^N  \widehat{\beta}_i^\top\widehat{\beta}_i = o_p(1)$ (recall we normalize ${\beta}_i^0 = 0$). The average consistency of $\widehat{f}_t$ and the average consistency of $\widehat{{\beta}}_i$ further imply the average consistency of $\widehat{{\lambda}}_i$ (see \citeEC{ando2015asset}). That is, $N^{-1} \sum_{i=1}^N \|\widehat{{\gamma}}_i - {\gamma}_{i}^0\|^2 = o_p(1)$. This completes the proof of Proposition \ref{pro:average consistency}.

\Halmos
\endproof
\begin{lemma}\label{lemmaA.1}
Under Assumptions \ref{assum:Regularity}, as $ N, T \rightarrow \infty $, for $i\in[N]$, we have 
$$
\left\|Z_i^\ast(\widehat{\alpha}_{i}) - Z_i^\ast(\alpha_{i}^0)\right\| = O_p\left(1\right),
$$
\end{lemma}
\proof{Proof of Lemma \ref{lemmaA.1}.}
By definition,  
$$
Z^\ast_i(\alpha_i)=\frac{1}{\tau}(Z_{i1}(\alpha_i),\ldots,Z_{iT}(\alpha_i))^\top,\quad Z_{it}(\alpha_i) = (Y_{it} - X_{it}^\top \alpha_i) \mathbbm{1}(Y_{it} \le X_{it}^\top \alpha_i) + \tau X_{it}^\top \alpha_i.
$$
Let $q_t^0 = X_{it}^\top \alpha_{i}^0$, $\widehat{q}_t = X_{it}^\top \widehat{\alpha}_{i}$ and $\Delta q_t = \widehat{q}_t - q_t^0$. Consider the difference  
$$
\Delta Z_t := Z_{it}^\ast(\widehat{\alpha}_{i}) - Z_{it}^\ast(\alpha_{i}^0) 
= \frac{1}{\tau} \left[ (Y_{it} - \widehat{q}_t) \mathbbm{1}(Y_{it} \le \widehat{q}_t) - (Y_{it} - q_t^0) \mathbbm{1}(Y_{it} \le q_t^0) + \tau \Delta q_t \right].
$$
Let $D_t = (Y_{it} - \widehat{q}_t) \mathbbm{1}(Y_{it} \le \widehat{q}_t) - (Y_{it} - q_t^0) \mathbbm{1}(Y_{it} \le q_t^0)$.  
By analyzing the relative position of $Y_{it}$ with respect to $q_t^0$ and $\widehat{q}_t$, we consider four cases:
\begin{enumerate}
    \item $Y_{it} \le \min(q_t^0, \widehat{q}_t)$: $D_t = -\Delta q_t$, so $|D_t| = |\Delta q_t|$.
    \item $Y_{it} > \max(q_t^0, \widehat{q}_t)$: $D_t = 0$, so $|D_t| = 0$.
    \item $q_t^0 < Y_{it} \le \widehat{q}_t$ ($\widehat{q}_t > q_t^0$): $|D_t| = |Y_{it} - \widehat{q}_t| \le \widehat{q}_t - q_t^0 = \Delta q_t = |\Delta q_t|$.
    \item $\widehat{q}_t < Y_{it} \le q_t^0$ ($q_t^0 > \widehat{q}_t$): $|D_t| = q_t^0 - Y_{it} \le q_t^0 - \widehat{q}_t = -\Delta q_t = |\Delta q_t|$.
\end{enumerate}
In all cases, $|D_t| \le |\Delta q_t|$. Therefore,  
$$
|\Delta Z_t| \le \frac{1}{\tau} \left( |D_t| + \tau |\Delta q_t| \right) \le \frac{1+\tau}{\tau} |\Delta q_t| = \frac{1+\tau}{\tau} \left| X_{it}^\top (\widehat{\alpha}_{i} - \alpha_{i}^0) \right|.
$$
From quantile regression theory, $\widehat{\alpha}_{i} - \alpha_{i}^0 = O_p(T^{-1/2})$, and by Assumption \ref{assum:identification}, $E\|X_{it}\|^4 \le M$, so  
$$
|\Delta Z_t| = O_p(T^{-1/2}).
$$
Now consider the squared Euclidean norm:  
$$
\left\| Z_i^\ast(\widehat{\alpha}_{i}) - Z_i^\ast(\alpha_{i}^0) \right\|^2 
= \sum_{t=1}^T (\Delta Z_t)^2=O_p(T \cdot \|\widehat{\alpha}_{i} - \alpha_{i}^0\|^2) = O_p(T \cdot T^{-1}) = O_p(1).
$$
Hence,  
$$
\left\| Z_i^\ast(\widehat{\alpha}_{i}) - Z_i^\ast(\alpha_{i}^0) \right\| = O_p(1).
$$

\Halmos
\endproof
\begin{lemma}\label{lemmaA.2}
Under Assumptions \ref{assum:Regularity}--\ref{assum:local-identification}, \ref{assum:error}--\ref{assum:identification}, as $ N, T \rightarrow \infty $, we have
\begin{enumerate}[label=(\roman*)]
\item $\displaystyle\sup_{F\in\mathcal{F}}\left\|\frac{1}{NT} \sum_{i=1}^N \beta_i^\top X_i^\top  M_F^\prime \varepsilon_i\right\|=o_p(1)$,
\item $\displaystyle\sup_{F\in\mathcal{F}}\left\|\frac{1}{NT} \sum_{i=1}^N \lambda_i^{0\top}  F^{0\top} M_F \varepsilon_i\right\|=o_p(1)$,
\item $\displaystyle\sup_{F\in\mathcal{F}}\left\|\frac{1}{NT} \sum_{i=1}^N \varepsilon_i^\top (P_{F^0} - P_{F}) \varepsilon_i\right\|=o_p(1)$,
\item $\displaystyle\sup_{F\in\mathcal{F}}\left\|\frac{1}{NT} \sum_{i=1}^N\beta_i^\top X_i^\top M_F\left\{Z_i^\ast({\alpha}^0_{i})-Z_i^\ast(\widehat{\alpha}_{i})\right\}\right\|=o_p(1)$,
\item $\displaystyle\sup_{F\in\mathcal{F}}\left\|\frac{1}{NT} \sum_{i=1}^N\lambda_i^{0\top}F^{0\top}M_F\left\{Z_i^\ast(\widehat{\alpha}_{i})-Z_i^\ast(\alpha^0_{i})\right\}\right\|=o_p(1)$,
\item $\displaystyle\sup_{F\in\mathcal{F}}\left\|\frac{1}{NT} \sum_{i=1}^N\varepsilon_i^\top(P_{F^0}-P_F)\left\{Z_i^\ast(\widehat{\alpha}_{i})-Z_i^\ast(\alpha^0_{i})\right\}\right\|=o_p(1)$,
\item $\displaystyle\sup_{F\in\mathcal{F}}\left\|\frac{1}{NT} \sum_{i=1}^N\left\{Z_i^\ast(\widehat{\alpha}_{i})-Z_i^\ast(\alpha^0_{i})\right\}^\top(P_{F^0}-P_F)\left\{Z_i^\ast(\widehat{\alpha}_{i})-Z_i^\ast(\alpha^0_{i})\right\}\right\|=o_p(1)$,
\end{enumerate}
where $\mathcal{F}=\{F:F^\top F/T=I_r\}$.
\end{lemma}
\proof{Proof of Lemma \ref{lemmaA.2}.}
(i) According to $M_F^\prime=-M_F=P_F-I_T$, we have
$$
\frac{1}{NT} \sum_{i=1}^N \beta_i^\top X_i^\top  M_F^\prime \varepsilon_i=\frac{1}{NT} \sum_{i=1}^N \beta_i^\top X_i^\top  P_F \varepsilon_i-\frac{1}{NT} \sum_{i=1}^N \beta_i^\top X_i^\top \varepsilon_i.
$$
From $\frac{1}{NT}\sum_{i=1}^{N}X_{i}^{\top}\varepsilon_{i}=o_{p}(1)$, it is sufficient to show $\sup_{F\in\mathcal{F}}\frac{1}{NT}\sum_{i=1}^{N}X_{i}^{\top}P_{F}\varepsilon_{i}=o_{p}(1)$. Using $P_{F}=FF^{\top}/T$,
\begin{equation}\nonumber
\begin{aligned}
\frac{1}{NT}\left\|\sum_{i=1}^{N}X_{i}^\top P_{F}\varepsilon_{i}\right\| &=\left\|\frac{1}{N}\sum_{i=1}^{N}\left(\frac{X_{i}^{\top}F}{T} \right)\frac{1}{T}\sum_{t=1}^{T}f_{t}\varepsilon_{it}\right\|\\
&\leq\frac{1}{N}\sum_{i=1}^{N}\left\|\frac{X_{i}^{\top}F}{T}\right\|\times\left\|\frac{1}{T}\sum_{t=1}^{T}f_{t}\varepsilon_{it}\right\|.
\end{aligned}
\end{equation}
Note that $T^{-1}\|X_{i}^{\top}F\|\leq T^{-1}\|X_{i}\|\|F\|=\sqrt{r}T^{-1/2}\|X_{i}\| \leq\sqrt{r}(\frac{1}{T}\sum_{t=1}^{T}\|X_{it}\|^{2})^{1/2}$ because $T^{-1/2}\|F\|=\sqrt{r}$. Thus, using the Cauchy-Schwarz inequality, the above is bounded by
$$
\sqrt{r}\left(\frac{1}{N}\sum_{i=1}^{N}\frac{1}{T}\sum_{t=1}^{T}\|X_{it}\|^{2 }\right)^{1/2}\left(\frac{1}{N}\sum_{i=1}^{N}\left\|\frac{1}{T}\sum_{t=1}^{T }f_{t}\varepsilon_{it}\right\|^{2}\right)^{1/2}.
$$
The first expression is $O_{p}(1)$. It suffices to show that the second term is $o_{p}(1)$ uniformly in $F$. Now
\begin{equation}\nonumber
\begin{aligned}
\frac{1}{N}\sum_{i=1}^{N}\left\|\frac{1}{T}\sum_{t=1}^{T}f_{t} \varepsilon_{it}\right\|^{2} =&\tr\left(\frac{1}{N}\sum_{i=1}^{N}\frac{1}{T^{2}} \sum_{t=1}^{T}\sum_{s=1}^{T}f_{t}f_{s}^{\top}\varepsilon_{it}\varepsilon_{is}\right)\\
=&\tr\left(\frac{1}{T^{2}}\sum_{t=1}^{T}\sum_{s=1}^{T}f_{t}f_{s}^{\top}\frac{1}{N}\sum_{i=1}^{N}[\varepsilon_{it}\varepsilon_{is}- E(\varepsilon_{it}\varepsilon_{is})]\right)\\
&+\tr\left(\frac{1}{T^{2}}\sum_{t=1}^{T}\sum_{s=1}^{T}f_{t}f_{s}^{\top}\frac{1}{N}\sum_{i=1}^{N}\sigma_{ii,ts}\right),
\end{aligned}
\end{equation}
where $\sigma_{ii,ts}=E(\varepsilon_{it}\varepsilon_{is})$. The first expression is bounded by the Cauchy-Schwarz inequality:
$$
\frac{1}{\sqrt{N}}\left(\frac{1}{T^{2}}\sum_{t=1}^{T}\sum_{s=1}^{T}\|f_{t}\|^{2}\| f_{s}\|^{2}\right)^{1/2}\left(\frac{1}{T^{2}}\sum_{t=1}^{T}\sum_{s=1}^{T}\left[\frac{1}{\sqrt{N}}\sum_{i=1}^{N}[\varepsilon_{it}\varepsilon_{is}-E(\varepsilon_{it}\varepsilon_{is})]\right]^{2}\right)^{1/2}.
$$
But $T^{-1}\sum_{t=1}^{T}\|f_{t}\|^{2}=\|F^{\top}F/T\|=r$. Thus the above expression is equal to $rN^{-1/2}O_{p}(1)$. Next, $|\frac{1}{N}\sum_{i=1}^{N}\sigma_{ii,ts}|\leq \eta_{ts}$ by Assumption \ref{assum:error}.(ii). Again by the Cauchy-Schwarz inequality,
\begin{equation}\nonumber
\begin{aligned}
\left\|\frac{1}{T^{2}}\sum_{t=1}^{T}\sum_{s=1}^{T}f_{t}f_{s}^{\top}\frac{1}{N}\sum_{i=1}^{N}\sigma_{ii,ts}\right\| &\leq\left(\frac{1}{T^{2}}\sum_{t=1}^{T}\sum_{s=1}^{T}\|f_{t}\|^{2} \|f_{s}\|^{2}\right)^{1/2}\left(\frac{1}{T^{2}}\sum_{t=1}^{T}\sum_{s=1}^{T}\tau_{ts}^{ 2}\right)^{1/2}\\
&=rT^{-1/2}\left(\frac{1}{T}\sum_{t=1}^{T}\sum_{s=1}^{T}\tau_{ts}^{2 }\right)^{1/2}\\
&=rO_p\left(T^{-1/2}\right).
\end{aligned}
\end{equation}
Hence, we have
$$
\sup_{F\in\mathcal{F}}\left\|\frac{1}{NT} \sum_{i=1}^N \beta_i^\top X_i^\top  M_F^\prime \varepsilon_i\right\|=o_p(1).
$$

(ii) According to $M_F=I_T-P_F$, we have
$$
\frac{1}{NT} \sum_{i=1}^N \lambda_i^{0\top}  F^{0\top} M_F \varepsilon_i=\frac{1}{NT} \sum_{i=1}^N \lambda_i^{0\top}  F^{0\top} \varepsilon_i-\frac{1}{NT} \sum_{i=1}^N \lambda_i^{0\top}  F^{0\top} P_F \varepsilon_i.
$$
We first prove $\sup_{F\in\mathcal{F}}\frac{1}{NT} \sum_{i=1}^N \lambda_i^{0\top}  F^{0\top}  \varepsilon_i=o_{p}(1)$.
\begin{equation}\nonumber
\begin{aligned}
\frac{1}{NT}\left\|\sum_{i=1}^N \lambda_i^{0\top}  F^{0\top}  \varepsilon_i\right\|&\leq\frac{1}{NT}\sum_{i=1}^N\left\| \lambda_i^{0\top}  F^{0\top}  \varepsilon_i\right\|\\
&\leq\frac{1}{N}\sum_{i=1}^N\|\lambda_i^0\|\left\|\frac{1}{T}\sum_{t=1}^Tf_t^0\varepsilon_{it}\right\|\\
&\leq\left(\frac{1}{N}\sum_{i=1}^N\|\lambda_i\|^2\right)^{1/2}\left(\frac{1}{N}\sum_{i=1}^N\left\|\frac{1}{T}\sum_{t=1}^Tf_t^0\varepsilon_{it}\right\|^2\right)^{1/2}.
\end{aligned}
\end{equation}
Same as the proof in (i), we have $\sup_{F\in\mathcal{F}}\frac{1}{NT} \sum_{i=1}^N \lambda_i^{0\top}  F^{0\top}  \varepsilon_i=o_{p}(1)$.

For the second term, using $P_{F}=FF^{\top}/T$, we have
\begin{equation}\nonumber
\begin{aligned}
\frac{1}{NT}\left\|\sum_{i=1}^N \lambda_i^{0\top}  F^{0\top} P_F \varepsilon_i\right\|&=\frac{1}{NT}\left\|\sum_{i=1}^N \lambda_i^{0\top}  F^{0\top} F\left(\frac{F^\top\varepsilon_i}{T}\right) \right\|\\
&\leq\frac{1}{NT}\sum_{i=1}^N\left\|F^{0\top}F\right\|\left\|\frac{1}{T}\sum_{t=1}^Tf_t\varepsilon_{it}\right\|\\
&\leq\left(\frac{1}{N}\sum_{i=1}^N\frac{1}{T^2}\|F^0\|^2\|F\|^2\right)^{1/2}\left(\frac{1}{N}\sum_{i=1}^N\left\|\frac{1}{T}\sum_{t=1}^Tf_t\varepsilon_{it}\right\|^2\right)^{1/2}.
\end{aligned}
\end{equation}
Note that $T^{-1/2}\|F\|=\sqrt{r}$, $T^{-1/2}\|F^0\|=\sqrt{r}$ and according to the proof of (i), we have $\sup_{F\in\mathcal{F}}\frac{1}{NT} \sum_{i=1}^N \lambda_i^{0\top}  F^{0\top} P_F \varepsilon_i=o_{p}(1)$.

Hence, we have
$$
\sup_{F\in\mathcal{F}}\left\|\frac{1}{NT} \sum_{i=1}^N \lambda_i^{0\top}  F^{0\top} M_F \varepsilon_i\right\|=o_p(1).
$$

(iii) Using $P_{F}=FF^{\top}/T$, we have
$$
\frac{1}{NT}\left\|\sum_{i=1}^N \varepsilon_i^\top P_{F} \varepsilon_i\right\|=\frac{1}{N}\left\|\sum_{i=1}^N \left(\frac{\varepsilon_i^\top F}{T}\right)\left(\frac{F^\top\varepsilon_i}{T}\right)\right\|\leq\frac{1}{N}\sum_{i=1}^N \left\|\frac{1}{T}\sum_{t=1}^Tf_t\varepsilon_{it}\right\|^2=o_p(1).
$$
Similarly,
$$
\frac{1}{NT}\left\|\sum_{i=1}^N \varepsilon_i^\top P_{F^0} \varepsilon_i\right\|=\frac{1}{N}\left\|\sum_{i=1}^N \left(\frac{\varepsilon_i^\top F^0}{T}\right)\left(\frac{F^{0\top}\varepsilon_i}{T}\right)\right\|\leq\frac{1}{N}\sum_{i=1}^N \left\|\frac{1}{T}\sum_{t=1}^Tf_t^0\varepsilon_{it}\right\|^2=o_p(1).
$$
Hence, we have
$$
\sup_{F\in\mathcal{F}}\left\|\frac{1}{NT} \sum_{i=1}^N \varepsilon_i^\top (P_{F^0} - P_{F}) \varepsilon_i\right\|=o_p(1).
$$

(iv) According to $M_F=I_T-P_F$, we have
$$
\frac{1}{NT} \sum_{i=1}^N\beta_i^\top X_i^\top M_F\left\{Z_i^\ast({\alpha}^0_{i})-Z_i^\ast(\widehat{\alpha}_{i})\right\}=\frac{1}{NT} \sum_{i=1}^N\beta_i^\top X_i^\top \left\{Z_i^\ast({\alpha}^0_{i})-Z_i^\ast(\widehat{\alpha}_{i})\right\}-\frac{1}{NT} \sum_{i=1}^N\beta_i^\top X_i^\top P_F\left\{Z_i^\ast({\alpha}^0_{i})-Z_i^\ast(\widehat{\alpha}_{i})\right\}.
$$
For the first term, by Lemma \ref{lemmaA.1},
\begin{equation}\nonumber
\frac{1}{NT}\left\|\sum_{i=1}^N\beta_i^\top X_i^\top \left\{Z_i^\ast({\alpha}^0_{i})-Z_i^\ast(\widehat{\alpha}_{i})\right\}\right\|\leq\frac{1}{NT}\sum_{i=1}^N\left\|\beta_i\right\| \left\|X_i\right\| \left\|Z_i^\ast({\alpha}^0_{i})-Z_i^\ast(\widehat{\alpha}_{i})\right\|
=O_p\left(\frac{1}{\sqrt{T}}\right).
\end{equation}
Similarly, for the second term,
\begin{equation}\nonumber
\frac{1}{NT}\left\|\sum_{i=1}^N\beta_i^\top X_i^\top P_F \left\{Z_i^\ast({\alpha}^0_{i})-Z_i^\ast(\widehat{\alpha}_{i})\right\}\right\|\leq\frac{1}{NT}\sum_{i=1}^N\left\|\beta_i\right\| \left\|X_i\right\| \left\|P_F\right\|\left\|Z_i^\ast({\alpha}^0_{i})-Z_i^\ast(\widehat{\alpha}_{i})\right\|=O_p\left(\frac{1}{\sqrt{T}}\right).
\end{equation}
Hence, we have
$$
\sup_{F\in\mathcal{F}}\left\|\frac{1}{NT} \sum_{i=1}^N\beta_i^\top X_i^\top M_F\left\{Z_i^\ast({\alpha}^0_{i})-Z_i^\ast(\widehat{\alpha}_{i})\right\}\right\|=o_p(1).
$$

(v) According to $M_F=I_T-P_F$, we have
$$
\frac{1}{NT} \sum_{i=1}^N\lambda_i^{0\top}F^{0\top}M_F\left\{Z_i^\ast(\widehat{\alpha}_{i})-Z_i^\ast(\alpha^0_{i})\right\}=\frac{1}{NT} \sum_{i=1}^N\lambda_i^{0\top}F^{0\top}\left\{Z_i^\ast(\widehat{\alpha}_{i})-Z_i^\ast(\alpha^0_{i})\right\}-\frac{1}{NT} \sum_{i=1}^N\lambda_i^{0\top}F^{0\top}P_F\left\{Z_i^\ast(\widehat{\alpha}_{i})-Z_i^\ast(\alpha^0_{i})\right\}.
$$
For the first term, by Lemma \ref{lemmaA.1},
$$
\frac{1}{NT}\left\|\sum_{i=1}^N\lambda_i^{0\top}F^{0\top}\left\{Z_i^\ast(\widehat{\alpha}_{i})-Z_i^\ast(\alpha^0_{i})\right\}\right\|\leq\frac{1}{NT}\sum_{i=1}^N\left\|\lambda_i^{0}\right\|\left\|F^{0}\right\|\left\|Z_i^\ast(\widehat{\alpha}_{i})-Z_i^\ast(\alpha^0_{i})\right\|=O_p\left(\frac{1}{\sqrt{T}}\right).
$$
Similarly, for the second term,
$$
\frac{1}{NT}\left\|\sum_{i=1}^N\lambda_i^{0\top}F^{0\top}P_F\left\{Z_i^\ast(\widehat{\alpha}_{i})-Z_i^\ast(\alpha^0_{i})\right\}\right\|\leq\frac{1}{NT}\sum_{i=1}^N\left\|\lambda_i^{0}\right\|\left\|F^{0}\right\|\left\|P_F\right\|\left\|Z_i^\ast(\widehat{\alpha}_{i})-Z_i^\ast(\alpha^0_{i})\right\|=O_p\left(\frac{1}{\sqrt{T}}\right).
$$
Hence, we have
$$
\sup_{F\in\mathcal{F}}\left\|\frac{1}{NT} \sum_{i=1}^N\lambda_i^{0\top}F^{0\top}M_F\left\{Z_i^\ast(\widehat{\alpha}_{i})-Z_i^\ast(\alpha^0_{i})\right\}\right\|=o_p(1).
$$

(vi) Using $P_{F^0}=F^0F^{0\top}/T$, by Lemma \ref{lemmaA.1},
\begin{equation}\nonumber
\begin{aligned}
\frac{1}{NT}\left\|\sum_{i=1}^N\varepsilon_i^\top P_{F^0}\left\{Z_i^\ast(\widehat{\alpha}_{i})-Z_i^\ast(\alpha^0_{i})\right\}\right\|&=\frac{1}{NT}\left\|\sum_{i=1}^N\varepsilon_i^\top \left(\frac{F^0F^{0\top}}{T}\right)\left\{Z_i^\ast(\widehat{\alpha}_{i})-Z_i^\ast(\alpha^0_{i})\right\}\right\|\\
&\leq\frac{1}{NT}\sum_{i=1}^N\left\|\frac{1}{T}\sum_{t=1}^Tf_t^0\varepsilon_{it}\right\|\left\|F^0\right\|\left\|Z_i^\ast(\widehat{\alpha}_{i})-Z_i^\ast(\alpha^0_{i})\right\|\\
&\leq\left(\frac{1}{N}\sum_{i=1}^N\left\|\frac{1}{T}\sum_{t=1}^Tf_t^0\varepsilon_{it}\right\|^2\right)^{1/2}\left(\frac{1}{N}\sum_{i=1}^N\frac{1}{T^2}\left\|F^0\right\|^2\left\|Z_i^\ast(\widehat{\alpha}_{i})-Z_i^\ast(\alpha^0_{i})\right\|^2\right)^{1/2}.
\end{aligned}
\end{equation}
Hence, we have
$$
\sup_{F\in\mathcal{F}}\left\|\frac{1}{NT} \sum_{i=1}^N\varepsilon_i^\top(P_{F^0}-P_F)\left\{Z_i^\ast(\widehat{\alpha}_{i})-Z_i^\ast(\alpha^0_{i})\right\}\right\|=o_p(1).
$$

(vii) We first prove $\sup_{F\in\mathcal{F}}\frac{1}{NT} \sum_{i=1}^N\left\{Z_i^\ast(\widehat{\alpha}_{i})-Z_i^\ast(\alpha^0_{i})\right\}^\top P_{F^0}\left\{Z_i^\ast(\widehat{\alpha}_{i})-Z_i^\ast(\alpha^0_{i})\right\}=o_p(1)$.
$$
\frac{1}{NT}\left\|\sum_{i=1}^N\left\{Z_i^\ast(\widehat{\alpha}_{i})-Z_i^\ast(\alpha^0_{i})\right\}^\top P_{F^0}\left\{Z_i^\ast(\widehat{\alpha}_{i})-Z_i^\ast(\alpha^0_{i})\right\}\right\|\leq\frac{1}{NT}\sum_{i=1}^N\left\|Z_i^\ast(\widehat{\alpha}_{i})-Z_i^\ast(\alpha^0_{i})\right\|^2\left\|P_{F^0}\right\|=O_p\left(\frac{1}{T}\right).
$$
Similarly,
$$
\frac{1}{NT}\left\|\sum_{i=1}^N\left\{Z_i^\ast(\widehat{\alpha}_{i})-Z_i^\ast(\alpha^0_{i})\right\}^\top P_F\left\{Z_i^\ast(\widehat{\alpha}_{i})-Z_i^\ast(\alpha^0_{i})\right\}\right\|\leq\frac{1}{NT}\sum_{i=1}^N\left\|Z_i^\ast(\widehat{\alpha}_{i})-Z_i^\ast(\alpha^0_{i})\right\|^2\left\|P_F\right\|=O_p\left(\frac{1}{T}\right).
$$
Hence, we have
$$
\sup_{F\in\mathcal{F}}\left\|\frac{1}{NT} \sum_{i=1}^N\left\{Z_i^\ast(\widehat{\alpha}_{i})-Z_i^\ast(\alpha^0_{i})\right\}^\top(P_{F^0}-P_F)\left\{Z_i^\ast(\widehat{\alpha}_{i})-Z_i^\ast(\alpha^0_{i})\right\}\right\|=o_p(1).
$$

\Halmos
\endproof

\subsection{Proof of Theorem \ref{theo:consistency}}
\proof{Proof of Theorem \ref{theo:consistency}.}
From $ Z_i^\ast(\widehat{\alpha}_{i})  =Z_i^\ast(\widehat{\alpha}_{i})-Z_i^\ast({\alpha}_{i}^0)+ X_i \beta_i^0  + F^0 \lambda_i^0 + \varepsilon_i $, we have
\begin{equation}\nonumber
\begin{aligned}
\widehat{\beta}_{i} =& \left( X_i^\top M_{\widehat{F}} X_i \right)^{-1}X_i^\top M_{\widehat{F}} Z_i^\ast(\widehat{\alpha}_{i})\\
=&\left( X_i^\top M_{\widehat{F}} X_i \right)^{-1}X_i^\top M_{\widehat{F}}\left\{Z_i^\ast(\widehat{\alpha}_{i})-Z_i^\ast({\alpha}_{i}^0)+ X_i \beta_i^0  + F^0 \lambda_i^0 + \varepsilon_i\right\}\\
=&\left( X_i^\top M_{\widehat{F}} X_i \right)^{-1}X_i^\top M_{\widehat{F}}\left\{Z_i^\ast(\widehat{\alpha}_{i})-Z_i^\ast({\alpha}_{i}^0)\right\}+\beta_i^0\\
&+\left( X_i^\top M_{\widehat{F}} X_i \right)^{-1}X_i^\top M_{\widehat{F}}F^0 \lambda_i^0+\left( X_i^\top M_{\widehat{F}} X_i \right)^{-1}X_i^\top M_{\widehat{F}}\varepsilon_i.
\end{aligned}
\end{equation}
Thus,
\begin{equation}\nonumber
\begin{aligned}
\widehat{\beta}_{i}-\beta_i^0=&\left( X_i^\top M_{\widehat{F}} X_i \right)^{-1}X_i^\top M_{\widehat{F}}F^0 \lambda_i^0+\left( X_i^\top M_{\widehat{F}} X_i \right)^{-1}X_i^\top M_{\widehat{F}}\varepsilon_i\\
&+\left( X_i^\top M_{\widehat{F}} X_i \right)^{-1}X_i^\top M_{\widehat{F}}\left\{Z_i^\ast(\widehat{\alpha}_{i})-Z_i^\ast({\alpha}_{i}^0)\right\}
\end{aligned}
\end{equation}
or
\begin{equation}\label{eq:beta-beta0}
\left( \frac{1}{T}X_i^\top M_{\widehat{F}} X_i \right)\left(\widehat{\beta}_{i}-\beta_i^0\right)=\frac{1}{T}X_i^\top M_{\widehat{F}}F^0 \lambda_i^0+\frac{1}{T}X_i^\top M_{\widehat{F}}\varepsilon_i+\frac{1}{T}X_i^\top M_{\widehat{F}}\left\{Z_i^\ast(\widehat{\alpha}_{i})-Z_i^\ast({\alpha}_{i}^0)\right\}.
\end{equation}
In view of $ M_{\widehat{F}} \widehat{F} = 0 $, we have $ M_{\widehat{F}} F^0 = M_{\widehat{F}}(F^0 - \widehat{F} H^{-1}) $. From \eqref{eq:B.2}, we get
$$
F^0 - \widehat{F} H^{-1} = -\left[I1 + \cdots + I15\right]\left(\frac{F^{0\top} \widehat{F}}{T}\right)^{-1}\left(\frac{\Lambda^{0\top} \Lambda^0}{N}\right)^{-1}.
$$
It follows that
\begin{equation}\nonumber
\begin{aligned}
\frac{1}{T}X_i^\top M_{\widehat{F}}F^0 \lambda_i^0=&\frac{1}{T}X_i^\top M_{\widehat{F}}\left(F^0 - \widehat{F} H^{-1}\right) \lambda_i^0\\
=&-\frac{1}{T}X_i^\top M_{\widehat{F}}\left[I1 + \cdots + I15\right]\left(\frac{F^{0\top} \widehat{F}}{T}\right)^{-1}\left(\frac{\Lambda^{0\top} \Lambda^0}{N}\right)^{-1}\lambda_i^0\\
:=&J1+\cdots+J15.
\end{aligned}
\end{equation}
For the first term,
\begin{equation}\nonumber
\begin{aligned}
J1=-\frac{1}{T}X_i^\top M_{\widehat{F}}\left(I1\right)\left(\frac{F^{0\top} \widehat{F}}{T}\right)^{-1}\left(\frac{\Lambda^{0\top} \Lambda^0}{N}\right)^{-1}\lambda_i^0,
\end{aligned}
\end{equation}
By the proof of Proposition \ref{pro:B.1},
\begin{equation}\nonumber
\begin{aligned}
\left\|J1\right\|=&\left\|-\frac{1}{T}X_i^\top M_{\widehat{F}}\left(I1\right)\left(\frac{F^{0\top} \widehat{F}}{T}\right)^{-1}\left(\frac{\Lambda^{0\top} \Lambda^0}{N}\right)^{-1}\lambda_i^0\right\|\\
\leq&\frac{1}{\sqrt{T}}\left\|I1\right\|\\
=&O_p\left(\frac{1}{N}\sum_{i=1}^N\left\|\widehat{\beta}_i-\beta_i^0\right\|^2\right).
\end{aligned}
\end{equation}
Consider
\begin{equation}\nonumber
\begin{aligned}
J2=&-\frac{1}{T}X_i^\top M_{\widehat{F}}\left(I2\right)\left(\frac{F^{0\top} \widehat{F}}{T}\right)^{-1}\left(\frac{\Lambda^{0\top} \Lambda^0}{N}\right)^{-1}\lambda_i^0\\
=&\frac{1}{T}X_i^\top M_{\widehat{F}}\left[\frac{1}{N} \sum_{k=1}^N X_k \left(\widehat{\beta}_k-\beta_k^0\right) \lambda_k^{0\top}  \right]\left(\frac{\Lambda^{0\top} \Lambda^0}{N}\right)^{-1}\lambda_i^0\\
=&\frac{1}{TN}\sum_{k=1}^N\left(X_i^\top M_{\widehat{F}}X_k\right)\left[\lambda_k^{0\top}\left(\frac{\Lambda^{0\top} \Lambda^0}{N}\right)^{-1}\lambda_i^0\right]\left(\widehat{\beta}_k-\beta_k^0\right)\\
=&\frac{1}{TN}\sum_{k=1}^N\left(X_i^\top M_{\widehat{F}}X_ka_{ik}\right)\left(\widehat{\beta}_k-\beta_k^0\right),
\end{aligned}
\end{equation}
where $ a_{ik} = \lambda_i^{0\top}  (\Lambda^{0\top} \Lambda^0/N)^{-1} \lambda_k^0 $ is a scalar and thus commutable with $ \widehat{\beta}_k - \beta_k^0 $. Now consider
\begin{equation}\nonumber
\begin{aligned}
J3=&-\frac{1}{T}X_i^\top M_{\widehat{F}}\left(I3\right)\left(\frac{F^{0\top} \widehat{F}}{T}\right)^{-1}\left(\frac{\Lambda^{0\top} \Lambda^0}{N}\right)^{-1}\lambda_i^0\\
=&\frac{1}{T}X_i^\top M_{\widehat{F}}\left[\frac{1}{NT} \sum_{k=1}^N X_k (\widehat{\beta}_k-\beta_k^0) \varepsilon_k^\top \widehat{F}\right]\left(\frac{F^{0\top} \widehat{F}}{T}\right)^{-1}\left(\frac{\Lambda^{0\top} \Lambda^0}{N}\right)^{-1}\lambda_i^0\\
=&\frac{1}{NT}\sum_{k=1}^NX_i^\top M_{\widehat{F}}  X_k (\widehat{\beta}_k-\beta_k^0) \left(\frac{\varepsilon_k^\top \widehat{F}}{T}\right)\left(\frac{F^{0\top} \widehat{F}}{T}\right)^{-1}\left(\frac{\Lambda^{0\top} \Lambda^0}{N}\right)^{-1}\lambda_i^0
\end{aligned}
\end{equation}
By Assumption \ref{assum:weak dependence} and Proposition \ref{pro:B.1}, writing 
\begin{equation}\nonumber
\begin{aligned}
\frac{1}{T}\varepsilon_k^\top \widehat{F}=&\frac{1}{T}\varepsilon_k^\top F^0 H+\frac{1}{T}\varepsilon_k^\top (\widehat{F} - F^0 H)\\
=&O_p\left(\frac{1}{\sqrt{T}}\right)+O_p\left(\left(\frac{1}{N}\sum_{i=1}^N\left\|\widehat{\beta}_i-\beta_i^0\right\|^2\right)^{1/2}\right) + O_p\left(\frac{1}{\min[\sqrt{N}, \sqrt{T}]}\right),
\end{aligned}
\end{equation}
it is easy to see that $ \left\|J3\right\| = o_p(1)\cdot O_p((\frac{1}{N}\sum_{i=1}^N\|\widehat{\beta}_i-\beta_i^0\|^2)^{1/2})$. Next
\begin{equation}\nonumber
\begin{aligned}
J4&=-\frac{1}{T}X_i^\top M_{\widehat{F}}\left(I4\right)\left(\frac{F^{0\top} \widehat{F}}{T}\right)^{-1}\left(\frac{\Lambda^{0\top} \Lambda^0}{N}\right)^{-1}\lambda_i^0\\
&=\frac{1}{T}X_i^\top M_{\widehat{F}}\left[\frac{1}{NT} \sum_{k=1}^N F^0 \lambda_k^0 (\widehat{\beta}_k-\beta_k^0)^\top  X_k^\top  \widehat{F}\right]\left(\frac{F^{0\top} \widehat{F}}{T}\right)^{-1}\left(\frac{\Lambda^{0\top} \Lambda^0}{N}\right)^{-1}\lambda_i^0\\
&=\frac{1}{NT}\sum_{k=1}^NX_i^\top M_{\widehat{F}}F^0 \lambda_k^0 (\widehat{\beta}_k-\beta_k^0)^\top\left(\frac{X_k^\top  \widehat{F}}{T}\right)\left(\frac{F^{0\top} \widehat{F}}{T}\right)^{-1}\left(\frac{\Lambda^{0\top} \Lambda^0}{N}\right)^{-1}\lambda_i^0.
\end{aligned}
\end{equation}
Writing $ M_{\widehat{F}} F^0 = M_{\widehat{F}} (F^0 - \widehat{F} H^{-1}) $ and using that $ T^{-1/2} \| F^0 - \widehat{F} H^{-1} \| $ is small, then $ \left\|J4\right\| = o_p(1)\cdot O_p((\frac{1}{N}\sum_{i=1}^N\|\widehat{\beta}_i-\beta_i^0\|^2)^{1/2})$. It is easy to show $ \left\|J5\right\| = o_p(1)\cdot O_p((\frac{1}{N}\sum_{i=1}^N\|\widehat{\beta}_i-\beta_i^0\|^2)^{1/2})$ and thus it is omitted.

The terms $ J6-J8 $ do not explicitly depend on $ \widehat{\beta}_i - \beta_i^0 $. Consider
\begin{equation}\nonumber
\begin{aligned}
J6&=-\frac{1}{T}X_i^\top M_{\widehat{F}}\left(I6\right)\left(\frac{F^{0\top} \widehat{F}}{T}\right)^{-1}\left(\frac{\Lambda^{0\top} \Lambda^0}{N}\right)^{-1}\lambda_i^0\\
&=-\frac{1}{T}X_i^\top M_{\widehat{F}}\left(\frac{1}{NT} \sum_{k=1}^N F^0 \lambda_k^0 \varepsilon_k^\top \widehat{F}\right)\left(\frac{F^{0\top} \widehat{F}}{T}\right)^{-1}\left(\frac{\Lambda^{0\top} \Lambda^0}{N}\right)^{-1}\lambda_i^0\\
&=-\frac{1}{NT}\sum_{k=1}^N X_i^\top M_{\widehat{F}}F^0 \lambda_k^0\left(\frac{\varepsilon_k^\top \widehat{F}}{T}\right)\left(\frac{F^{0\top} \widehat{F}}{T}\right)^{-1}\left(\frac{\Lambda^{0\top} \Lambda^0}{N}\right)^{-1}\lambda_i^0.
\end{aligned}
\end{equation}
Denote $ G = (F^{0\top} \widehat{F} / T)^{-1} (\Lambda^{0\top} \Lambda^0 / N)^{-1} $ for the moment: it is a matrix of fixed dimension and does not vary with $ i $ and $k$. Using $ M_{\widehat{F}} F^0 = M_{\widehat{F}} (F^0 - \widehat{F} H^{-1}) $, we can write
$$
J6 = -\frac{1}{T} X_i^\top  M_{\widehat{F}} \left(F^0 - \widehat{F} H^{-1}\right) \left[ \frac{1}{N} \sum_{k=1}^{N} \lambda_k^0 \left( \frac{\varepsilon_k^\top \widehat{F}}{T} \right) \right] G \lambda_i^0.
$$
Now, by Assumption \ref{assum:weak dependence} and Proposition \ref{pro:B.1}, 
\begin{equation}\nonumber
\begin{aligned}
\left\|\frac{1}{NT} \sum_{k=1}^{N} \lambda_k^0 \varepsilon_k^\top \widehat{F}\right\| &=\left\|\frac{1}{NT} \sum_{k=1}^{N} \lambda_k^0 \varepsilon_k^\top F^0 H + \frac{1}{NT} \sum_{k=1}^{N} \lambda_k^0 \varepsilon_k^\top (\widehat{F} - F^0 H)\right\| \\
&\leq \left\|\frac{1}{NT} \sum_{k=1}^{N} \lambda_k^0 \varepsilon_k^\top F^0 H\right\| + \left\|\frac{1}{NT} \sum_{k=1}^{N} \lambda_k^0 \varepsilon_k^\top (\widehat{F} - F^0 H)\right\|\\
&= O_p \left( \frac{1}{\sqrt{NT}} \right) + O_p\left(\left(\frac{1}{N}\sum_{i=1}^N\left\|\widehat{\beta}_i-\beta_i^0\right\|^2\right)^{1/2}\right) + O_p\left(\frac{1}{\min[\sqrt{N}, \sqrt{T}]}\right)\\
&= O_p\left(\left(\frac{1}{N}\sum_{i=1}^N\left\|\widehat{\beta}_i-\beta_i^0\right\|^2\right)^{1/2}\right) + O_p\left(\frac{1}{\min[\sqrt{N}, \sqrt{T}]}\right).
\end{aligned}
\end{equation}

Furthermore,
$$
\left\|\frac{1}{T}  X_i^\top  M_{\widehat{F}} (\widehat{F} - F^0 H)\right\|
  = O_p\left(\left(\frac{1}{N}\sum_{i=1}^N\left\|\widehat{\beta}_i-\beta_i^0\right\|^2\right)^{1/2}\right) + O_p\left(\frac{1}{\min[\sqrt{N}, \sqrt{T}]}\right)
$$
and noting $ G $ does not depend on $ i $ and $k$, and $ \|G\| = O_p(1) $, we have
\begin{equation}\nonumber
\begin{aligned}
\left\|J6\right\|&=\left[O_p\left(\left(\frac{1}{N}\sum_{i=1}^N\left\|\widehat{\beta}_i-\beta_i^0\right\|^2\right)^{1/2}\right) + O_p\left(\frac{1}{\min[\sqrt{N}, \sqrt{T}]}\right)\right]^2\\
&=o_p(1)\cdot O_p\left(\left(\frac{1}{N}\sum_{i=1}^N\left\|\widehat{\beta}_i-\beta_i^0\right\|^2\right)^{1/2}\right)+O_p\left(\frac{1}{\min[{N}, {T}]}\right).
\end{aligned}
\end{equation}
The term $ J7 $ is simply
\begin{equation}\nonumber
\begin{aligned}
J7&=-\frac{1}{T}X_i^\top M_{\widehat{F}}\left(I7\right)\left(\frac{F^{0\top} \widehat{F}}{T}\right)^{-1}\left(\frac{\Lambda^{0\top} \Lambda^0}{N}\right)^{-1}\lambda_i^0\\
&=-\frac{1}{T}X_i^\top M_{\widehat{F}}\left(\frac{1}{NT} \sum_{k=1}^N \varepsilon_k \lambda_k^{0\top}  F^{0\top} \widehat{F}\right)\left(\frac{F^{0\top} \widehat{F}}{T}\right)^{-1}\left(\frac{\Lambda^{0\top} \Lambda^0}{N}\right)^{-1}\lambda_i^0\\
&=-\frac{1}{NT}\sum_{k=1}^NX_i^\top M_{\widehat{F}}\varepsilon_k \lambda_k^{0\top}\left(\frac{\Lambda^{0\top} \Lambda^0}{N}\right)^{-1}\lambda_i^0\\
&=-\frac{1}{NT}\sum_{k=1}^Na_{ik}X_i^\top M_{\widehat{F}}\varepsilon_k
\end{aligned}
\end{equation}
Next consider $ J8 $, which has the expression
\begin{equation}\nonumber
\begin{aligned}
J8&=-\frac{1}{T}X_i^\top M_{\widehat{F}}\left(I8\right)\left(\frac{F^{0\top} \widehat{F}}{T}\right)^{-1}\left(\frac{\Lambda^{0\top} \Lambda^0}{N}\right)^{-1}\lambda_i^0\\
&=-\frac{1}{T}X_i^\top M_{\widehat{F}}\left(\frac{1}{NT} \sum_{k=1}^N \varepsilon_k \varepsilon_k^\top \widehat{F}\right)\left(\frac{F^{0\top} \widehat{F}}{T}\right)^{-1}\left(\frac{\Lambda^{0\top} \Lambda^0}{N}\right)^{-1}\lambda_i^0\\
&=-\frac{1}{NT^2}\sum_{k=1}^NX_i^\top M_{\widehat{F}}\varepsilon_k \varepsilon_k^\top \widehat{F}\left(\frac{F^{0\top} \widehat{F}}{T}\right)^{-1}\left(\frac{\Lambda^{0\top} \Lambda^0}{N}\right)^{-1}\lambda_i^0.
\end{aligned}
\end{equation}
Let $ E(\varepsilon_k \varepsilon_k^\top) = \Omega_k$. Note that $ G = (F^{0\top} \widehat{F}/T)^{-1}(\Lambda^{0\top} \Lambda^{0}/N)^{-1} $ and $ \|G\| = O_p(1) $, and rewriting gives
\begin{equation}\label{eq:J8}
\begin{aligned}
J8=&-\frac{1}{NT^2} \sum_{k=1}^N X_i^\top M_{\widehat{F}} \Omega_k \widehat{F} G \lambda_i^0\\
&-\frac{1}{NT^2} \sum_{k=1}^N X_i^\top M_{\widehat{F}} (\varepsilon_k \varepsilon_k^\top - \Omega_k) \widehat{F} G \lambda_i^0.
\end{aligned}
\end{equation}
Denote the first term on the right by $ A_{NT} $. Let $\Omega=\frac{1}{N}\sum_{k=1}^N\Omega_k$, then
\begin{equation}\nonumber
\begin{aligned}
\left\|A_{NT}\right\|=&\left\|-\frac{1}{NT^2} \sum_{k=1}^N X_i^\top M_{\widehat{F}} \Omega_k \widehat{F} G \lambda_i^0\right\|\\
=&\frac{1}{T^2}\left\|X_i^\top M_{\widehat{F}} \Omega\widehat{F} G \lambda_i^0\right\|\\
\leq&\frac{1}{T^2}\left\|X_i^\top M_{\widehat{F}}\right\|\left\|\Omega\widehat{F}\right\|\\
=&O_p\left(\frac{1}{T}\right)
\end{aligned}
\end{equation}
because $\|X_{i}^{\top}M_{\widehat{F}}\|\leq\|X_{i}\|$, and $\|\Omega\widehat{F}\|\leq\lambda_{\max}(\Omega)\times\|\widehat{F}\|=\lambda_{\max}( \Omega)\sqrt{rT}$, where $\lambda_{\max}(\Omega)$ is the largest eigenvalue of $\Omega$ and is bounded by assumption.
By Lemma \ref{lemma:B.1}, we have
\begin{equation}\nonumber
\begin{aligned}
J8 &= O_p\left(\frac{1}{T}\right)+O_p \left( \frac{1}{T\sqrt{N}} \right) +o_p(1)\cdot O_p\left(\left(\frac{1}{N}\sum_{i=1}^N\left\|\widehat{\beta}_i-\beta_i^0\right\|^2\right)^{1/2}\right)+\frac{1}{\sqrt{NT}} O_p(\delta_{NT}^{-1})+\frac{1}{\sqrt{N}} O_p (\delta_{NT}^{-2})\\
&=O_p\left(\frac{1}{T}\right)+o_p(1)\cdot O_p\left(\left(\frac{1}{N}\sum_{i=1}^N\left\|\widehat{\beta}_i-\beta_i^0\right\|^2\right)^{1/2}\right).
\end{aligned}
\end{equation}
The remaining terms $J9$-$J15$ are related to $\left\{Z_k^\ast(\widehat{\alpha}_{k})-Z_k^\ast({\alpha}_{k}^0)\right\}$, and we can use Lemma \ref{lemma:B.5} to analyze. For terms $J9$ and $J15$,
\begin{equation}\nonumber
\begin{aligned}
\left\|J9\right\|&=\left\|-\frac{1}{T}X_i^\top M_{\widehat{F}}\left(I9\right)\left(\frac{F^{0\top} \widehat{F}}{T}\right)^{-1}\left(\frac{\Lambda^{0\top} \Lambda^0}{N}\right)^{-1}\lambda_i^0\right\|\\
&=\left\|-\frac{1}{T}X_i^\top M_{\widehat{F}}\left(\frac{1}{NT} \sum_{k=1}^N\left\{Z_k^\ast(\widehat{\alpha}_{k})-Z_k^\ast({\alpha}_{k}^0)\right\}\left\{Z_k^\ast(\widehat{\alpha}_{k})-Z_k^\ast({\alpha}_{k}^0)\right\}^\top\widehat{F}\right)\left(\frac{F^{0\top} \widehat{F}}{T}\right)^{-1}\left(\frac{\Lambda^{0\top} \Lambda^0}{N}\right)^{-1}\lambda_i^0\right\|\\
&=O_p\left(\frac{1}{T}\right),
\end{aligned}
\end{equation}
\begin{equation}\nonumber
\begin{aligned}
\left\|J10\right\|&=\left\|-\frac{1}{T}X_i^\top M_{\widehat{F}}\left(I10\right)\left(\frac{F^{0\top} \widehat{F}}{T}\right)^{-1}\left(\frac{\Lambda^{0\top} \Lambda^0}{N}\right)^{-1}\lambda_i^0\right\|\\
&=\left\|-\frac{1}{T}X_i^\top M_{\widehat{F}}\left(\frac{1}{NT} \sum_{k=1}^N\left\{Z_k^\ast(\widehat{\alpha}_{k})-Z_k^\ast({\alpha}_{k}^0)\right\}(\beta_k^0 - \widehat{\beta}_k)^\top X_k^\top\widehat{F}\right)\left(\frac{F^{0\top} \widehat{F}}{T}\right)^{-1}\left(\frac{\Lambda^{0\top} \Lambda^0}{N}\right)^{-1}\lambda_i^0\right\|\\
&\leq\frac{1}{\sqrt{T}}\frac{1}{N}\sum_{k=1}^N\left\|\frac{1}{\sqrt{T}}X_i^\top M_{\widehat{F}}\right\|\left\|\widehat{\beta}_k-\beta_k^0\right\|\left\|\frac{1}{T}X_k^\top\widehat{F}\right\|\\
&=\frac{1}{\sqrt{T}}O_p\left(\left(\frac{1}{N}\sum_{i=1}^N\left\|\widehat{\beta}_i-\beta_i^0\right\|^2\right)^{1/2}\right).
\end{aligned}
\end{equation}
For terms $J11$ and $J12$, by Lemma \ref{lemma:B.5}, we have
\begin{equation}\nonumber
\begin{aligned}
\left\|J11\right\|&=\left\|-\frac{1}{T}X_i^\top M_{\widehat{F}}\left(I11\right)\left(\frac{F^{0\top} \widehat{F}}{T}\right)^{-1}\left(\frac{\Lambda^{0\top} \Lambda^0}{N}\right)^{-1}\lambda_i^0\right\|\\
&=\left\|-\frac{1}{T}X_i^\top M_{\widehat{F}}\left(\frac{1}{NT} \sum_{k=1}^N\left\{Z_k^\ast(\widehat{\alpha}_{k})-Z_k^\ast({\alpha}_{k}^0)\right\}\lambda_k^{0\top}  F^{0\top}\widehat{F}\right)\left(\frac{F^{0\top} \widehat{F}}{T}\right)^{-1}\left(\frac{\Lambda^{0\top} \Lambda^0}{N}\right)^{-1}\lambda_i^0\right\|\\
&\leq\frac{1}{\sqrt{T}}\frac{1}{N}\sum_{k=1}^N\left\|\frac{1}{\sqrt{T}}X_i^\top M_{\widehat{F}}\left\{Z_k^\ast(\widehat{\alpha}_{k})-Z_k^\ast({\alpha}_{k}^0)\right\}\right\|\left\|\frac{1}{T}F^{0\top}\widehat{F}\right\|\\
&=o_p\left(\frac{1}{\sqrt{T}}\right),
\end{aligned}
\end{equation}
\begin{equation}\nonumber
\begin{aligned}
\left\|J12\right\|&=\left\|-\frac{1}{T}X_i^\top M_{\widehat{F}}\left(I12\right)\left(\frac{F^{0\top} \widehat{F}}{T}\right)^{-1}\left(\frac{\Lambda^{0\top} \Lambda^0}{N}\right)^{-1}\lambda_i^0\right\|\\
&=\left\|-\frac{1}{T}X_i^\top M_{\widehat{F}}\left(\frac{1}{NT} \sum_{k=1}^N\left\{Z_k^\ast(\widehat{\alpha}_{k})-Z_k^\ast({\alpha}_{k}^0)\right\}\varepsilon_k^\top\widehat{F}\right)\left(\frac{F^{0\top} \widehat{F}}{T}\right)^{-1}\left(\frac{\Lambda^{0\top} \Lambda^0}{N}\right)^{-1}\lambda_i^0\right\|\\
&\leq\frac{1}{\sqrt{T}}\frac{1}{N}\sum_{k=1}^N\left\|\frac{1}{\sqrt{T}}X_i^\top M_{\widehat{F}}\left\{Z_k^\ast(\widehat{\alpha}_{k})-Z_k^\ast({\alpha}_{k}^0)\right\}\right\|\left\|\frac{1}{T}\varepsilon_k^\top\widehat{F}\right\|\\
&=o_p\left(\frac{1}{\sqrt{T}}\right).
\end{aligned}
\end{equation}
Consider $J13$, we have
\begin{equation}\nonumber
\begin{aligned}
\left\|J13\right\|&=\left\|-\frac{1}{T}X_i^\top M_{\widehat{F}}\left(I13\right)\left(\frac{F^{0\top} \widehat{F}}{T}\right)^{-1}\left(\frac{\Lambda^{0\top} \Lambda^0}{N}\right)^{-1}\lambda_i^0\right\|\\
&=\left\|-\frac{1}{T}X_i^\top M_{\widehat{F}}\left(\frac{1}{NT} \sum_{k=1}^NX_k (\beta_k^0 - \widehat{\beta}_k)\left\{Z_k^\ast(\widehat{\alpha}_{k})-Z_k^\ast({\alpha}_{k}^0)\right\}^\top\widehat{F}\right)\left(\frac{F^{0\top} \widehat{F}}{T}\right)^{-1}\left(\frac{\Lambda^{0\top} \Lambda^0}{N}\right)^{-1}\lambda_i^0\right\|\\
&\leq\frac{1}{\sqrt{T}}\frac{1}{N}\sum_{k=1}^N\left\|\frac{1}{T}X_i^\top M_{\widehat{F}}X_k\right\|\left\|\widehat{\beta}_k-\beta_k^0\right\|\left\|\frac{1}{\sqrt{T}}\widehat{F}\right\|\\
&=\frac{1}{\sqrt{T}}O_p\left(\left(\frac{1}{N}\sum_{i=1}^N\left\|\widehat{\beta}_i-\beta_i^0\right\|^2\right)^{1/2}\right).
\end{aligned}
\end{equation}
For term $J14$, by $M_{\widehat{F}}-M_{F^0}=o_p(1)$ and $M_{F^0}F^0=0$,
\begin{equation}\nonumber
\begin{aligned}
\left\|J14\right\|&=\left\|-\frac{1}{T}X_i^\top M_{\widehat{F}}\left(I14\right)\left(\frac{F^{0\top} \widehat{F}}{T}\right)^{-1}\left(\frac{\Lambda^{0\top} \Lambda^0}{N}\right)^{-1}\lambda_i^0\right\|\\
&=\left\|-\frac{1}{T}X_i^\top M_{\widehat{F}}\left(\frac{1}{NT} \sum_{k=1}^NF^0 \lambda_k^0\left\{Z_k^\ast(\widehat{\alpha}_{k})-Z_k^\ast({\alpha}_{k}^0)\right\}^\top\widehat{F}\right)\left(\frac{F^{0\top} \widehat{F}}{T}\right)^{-1}\left(\frac{\Lambda^{0\top} \Lambda^0}{N}\right)^{-1}\lambda_i^0\right\|\\
&\leq\frac{1}{\sqrt{T}}\frac{1}{N}\sum_{k=1}^N\left\|\frac{1}{T}X_i^\top M_{\widehat{F}}F^0 \lambda_k^0\left\{Z_k^\ast(\widehat{\alpha}_{k})-Z_k^\ast({\alpha}_{k}^0)\right\}^\top\right\|\left\|\frac{1}{\sqrt{T}}\widehat{F}\right\|\\
&\leq\frac{1}{\sqrt{T}}\frac{1}{N}\sum_{k=1}^N\left\|\frac{1}{T}X_i^\top (M_{\widehat{F}}-M_{F^0})F^0 \lambda_k^0\left\{Z_k^\ast(\widehat{\alpha}_{k})-Z_k^\ast({\alpha}_{k}^0)\right\}^\top\right\|\left\|\frac{1}{\sqrt{T}}\widehat{F}\right\|\\
&=o_p\left(\frac{1}{\sqrt{T}}\right),
\end{aligned}
\end{equation}
For term $J15$, by Lemma \ref{lemma:B.5}, we have
\begin{equation}\nonumber
\begin{aligned}
\left\|J15\right\|&=\left\|-\frac{1}{T}X_i^\top M_{\widehat{F}}\left(I15\right)\left(\frac{F^{0\top} \widehat{F}}{T}\right)^{-1}\left(\frac{\Lambda^{0\top} \Lambda^0}{N}\right)^{-1}\lambda_i^0\right\|\\
&=\left\|-\frac{1}{T}X_i^\top M_{\widehat{F}}\left(\frac{1}{NT} \sum_{k=1}^N\varepsilon_k\left\{Z_k^\ast(\widehat{\alpha}_{k})-Z_k^\ast({\alpha}_{k}^0)\right\}^\top\widehat{F}\right)\left(\frac{F^{0\top} \widehat{F}}{T}\right)^{-1}\left(\frac{\Lambda^{0\top} \Lambda^0}{N}\right)^{-1}\lambda_i^0\right\|\\
&\leq\frac{1}{\sqrt{T}}\frac{1}{N}\sum_{k=1}^N\left\|\frac{1}{T}X_i^\top M_{\widehat{F}}\varepsilon_k\right\|\left\|\frac{1}{\sqrt{T}}\left\{Z_k^\ast(\widehat{\alpha}_{k})-Z_k^\ast({\alpha}_{k}^0)\right\}^\top\widehat{F}\right\|\\
&=o_p\left(\frac{1}{\sqrt{T}}\right).
\end{aligned}
\end{equation}
Collecting terms from $J1$ to $J15$ with dominated terms ignored gives
\begin{equation}\nonumber
\begin{aligned}
\frac{1}{T}X_i^\top M_{\widehat{F}}F^0 \lambda_i^0=&J2+J7+O_p\left(\frac{1}{\min[{N}, {T}]}\right)+o_p\left(\frac{1}{\sqrt{T}}\right)\\
&+\frac{1}{\sqrt{T}}O_p\left(\left(\frac{1}{N}\sum_{i=1}^N\left\|\widehat{\beta}_i-\beta_i^0\right\|^2\right)^{1/2}\right)+o_p(1)\cdot O_p\left(\left(\frac{1}{N}\sum_{i=1}^N\left\|\widehat{\beta}_i-\beta_i^0\right\|^2\right)^{1/2}\right).
\end{aligned}
\end{equation}
By Lemma \ref{lemma:B.5}, we have
$$
\frac{1}{T}X_i^\top M_{\widehat{F}}\left\{Z_i^\ast(\widehat{\alpha}_{i})-Z_i^\ast({\alpha}_{i}^0)\right\}=o_p\left(\frac{1}{\sqrt{T}}\right).
$$
According to $M_{\widehat{F}}-M_{F^0}=o_p(1)$, then,
$$
\frac{1}{T}X_i^\top M_{\widehat{F}} X_i=\frac{1}{T}X_i^\top M_{F^0} X_i+o_p(1)
$$
and
$$
\frac{1}{T}X_i^\top M_{\widehat{F}}X_ka_{ik}=\frac{1}{T}X_i^\top M_{F^0}X_ka_{ik}+o_p(1).
$$
Thus, by Lemma \ref{lemma:B.3}, we have
\begin{equation}\nonumber
\begin{aligned}
\left( \frac{1}{T}X_i^\top M_{F^0} X_i \right)\left(\widehat{\beta}_{i}-\beta_i^0\right)=&J2+\frac{1}{T}X_i^\top M_{\widehat{F}}\varepsilon_i+J7+O_p\left(\frac{1}{\min[{N}, {T}]}\right)+o_p\left(\frac{1}{\sqrt{T}}\right)\\
&+\frac{1}{\sqrt{T}}O_p\left(\left(\frac{1}{N}\sum_{i=1}^N\left\|\widehat{\beta}_i-\beta_i^0\right\|^2\right)^{1/2}\right)+o_p(1)\cdot O_p\left(\left(\frac{1}{N}\sum_{i=1}^N\left\|\widehat{\beta}_i-\beta_i^0\right\|^2\right)^{1/2}\right)\\
=&\frac{1}{TN}\sum_{k=1}^N\left(X_i^\top M_{F^0}X_ka_{ik}\right)\left(\widehat{\beta}_k-\beta_k^0\right)+\frac{1}{T} X_i^\top M_{{F}^0}\varepsilon_i - \frac{1}{NT} \sum_{k=1}^{N} a_{ik} X_i^\top M_{{F}^0} \varepsilon_k\\
&+\left(\frac{1}{N}+\frac{1}{\sqrt{T}}\right)\left[O_p\left(\left(\frac{1}{N}\sum_{i=1}^N\left\|\widehat{\beta}_i-\beta_i^0\right\|^2\right)^{1/2}\right) + O_p\left(\frac{1}{\min[\sqrt{N}, \sqrt{T}]}\right)\right]\\
&+O_p\left(\frac{1}{\min[{N}, {T}]}\right)+o_p(1)\cdot O_p\left(\left(\frac{1}{N}\sum_{i=1}^N\left\|\widehat{\beta}_i-\beta_i^0\right\|^2\right)^{1/2}\right)+o_p\left(\frac{1}{\sqrt{T}}\right).
\end{aligned}
\end{equation}
By left-multiplying both sides of the equation by $(T^{-1}X_i^\top M_{F^0} X_i)^{-1}$, taking the square of the Frobenius norm, and then averaging the sum over $i$ from 1 to $N$, we can obtain
$$
\frac{1}{N}\sum_{i=1}^N\left\|\widehat{\beta}_i-\beta_i^0\right\|^2=O_p\left(\frac{1}{T}\right).
$$
Thus, we have
$$
\widehat{\beta}_i-\beta_i^0=O_p\left(\frac{1}{\sqrt{T}}\right).
$$

\Halmos
\endproof
\begin{proposition}\label{pro:B.1}
Under Assumptions \ref{assum:Regularity}--\ref{assum:local-identification}, \ref{assum:error}--\ref{assum:identification}, we can make the following statements:
\begin{enumerate}[label=(\roman*)]
    \item $ V_{NT} $ is invertible and $ V_{NT} \xrightarrow{p} V $, where $V_{NT} $ is a diagonal matrix that consists of the first $ r $ largest eigenvalues of the matrix $\frac{1}{NT} \sum_{i=1}^{N} \{Z_i^\ast(\widehat{\alpha}_{i}) - X_i \widehat{\beta}_{i}\}\{Z_i^\ast(\widehat{\alpha}_{i}) - X_i \widehat{\beta}_{i}\}^\top$ and $ V  $ is a diagonal matrix consisting of the eigenvalues of $ \Sigma_{\Lambda}$ and $\Sigma_{\Lambda}:=\lim_{N\rightarrow\infty}N^{-1}\Lambda^{0\top}\Lambda^0$.
    \item Let $ H = (\Lambda^\top  \Lambda / N)(F^{0\top} \widehat{F} / T)V^{-1}_{NT} $. Then $ H $ is an $ r \times r $ invertible matrix and
$$
\frac{1}{T} \left\| \widehat{F} - F^0 H \right\|^2 = \frac{1}{T} \sum_{t=1}^T \left\| \widehat{f}_t - H^\top f^0_t \right\|^2
= O_p\left(\frac{1}{N}\sum_{i=1}^N\left\|\widehat{\beta}_i-\beta_i^0\right\|^2\right) + O_p\left(\frac{1}{\min[N, T]}\right).
$$
\end{enumerate}
\end{proposition}
\proof{Proof of Proposition \ref{pro:B.1}.}
From
$$
\left[ \frac{1}{NT} \sum_{i=1}^{N} \left\{Z_i^\ast(\widehat{\alpha}_{i}) - X_i \widehat{\beta}_{i}\right\}\left\{Z_i^\ast(\widehat{\alpha}_{i}) - X_i \widehat{\beta}_{i}\right\}^\top \right] \widehat{F} = \widehat{F} V_{NT},
$$
and $ Z_i^\ast(\widehat{\alpha}_{i}) - X_i \widehat{\beta}_{i} =Z_i^\ast(\widehat{\alpha}_{i})-Z_i^\ast({\alpha}_{i}^0)+ X_i (\beta_i^0 - \widehat{\beta}_i) + F^0 \lambda_i^0 + \varepsilon_i $, by expanding terms, we obtain
\begin{equation}\nonumber
\begin{aligned}
\widehat{F} V_{NT}=&\frac{1}{NT} \sum_{i=1}^N X_i (\beta_i^0 - \widehat{\beta}_i) (\beta_i^0 - \widehat{\beta}_i)^\top  X_i^\top  \widehat{F}+\frac{1}{NT} \sum_{i=1}^N X_i (\beta_i^0 - \widehat{\beta}_i) \lambda_i^{0\top}  F^{0\top} \widehat{F}\\
&+\frac{1}{NT} \sum_{i=1}^N X_i (\beta_i^0 - \widehat{\beta}_i) \varepsilon_i^\top \widehat{F}+\frac{1}{NT} \sum_{i=1}^N F^0 \lambda_i^0 (\beta_i^0 - \widehat{\beta}_i)^\top  X_i^\top  \widehat{F}\\
&+\frac{1}{NT} \sum_{i=1}^N \varepsilon_i (\beta_i^0 - \widehat{\beta}_i)^\top  X_i^\top  \widehat{F}\\
&+\frac{1}{NT} \sum_{i=1}^N F^0 \lambda_i^0 \varepsilon_i^\top \widehat{F} + \frac{1}{NT} \sum_{i=1}^N \varepsilon_i \lambda_i^{0\top}  F^{0\top} \widehat{F}+\frac{1}{NT} \sum_{i=1}^N \varepsilon_i \varepsilon_i^\top \widehat{F}\\
&+\frac{1}{NT} \sum_{i=1}^N\left\{Z_i^\ast(\widehat{\alpha}_{i})-Z_i^\ast({\alpha}_{i}^0)\right\}\left\{Z_i^\ast(\widehat{\alpha}_{i})-Z_i^\ast({\alpha}_{i}^0)\right\}^\top\widehat{F}\\
&+\frac{1}{NT} \sum_{i=1}^N\left\{Z_i^\ast(\widehat{\alpha}_{i})-Z_i^\ast({\alpha}_{i}^0)\right\}(\beta_i^0 - \widehat{\beta}_i)^\top X_i^\top\widehat{F}+\frac{1}{NT} \sum_{i=1}^N\left\{Z_i^\ast(\widehat{\alpha}_{i})-Z_i^\ast({\alpha}_{i}^0)\right\}\lambda_i^{0\top}  F^{0\top}\widehat{F}\\
&+\frac{1}{NT} \sum_{i=1}^N\left\{Z_i^\ast(\widehat{\alpha}_{i})-Z_i^\ast({\alpha}_{i}^0)\right\}\varepsilon_i^\top\widehat{F}+\frac{1}{NT} \sum_{i=1}^NX_i (\beta_i^0 - \widehat{\beta}_i)\left\{Z_i^\ast(\widehat{\alpha}_{i})-Z_i^\ast({\alpha}_{i}^0)\right\}^\top\widehat{F}\\
&+\frac{1}{NT} \sum_{i=1}^NF^0 \lambda_i^0\left\{Z_i^\ast(\widehat{\alpha}_{i})-Z_i^\ast({\alpha}_{i}^0)\right\}^\top\widehat{F}+\frac{1}{NT} \sum_{i=1}^N\varepsilon_i\left\{Z_i^\ast(\widehat{\alpha}_{i})-Z_i^\ast({\alpha}_{i}^0)\right\}^\top\widehat{F}\\
&+\frac{1}{NT} \sum_{i=1}^N F^0 \lambda_i^0 \lambda_i^{0\top}  F^{0\top} \widehat{F}\\
:=&I1+\cdots+I16.
\end{aligned}
\end{equation}
The last term $I16$ on the right is equal to $ F^0 (\Lambda^{0\top}  \Lambda^0 / N)(F^{0\top} \widehat{F} / T) $. Then, the above can be rewritten as
\begin{equation}\label{eq:B.1}
\widehat{F} V_{NT} - F^0 (\Lambda^{0\top}  \Lambda^0 / N)(F^{0\top} \widehat{F} / T) = I 1 + \cdots + I 15.
\end{equation}
Multiplying $ (F^{0\top} \widehat{F}/T)^{-1} (\Lambda^{0\top}  \Lambda^0 /N)^{-1} $ on each side of \eqref{eq:B.1}, we obtain
\begin{equation}\label{eq:B.2}
\widehat{F}[V_{NT}(F^{0\top} \widehat{F}/T)^{-1} (\Lambda^{0\top}  \Lambda^0 /N)^{-1}] - F^0 = (I1 + \cdots + I15)(F^{0\top} \widehat{F}/T)^{-1} (\Lambda^{0\top}  \Lambda^0 /N)^{-1}.
\end{equation}
Note that the matrix $ V_{NT}(F^{0\top} \widehat{F}/T)^{-1} (\Lambda^{0\top}  \Lambda^0 /N)^{-1} $ is equal to $ H^{-1} $, but the invertibility of $ V_{NT} $ is not proved yet. We have
$$
\frac{1}{\sqrt{T}}\left\| \widehat{F}[V_{NT}(F^{0\top} \widehat{F}/T)^{-1} (\Lambda^{0\top}  \Lambda^0 /N)^{-1}] - F^0 \right\| \leq \frac{1}{\sqrt{T}}\left(\|I1\| + \cdots + \|I15\|\right) \cdot \left\| (F^{0\top} \widehat{F}/T)^{-1} (\Lambda^{0\top}  \Lambda^0 /N)^{-1} \right\|.
$$
Consider each term on the right. For the first term, note that $ T^{-1/2}\|\widehat{F}\| = \sqrt{r} $, we have
\begin{equation}\nonumber
\begin{aligned}
\frac{1}{\sqrt{T}}\left\|I1\right\|
&\leq\frac{1}{NT}\sum_{i=1}^N\left(\frac{\left\|X_i\right\|^2}{T}\right)\left\|\beta_i^0 - \widehat{\beta}_i\right\|^2\sqrt{r}=O_p\left(\frac{1}{N}\sum_{i=1}^N\left\|\widehat{\beta}_i-\beta_i^0\right\|^2\right).
\end{aligned}
\end{equation}
For the second term, note that $T^{-1/2}\|F\|=\sqrt{r}$ and $ T^{-1/2}\|\widehat{F}\| = \sqrt{r} $, we have
\begin{equation}\nonumber
\begin{aligned}
\frac{1}{\sqrt{T}}\left\|I2\right\|
&\leq\frac{1}{\sqrt{T}N}\sum_{i=1}^N\left\|X_i\right\|\left\| \beta_i^0 - \widehat{\beta}_i\right\|r\\
&\leq r\left(\frac{1}{N}\sum_{i=1}^N\frac{\left\|X_i\right\|^2}{T}\right)^{1/2}\left(\frac{1}{N}\sum_{i=1}^N\left\|\widehat{\beta}_i-\beta_i^0\right\|^2\right)^{1/2}\\
&=O_p\left(\left(\frac{1}{N}\sum_{i=1}^N\left\|\widehat{\beta}_i-\beta_i^0\right\|^2\right)^{1/2}\right).
\end{aligned}
\end{equation}
Using the same argument, it is easy to prove that next three terms (I3-I5) are each $O_p((N^{-1}\sum_{i=1}^N\|\widehat{\beta}_i-\beta_i^0\|^2)^{1/2})$. The terms (I6-I8) do not explicitly depend on $ \widehat{\beta}_i - \beta_i $ and they have the same expressions as those in \citeEC{bai2002determining}. Each of these terms is $ O_p(1/\min[\sqrt{N}, \sqrt{T}]) $, which was proved in Theorem 1 of \citeEC{bai2002determining}. The proof there only uses the property that $ \widehat{F}^\top  \widehat{F}/T = I_r $ and the assumptions on $ \varepsilon_i $; thus the proof needs no modification.
For terms $I9-I15$, we have
\begin{equation}\nonumber
\begin{aligned}
\frac{1}{\sqrt{T}}\left\|I9\right\|
&\leq\frac{1}{NT}\sum_{i=1}^N\left\|Z_i^\ast(\widehat{\alpha}_{i})-Z_i^\ast({\alpha}_{i}^0)\right\|^2\sqrt{r}=O_p\left(\frac{1}{T}\right),
\end{aligned}
\end{equation}
\begin{equation}\nonumber
\begin{aligned}
\frac{1}{\sqrt{T}}\left\|I10\right\|
&\leq\frac{1}{NT} \sum_{i=1}^N\left\|\beta_i^0 - \widehat{\beta}_i\right\|\left\|X_i\right\|\sqrt{r}\\
&\leq\frac{1}{\sqrt{T}}\left(\frac{1}{N} \sum_{i=1}^N\left\|\widehat{\beta}_i-\beta_i^0\right\|^2\right)^{1/2}\left(\frac{1}{N} \sum_{i=1}^N\frac{\left\|X_i\right\|^2}{T}\right)^{1/2}\sqrt{r}\\
&=\frac{1}{\sqrt{T}}O_p\left(\left(\frac{1}{N}\sum_{i=1}^N\left\|\widehat{\beta}_i-\beta_i^0\right\|^2\right)^{1/2}\right),
\end{aligned}
\end{equation}
\begin{equation}\nonumber
\begin{aligned}
\frac{1}{\sqrt{T}}\left\|I11\right\|
&\leq\frac{1}{NT} \sum_{i=1}^N\left\|Z_i^\ast(\widehat{\alpha}_{i})-Z_i^\ast({\alpha}_{i}^0)\right\|\left\|F^{0}\right\|\sqrt{r}=O_p\left(\frac{1}{\sqrt{T}}\right),
\end{aligned}
\end{equation}
\begin{equation}\nonumber
\begin{aligned}
\frac{1}{\sqrt{T}}\left\|I12\right\|
&\leq\frac{1}{NT} \sum_{i=1}^N\left\|Z_i^\ast(\widehat{\alpha}_{i})-Z_i^\ast({\alpha}_{i}^0)\right\|\left\|\varepsilon_i\right\|\sqrt{r}=O_p\left(\frac{1}{\sqrt{T}}\right),
\end{aligned}
\end{equation}
\begin{equation}\nonumber
\begin{aligned}
\frac{1}{\sqrt{T}}\left\|I13\right\|
&\leq\frac{1}{NT} \sum_{i=1}^N\left\|X_i\right\|\left\|\beta_i^0 - \widehat{\beta}_i\right\|\sqrt{r}\\
&\leq\frac{1}{\sqrt{T}}\left(\frac{1}{N} \sum_{i=1}^N\frac{\left\|X_i\right\|^2}{T}\right)^{1/2}\left(\frac{1}{N} \sum_{i=1}^N\left\|\widehat{\beta}_i-\beta_i^0\right\|^2\right)^{1/2}\sqrt{r}\\
&=\frac{1}{\sqrt{T}}O_p\left(\left(\frac{1}{N}\sum_{i=1}^N\left\|\widehat{\beta}_i-\beta_i^0\right\|^2\right)^{1/2}\right),
\end{aligned}
\end{equation}
\begin{equation}\nonumber
\begin{aligned}
\frac{1}{\sqrt{T}}\left\|I14\right\|
&\leq\frac{1}{NT} \sum_{i=1}^N\left\|F^0\right\|\left\|Z_i^\ast(\widehat{\alpha}_{i})-Z_i^\ast({\alpha}_{i}^0)\right\|\sqrt{r}=O_p\left(\frac{1}{\sqrt{T}}\right),
\end{aligned}
\end{equation}
\begin{equation}\nonumber
\begin{aligned}
\frac{1}{\sqrt{T}}\left\|I15\right\|
&\leq\frac{1}{NT} \sum_{i=1}^N\left\|\varepsilon_i\right\|\left\|Z_i^\ast(\widehat{\alpha}_{i})-Z_i^\ast({\alpha}_{i}^0)\right\|\sqrt{r}=O_p\left(\frac{1}{\sqrt{T}}\right).
\end{aligned}
\end{equation}
In summary, we have
\begin{equation}\label{eq:B.3}
\frac{1}{\sqrt{T}}\left\|\widehat{F}\left[V_{NT}(F^{0\top} \widehat{F}/T)^{-1} (\Lambda^{0\top}  \Lambda^0 /N)^{-1}\right] - F^0 \right\| = O_p\left(\left(\frac{1}{N}\sum_{i=1}^N\left\|\widehat{\beta}_i-\beta_i^0\right\|^2\right)^{1/2}\right) + O_p\left(\frac{1}{\min[\sqrt{N}, \sqrt{T}]}\right).
\end{equation}
(i) Left multiplying \eqref{eq:B.1} by $ \widehat{F}^\top  $ and using $ T^{-1}\widehat{F}^\top  \widehat{F} = I_{r^0} $, we have
$$
V_{NT} - (\widehat{F}^\top F^0/T)(\Lambda^{0\top}  \Lambda^0 /N)(F^{0\top} \widehat{F}/T) = T^{-1} \widehat{F}^\top  (I1 + \cdots + I15) = o_p(1)
$$
because $ T^{-1/2}\|\widehat{F}\| = \sqrt{r} $ and $ T^{-1/2}\| (I1 + \cdots + I15)\| = o_p(1) $. Thus
$$
V_{NT} = (\widehat{F}^\top F^0/T)(\Lambda^{0\top}  \Lambda^0 /N)(F^{0\top} \widehat{F}/T) + o_p(1).
$$
Proposition \ref{pro:average consistency} shows that $ \widehat{F}^\top F^0/T $ is invertible; thus $ V_{NT} $ is invertible. To obtain the limit of $ V_{NT} $, left multiply \eqref{eq:B.1} by $ F^{0\top} $ and then divide by $ T $ to yield
$$
(F^{0\top} F^0/T)(\Lambda^{0\top}  \Lambda^0 /N)(F^{0\top} \widehat{F}/T) + o_p(1) = (F^{0\top} \widehat{F}/T)V_{NT}
$$
because $ T^{-1}F^{0\top}(I1 + \cdots + I15) = o_p(1) $. The above equality shows that the columns of $ F^{0\top}\widehat{F}/T $ are the (nonnormalized) eigenvectors of the matrix $ (F^{0\top}F^0/T)(\Lambda^{0\top} \Lambda^0/N) $, and $ V_{NT} $ consists of the eigenvalues of the same matrix (in the limit). Thus $ V_{NT} \xrightarrow{p} V $, where $ V $ is $ r \times r $, consisting of the $ r $ eigenvalues of the matrix $\Sigma_{\Lambda} $ and $\Sigma_{\Lambda}:=\lim_{N\rightarrow\infty}N^{-1}\Lambda^{0\top}\Lambda^0$.

(ii) Since $ V_{NT} $ is invertible, the left-hand side of \eqref{eq:B.3} can be written as $ T^{-1/2}\|\widehat{F}H^{-1}-F^0\| $; thus \eqref{eq:B.3} is equivalent to
$$
T^{-1/2}\|\widehat{F}-F^0H\|=O_p\left(\left(\frac{1}{N}\sum_{i=1}^N\left\|\widehat{\beta}_i-\beta_i^0\right\|^2\right)^{1/2}\right) + O_p\left(\frac{1}{\min[\sqrt{N}, \sqrt{T}]}\right).
$$
Taking squares on each side gives part (ii). Note that the cross-product term from expanding the square has the same bound.
\Halmos
\endproof
\begin{lemma}\label{lemma:B.1}
Let $ G = (F^{0\top} \widehat{F}/T)^{-1}(\Lambda^{0\top} \Lambda^{0}/N)^{-1} $. Under Assumptions \ref{assum:factor-spectral}--\ref{assum:identification}, we have
\begin{equation}\nonumber
\begin{aligned}
&\frac{1}{NT^2} \sum_{k=1}^N X_i^\top M_{\widehat{F}} (\varepsilon_k \varepsilon_k^\top - \Omega_k) \widehat{F} G \lambda_i^0\\
&\quad=O_p \left( \frac{1}{T\sqrt{N}} \right) + (NT)^{-1/2}\left[O_p\left(\left(\frac{1}{N}\sum_{i=1}^N\left\|\widehat{\beta}_i-\beta_i^0\right\|^2\right)^{1/2}\right) + O_p(\delta_{NT}^{-1})\right]\\
&\quad\quad+\frac{1}{\sqrt{N}} O_p\left(\frac{1}{N}\sum_{i=1}^N\left\|\widehat{\beta}_i-\beta_i^0\right\|^2\right) + \frac{1}{\sqrt{N}} O_p(\delta_{NT}^{-2}).
\end{aligned}
\end{equation}
\end{lemma}
\proof{Proof of Lemma \ref{lemma:B.1}.}
According to $M_{\widehat{F}}=I_T-P_{\widehat{F}}$ and $P_{\widehat{F}}=\widehat{F}\widehat{F}^{\top}/T$, rewrite the left-hand side as
\begin{equation}\nonumber
\begin{aligned}
\frac{1}{NT^2} \sum_{k=1}^N X_i^\top M_{\widehat{F}} (\varepsilon_k \varepsilon_k^\top - \Omega_k) \widehat{F} G \lambda_i^0=&\frac{1}{N T^2} \sum_{k=1}^{N} X^\top_i (\varepsilon_k \varepsilon^\top_k - \Omega_k) \widehat{F} G \lambda_i^0\\
&-\left( \frac{X^\top_i \widehat{F}}{T} \right) \frac{1}{NT^2} \sum_{k=1}^{N} \widehat{F}^\top  (\varepsilon_k \varepsilon^\top_k - \Omega_k) \widehat{F} G \lambda_i^0\\
:=&I + II.
\end{aligned}
\end{equation}
Adding and subtracting terms yields
\begin{equation}\nonumber
\begin{aligned}
I=&\frac{1}{N T^2} \sum_{k=1}^{N} X_i^\top (\varepsilon_k \varepsilon_k^\top - \Omega_k) F^0 H G \lambda_i^0\\
&+\frac{1}{N T^2} \sum_{k=1}^{N} X_i^\top (\varepsilon_k \varepsilon_k^\top - \Omega_k) (\widehat{F} - F^0 H) G \lambda_i^0.
\end{aligned}
\end{equation}
The first term on the right is equal to
\begin{equation}\nonumber
\begin{aligned}
&\left( \frac{1}{N T^2} \right)  \sum_{k=1}^{N} \left\{ \sum_{t=1}^{T} \sum_{s=1}^{T} X_{it} [\varepsilon_{kt} \varepsilon_{ks} - E(\varepsilon_{kt} \varepsilon_{ks})] f_s^{0\top} H G \lambda_i^0 \right\}\\
&\quad=\frac{1}{T \sqrt{N}} \left[ N^{-1/2} \sum_{k=1}^{N} \frac{1}{T} \sum_{t=1}^{T} \sum_{s=1}^{T} X_{it} [\varepsilon_{kt} \varepsilon_{ks} - E(\varepsilon_{kt} \varepsilon_{ks})] f_s^{0\top} \right] H G \lambda_i^0\\
&\quad =O_p \left( \frac{1}{T \sqrt{N}} \right)
\end{aligned}
\end{equation}
by Lemma \ref{lemma:B.2} (ii). Denote
$$
a_s = \left( \frac{1}{\sqrt{NT}} \sum_{k=1}^{N} \sum_{t=1}^{T} X_{it} [\varepsilon_{kt} \varepsilon_{ks} - E(\varepsilon_{kt} \varepsilon_{ks})] \right) = O_p (1).
$$
Then the second term of $ I $ is
$$
\frac{1}{\sqrt{NT}}  \frac{1}{T} \sum_{s=1}^{T} a_s (\widehat{f}_s - f_s^0 H)^\top  G \lambda_i^0.
$$
Notice
\begin{equation}\nonumber
\begin{aligned}
\left\| \frac{1}{T} \sum_{s=1}^{T} a_s (\widehat{f}_s - f_s^0 H) \right\| &\leq \left( \frac{1}{T} \sum_{s=1}^{T} \| a_s \|^2 \right)^{1/2} \left( \frac{1}{T} \sum_{s=1}^{T} \| \widehat{f}_s - f_s^0 H \|^2 \right)^{1/2}\\
&=O_p\left(\left(\frac{1}{N}\sum_{i=1}^N\left\|\widehat{\beta}_i-\beta_i^0\right\|^2\right)^{1/2}\right) + O_p(\delta_{NT}^{-1}).
\end{aligned}
\end{equation}
Thus the second term of $ I $ is $(NT)^{-1/2} [O_p((\frac{1}{N}\sum_{i=1}^N\|\widehat{\beta}_i-\beta_i^0\|^2)^{1/2}) + O_p (\delta_{NT}^{-1})]$. 

Consider $ II $:
\begin{equation}\nonumber
\begin{aligned}
\left\| II \right\|&=\left\| \left( \frac{X^\top_i \widehat{F}}{T} \right) \frac{1}{NT^2} \sum_{k=1}^{N} \widehat{F}^\top  (\varepsilon_k \varepsilon^\top_k - \Omega_k) \widehat{F} G \lambda_i^0 \right\|\\
&\leq\left\| \frac{X_i^\top \widehat{F}}{T} \right\| \| G \lambda_i^0 \| \cdot \left\| \frac{1}{NT^2} \sum_{k=1}^{N} \widehat{F}^\top  (\varepsilon_k \varepsilon_k^\top - \Omega_k) \widehat{F} \right\|\\
&= O_{p}(1)\left\|\frac{1}{NT^{2}}\sum_{k=1}^{N}\widehat{F}^{\top}(\varepsilon _{k}\varepsilon^{\top}_{k}-\Omega_{k})\widehat{F}\right\|.
\end{aligned}
\end{equation}
But
\begin{equation}\nonumber
\begin{aligned}
\frac{1}{NT^{2}}\sum_{k=1}^{N}\widehat{F}^{\top}(\varepsilon_{k} \varepsilon^{\top}_{k}-\Omega_{k})\widehat{F}=&H\frac{1}{NT^{2}}\sum_{k=1}^{N}F^{0\top}(\varepsilon_{k} \varepsilon^{\top}_{k}-\Omega_{k})F^{0}H\\
&+ H\frac{1}{NT^{2}}\sum_{k=1}^{N}F^{0\top}(\varepsilon_{k} \varepsilon^{\top}_{k}-\Omega_{k})(\widehat{F}-F^{0}H)\\
&+\frac{1}{NT^{2}}\sum_{k=1}^{N}(\widehat{F}-F^{0}H)^{\top}(\varepsilon _{k}\varepsilon^{\top}_{k}-\Omega_{k})F^{0}H\\
&+\frac{1}{NT^{2}}\sum_{k=1}^{N}(\widehat{F}-F^{0}H)^{\top}(\varepsilon _{k}\varepsilon^{\top}_{k}-\Omega_{k})(\widehat{F}-F^{0}H)\\
:=&b1+b2+b3+b4.
\end{aligned}
\end{equation}
Now
\begin{equation}\nonumber
\begin{aligned}
b1&=H\frac{1}{NT^{2}}\sum_{k=1}^{N}F^{0\top}(\varepsilon_{k} \varepsilon^{\top}_{k}-\Omega_{k})F^{0}H\\
&=H\left(\frac{1}{T^{2}N}\right)\sum_{k=1}^{N}\sum_{t=1}^{T}\sum_ {s=1}^{T}f_{s}^0f_{t}^{0\top}[\varepsilon_{kt}\varepsilon_{ks}-E(\varepsilon_{kt}\varepsilon_{ks})]H\\
&=O_{p}\left(\frac{1}{T\sqrt{N}}\right)
\end{aligned}
\end{equation}
by Lemma \ref{lemma:B.2} (i). Next
\begin{equation}\nonumber
\begin{aligned}
b2&=H\frac{1}{NT^{2}}\sum_{k=1}^{N}F^{0\top}(\varepsilon_{k} \varepsilon^{\top}_{k}-\Omega_{k})(\widehat{F}-F^{0}H)\\
&=H\frac{1}{\sqrt{NT}}\frac{1}{T}\sum_{s=1}^{T}\left[\frac{1}{ \sqrt{NT}}\sum_{t=1}^{T}\sum_{k=1}^{N}f_{t}^{0}[\varepsilon_{kt}\varepsilon_{ks }-E(\varepsilon_{kt}\varepsilon_{ks})]\right](\widehat{f}_{s}-H^\top f_{s}^{0}).
\end{aligned}
\end{equation}
Thus if we let $A_{s}=\frac{1}{\sqrt{NT}}\sum_{t=1}^{T}\sum_{k=1}^{N}f_{t}^{0}[\varepsilon_{kt}\varepsilon_{ks}-E(\varepsilon_{kt}\varepsilon_{ks})]$, then
\begin{equation}\nonumber
\begin{aligned}
\|b2\| &\leq \|H\|\frac{1}{\sqrt{NT}}\left(\frac{1}{T}\sum_{s=1}^{T}\|A_{s}\| ^{2}\right)^{1/2}\left(\frac{1}{T}\sum_{s=1}^{T}\|\widehat{f}_{s}-H^\top f_{s}^{0}\|^ {2}\right)^{1/2}\\
&= \frac{1}{\sqrt{NT}}\left[O_p\left(\left(\frac{1}{N}\sum_{i=1}^N\left\|\widehat{\beta}_i-\beta_i^0\right\|^2\right)^{1/2}\right) + O_p(\delta_{NT}^{-1})\right].
\end{aligned}
\end{equation}
The term $b3$ has the same upper bound because it is the transpose of $b2$. The last term is
\begin{equation}\nonumber
\begin{aligned}
b4&=\frac{1}{NT^{2}}\sum_{k=1}^{N}(\widehat{F}-F^{0}H)^{\top}(\varepsilon _{k}\varepsilon^{\top}_{k}-\Omega_{k})(\widehat{F}-F^{0}H)\\
&=\frac{1}{\sqrt{N}}\frac{1}{T^{2}}\sum_{t=1}^{T}\sum_{s=1}^{T}( \widehat{f}_{t}-H^{\top}f_{t}^{0})(\widehat{f}_{s}-H^{\top}f_{s}^{0})^{\top}\left[\frac{1}{\sqrt{N}}\sum_{k=1}^{N}\left[\varepsilon _{kt}\varepsilon_{ks}-E(\varepsilon_{kt}\varepsilon_{ks})\right]\right].
\end{aligned}
\end{equation}
Thus by the Cauchy-Schwarz inequality,
\begin{equation}\nonumber
\begin{aligned}
\|b4\| &\leq\frac{1}{\sqrt{N}}\left(\frac{1}{T}\sum_{t=1}^{T}\|\widehat {f}_{t}-H^{ \top}f_{t}^{0}\|^{2}\right)\left(\frac{1}{T^{2}}\sum_{t=1}^{T}\sum_{s=1}^{T}\left[ \frac{1}{\sqrt{N}}\sum_{k=1}^{N}\left[\varepsilon_{kt}\varepsilon_{ks}-E( \varepsilon_{kt}\varepsilon_{ks})\right]\right]^{2}\right)^{1/2}\\
&=\frac{1}{\sqrt{N}}O_p\left(\frac{1}{N}\sum_{i=1}^N\left\|\widehat{\beta}_i-\beta_i^0\right\|^2\right)+\frac{1}{ \sqrt{N}}O_{p}(\delta_{NT}^{-2}).
\end{aligned}
\end{equation}
Now collecting terms yields the lemma. 
\Halmos
\endproof
\begin{lemma}\label{lemma:B.2}
Under Assumptions \ref{assum:Regularity}--\ref{assum:local-identification}, \ref{assum:error}--\ref{assum:identification}, there exists an $ M < \infty$ such that statements (i) and (ii) hold:
\begin{enumerate}[label=(\roman*)]
\item We have
$$
E\left\| N^{-1/2} \sum_{k=1}^{N} \frac{1}{T} \sum_{t=1}^{T} \sum_{s=1}^{T} f_s f_t^\top [\varepsilon_{kt} \varepsilon_{ks} - E(\varepsilon_{kt} \varepsilon_{ks})] \right\|^2 \leq M.
$$
\item For all $ i \in[N] $ and $ h \in[r] $, we have
$$
E\left\| N^{-1/2} \sum_{k=1}^{N} \frac{1}{T} \left\{ \sum_{t=1}^{T} \sum_{s=1}^{T} X_{it} [\varepsilon_{kt} \varepsilon_{ks} - E(\varepsilon_{kt} \varepsilon_{ks})] f_{hs} \right\} \right\|^2 \leq M.
$$
\end{enumerate}
\end{lemma}
\proof{Proof of Lemma \ref{lemma:B.2}.}
(i) Let 
$$
A = N^{-1/2} \sum_{k=1}^{N} \frac{1}{T} \sum_{t=1}^{T} \sum_{s=1}^{T} f_s f_t^\top [\varepsilon_{kt} \varepsilon_{ks} - {E}(\varepsilon_{kt} \varepsilon_{ks})].
$$
We need to show ${E}\|A\|^2 \leq M$. Note that for any matrix $A$, $\|A\|^2 = \tr(AA^\top)$. Therefore,
\begin{equation}\nonumber
\begin{aligned}
{E}\|A\|^2 &= {E}[\tr(AA^\top)] \\
&= {E}\left[ \tr\left( \frac{1}{N} \sum_{k=1}^{N} \sum_{\ell=1}^{N} \frac{1}{T^2} \sum_{t,s,u,v} f_s f_t^\top [\varepsilon_{kt} \varepsilon_{ks} - {E}(\varepsilon_{kt} \varepsilon_{ks})]   [\varepsilon_{\ell u} \varepsilon_{\ell v} - {E}(\varepsilon_{\ell u} \varepsilon_{\ell v})] f_u f_v^\top \right) \right].
\end{aligned}
\end{equation}
Using the cyclic property of the trace, we have
\begin{equation}\nonumber
\begin{aligned}
\tr(AA^\top) &= \frac{1}{N} \sum_{k,\ell} \frac{1}{T^2} \sum_{t,s,u,v} \tr\left(f_s f_t^\top [\varepsilon_{kt} \varepsilon_{ks} - {E}(\varepsilon_{kt} \varepsilon_{ks})] [\varepsilon_{\ell u} \varepsilon_{\ell v} - {E}(\varepsilon_{\ell u} \varepsilon_{\ell v})] f_u f_v^\top\right) \\
&= \frac{1}{N} \sum_{k,\ell} \frac{1}{T^2} \sum_{t,s,u,v} f_v^\top f_s \cdot f_t^\top f_u \cdot [\varepsilon_{kt} \varepsilon_{ks} - {E}(\varepsilon_{kt} \varepsilon_{ks})] \cdot [\varepsilon_{\ell u} \varepsilon_{\ell v} - {E}(\varepsilon_{\ell u} \varepsilon_{\ell v})].
\end{aligned}
\end{equation}
Taking expectations, we have
$$
{E}[\tr(AA^\top)]= \frac{1}{N} \sum_{k,\ell} \frac{1}{T^2} \sum_{t,s,u,v} f_v^\top f_s \cdot f_t^\top f_u \cdot \mathrm{cov}(\varepsilon_{kt} \varepsilon_{ks}, \varepsilon_{\ell u} \varepsilon_{\ell v}).
$$
By Assumption \ref{assum:error}, specifically,
$$
T^{-2} N^{-1} \sum_{t,s,u,v} \sum_{k,\ell} |\mathrm{cov}(\varepsilon_{kt} \varepsilon_{ks}, \varepsilon_{\ell u} \varepsilon_{\ell v})| \leq M,
$$
Combined with Assumption \ref{assum:factor-spectral}, we have
\begin{equation}\nonumber
\begin{aligned}
{E}\|A\|^2 &= {E}[\tr(AA^\top)]\\
&\leq \frac{1}{N} \sum_{k,\ell} \frac{1}{T^2} \sum_{t,s,u,v} \|f_v\| \|f_s\| \|f_t\| \|f_u\| \cdot |\mathrm{cov}(\varepsilon_{kt} \varepsilon_{ks}, \varepsilon_{\ell u} \varepsilon_{\ell v})|\\
&\leq M.
\end{aligned}
\end{equation}

(ii) For fixed $i$ and $h$, let
$$
B = N^{-1/2} \sum_{k=1}^{N} \frac{1}{T} \left\{ \sum_{t=1}^{T} \sum_{s=1}^{T} X_{it} [\varepsilon_{kt} \varepsilon_{ks} - {E}(\varepsilon_{kt} \varepsilon_{ks})] f_{hs} \right\}.
$$
We need to show ${E}\|B\|^2 \leq M$. Again, using $\|B\|^2 = \tr(BB^\top)$, we have
$$
\tr(BB^\top)= \frac{1}{N} \sum_{k,\ell} \frac{1}{T^2} \sum_{t,s,u,v} X_{it}^\top X_{iu} f_{hs} f_{hv} [\varepsilon_{kt} \varepsilon_{ks} - {E}(\varepsilon_{kt} \varepsilon_{ks})] [\varepsilon_{\ell u} \varepsilon_{\ell v} - {E}(\varepsilon_{\ell u} \varepsilon_{\ell v})].
$$
Taking expectations,
$$
{E}[\tr(BB^\top)] = \frac{1}{N} \sum_{k,\ell} \frac{1}{T^2} \sum_{t,s,u,v} X_{it}^\top X_{iu} f_{hs} f_{hv} \cdot \mathrm{cov}(\varepsilon_{kt} \varepsilon_{ks}, \varepsilon_{\ell u} \varepsilon_{\ell v}).
$$
By Assumption \ref{assum:error}, specifically,
$$
T^{-1} N^{-2} \sum_{t,s} \sum_{i,j,k,\ell} |\mathrm{cov}(\varepsilon_{it} \varepsilon_{jt}, \varepsilon_{ks} \varepsilon_{\ell s})| \leq M.
$$
Combined with Assumptions \ref{assum:factor-spectral} and \ref{assum:identification}, we obtain
\begin{equation}\nonumber
\begin{aligned}
{E}\|B\|^2 &= {E}[\tr(BB^\top)]\\
&\leq \frac{1}{N} \sum_{k,\ell} \frac{1}{T^2} \sum_{t,s,u,v} |X_{it}| |X_{iu}| |f_{hs}| |f_{hv}| \cdot |\mathrm{cov}(\varepsilon_{kt} \varepsilon_{ks}, \varepsilon_{\ell u} \varepsilon_{\ell v})|\\
&\leq M.
\end{aligned}
\end{equation}
This completes the proof.

\Halmos
\endproof
\begin{lemma}\label{lemma:B.3}
Under Assumptions \ref{assum:Regularity}--\ref{assum:local-identification}, \ref{assum:error}--\ref{assum:identification},
\begin{equation}\nonumber
\begin{aligned}
\frac{1}{\sqrt{T}} X_i^\top M_{\widehat{F}}\varepsilon_i - \frac{1}{\sqrt{T}}\frac{1}{N} \sum_{k=1}^{N} a_{ik} X_i^\top M_{\widehat{F}} \varepsilon_k=&\frac{1}{\sqrt{T}} X_i^\top M_{{F}^0}\varepsilon_i - \frac{1}{\sqrt{T}}\frac{1}{N} \sum_{k=1}^{N} a_{ik} X_i^\top M_{{F}^0} \varepsilon_k\\
&+\left(\frac{\sqrt{T}}{N}+1\right)\left[O_p\left(\left(\frac{1}{N}\sum_{i=1}^N\left\|\widehat{\beta}_i-\beta_i^0\right\|^2\right)^{1/2}\right) + O_p\left(\delta_{NT}^{-1}\right)\right]\\
&+\sqrt{T}\left[O_p\left(\frac{1}{N}\sum_{i=1}^N\left\|\widehat{\beta}_i-\beta_i^0\right\|^2\right)+O_p(\delta_{NT}^{-2})\right]+o_p(1).
\end{aligned}
\end{equation}
\end{lemma}
\proof{Proof of Lemma \ref{lemma:B.3}.}
First consider $\frac{1}{\sqrt{T}}X_{i}^\top(M_{F^{0}}-M_{\widehat{F}})\varepsilon_{i}$. Note that $M_{F^{0}}-M_{\widehat{F}}=P_{\widehat{F}}-P_{F^{0}}$ and $P_{\widehat{F}}=\widehat{F}\widehat{F}^{\top}/T$. By adding and subtracting terms,
\begin{equation}\nonumber
\begin{aligned}
\frac{1}{\sqrt{T}}\frac{X^{\top}_{i}\widehat{F}}{T}\,\widehat{F}^{ \top}\varepsilon_{i}-\frac{1}{\sqrt{T}}X^{\top}_{i}P_{F^{0}} \varepsilon_{i}=&\frac{1}{\sqrt{T}}\frac{X^{\top}_{i}(\widehat{F}-F^{0}H)}{T}H^{ \top}F^{0\top}\varepsilon_{i}\\
&+\frac{1}{\sqrt{T}}\frac{X^{\top}_{i}(\widehat{F}-F^{0}H)}{T}( \widehat{F}-F^{0}H)^{\top}\varepsilon_{i}\\
&+\frac{1}{\sqrt{T}}\frac{X^{\top}_{i}F^{0}H}{T}\,(\widehat{F}-F^{0 }H)^{\top}\varepsilon_{i}\\
&+\frac{1}{\sqrt{T}}\frac{X^{\top}_{i}F^{0}}{T}\left[HH^{ \top}-\left(\frac{F^{0\top}F^{0}}{T}\right)^{-1}\right]F^{0\top}\varepsilon _{i}\\
:=&a+b+c+d.
\end{aligned}
\end{equation}

Consider $a$. Note that $(\widehat{f}_{s}-H^{\top}f_{s}^{0})^{\top}H^{\top}f_{t}^{0}$ is scalar and thus commutable with $X_{it}$:
\begin{equation}\nonumber
\begin{aligned}
a=&\frac{1}{\sqrt{T}}\frac{X^{\top}_{i}(\widehat{F}-F^{0}H)}{T}H^{ \top}F^{0\top}\varepsilon_{i}\\
=&\frac{1}{T}\sum_{s=1}^{T}(\widehat{f}_{s}-H^{\top}f_{s}^{0})^{\top}H^{\top }\left(\frac{1}{\sqrt{T}}\sum_{t=1}^{T}f_{t}^{0}X_{is} \varepsilon_{it}\right).
\end{aligned}
\end{equation}
Thus
\begin{equation}\nonumber
\begin{aligned}
\left\|a\right\|=&\left\|\frac{1}{T}\sum_{s=1}^{T}(\widehat{f}_{s}-H^{\top}f_{s}^{0})^{\top}H^{\top }\left(\frac{1}{\sqrt{T}}\sum_{t=1}^{T}f_{t}^{0}X_{is} \varepsilon_{it}\right)\right\|\\
\leq&\left[\frac{1}{T}\sum_{s=1}^{T}\left\|\widehat{f}_{s}-H^{\top}f_{s}^{0}\right\| ^{2}\right]^{1/2}\left\|H\right\|\left[\frac{1}{T}\sum_{s=1}^{T}\left\|\left(\frac{1}{\sqrt{T} }\sum_{t=1}^{T}f_{t}^{0}X_{is}\varepsilon_{it}\right)\right\|^2\right]^{1/2}\\
=&O_p\left(\left(\frac{1}{N}\sum_{i=1}^N\left\|\widehat{\beta}_i-\beta_i^0\right\|^2\right)^{1/2}\right) + O_p\left(\delta_{NT}^{-1 }\right).\\
\end{aligned}
\end{equation}
Similarly,
\begin{equation}\nonumber
\begin{aligned}
b=&\frac{1}{\sqrt{T}}\frac{X^{\top}_{i}(\widehat{F}-F^{0}H)}{T}( \widehat{F}-F^{0}H)^{\top}\varepsilon_{i}\\
=&\sqrt{T}\frac{1}{T^2}\sum_{s=1}^T\sum_{t=1}^T(\widehat{f}_s - H^\top f_s^0)^\top(\widehat{f}_t - H^\top f_t^0)\left( X_{is} \varepsilon_{it}\right)
\end{aligned}
\end{equation}
and
\begin{equation}\nonumber
\begin{aligned}
\left\|b\right\|=&\left\|\sqrt{T}\frac{1}{T^2}\sum_{s=1}^T\sum_{t=1}^T(\widehat{f}_s - H^\top f_s^0)^\top(\widehat{f}_t - H^\top f_t^0)\left( X_{is} \varepsilon_{it}\right)\right\|\\
\leq&\sqrt{T}\left(\frac{1}{T}\sum_{t=1}^T\|\widehat{f}_t - H^\top f_t^0\|^2\right)\left(\frac{1}{T^2}\sum_{s=1}^T\sum_{t=1}^T\left\|  X_{it}\varepsilon_{it}\right\|^2\right)^{1/2}\\
=&\sqrt{T}\left[O_p\left(\frac{1}{N}\sum_{i=1}^N\left\|\widehat{\beta}_i-\beta_i^0\right\|^2\right)+O_p(\delta_{NT}^{-2})\right].
\end{aligned}
\end{equation}
Consider $c$:
\begin{equation}\nonumber
\begin{aligned}
c=&\frac{1}{\sqrt{T}}\frac{X^{\top}_{i}F^{0}H}{T}\,(\widehat{F}-F^{0 }H)^{\top}\varepsilon_{i}\\
=&\frac{1}{\sqrt{T}}\frac{X_i^\top F^0}{T}HH^\top (\widehat{F} H^{-1} - F^{0})^\top \varepsilon_i\\
=&\frac{1}{\sqrt{T}}\frac{X_i^\top F^0}{T}\left(\frac{F^{0\top}F^0}{T}\right)^{-1}(\widehat{F} H^{-1} - F^{0})^\top \varepsilon_i\\
&+\frac{1}{\sqrt{T}}\frac{X_i^\top F^0}{T}\left[HH^\top  - \left(\frac{F^{0\top}F^0}{T}\right)^{-1}\right](\widehat{F} H^{-1} - F^{0})^\top \varepsilon_i\\
:=&c1+c2.
\end{aligned}
\end{equation}
Denote $Q = HH^\top  - (F^{0\top}F^0/T)^{-1}$ for the moment. For the term $c2$,
\begin{equation}\nonumber
\begin{aligned}
c2=&\frac{1}{\sqrt{T}}\frac{X_i^\top F^0}{T}\left[HH^\top  - \left(\frac{F^{0\top}F^0}{T}\right)^{-1}\right](\widehat{F} H^{-1} - F^{0})^\top \varepsilon_i\\
=&\sqrt{T}\left(\frac{1}{T}\left[\varepsilon_i^\top (\widehat{F} H^{-1} - F_0)\otimes\left(\frac{X_i^\top F^0}{T}\right)\right]\right)\text{Vec}(Q)
\end{aligned}
\end{equation}
We first analyze $\text{Vec}(Q)$. From the Proposition \ref{pro:B.1}, we have
$$
T^{-1/2}\|\widehat{F}-F^0H\|=O_p\left(\left(\frac{1}{N}\sum_{i=1}^N\left\|\widehat{\beta}_i-\beta_i^0\right\|^2\right)^{1/2}\right) + O_p\left(\delta_{NT}^{-1}\right).
$$
Then, 
\begin{equation}\nonumber
\begin{aligned}
\frac{1}{T}\left\|F^{0\top}(\widehat{F}-F^0H)\right\|\leq&\frac{1}{\sqrt{T}}\left\|F^{0}\right\|\cdot\frac{1}{\sqrt{T}}\left\|\widehat{F}-F^0H\right\|\\
=&O_p\left(\left(\frac{1}{N}\sum_{i=1}^N\left\|\widehat{\beta}_i-\beta_i^0\right\|^2\right)^{1/2}\right) + O_p\left(\delta_{NT}^{-1}\right)
\end{aligned}
\end{equation}
and
\begin{equation}\nonumber
\begin{aligned}
\frac{1}{T}\left\|\widehat{F}^{\top}(\widehat{F}-F^0H)\right\|\leq&\frac{1}{\sqrt{T}}\left\|\widehat{F}\right\|\cdot\frac{1}{\sqrt{T}}\left\|\widehat{F}-F^0H\right\|\\
=&O_p\left(\left(\frac{1}{N}\sum_{i=1}^N\left\|\widehat{\beta}_i-\beta_i^0\right\|^2\right)^{1/2}\right) + O_p\left(\delta_{NT}^{-1}\right).
\end{aligned}
\end{equation}
The two results can be rewritten as
$$
\left\|\frac{F^{0\top}\widehat{F}}{T}-\frac{F^{0\top}F^{0}}{T}H\right\|=O_p\left(\left(\frac{1}{N}\sum_{i=1}^N\left\|\widehat{\beta}_i-\beta_i^0\right\|^2\right)^{1/2}\right) + O_p\left(\delta_{NT}^{-1}\right)
$$
and
$$
\left\|I_{r^0}-\frac{\widehat{F}^{\top}F^{0}}{T}H\right\|=O_p\left(\left(\frac{1}{N}\sum_{i=1}^N\left\|\widehat{\beta}_i-\beta_i^0\right\|^2\right)^{1/2}\right) + O_p\left(\delta_{NT}^{-1}\right).
$$
Left multiply the first equation by $H^{\top}$ and use the transpose of the second equation to obtain
\begin{equation}\nonumber
\begin{aligned}
\left\|H^{\top}\left(\frac{F^{0\top}\widehat{F}}{T}-\frac{F^{0\top}F^{0}}{T}H\right)+\left(I_{r^0}-\frac{\widehat{F}^{\top}F^{0}}{T}H\right)^\top\right\|\leq&\left\|\frac{F^{0\top}\widehat{F}}{T}-\frac{F^{0\top}F^{0}}{T}H\right\|+\left\|I_{r^0}-\frac{\widehat{F}^{\top}F^{0}}{T}H\right\|\\
=&O_p\left(\left(\frac{1}{N}\sum_{i=1}^N\left\|\widehat{\beta}_i-\beta_i^0\right\|^2\right)^{1/2}\right) + O_p\left(\delta_{NT}^{-1}\right).
\end{aligned}
\end{equation}
That is,
$$
\left\|I_{r^0}-H^{\top}\frac{F^{0\top}F^{0}}{T}H\right\|=O_p\left(\left(\frac{1}{N}\sum_{i=1}^N\left\|\widehat{\beta}_i-\beta_i^0\right\|^2\right)^{1/2}\right) + O_p\left(\delta_{NT}^{-1}\right).
$$
Right multiplying by $H^{\top}$ and left multiplying by $H^{\top-1}$, we obtain
$$
\left\|I_{r^0}-\frac{F^{0\top}F^{0}}{T}HH^\top\right\|=O_p\left(\left(\frac{1}{N}\sum_{i=1}^N\left\|\widehat{\beta}_i-\beta_i^0\right\|^2\right)^{1/2}\right) + O_p\left(\delta_{NT}^{-1}\right).
$$
Then,
$$
\left\|\text{Vec}(Q)\right\|=O_p\left(\left(\frac{1}{N}\sum_{i=1}^N\left\|\widehat{\beta}_i-\beta_i^0\right\|^2\right)^{1/2}\right) + O_p\left(\delta_{NT}^{-1}\right).
$$
Moreover,
\begin{equation}\nonumber
\begin{aligned}
\left\|\frac{1}{T}\varepsilon_i^\top (\widehat{F} H^{-1} - F_0)\right\|\leq&\frac{1}{T}\left\|\varepsilon_i\right\|\cdot\frac{1}{T}\left\|\widehat{F} H^{-1} - F_0\right\|\\
=&O_p\left(\left(\frac{1}{N}\sum_{i=1}^N\left\|\widehat{\beta}_i-\beta_i^0\right\|^2\right)^{1/2}\right) + O_p\left(\delta_{NT}^{-1}\right).
\end{aligned}
\end{equation}
Therefore,
\begin{equation}\nonumber
\begin{aligned}
\left\|c2\right\|&=\sqrt{T}\left[O_p\left(\left(\frac{1}{N}\sum_{i=1}^N\left\|\widehat{\beta}_i-\beta_i^0\right\|^2\right)^{1/2}\right) + O_p\left(\delta_{NT}^{-1}\right)\right]^2\\
&=\sqrt{T} O_p\left(\frac{1}{N}\sum_{i=1}^N\left\|\widehat{\beta}_i-\beta_i^0\right\|^2\right)+ \sqrt{T} O_p (\delta_{NT}^{-2}).
\end{aligned}
\end{equation}
For the term $c1$, by Lemma \ref{lemma:B.4}, 
\begin{equation}\nonumber
\begin{aligned}
\left\|c1\right\|=&\left\|\frac{1}{\sqrt{T}}\frac{X_i^\top F^0}{T}\left(\frac{F^{0\top}F^0}{T}\right)^{-1}(\widehat{F} H^{-1} - F^{0})^\top \varepsilon_i\right\|\\
\leq&\left\|\frac{1}{\sqrt{T}}\varepsilon_i^\top(\widehat{F}-F^0H)\right\|\\
=&O_p\left(\left(\frac{1}{N}\sum_{i=1}^N\left\|\widehat{\beta}_i-\beta_i^0\right\|^2\right)^{1/2}\right) + O_p\left(\delta_{NT}^{-1}\right)\\
&+\frac{\sqrt{T}}{N}\left[O_p\left(\left(\frac{1}{N}\sum_{i=1}^N\left\|\widehat{\beta}_i-\beta_i^0\right\|^2\right)^{1/2}\right) + O_p\left(\delta_{NT}^{-1}\right)\right]+o_p(1).
\end{aligned}
\end{equation}
For $d$, again let $Q = HH^\top  - (F^{0\top} F^0 / T)^{-1}$. Then
\begin{equation}\nonumber
\begin{aligned}
d=&\frac{1}{\sqrt{T}}\frac{X^{\top}_{i}F^{0}}{T}\left[HH^{ \top}-\left(\frac{F^{0\top}F^{0}}{T}\right)^{-1}\right]F^{0\top}\varepsilon _{i}\\
=&\frac{1}{\sqrt{T}}\left[\varepsilon_i^\top F^0\otimes\left( \frac{X^\top_i F^0}{T} \right)\right]\text{Vec}(Q)\\
=&O_p(1) \text{Vec}(Q)\\
=&O_p\left(\left(\frac{1}{N}\sum_{i=1}^N\left\|\widehat{\beta}_i-\beta_i^0\right\|^2\right)^{1/2}\right) + O_p\left(\delta_{NT}^{-1}\right).
\end{aligned}
\end{equation}
In summary,
\begin{equation}\label{eq:B.6}
\begin{aligned}
\frac{1}{\sqrt{T}}X_{i}^\top(M_{F^{0}}-M_{\widehat{F}})\varepsilon_{i}=&\left(\frac{\sqrt{T}}{N}+1\right)\left[O_p\left(\left(\frac{1}{N}\sum_{i=1}^N\left\|\widehat{\beta}_i-\beta_i^0\right\|^2\right)^{1/2}\right) + O_p\left(\delta_{NT}^{-1}\right)\right]\\
&+\sqrt{T}\left[O_p\left(\frac{1}{N}\sum_{i=1}^N\left\|\widehat{\beta}_i-\beta_i^0\right\|^2\right)+O_p(\delta_{NT}^{-2})\right]+o_p(1).
\end{aligned}
\end{equation}
Let $V_{ik} =  a_{ik} X_i^\top$. Then replacing $X_i^\top$ with $V_{ik}$, the same argument leads to
\begin{equation}\label{eq:B.7}
\begin{aligned}
\frac{1}{\sqrt{T}}\frac{1}{N} \sum_{k=1}^{N}  V_{ik}(M_{F^0}-M_{\widehat{F}}) \varepsilon_k=&\left(\frac{\sqrt{T}}{N}+1\right)\left[O_p\left(\left(\frac{1}{N}\sum_{i=1}^N\left\|\widehat{\beta}_i-\beta_i^0\right\|^2\right)^{1/2}\right) + O_p\left(\delta_{NT}^{-1}\right)\right]\\
&+\sqrt{T}\left[O_p\left(\frac{1}{N}\sum_{i=1}^N\left\|\widehat{\beta}_i-\beta_i^0\right\|^2\right)+O_p(\delta_{NT}^{-2})\right]+o_p(1).
\end{aligned}
\end{equation}
Combining \eqref{eq:B.6} and \eqref{eq:B.7}, we obtain the lemma:

\begin{equation}\nonumber
\begin{aligned}
\frac{1}{\sqrt{NT}} \sum_{i=1}^N \left[ X^\top_i M_{\widehat{F}} - \frac{1}{N} \sum_{k=1}^N a_{ik} X^\top_k M_{\widehat{F}} \right] \varepsilon_i&=\frac{1}{\sqrt{NT}} \sum_{i=1}^N \left[ X^\top_i M_{F^0} - \frac{1}{N} \sum_{k=1}^N a_{ik} X^\top_k M_{F^0} \right] \varepsilon_i\\
&\quad-\left(\frac{\sqrt{NT}}{N}\right)(\psi_{NT}-\psi^{\ast}_{NT})+\sqrt{T}O_p(\| \widehat{\beta}-\beta^{0}\|^{2})\\
&\quad+O_p(\|\widehat{\beta}-\beta^{0}\|)+\sqrt{T}O_p(\delta^{-2}_{NT}).
\end{aligned}
\end{equation}

\Halmos
\endproof
\begin{lemma}\label{lemma:B.4}
Under Assumptions \ref{assum:Regularity}--\ref{assum:local-identification}, \ref{assum:error}--\ref{assum:identification}, we have
\begin{equation}\nonumber
\begin{aligned}
\frac{1}{\sqrt{T}}\varepsilon_i^\top(\widehat{F}H^{-1}-F^0)=&O_p\left(\left(\frac{1}{N}\sum_{i=1}^N\left\|\widehat{\beta}_i-\beta_i^0\right\|^2\right)^{1/2}\right) + O_p\left(\delta_{NT}^{-1}\right)\\
&+\frac{\sqrt{T}}{N}\left[O_p\left(\left(\frac{1}{N}\sum_{i=1}^N\left\|\widehat{\beta}_i-\beta_i^0\right\|^2\right)^{1/2}\right) + O_p\left(\delta_{NT}^{-1}\right)\right]+o_p(1).
\end{aligned}
\end{equation}
\end{lemma}

\proof{Proof of Lemma \ref{lemma:B.4}.}
From \eqref{eq:B.1} and denoting $G=(F^{0\top}\widehat{F}/T)^{-1}(\Lambda^{0\top}\Lambda^0/N)^{-1}$ for the moment,
\begin{equation}\nonumber
\begin{aligned}
\frac{1}{\sqrt{T}}\varepsilon^{\top}_{i}(\widehat{F} H^{-1}-F^{0}) =&\frac{1}{\sqrt{T}}\varepsilon^{\top}_{k}(I1+\cdots+ I8)G\\
:=&a1+\cdots+a8.
\end{aligned}
\end{equation}
For the first four terms, we have
\begin{equation}\nonumber
\begin{aligned}
\left\|a1\right\|
=&\left\|\frac{1}{\sqrt{T}}\varepsilon^{\top}_{k}\left(\frac{1}{NT} \sum_{i=1}^N X_i (\beta_i^0 - \widehat{\beta}_i) (\beta_i^0 - \widehat{\beta}_i)^\top  X_i^\top  \widehat{F}\right)G\right\|\\
\leq&\frac{1}{N}\sum_{i=1}^N\left\|\frac{\varepsilon^{\top}_{k}X_i}{\sqrt{T}}\right\|\left\|\widehat{\beta}_i-\beta_i^0\right\|^2\\
=&\frac{1}{N}\sum_{i=1}^N\left\|\frac{1}{\sqrt{T}}\sum_{t=1}^TX_{it}\varepsilon_{kt}\right\|\left\|\widehat{\beta}_i-\beta_i^0\right\|^2\\
=&O_p\left(\frac{1}{N}\sum_{i=1}^N\left\|\widehat{\beta}_i-\beta_i^0\right\|^2\right),
\end{aligned}
\end{equation}
\begin{equation}\nonumber
\begin{aligned}
\left\|a2\right\|
=&\left\|\frac{1}{\sqrt{T}}\varepsilon^{\top}_{k}\left(\frac{1}{NT} \sum_{i=1}^N X_i (\beta_i^0 - \widehat{\beta}_i) \lambda_i^{0\top}  F^{0\top} \widehat{F}\right)G\right\|\\
=&\left\|\frac{1}{N\sqrt{T}}\sum_{i=1}^N\varepsilon^{\top}_{k}X_i (\beta_i^0 - \widehat{\beta}_i) \lambda_i^{0\top}\left(\frac{\Lambda ^{0\top}\Lambda^0}{N}\right)\right\|\\
\leq&\frac{1}{N}\sum_{i=1}^N\left\|\frac{\varepsilon^{\top}_{k}X_i}{\sqrt{T}}\right\|\left\|\widehat{\beta}_i-\beta_i^0\right\|\\
=&O_p\left(\left(\frac{1}{N}\sum_{i=1}^N\left\|\widehat{\beta}_i-\beta_i^0\right\|^2\right)^{1/2}\right),
\end{aligned}
\end{equation}
\begin{equation}\nonumber
\begin{aligned}
\left\|a3\right\|
=&\left\|\frac{1}{\sqrt{T}}\varepsilon^{\top}_{k}\left(\frac{1}{NT} \sum_{i=1}^N X_i (\beta_i^0 - \widehat{\beta}_i) \varepsilon_i^\top \widehat{F}\right)G\right\|\\
\leq&\frac{1}{N}\sum_{i=1}^N\left\|\frac{\varepsilon^{\top}_{k}X_i}{\sqrt{T}}\right\|\left\|\widehat{\beta}_i-\beta_i^0\right\|\left\|\frac{\varepsilon_{i}}{\sqrt{T}}\right\|\\
=&O_p\left(\left(\frac{1}{N}\sum_{i=1}^N\left\|\widehat{\beta}_i-\beta_i^0\right\|^2\right)^{1/2}\right),
\end{aligned}
\end{equation}
\begin{equation}\nonumber
\begin{aligned}
\left\|a4\right\|
=&\left\|\frac{1}{\sqrt{T}}\varepsilon^{\top}_{k}\left(\frac{1}{NT} \sum_{i=1}^N F^0 \lambda_i^0 (\beta_i^0 - \widehat{\beta}_i)^\top  X_i^\top  \widehat{F}\right)G\right\|\\
\leq&\frac{1}{N}\sum_{i=1}^N\left\|\frac{\varepsilon^{\top}_{k}F^0}{\sqrt{T}}\right\|\left\|\widehat{\beta}_i-\beta_i^0\right\|\left\|\frac{X_i}{\sqrt{T}}\right\|\\
=&\frac{1}{N}\sum_{i=1}^N\left\|\frac{1}{\sqrt{T}}\sum_{t=1}^Tf_t^0\varepsilon_{kt}\right\|\left\|\widehat{\beta}_i-\beta_i^0\right\|\left\|\frac{X_i}{\sqrt{T}}\right\|\\
=&O_p\left(\left(\frac{1}{N}\sum_{i=1}^N\left\|\widehat{\beta}_i-\beta_i^0\right\|^2\right)^{1/2}\right).
\end{aligned}
\end{equation}
For $ a5 $, let $ W_i = X_i^\top  \widehat{F}/T $ and note that $\|W_i\|^2 \leq \|X_i\|^2/T$:
\begin{equation}\nonumber
\begin{aligned}
\left\|a5\right\|
=&\left\|\frac{1}{\sqrt{T}}\varepsilon^{\top}_{k}\left(\frac{1}{NT} \sum_{i=1}^N \varepsilon_i (\beta_i^0 - \widehat{\beta}_i)^\top  X_i^\top  \widehat{F}\right)G\right\|\\
=&\left\|\frac{1}{N\sqrt{T}}\sum_{i=1}^N\varepsilon^{\top}_{k}\varepsilon_i (\beta_i^0 - \widehat{\beta}_i)^\top W_iG\right\|\\
\leq&\frac{1}{N}\sum_{i=1}^N\left\|\frac{1}{\sqrt{T}}\sum_{t=1}^T\varepsilon_{kt}\varepsilon_{it}\right\|\left\|\widehat{\beta}_i-\beta_i^0\right\|\\
=&O_p\left(\left(\frac{1}{N}\sum_{i=1}^N\left\|\widehat{\beta}_i-\beta_i^0\right\|^2\right)^{1/2}\right).
\end{aligned}
\end{equation}
For $ a6 $,
\begin{equation}\nonumber
\begin{aligned}
a6
=&\frac{1}{\sqrt{T}}\varepsilon^{\top}_{k}\left(\frac{1}{NT} \sum_{i=1}^N F^0 \lambda_i^0 \varepsilon_i^\top \widehat{F}\right)G\\
=&\left(\frac{1}{\sqrt{T}}\varepsilon^{\top}_{k}F^0\right)\left(\frac{1}{NT} \sum_{i=1}^N  \lambda_i^0 \varepsilon_i^\top F^0HG\right)+\left(\frac{1}{\sqrt{T}}\varepsilon^{\top}_{k}F^0\right)\frac{1}{NT} \sum_{i=1}^N  \lambda_i^0 \varepsilon_i^\top (\widehat{F}-F^0H)G\\
:=&a6.1 + a6.2.
\end{aligned}
\end{equation}
For terms $a6.1$ and $a6.2$, we have
\begin{equation}\nonumber
\begin{aligned}
\left\|a6.1\right\|
\leq&\frac{1}{\sqrt{NT}}\left\|\frac{1}{\sqrt{T}}\sum_{t=1}^Tf^0_t\varepsilon_{kt}\right\|\left\|\frac{1}{\sqrt{NT}}\sum_{i=1}^N\sum_{t=1}^T\lambda_i^0f_t^{0\top}\varepsilon_{it}\right\|=O_p\left(\frac{1}{\sqrt{NT}}\right)
\end{aligned}
\end{equation}
and
\begin{equation}\nonumber
\begin{aligned}
\left\|a6.2\right\|=&\left\|\left(\frac{1}{\sqrt{T}}\varepsilon^{\top}_{k}F^0\right)\frac{1}{NT} \sum_{i=1}^N  \lambda_i^0 \varepsilon_i^\top (\widehat{F}-F^0H)G\right\|\\
\leq&\frac{1}{\sqrt{N}}\left\|\frac{1}{\sqrt{T}}\sum_{t=1}^Tf^0_t\varepsilon_{kt}\right\|\left\|\frac{1}{\sqrt{NT}}\sum_{i=1}^N  \lambda_i^0 \varepsilon_i^\top\right\|\left\|\frac{\widehat{F}-F^0H}{\sqrt{T}}\right\|\\
=&\frac{1}{\sqrt{N}}\left[O_p\left(\left(\frac{1}{N}\sum_{i=1}^N\left\|\widehat{\beta}_i-\beta_i^0\right\|^2\right)^{1/2}\right) + O_p\left(\delta_{NT}^{-1}\right)\right].
\end{aligned}
\end{equation}
Next consider $a7$ and $a8$:
\begin{equation}\nonumber
\begin{aligned}
a7
=&\frac{1}{\sqrt{T}}\varepsilon^{\top}_{k}\left( \frac{1}{NT} \sum_{i=1}^N \varepsilon_i \lambda_i^{0\top}  F^{0\top} \widehat{F}\right)G\\
=&\frac{1}{\sqrt{T}}\frac{1}{NT} \sum_{i=1}^N \varepsilon^{\top}_{k}\varepsilon_i \lambda_i^{0\top}\left(\frac{\Lambda^{0\top}\Lambda^0}{N}\right)^{-1}\\
=&\frac{1}{\sqrt{NT}}\frac{1}{T}\sum_{t=1}^T \left[\varepsilon_{kt}\left(\frac{1}{\sqrt{N}}\sum_{i=1}^N \varepsilon_{it}\lambda_i^{0\top}\right)\right]\left(\frac{\Lambda^{0\top}\Lambda^0}{N}\right)^{-1}\\
=&O_p\left(\frac{1}{\sqrt{NT}}\right),
\end{aligned}
\end{equation}
\begin{equation}\nonumber
\begin{aligned}
a8
=&\frac{1}{\sqrt{T}}\varepsilon^{\top}_{k}\left(\frac{1}{NT} \sum_{i=1}^N \varepsilon_i \varepsilon_i^\top \widehat{F}\right)G\\
=&\frac{1}{\sqrt{T}}\frac{1}{NT}\sum_{i=1}^N\varepsilon^{\top}_{k}\varepsilon_i (\varepsilon_i^\top \widehat{F})G\\
=&\frac{1}{\sqrt{T}}\frac{1}{NT}\sum_{i=1}^N\varepsilon^{\top}_{k}\varepsilon_i (\varepsilon_i^\top F^0)HG+\frac{1}{\sqrt{T}}\frac{1}{NT}\sum_{i=1}^N\varepsilon^{\top}_{k}\varepsilon_i (\varepsilon_i^\top (\widehat{F}-F^0H))G\\
:=&b8+c8,
\end{aligned}
\end{equation}
\begin{equation}\nonumber
b8=\frac{1}{\sqrt{T}}\frac{1}{NT}\sum_{i=1}^N\varepsilon^{\top}_{k}\varepsilon_i (\varepsilon_i^\top F^0)HG=O_p\left(\frac{1}{N}+\frac{1}{\sqrt{T}}\right)
\end{equation}
because $\sum_{i=1}^N\varepsilon^{\top}_{k}\varepsilon_i=O_p(T+N\sqrt{T})$.
Moreover,
\begin{equation}\nonumber
\begin{aligned}
\left\|c8\right\|=&\left\|\frac{1}{\sqrt{T}}\frac{1}{NT}\sum_{i=1}^N\varepsilon^{\top}_{k}\varepsilon_i (\varepsilon_i^\top (\widehat{F}-F^0H))G\right\|\\
\leq&\frac{1}{N\sqrt{T}}\sum_{i=1}^N\left\|\varepsilon^{\top}_{k}\varepsilon_i\right\|\left\|\frac{\widehat{F}-F^0H}{\sqrt{T}}\right\|\\
=&\left(\frac{\sqrt{T}}{N}+1\right)\left[O_p\left(\left(\frac{1}{N}\sum_{i=1}^N\left\|\widehat{\beta}_i-\beta_i^0\right\|^2\right)^{1/2}\right) + O_p\left(\delta_{NT}^{-1}\right)\right].
\end{aligned}
\end{equation}
For $a9$ and $a10$,
\begin{equation}\nonumber
\begin{aligned}
a9
=&\frac{1}{\sqrt{T}}\varepsilon^{\top}_{k}\left(\frac{1}{NT} \sum_{i=1}^N\left\{Z_i^\ast(\widehat{\alpha}_{i})-Z_i^\ast({\alpha}_{i}^0)\right\}\left\{Z_i^\ast(\widehat{\alpha}_{i})-Z_i^\ast({\alpha}_{i}^0)\right\}^\top\widehat{F}\right)G=O_p\left(\frac{1}{\sqrt{T}}\right),
\end{aligned}
\end{equation}
\begin{equation}\nonumber
\begin{aligned}
\left\|a10\right\|
=&\left\|\frac{1}{\sqrt{T}}\varepsilon^{\top}_{k}\left(\frac{1}{NT} \sum_{i=1}^N\left\{Z_i^\ast(\widehat{\alpha}_{i})-Z_i^\ast({\alpha}_{i}^0)\right\}(\beta_i^0 - \widehat{\beta}_i)^\top X_i^\top\widehat{F}\right)G\right\|=O_p\left(\left(\frac{1}{N}\sum_{i=1}^N\left\|\widehat{\beta}_i-\beta_i^0\right\|^2\right)^{1/2}\right).
\end{aligned}
\end{equation}
By Lemma \ref{lemma:B.5},
\begin{equation}\nonumber
\begin{aligned}
a11
=&\frac{1}{\sqrt{T}}\varepsilon^{\top}_{k}\left(\frac{1}{NT} \sum_{i=1}^N\left\{Z_i^\ast(\widehat{\alpha}_{i})-Z_i^\ast({\alpha}_{i}^0)\right\}\lambda_i^{0\top}  F^{0\top}\widehat{F}\right)G\\
=&\frac{1}{N} \sum_{i=1}^N\left(\frac{1}{\sqrt{T}}\varepsilon^{\top}_{k}\right)\left(\frac{1}{\sqrt{T}}\left\{Z_i^\ast(\widehat{\alpha}_{i})-Z_i^\ast({\alpha}_{i}^0)\right\}\lambda_i^{0\top}  F^{0\top}\right)\left(\frac{1}{\sqrt{T}}\widehat{F}\right)G\\
=&o_p(1).
\end{aligned}
\end{equation}
For terms $a12$-$a15$,
\begin{equation}\nonumber
\begin{aligned}
\left\|a12\right\|
=&\left\|\frac{1}{\sqrt{T}}\varepsilon^{\top}_{k}\left(\frac{1}{NT} \sum_{i=1}^N\left\{Z_i^\ast(\widehat{\alpha}_{i})-Z_i^\ast({\alpha}_{i}^0)\right\}\varepsilon_i^\top\widehat{F}\right)G\right\|\\
\leq&\frac{1}{NT} \frac{1}{\sqrt{T}}\sum_{i=1}^N\left\|\varepsilon_{k}\varepsilon_i^\top\widehat{F}G\right\|\\
=&O_p\left(\frac{1}{N}+\frac{1}{\sqrt{T}}\right),
\end{aligned}
\end{equation}
\begin{equation}\nonumber
\begin{aligned}
\left\|a13\right\|
=&\left\|\frac{1}{\sqrt{T}}\varepsilon^{\top}_{k}\left(\frac{1}{NT} \sum_{i=1}^NX_i (\beta_i^0 - \widehat{\beta}_i)\left\{Z_i^\ast(\widehat{\alpha}_{i})-Z_i^\ast({\alpha}_{i}^0)\right\}^\top\widehat{F}\right)G\right\|=O_p\left(\left(\frac{1}{N}\sum_{i=1}^N\left\|\widehat{\beta}_i-\beta_i^0\right\|^2\right)^{1/2}\right),
\end{aligned}
\end{equation}
\begin{equation}\nonumber
\begin{aligned}
\left\|a14\right\|
=&\left\|\frac{1}{\sqrt{T}}\varepsilon^{\top}_{k}\left(\frac{1}{NT} \sum_{i=1}^NF^0 \lambda_i^0\left\{Z_i^\ast(\widehat{\alpha}_{i})-Z_i^\ast({\alpha}_{i}^0)\right\}^\top\widehat{F}\right)G\right\|=O_p\left(\frac{1}{\sqrt{T}}\right),
\end{aligned}
\end{equation}
\begin{equation}\nonumber
\begin{aligned}
\left\|a15\right\|
=&\left\|\frac{1}{\sqrt{T}}\varepsilon^{\top}_{k}\left(\frac{1}{NT} \sum_{i=1}^N\varepsilon_i\left\{Z_i^\ast(\widehat{\alpha}_{i})-Z_i^\ast({\alpha}_{i}^0)\right\}^\top\widehat{F}\right)G\right\|\\
\leq&\frac{1}{NT} \sum_{i=1}^N\left\|\varepsilon^{\top}_{k}\varepsilon_i\right\|\\
=&O_p\left(\frac{1}{N}+\frac{1}{\sqrt{T}}\right).
\end{aligned}
\end{equation}
Hence, we have
\begin{equation}\nonumber
\begin{aligned}
\frac{1}{\sqrt{T}}\varepsilon_i^\top(\widehat{F}H^{-1}-F^0)=&O_p\left(\left(\frac{1}{N}\sum_{i=1}^N\left\|\widehat{\beta}_i-\beta_i^0\right\|^2\right)^{1/2}\right) + O_p\left(\delta_{NT}^{-1}\right)\\
&+\frac{\sqrt{T}}{N}\left[O_p\left(\left(\frac{1}{N}\sum_{i=1}^N\left\|\widehat{\beta}_i-\beta_i^0\right\|^2\right)^{1/2}\right) + O_p\left(\delta_{NT}^{-1}\right)\right]+o_p(1).
\end{aligned}
\end{equation}
\Halmos
\endproof
\begin{lemma}\label{lemma:B.5}
Under Assumptions \ref{assum:Regularity}--\ref{assum:local-identification}, \ref{assum:error}--\ref{assum:identification}, we have
\begin{enumerate}[label=(\roman*)]
\item $\displaystyle\frac{1}{\sqrt{T}}X_i^\top M_{\widehat{F}}\left\{Z_k^\ast(\widehat{\alpha}_{k})-Z_k^\ast({\alpha}_{k}^0)\right\}=o_p(1)$, for $i,k\in[N]$.
\item $\displaystyle\frac{1}{\sqrt{T}}\left\{Z_k^\ast(\widehat{\alpha}_{k})-Z_k^\ast({\alpha}_{k}^0)\right\}^\top\widehat{F}=o_p(1)$, for $k\in[N]$.
\item $\displaystyle\frac1{\sqrt T}(F^0\lambda_i^0)^{\top}\{Z_i^\ast(\widehat\alpha_i)-Z_i^\ast(\alpha_i^0)\}=o_p(1)$, for $i\in[N]$.
\end{enumerate}
\end{lemma}
\proof{Proof of Lemma \ref{lemma:B.5}.}
Write $\delta_k=\widehat\alpha_k-\alpha_k^0$, and $\epsilon_{kt}=Y_{kt}-X_{kt}^{\top}\alpha_k^0$. For $\delta\in\mathbb R^{p+1}$, set $v_{kt}(\delta)=X_{kt}^{\top}\delta$. Let $\rho_\tau(u)=u\{\tau-\mathbbm{1}(u\leq0)\}$, and $\psi(u)=\tau-\mathbbm{1}(u\leq0)$. The ES generated response and the quantile check loss are related by $Z_{kt}^\ast(\alpha_k^0+\delta)-Z_{kt}^\ast(\alpha_k^0)=\tau^{-1}\{\rho_\tau(\epsilon_{kt})-\rho_\tau(\epsilon_{kt}-v_{kt}(\delta))\}$. Knight's identity states that
$$
\rho_\tau(u-v)-\rho_\tau(u)=-v\psi(u)+\int_0^v\{\mathbbm{1}(u\leq s)-\mathbbm{1}(u\leq0)\}\,ds.
$$
Let $m_{kt}=1-\tau^{-1}\mathbbm{1}\{\epsilon_{kt}\leq0\}$, and $D_k=\operatorname{diag}(m_{k1},\ldots,m_{kT})$. We obtain the following exact expansion
\begin{equation}\label{eq:B5-exact-identity}
Z_{kt}^\ast(\alpha_k^0+\delta)-Z_{kt}^\ast(\alpha_k^0)=m_{kt}v_{kt}(\delta)+r_{kt}(\delta),
\end{equation}
where
\begin{equation}\nonumber
\begin{aligned}
r_{kt}(\delta)&=-\tau^{-1}\int_0^{v_{kt}(\delta)}\left[\mathbbm{1}\{\epsilon_{kt}\leq s\}-\mathbbm{1}\{\epsilon_{kt}\leq0\}\right]ds\\
&=\tau^{-1}\{\epsilon_{kt}-v_{kt}(\delta)\}\left[\mathbbm{1}\{\epsilon_{kt}\leq v_{kt}(\delta)\}-\mathbbm{1}\{\epsilon_{kt}\leq0\}\right].
\end{aligned}
\end{equation}
In particular, we have
\begin{equation}
\label{eq:B5-remainder-envelope}
|r_{kt}(\delta)|\leq\tau^{-1}|v_{kt}(\delta)|\mathbbm{1}\{|\epsilon_{kt}|\leq|v_{kt}(\delta)|\}.
\end{equation}
The conditional quantile restriction implies $E(m_{kt}\mid\mathcal C_{kt})=0$, where $\mathcal C_{kt}=\sigma(X_{kt},f_{t}^{0})$. Moreover, differentiating the conditional expectation, rather than the sample indicator, gives
\begin{equation}\nonumber
E\{Z_{kt}^\ast(\alpha_k^0+\delta)-Z_{kt}^\ast(\alpha_k^0)\mid\mathcal C_{kt}\}=\tau^{-1}\int_0^{v_{kt}(\delta)}\{\tau-F_{kt}(s\mid\mathcal C_{kt})\}\,ds.
\end{equation}
Consequently, by the Neyman-orthogonality property, the derivative of the population generated response with respect to the first-stage quantile parameter is zero at $\delta=0$. Therefore, we have
\begin{equation}\nonumber
\left|E\{Z_{kt}^\ast(\alpha_k^0+\delta)-Z_{kt}^\ast(\alpha_k^0)\mid\mathcal C_{kt}\}\right|\leq\frac{\overline f}{2\tau}|v_{kt}(\delta)|^2.
\end{equation}

Fix $M<\infty$ and let $a_T=M/\sqrt T$. From \eqref{eq:B5-remainder-envelope},
\begin{equation}\nonumber
\begin{aligned}
E\left[\sup_{\|\delta\|\leq a_T}\sum_{t=1}^T r_{kt}(\delta)^2\right]\nonumber&\leq\tau^{-2}a_T^2\sum_{t=1}^T E\left[\|X_{kt}\|^2\mathbbm{1}\{|\epsilon_{kt}|\leq a_T\|X_{kt}\|\}\right]\nonumber\\
&\leq\frac{2\overline f}{\tau^2}a_T^3\sum_{t=1}^TE\|X_{kt}\|^3\\
&=O(T^{-1/2}).
\end{aligned}
\end{equation}
Markov's inequality and $\|\delta_k\|=O_p(T^{-1/2})$ therefore yield
\begin{equation}
\label{eq:B5-r-norm}
\|r_k(\delta_k)\|=O_p(T^{-1/4})=o_p(1),
\end{equation}
where $r_k(\delta)=(r_{k1}(\delta),\ldots,r_{kT}(\delta))^{\top}$. Because $\|D_k\|_{\mathrm{op}}\leq1+\tau^{-1}$,
\eqref{eq:B5-exact-identity} also gives
\begin{equation}
\label{eq:B5-Z-norm}
\|Z_k^\ast(\widehat\alpha_k)-Z_k^\ast(\alpha_k^0)\|\leq\|D_k\|_{\mathrm{op}}\|X_k\|_{\mathrm{op}}\|\delta_k\|
+\|r_k(\delta_k)\|=O_p(1).
\end{equation}

For part (i), first replace $M_{\widehat F}$ by $M_{F^0}$. Equations \eqref{eq:B5-exact-identity} and Lemma \ref{lem:B5-score-bounds} imply

\begin{equation}\nonumber
\begin{aligned}
\frac1{\sqrt T}X_i^{\top}M_{F^0}\{Z_k^\ast(\widehat\alpha_k)-Z_k^\ast(\alpha_k^0)\}&=\left(\frac1{\sqrt T}X_i^{\top}M_{F^0}D_kX_k\right)\delta_k+\frac1{\sqrt T}X_i^{\top}M_{F^0}r_k(\delta_k)\nonumber\\
&=O_p(1)O_p(T^{-1/2})+O_p(1)o_p(1)\\
&=o_p(1).
\end{aligned}
\end{equation}
Furthermore, by \eqref{eq:B5-Z-norm},
\begin{equation}\nonumber
\left\|\frac1{\sqrt T}X_i^{\top}(M_{\widehat F}-M_{F^0})
\{Z_k^\ast(\widehat\alpha_k)-Z_k^\ast(\alpha_k^0)\}\right\|\leq\frac{\|X_i\|_{\mathrm{op}}}{\sqrt T}\|M_{\widehat F}-M_{F^0}\|_{\mathrm{op}}
\|Z_k^\ast(\widehat\alpha_k)-Z_k^\ast(\alpha_k^0)\|=o_p(1).
\end{equation}
This proves part (i).

For part (ii), use the factor alignment in Proposition \ref{pro:B.1} to write
\begin{equation}\nonumber
\begin{aligned}
\frac1{\sqrt T}\{Z_k^\ast(\widehat\alpha_k)-Z_k^\ast(\alpha_k^0)\}^{\top}\widehat F=&\delta_k^{\top}\left(\frac1{\sqrt T}X_k^{\top}D_kF^0\right)H+\frac1{\sqrt T}r_k(\delta_k)^{\top}F^0H\\
&+\frac1{\sqrt T}\{Z_k^\ast(\widehat\alpha_k)-Z_k^\ast(\alpha_k^0)\}^{\top}(\widehat F-F^0H).
\end{aligned}
\end{equation}
The first term is $O_p(T^{-1/2})$ by Lemma \ref{lem:B5-score-bounds}. The second is $o_p(1)$ because $\|F^0\|/\sqrt T=O_p(1)$ and \eqref{eq:B5-r-norm} holds. The last term is $o_p(1)$ by \eqref{eq:B5-Z-norm} and $T^{-1/2}\|\widehat F-F^0H\|=o_p(1)$. This proves part (ii).

For part (iii),
\begin{equation}\nonumber
\frac1{\sqrt T}(F^0\lambda_i^0)^{\top}\{Z_i^\ast(\widehat\alpha_i)-Z_i^\ast(\alpha_i^0)\}=\lambda_i^{0\top}\left(\frac1{\sqrt T}F^{0\top}D_iX_i\right)\delta_i+\frac1{\sqrt T}\lambda_i^{0\top}F^{0\top}r_i(\delta_i)=o_p(1),
\end{equation}
by Lemma \ref{lem:B5-score-bounds}, \eqref{eq:B5-r-norm}, and the boundedness of $\lambda_i^0$.

\Halmos
\endproof
\begin{lemma}\label{lem:B5-score-bounds}
Under Assumptions \ref{assum:Regularity}--\ref{assum:local-identification}, \ref{assum:error}--\ref{assum:identification}, for any fixed $i,k\in[N]$,
\begin{equation}\nonumber
\left\|\frac{1}{\sqrt T}X_i^\top M_{F^0}D_kX_k\right\|= O_p(1),\ \left\|\frac{1}{\sqrt T}F^{0\top}D_kX_k\right\|= O_p(1).
\end{equation}
\end{lemma}
\proof{Proof of Lemma \ref{lem:B5-score-bounds}.}
By Assumption \ref{assum:ES-score}(i), the conditional quantile restriction implies $E\left[\mathbbm{1}\{\epsilon_{kt}\leq 0\}\mid\mathcal C_{kt}\right]=\tau$, where $\mathcal C_{kt}=\sigma(X_{kt},f_t^0)$. Hence, by the definition of $m_{kt}$,
\begin{equation}\label{eq:tagB.2}
E(m_{kt}\mid\mathcal C_{kt})=1-\tau^{-1}E\left[\mathbbm{1}\{\epsilon_{kt}\leq 0\}\mid\mathcal C_{kt}\right]=1-\tau^{-1}\tau=0.
\end{equation}
This zero-mean property is the key ingredient. It suffices to show that each element of the matrix $T^{-1/2}F^{0\top}D_kX_k$ is $O_p(1)$. The $(a,b)$-th element is ${T}^{-1/2}\sum_{t=1}^T f_{t,a}^0 X_{kt,b} m_{kt}$. Define the time-local score variable $Z_{kt,ab} =f_{t,a}^0 X_{kt,b} m_{kt}$.
By \eqref{eq:tagB.2}, $E(Z_{kt,ab}\mid\mathcal C_{kt})= f_{t,a}^0 X_{kt,b} E(m_{kt}\mid\mathcal C_{kt})=0$. Thus $\{Z_{kt,ab}:t\in[T]\}$ is a zero-mean, $\alpha$-mixing process under Assumption \ref{assum:Regularity}. By Assumption \ref{assum:ES-score}(ii), it has uniformly bounded $(2+\nu)$-th moments for some $\nu>0$. Applying Davydov's inequality, we obtain
$$
\mathrm{Var}\left(\sum_{t=1}^T Z_{kt,ab}\right)= \sum_{t,s=1}^T \mathrm{Cov}(Z_{kt,ab},Z_{ks,ab})= O(T),
$$
because the mixing coefficients decay geometrically, so $\sum_{h=1}^\infty \alpha(h)^{\nu/(2+\nu)}<\infty$. Consequently, 
$$
\frac1{\sqrt T}\sum_{t=1}^T Z_{kt,ab}= O_p(1).
$$
Since there are only finitely many pairs $(a,b)$ with $a\leq r^0$ and $b\leq p+1$,
$$
\left\|\frac{1}{\sqrt T}F^{0\top}D_kX_k\right\|= O_p(1).
$$

Decompose
\begin{equation}\nonumber
\frac{1}{\sqrt T}X_i^\top M_{F^0}D_kX_k=\frac{1}{\sqrt T}X_i^\top D_kX_k-\frac{1}{\sqrt T}X_i^\top F^0(F^{0\top}F^0)^{-1}F^{0\top}D_kX_k.
\end{equation}
Using $T^{-1}F^{0\top}F^0=I_{r^0}$, the second term becomes
\begin{equation}\nonumber
\begin{aligned}
\frac{1}{\sqrt T}X_i^\top F^0(F^{0\top}F^0)^{-1}F^{0\top}D_kX_k&=\frac{1}{\sqrt T}X_i^\top F^0\cdot \frac{1}{T}I_{r^0}\cdot F^{0\top}D_kX_k\\
&=\left(\frac{1}{T}X_i^\top F^0\right)\left(\frac{1}{\sqrt T}F^{0\top}D_kX_k\right).
\end{aligned}
\end{equation}
Therefore,
$$
\left\|\frac{1}{\sqrt T}X_i^\top F^0(F^{0\top}F^0)^{-1}F^{0\top}D_kX_k\right\|=O_p(1)\cdot O_p(1)=O_p(1).
$$
It remains to control the first term. Its $(a,b)$-th element is ${T}^{-1/2}\sum_{t=1}^T X_{it,a} X_{kt,b} m_{kt}$. Define $Z_{kt,ab}^{(2)}= X_{it,a}X_{kt,b}m_{kt}$. By \eqref{eq:tagB.2}, $E(Z_{kt,ab}^{(2)}\mid\mathcal C_{kt})= X_{it,a}X_{kt,b}E(m_{kt}\mid\mathcal C_{kt})=0$. The same Davydov's inequality argument gives
$$
\frac{1}{\sqrt T}\sum_{t=1}^T X_{it,a}X_{kt,b}m_{kt}
= O_p(1).
$$
Since $p$ is fixed, there are only finitely many $(a,b)$ pairs, so
$$
\left\|\frac{1}{\sqrt T}X_i^\top D_kX_k\right\|= O_p(1).
$$
Combining the two pieces,
$$
\left\|\frac{1}{\sqrt T}X_i^\top M_{F^0}D_kX_k\right\|\leq\left\|\frac{1}{\sqrt T}X_i^\top D_kX_k\right\|+\left\|\frac{1}{\sqrt T}X_i^\top F^0(F^{0\top}F^0)^{-1}F^{0\top}D_kX_k\right\|=O_p(1).
$$
\Halmos
\endproof

\subsection{Proof of Theorem \ref{theo:asymptotic}}

\proof{Proof of Theorem \ref{theo:asymptotic}.}
Fix a target unit $i$. The assumed growth conditions ensure that the conditions of Theorem \ref{theo:Gaussian} hold for all sufficiently large $N$ and $T$. Moreover, ${\log^2(2T)}/{\sqrt T}+\mathfrak b_{N,T,p}^{2/3}\rightarrow0$. It follows from Theorem \ref{theo:Gaussian} that
\begin{equation}
\label{eq:asymptotic-finite-variance}
\sup_{\substack{u\in\mathbb S^p\\z\in\mathbb R}}\left|P\left\{\frac{\tau\sqrt T\,u^\top(\widehat\beta_{i,\tau}-\beta_{i,\tau}^0)}{\sqrt{u^\top\Omega_{i,\tau,T}u}}\leq z\right\}-\Phi(z)\right|\rightarrow0.
\end{equation}

Assumption \ref{assum:Gaussian}(iii) gives $\lambda_{\min}(\Omega_{i,\tau,T})\geq c_\Omega$. Together with Assumption \ref{assum:asymptotic-variance} and Weyl's inequality, this implies $\lambda_{\min}(\Omega_{i,\tau})\geq c_\Omega$. Consequently, for every $u\in\mathbb S^p$,
\begin{align*}
&\sup_{u\in\mathbb S^p}\left|\frac{\sqrt{u^\top\Omega_{i,\tau,T}u}}{\sqrt{u^\top\Omega_{i,\tau}u}}-1\right|\leq\frac{\|\Omega_{i,\tau,T}-\Omega_{i,\tau}\|_{\mathrm{op}}}{2c_\Omega}\rightarrow0.
\end{align*}
For a fixed $u\in\mathbb S^p$, write
$$
Z_{i,T}(u)=\frac{\tau\sqrt T\,u^\top(\widehat\beta_{i,\tau}-\beta_{i,\tau}^0)}{\sqrt{u^\top\Omega_{i,\tau,T}u}}.
$$
Equation \eqref{eq:asymptotic-finite-variance} implies $Z_{i,T}(u)\xrightarrow{d}\mathcal N(0,1)$. Since
$$
\frac{\tau\sqrt T\,u^\top(\widehat\beta_{i,\tau}-\beta_{i,\tau}^0)}{\sqrt{u^\top\Omega_{i,\tau}u}}=Z_{i,T}(u)\frac{\sqrt{u^\top\Omega_{i,\tau,T}u}}{\sqrt{u^\top\Omega_{i,\tau}u}},
$$
Slutsky's theorem gives
$$
\frac{\tau\sqrt T\,u^\top(\widehat\beta_{i,\tau}-\beta_{i,\tau}^0)}{\sqrt{u^\top\Omega_{i,\tau}u}}\overset{d}{\longrightarrow}\mathcal N(0,1).
$$

It remains to establish the vector convergence. Let $Y_{i,T}=\tau\sqrt T(\widehat\beta_{i,\tau}-\beta_{i,\tau}^0)$. For any nonzero $a\in\mathbb R^{p+1}$, set $u=a/\|a\|$. The preceding scalar convergence gives
$$
a^\top Y_{i,T}=\|a\|u^\top Y_{i,T}\xrightarrow{d}\mathcal N\left(0,\|a\|^2u^\top\Omega_{i,\tau}u\right)=\mathcal N\left(
0,a^\top\Omega_{i,\tau}a\right).
$$
Because $p$ is fixed, the Cram\'er--Wold device therefore yields $\tau\sqrt T(\widehat\beta_{i,\tau}-\beta_{i,\tau}^0)\xrightarrow{d}\mathcal N(0,\Omega_{i,\tau})$. Finally, division by the fixed constant $\tau$ gives $\sqrt T(\widehat\beta_{i,\tau}-\beta_{i,\tau}^0)\xrightarrow{d}\mathcal N(0,\tau^{-2}\Omega_{i,\tau})$.
\Halmos
\endproof

\section{Additional Simulation Results}\label{sec:Additional Simulation Results}

\subsection{Simulation Design}

For $i=1,\ldots,N$ and $t=1,\ldots,T$, we generate a covariate vector $X_{it}\in\mathbb{R}^{p+1}$ that includes an intercept and $p=3$ observed regressors. The latent tail component is driven by a factor structure with true dimension $r_\tau^0=2$. To align with the two-stage ESFM estimation strategy, we construct the data so that the conditional $\tau$-quantile satisfies a standard panel quantile regression,
$Q_{\tau}(Y_{it}\mid X_{it}) = X_{it}^\top\alpha_i$, and the conditional expected shortfall follows a location--scale form in which latent factors enter through the tail scale:
\begin{equation}
Y_{it} = \mu_{it} + \sigma_{it}\varepsilon_{it}^\ast,\qquad \mu_{it}=X_{it}^\top\alpha_i.
\label{eq:DGP_new}
\end{equation}
The innovation $\varepsilon_{it}^\ast$ is generated from a standardized Student-$t$ distribution with $\nu=5$ degrees of freedom and then shifted so that its unconditional $\tau$-quantile is approximately zero. This construction ensures that $Q_{\tau}(Y_{it}\mid X_{it})=\mu_{it}$ holds by design, while the magnitude of tail losses is governed by $\sigma_{it}$. The conditional ES is therefore $\mathrm{ES}_{\tau}(Y_{it}\mid X_{it},\sigma_{it})=\mu_{it}+\sigma_{it}\,\mathrm{ES}_{\tau}(\varepsilon_{it}^\ast)$, so that latent factors affect tail risk through the time variation in $\sigma_{it}$, without mechanically entering the quantile equation.

\paragraph{Latent factors and loadings.}
The ES factors $F_t\in\mathbb{R}^{r_\tau^0}$ follow a stationary AR(1) process (componentwise) and the loadings $\lambda_i\in\mathbb{R}^{r_\tau^0}$ are drawn independently across $i$ and scaled to control factor strength. The factor-driven tail scale is generated from the interactive component $\lambda_i^\top F_t$ through an exponential link, $\sigma_{it}=\exp\!\left(c_{\sigma}\cdot \frac{\lambda_i^\top F_t}{\mathrm{sd}(\lambda^\top F)}\right)$,
with truncation to avoid extreme values. The constant $c_{\sigma}$ controls the intensity of tail comovement.

\paragraph{Cross-sectional heterogeneity.}
To reflect heterogeneous exposures to observed risk drivers, we allow slope coefficients in the ES regression to vary across units via a grouped design: units are randomly assigned to a small number of groups, and the slope vector is shifted across groups. This introduces meaningful cross-sectional dispersion in $\beta_i$ while keeping the regressor dimension fixed.

\subsection{Data-generating Scenarios}\label{subsec:Simulation Scenarios}

We consider seven economically motivated scenarios (indexed by $\texttt{scenarioID}\in\{1,\ldots,7\}$) and differing only in how $\sigma_{it}$ or $\varepsilon_{it}^\ast$ are modified.

\begin{itemize}
\item \textbf{Scenario 1 (Baseline tail-factor structure).}
Tail dependence is introduced through a latent additive component in the conditional expected shortfall. 
Specifically, the ES contains a low-rank factor structure $\lambda_i'F_t$, while the conditional quantile location remains $\mu_{it} = X_{it}'\alpha_i$. This baseline design captures cross-sectional tail comovement driven by latent factors.

\item \textbf{Scenario 2 (Stronger tail-factor dependence).}
We increase the strength of the latent factor component in the ES structure. This design amplifies cross-sectional tail comovement and evaluates whether the estimator can accurately recover the factor space when the tail component dominates.

\item \textbf{Scenario 3 (Heterogeneous slopes).}
We strengthen cross-sectional heterogeneity in the slope coefficients by grouping units and shifting $\beta_i$ across groups, while maintaining factor-driven tail scale. This scenario evaluates robustness to heterogeneous covariate effects.

\item \textbf{Scenario 4 (Endogenous covariates).}
We introduce correlation between covariates and latent tail factors by letting one regressor load on the interactive component $\lambda_i'F_t$. This creates an omitted-factor-type bias for methods that do not account for latent tail components.

\item \textbf{Scenario 5 (Volatility-factor dominance).}
We increase the strength of the factor effect in $\sigma_{it}$, so that tail comovement is largely volatility-driven. This scenario stresses the setting where ignoring factors should be particularly costly for ES regression.

\item \textbf{Scenario 6 (Jump episodes).}
We add rare but large shocks to the innovation by augmenting $\varepsilon_{it}^\ast$ with a jump term that occurs with small probability. This generates episodic systemic stress while preserving the basic location--scale structure.

\item \textbf{Scenario 7 (Asymmetric tails).}
We modify the innovation distribution by introducing an additional asymmetric component that thickens the lower tail. This design captures asymmetric downside risk beyond symmetric heavy-tailed innovations.
\end{itemize}

Across these scenarios, ES regression is intentionally misspecified whenever latent tail factors drive $\sigma_{it}$, whereas ESFM explicitly models and estimates the latent tail component. The simulations therefore provide a direct check of the central mechanism of ESFM: incorporating tail factors improves the accuracy of $\beta_i$ estimation and yields reliable recovery of the factor space in finite samples.

\subsection{Simulation Results}\label{subsec:Simulation Results}

To evaluate the finite-sample accuracy of tail-risk estimation, we examine the bias in conditional ES across different cross-sectional and temporal dimensions. When latent common factors are present, misspecifying cross-sectional dependence can distort tail expectations, even if first-stage quantile estimation is accurate. Thus, we compare the results of ESR and ESFM. The true conditional ES is computed using Monte Carlo integration from the data-generating process.
\begin{table}
	\footnotesize
	\TABLE
	{Finite-sample bias of conditional expected shortfall (ES) estimators.
		\label{tab:Estimation error of ES}}
	{\footnotesize
		\setlength{\tabcolsep}{2.5pt}
		\renewcommand{\arraystretch}{0.85}
		\begin{tabular*}{\textwidth}
			{@{\extracolsep{\fill}}lcccccccccccccc}
			\toprule
			& \multicolumn{2}{c}{Scenario 1}
			& \multicolumn{2}{c}{Scenario 2}
			& \multicolumn{2}{c}{Scenario 3}
			& \multicolumn{2}{c}{Scenario 4}
			& \multicolumn{2}{c}{Scenario 5}
			& \multicolumn{2}{c}{Scenario 6}
			& \multicolumn{2}{c}{Scenario 7} \\
			& ESR & ESFM
			& ESR & ESFM
			& ESR & ESFM
			& ESR & ESFM
			& ESR & ESFM
			& ESR & ESFM
			& ESR & ESFM \\
			\midrule
			\multicolumn{15}{c}{Panel A: $\tau = 0.10$} \\
			\midrule
			\multicolumn{15}{c}{$N$ = 100} \\
			\midrule
			$T=100$ & 1.0207 & 0.9671 & 1.0336 & 0.9797 & 1.0286 & 0.9749 & 1.0274 & 0.9735 & 1.0274 & 0.9735 & 1.0333 & 0.9789 & 1.0375 & 0.9826 \\
			
			$T=200$ & 1.1538 & 1.0977 & 1.1577 & 1.1022 & 1.1644 & 1.1088 & 1.1536 & 1.0983 & 1.1619 & 1.1066 & 1.1592 & 1.1044 & 1.1512 & 1.0947 \\
			
			$T=300$ & 1.2128 & 1.1563 & 1.2327 & 1.1764 & 1.1977 & 1.1413 & 1.2117 & 1.1561 & 1.2119 & 1.1550 & 1.2156 & 1.1597 & 1.2110 & 1.1555 \\
			
			\midrule
			\multicolumn{15}{c}{$N$ = 200} \\
			\midrule
			
			$T=100$ & 1.0265 & 0.9727 & 1.0317 & 0.9771 & 1.0244 & 0.9709 & 1.0311 & 0.9773 & 1.0220 & 0.9685 & 1.0298 & 0.9754 & 1.0230 & 0.9683 \\
			
			$T=200$ & 1.1474 & 1.0921 & 1.1464 & 1.0911 & 1.1558 & 1.1000 & 1.1550 & 1.0993 & 1.1539 & 1.0984 & 1.1623 & 1.1067 & 1.1509 & 1.0949 \\
			
			$T=300$ & 1.2063 & 1.1501 & 1.2112 & 1.1552 & 1.2119 & 1.1558 & 1.2078 & 1.1518 & 1.2051 & 1.1493 & 1.2103 & 1.1540 & 1.2052 & 1.1494 \\
			
			\midrule
			\multicolumn{15}{c}{$N$ = 300} \\
			\midrule
			
			$T=100$ & 1.0356 & 0.9815 & 1.0337 & 0.9793 & 1.0285 & 0.9745 & 1.0309 & 0.9764 & 1.0434 & 0.9888 & 1.0278 & 0.9738 & 1.0276 & 0.9732 \\
			
			$T=200$ & 1.1603 & 1.1049 & 1.1497 & 1.0945 & 1.1520 & 1.0967 & 1.1535 & 1.0982 & 1.1574 & 1.1017 & 1.1486 & 1.0933 & 1.1556 & 1.1003 \\
			
			$T=300$ & 1.2173 & 1.1615 & 1.2087 & 1.1526 & 1.1989 & 1.1428 & 1.1943 & 1.1377 & 1.2090 & 1.1529 & 1.2104 & 1.1539 & 1.2163 & 1.1604 \\
            			\midrule
			\multicolumn{15}{c}{Panel B: $\tau = 0.05$} \\
			\midrule
			\multicolumn{15}{c}{$N$ = 100} \\
			\midrule
			
			$T=100$ & 0.8184 & 0.6970 & 0.8400 & 0.7189 & 0.8412 & 0.7188 & 0.8429 & 0.7212 & 0.8370 & 0.7150 & 0.8445 & 0.7215 & 0.8457 & 0.7232 \\
			
			$T=200$ & 1.0214 & 0.8966 & 1.0278 & 0.9041 & 1.0300 & 0.9077 & 1.0129 & 0.8895 & 1.0276 & 0.9050 & 1.0328 & 0.9100 & 1.0086 & 0.8833 \\
			
			$T=300$ & 1.1105 & 0.9859 & 1.1301 & 1.0049 & 1.0971 & 0.9720 & 1.1054 & 0.9812 & 1.1096 & 0.9836 & 1.1157 & 0.9906 & 1.1113 & 0.9865 \\
			
			\midrule
			\multicolumn{15}{c}{$N$ = 200} \\
			\midrule
			
			$T=100$ & 0.8310 & 0.7093 & 0.8464 & 0.7229 & 0.8291 & 0.7067 & 0.8372 & 0.7156 & 0.8371 & 0.7162 & 0.8395 & 0.7176 & 0.8241 & 0.7013 \\
			
			$T=200$ & 1.0109 & 0.8874 & 1.0067 & 0.8834 & 1.0224 & 0.8981 & 1.0197 & 0.8957 & 1.0229 & 0.8990 & 1.0317 & 0.9072 & 1.0146 & 0.8906 \\
			
			$T=300$ & 1.0927 & 0.9683 & 1.1123 & 0.9887 & 1.1076 & 0.9829 & 1.1009 & 0.9764 & 1.1021 & 0.9779 & 1.1027 & 0.9777 & 1.1091 & 0.9852 \\
			
			\midrule
			\multicolumn{15}{c}{$N$ = 300} \\
			\midrule
			
			$T=100$ & 0.8407 & 0.7181 & 0.8421 & 0.7189 & 0.8303 & 0.7085 & 0.8401 & 0.7171 & 0.8560 & 0.7334 & 0.8362 & 0.7140 & 0.8285 & 0.7058 \\
			
			$T=200$ & 1.0272 & 0.9038 & 1.0112 & 0.8886 & 1.0164 & 0.8933 & 1.0183 & 0.8950 & 1.0220 & 0.8983 & 1.0110 & 0.8875 & 1.0268 & 0.9031 \\
			
			$T=300$ & 1.1195 & 0.9952 & 1.1009 & 0.9760 & 1.0914 & 0.9668 & 1.0838 & 0.9589 & 1.1020 & 0.9777 & 1.1067 & 0.9809 & 1.1169 & 0.9928 \\
            			\midrule
			\multicolumn{15}{c}{Panel C: $\tau = 0.01$} \\
			\midrule
			\multicolumn{15}{c}{$N$ = 100} \\
			\midrule
			
			$T=100$ & 0.2899 & 0.2932 & 0.3666 & 0.3699 & 0.3314 & 0.3347 & 0.3711 & 0.3743 & 0.3236 & 0.3270 & 0.3251 & 0.3288 & 0.3539 & 0.3572 \\
			
			$T=200$ & 0.5127 & 0.3323 & 0.5344 & 0.3491 & 0.5467 & 0.3627 & 0.4970 & 0.3164 & 0.5315 & 0.3486 & 0.5493 & 0.3696 & 0.5409 & 0.3551 \\
			
			$T=300$ & 0.6798 & 0.3985 & 0.7094 & 0.4315 & 0.6411 & 0.3639 & 0.6617 & 0.3851 & 0.6949 & 0.4137 & 0.7022 & 0.4262 & 0.6835 & 0.4064 \\
			
			\midrule
			\multicolumn{15}{c}{$N$ = 200} \\
			\midrule
			
			$T=100$ & 0.3153 & 0.3187 & 0.3611 & 0.3645 & 0.3338 & 0.3374 & 0.3469 & 0.3500 & 0.3583 & 0.3617 & 0.3348 & 0.3384 & 0.3451 & 0.3486 \\
			
			$T=200$ & 0.5094 & 0.3272 & 0.4858 & 0.3077 & 0.5359 & 0.3558 & 0.5245 & 0.3414 & 0.5465 & 0.3636 & 0.5638 & 0.3786 & 0.5177 & 0.3368 \\
			
			$T=300$ & 0.6379 & 0.3612 & 0.6857 & 0.4120 & 0.6756 & 0.3985 & 0.6658 & 0.3905 & 0.6867 & 0.4076 & 0.6496 & 0.3739 & 0.6921 & 0.4120 \\
			
			\midrule
			\multicolumn{15}{c}{$N$ = 300} \\
			\midrule
			
			$T=100$ & 0.3440 & 0.3476 & 0.3375 & 0.3409 & 0.3236 & 0.3272 & 0.3418 & 0.3454 & 0.3295 & 0.3331 & 0.3071 & 0.3107 & 0.3290 & 0.3324 \\
			
			$T=200$ & 0.5573 & 0.3744 & 0.5050 & 0.3253 & 0.5385 & 0.3570 & 0.5263 & 0.3441 & 0.5343 & 0.3538 & 0.5301 & 0.3478 & 0.5568 & 0.3750 \\
			
			$T=300$ & 0.6992 & 0.4214 & 0.6543 & 0.3773 & 0.6502 & 0.3727 & 0.6326 & 0.3528 & 0.6348 & 0.3587 & 0.6488 & 0.3703 & 0.7159 & 0.4398 \\
			
			\bottomrule
		\end{tabular*}
	}
	{This table reports the bias of conditional expected shortfall estimators across Monte Carlo replications. Bias is defined as the difference between the estimated and true ES. Results are presented for seven data-generating scenarios, comparing the benchmark expected shortfall regression (ESR) and the proposed expected shortfall factor model (ESFM).}
\end{table}

Table \ref{tab:Estimation error of ES} reports the Monte Carlo bias of the ES estimators across various tail levels and sample sizes. A few key patterns emerge. First, for moderate tail levels ($\tau = 0.10$ and $\tau = 0.05$), ESFM consistently shows smaller bias than ESR across all scenarios. This improvement is more pronounced as the cross-sectional dimension increases, highlighting the advantage of latent factor modeling in estimating tail risk. Second, at the extreme tail level ($\tau = 0.01$), the performance difference narrows, especially in smaller samples. This is expected, as extreme tail extrapolation amplifies estimation noise, and the pseudo-response in ES construction increases lower-tail variability. Overall, these results confirm that modeling cross-sectional dependence significantly improves tail-risk estimation, especially when the tail probability is not too small.

The previous results assume the number of factors is known. Here, we examine the performance of the factor selection procedure. To assess the finite-sample behavior of the estimator, we report the Monte Carlo mean of the estimated factor number $\widehat{r}_{\tau}$ under different sample sizes and tail probability levels. For each simulation design, the factor number is estimated using the information criterion described in Section \ref{subsec:determine}. The true number of factors in the data-generating process is fixed at $r_\tau^0=2$, while the cross-sectional dimension $N$, the time dimension $T$, and the tail probability level $\tau$ vary across designs.
\begin{table}
	\footnotesize
	\TABLE
	{Monte Carlo mean of the estimated number of factors.
		\label{tab:Error for number of factors}}
	{\footnotesize
		\setlength{\tabcolsep}{2.5pt}
		\renewcommand{\arraystretch}{0.85}
		\begin{tabular*}{\textwidth}
			{@{\extracolsep{\fill}}lccccccc}
			\toprule
			& Scenario 1 & Scenario 2 & Scenario 3 & Scenario 4 
			& Scenario 5 & Scenario 6 & Scenario 7 \\
			\midrule
			
			\multicolumn{8}{c}{Panel A: $\tau = 0.10$} \\
			\midrule
			\multicolumn{8}{c}{$N$ = 100} \\
			\midrule
			
			$T=100$ & 2.1300 & 2.2500 & 2.1200 & 2.2400 & 2.2300 & 2.2200 & 2.2800 \\
			$T=200$ & 2.1100 & 2.1400 & 2.1000 & 2.1900 & 2.1600 & 2.1600 & 2.1600 \\
			$T=300$ & 2.0600 & 2.0500 & 2.0400 & 2.0900 & 2.0100 & 2.0600 & 2.1300 \\
			
			\midrule
			\multicolumn{8}{c}{$N$ = 200} \\
			\midrule
			
			$T=100$ & 2.1100 & 2.1600 & 2.1100 & 2.3500 & 2.1800 & 2.1800 & 2.1600 \\
			$T=200$ & 2.1000 & 2.1400 & 2.1000 & 2.1200 & 2.1100 & 2.1000 & 2.1400 \\
			$T=300$ & 2.0700 & 2.0500 & 2.0500 & 2.0300 & 2.0800 & 2.0500 & 2.0600 \\
			
			\midrule
			\multicolumn{8}{c}{$N$ = 300} \\
			\midrule
			
			$T=100$ & 2.1000 & 2.1500 & 2.1000 & 2.1000 & 2.1000 & 2.1300 & 2.1500 \\
			$T=200$ & 2.0900 & 2.1000 & 2.0700 & 2.0600 & 2.0300 & 1.9500 & 2.1100 \\
			$T=300$ & 2.0500 & 2.0400 & 2.0600 & 1.9800 & 1.9700 & 1.9600 & 2.0400 \\
			
			\midrule
			\multicolumn{8}{c}{Panel B: $\tau = 0.05$} \\
			\midrule
			\multicolumn{8}{c}{$N$ = 100} \\
			\midrule
			
			$T=100$ & 2.0900 & 2.0900 & 2.0900 & 2.1500 & 2.1500 & 2.2000 & 2.2000 \\
			$T=200$ & 2.0500 & 2.0400 & 2.0600 & 2.1300 & 2.1500 & 2.1500 & 2.1900 \\
			$T=300$ & 2.0000 & 2.0300 & 2.0500 & 2.0100 & 2.0500 & 2.1300 & 2.0600 \\
			
			\midrule
			\multicolumn{8}{c}{$N$ = 200} \\
			\midrule
			
			$T=100$ & 2.1500 & 2.1800 & 2.1900 & 2.1400 & 2.1300 & 2.1800 & 2.1900 \\
			$T=200$ & 2.1100 & 2.1200 & 2.0900 & 2.0600 & 2.0800 & 2.0900 & 2.0600 \\
			$T=300$ & 2.0900 & 2.0800 & 2.0200 & 2.0300 & 2.0500 & 2.0500 & 2.0200 \\
			
			\midrule
			\multicolumn{8}{c}{$N$ = 300} \\
			\midrule
			
			$T=100$ & 2.0900 & 2.1000 & 1.9100 & 1.9200 & 1.9200 & 1.9100 & 1.9000 \\
			$T=200$ & 2.0500 & 2.0900 & 1.9300 & 1.9400 & 1.9600 & 1.9500 & 1.9500 \\
			$T=300$ & 2.0800 & 2.0600 & 1.9500 & 1.9600 & 2.0200 & 1.9800 & 1.9900 \\
			
			\midrule
			\multicolumn{8}{c}{Panel C: $\tau = 0.01$} \\
			\midrule
			\multicolumn{8}{c}{$N$ = 100} \\
			\midrule
			
			$T=100$ & 1.9400 & 2.1000 & 1.9800 & 2.0600 & 1.9500 & 1.9400 & 1.9400 \\
			$T=200$ & 1.9600 & 2.0800 & 1.9600 & 2.0500 & 1.9600 & 1.9600 & 1.9500 \\
			$T=300$ & 2.0100 & 2.0200 & 2.0000 & 2.0100 & 1.9900 & 1.9800 & 1.9900 \\
			
			\midrule
			\multicolumn{8}{c}{$N$ = 200} \\
			\midrule
			
			$T=100$ & 1.9600 & 1.9600 & 1.9800 & 1.9600 & 1.9800 & 2.0400 & 2.0500 \\
			$T=200$ & 1.9800 & 1.9600 & 2.0200 & 1.9900 & 1.9900 & 2.0300 & 2.0200 \\
			$T=300$ & 1.9900 & 1.9900 & 2.0100 & 2.0000 & 2.0000 & 2.0100 & 2.0000 \\
			
			\midrule
			\multicolumn{8}{c}{$N$ = 300} \\
			\midrule
			
			$T=100$ & 1.9500 & 2.0400 & 1.9500 & 2.0300 & 1.9600 & 1.9600 & 2.0300 \\
			$T=200$ & 1.9900 & 1.9000 & 1.9800 & 1.9000 & 2.0100 & 2.0000 & 1.9900 \\
			$T=300$ & 2.0000 & 2.0000 & 2.0000 & 2.0000 & 2.0000 & 2.0000 & 2.0000 \\
			
			\bottomrule
		\end{tabular*}
	}
	{This table reports the Monte Carlo mean of the estimated number of factors $\widehat r$ across different simulation settings. The true number of factors in the data-generating process is $r=2$.}
\end{table}
Table \ref{tab:Error for number of factors} summarizes the Monte Carlo averages of $\widehat{r}_\tau$ across the considered simulation settings. Overall, the estimated factor number is close to the true value $r_\tau^0=2$ across all designs, with deviations diminishing as the sample size increases. In particular, when both $N$ and $T$ are moderately large, the average estimate concentrates tightly around 2, indicating reliable identification of the latent factor structure. The results are largely stable across different tail probability levels, suggesting that the factor-number selection procedure is not sensitive to the degree of tail sparsity.

\begin{figure}
\FIGURE
{\includegraphics[scale=0.45]{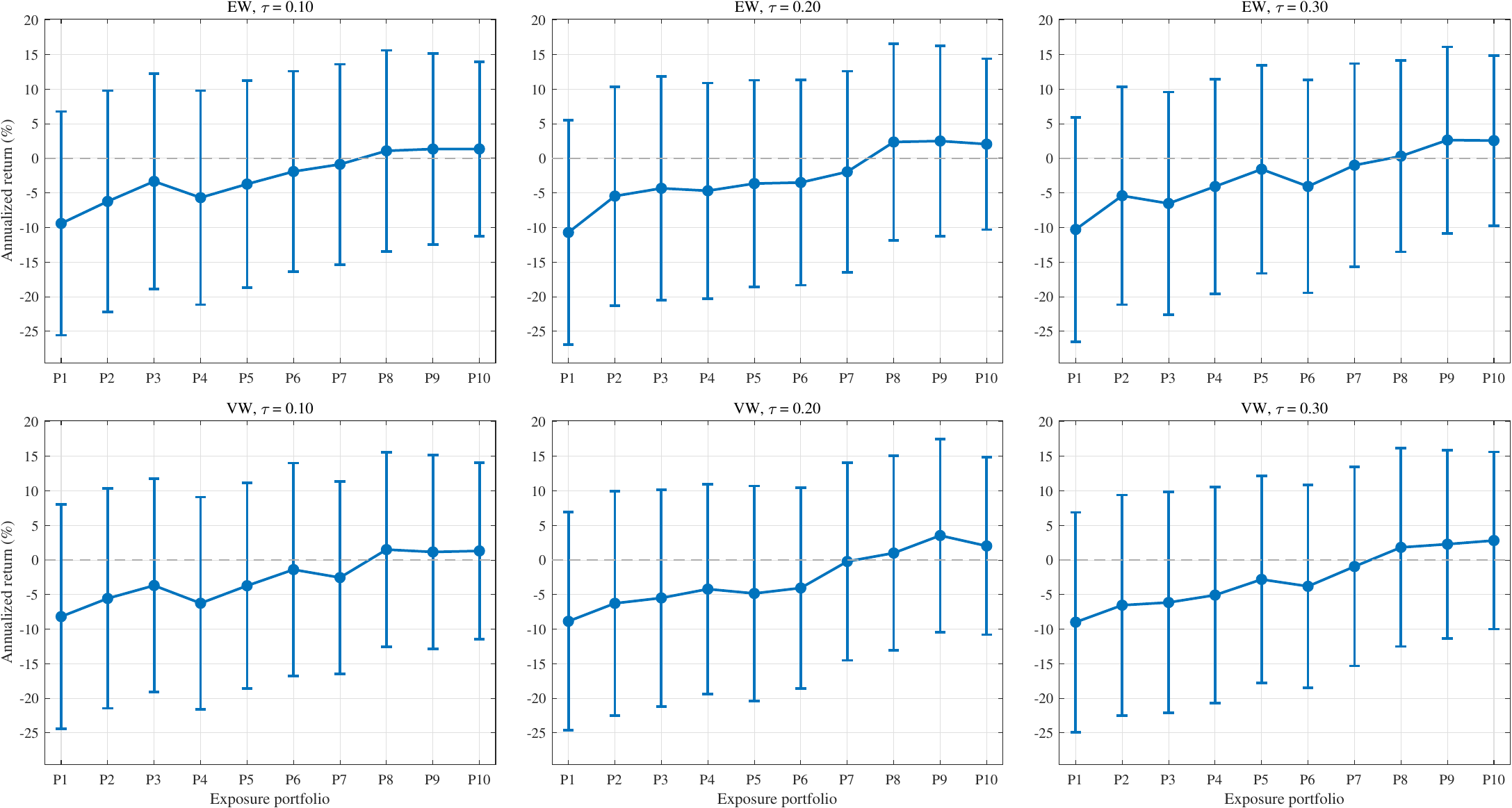}}
{Portfolio returns across ESFM exposure deciles.
\label{fig:ec_esfm_decile_returns}}
{This figure reports annualized average returns for decile portfolios sorted on estimated ESFM exposure. At each portfolio-formation date, stocks are assigned to ten portfolios, with P1 containing stocks with the lowest exposure and P10 those with the highest exposure. The top and bottom rows use latent factors estimated with equal-weighted and value-weighted versions of the observable factors, respectively. The columns correspond to $\tau=0.10$, $0.20$, and $0.30$. Sorting variables are estimated using information available at portfolio formation. Portfolios are equally weighted and held for one month. Returns are reported in percent per annum, and the vertical bars denote 95\% confidence intervals based on Newey--West standard errors with six lags.}
\end{figure}
\section{Additional Empirical Results}\label{sec:Additional Empirical Results}
This sample covers the most liquid and largest-cap stocks in the Chinese market, providing a representative laboratory for studying cross-sectional risk heterogeneity. To ensure data quality, we exclude assets with more than 20\% missing observations as well as stocks under special treatment (ST). This section includes additional empirical results not presented in the main text.

\subsection{Estimation Results}\label{subsec:additional_estimation_results}
The additional results reported in Figures \ref{fig:beta_EW}–\ref{fig:factor_EW_0p10} provide robustness checks under alternative specifications, including different tail levels and weighting schemes.

Figure \ref{fig:beta_EW} presents the corresponding coefficient estimates across models. The overall pattern remains stable relative to the main text: the ESFM continues to produce larger and more pronounced loadings across most observable factors, while the quantile-based model exhibits weaker and less stable estimates, particularly in the lower tail. The relative ordering across models is largely preserved across $\tau$ levels.

\begin{figure}
\FIGURE
{\includegraphics[scale=0.45]{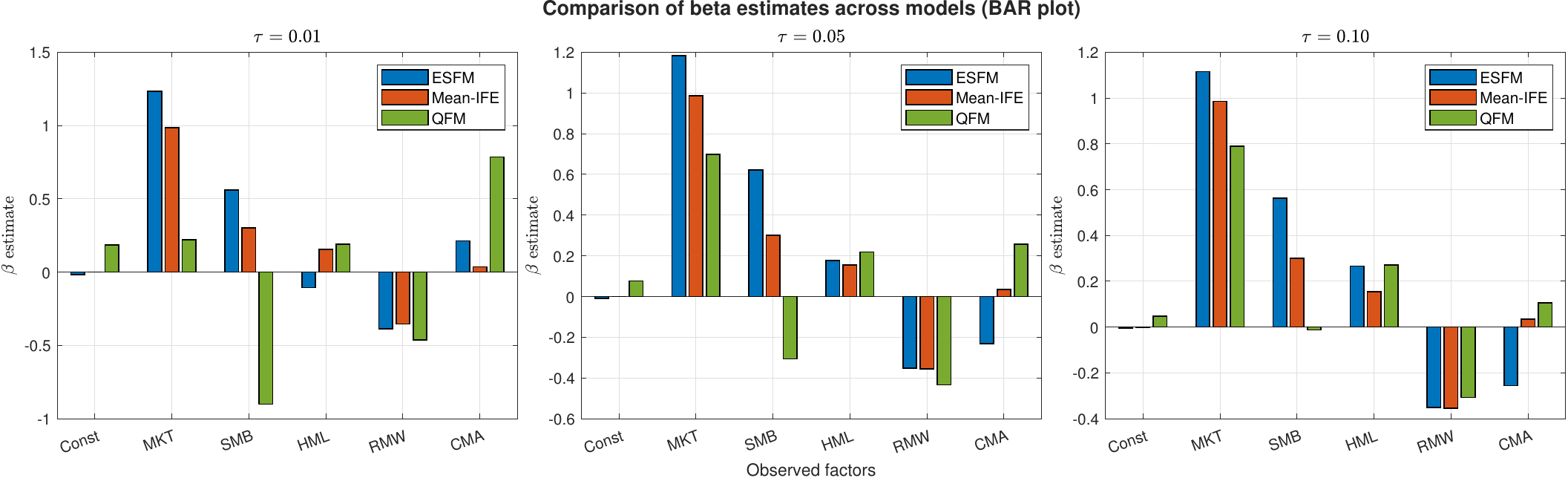}}
{Estimated coefficients ($\beta$) on observable risk factors.
\label{fig:beta_EW}}
{This figure reports the estimated coefficients on observable factors (MKT, SMB, HML, RMW, CMA, and a constant) from three models: ESFM, the mean factor model (Mean-IFE), and the quantile factor model (QFM). The estimates are obtained using EW observable factor returns, and results are shown for three tail levels ($\tau = 0.01, 0.05, 0.10$). Each panel corresponds to a different $\tau$, and bars represent the average coefficients across assets.}
\end{figure}

Figures \ref{fig:factor_VW_0p01} and \ref{fig:factor_VW_0p05} report the factor paths under value-weighted specifications for $\tau = 0.01$ and $\tau = 0.05$, respectively, while Figures \ref{fig:factor_EW_0p01}–\ref{fig:factor_EW_0p10} present the corresponding results under equal-weighted specifications. Several features are worth noting. First, the ES factors consistently display episodic spikes aligned with periods of market stress, whereas the mean factors remain comparatively smooth and the quantile factors exhibit persistent but less differentiated fluctuations. Second, these patterns are robust across weighting schemes, suggesting that the tail factor structure is not driven by portfolio construction choices. Third, as $\tau$ increases (Figures~\ref{fig:factor_VW_0p05}, \ref{fig:factor_VW_0p10}, \ref{fig:factor_EW_0p05}, and \ref{fig:factor_EW_0p10}), the distinction between ES and quantile factors becomes less pronounced, consistent with the interpretation that tail-specific information diminishes away from the extreme left tail.

\begin{figure}
\FIGURE
{\includegraphics[scale=0.59]{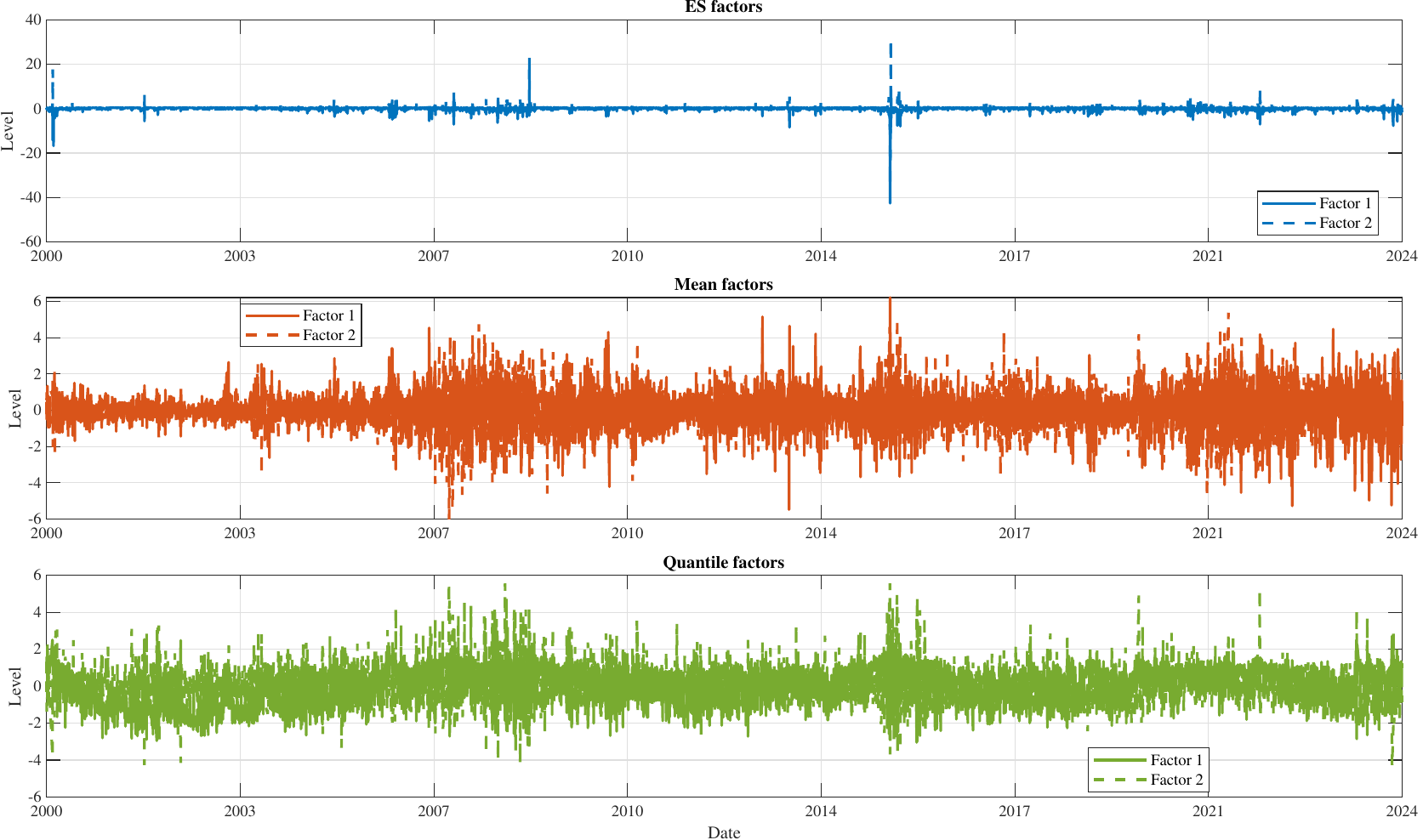}}
{Estimated factor paths across models ($r=2,\tau=0.01$).
\label{fig:factor_VW_0p01}}
{This figure plots the estimated latent factor paths from three models: ESFM (top panel), the mean factor model (middle panel), and the quantile factor model (bottom panel). Each model extracts two factors ($r = 2$) based on value-weighted Fama--French five factors. The quantile and ES factors are constructed using tail information at level $\tau = 0.01$.}
\end{figure}

\begin{figure}
\FIGURE
{\includegraphics[scale=0.59]{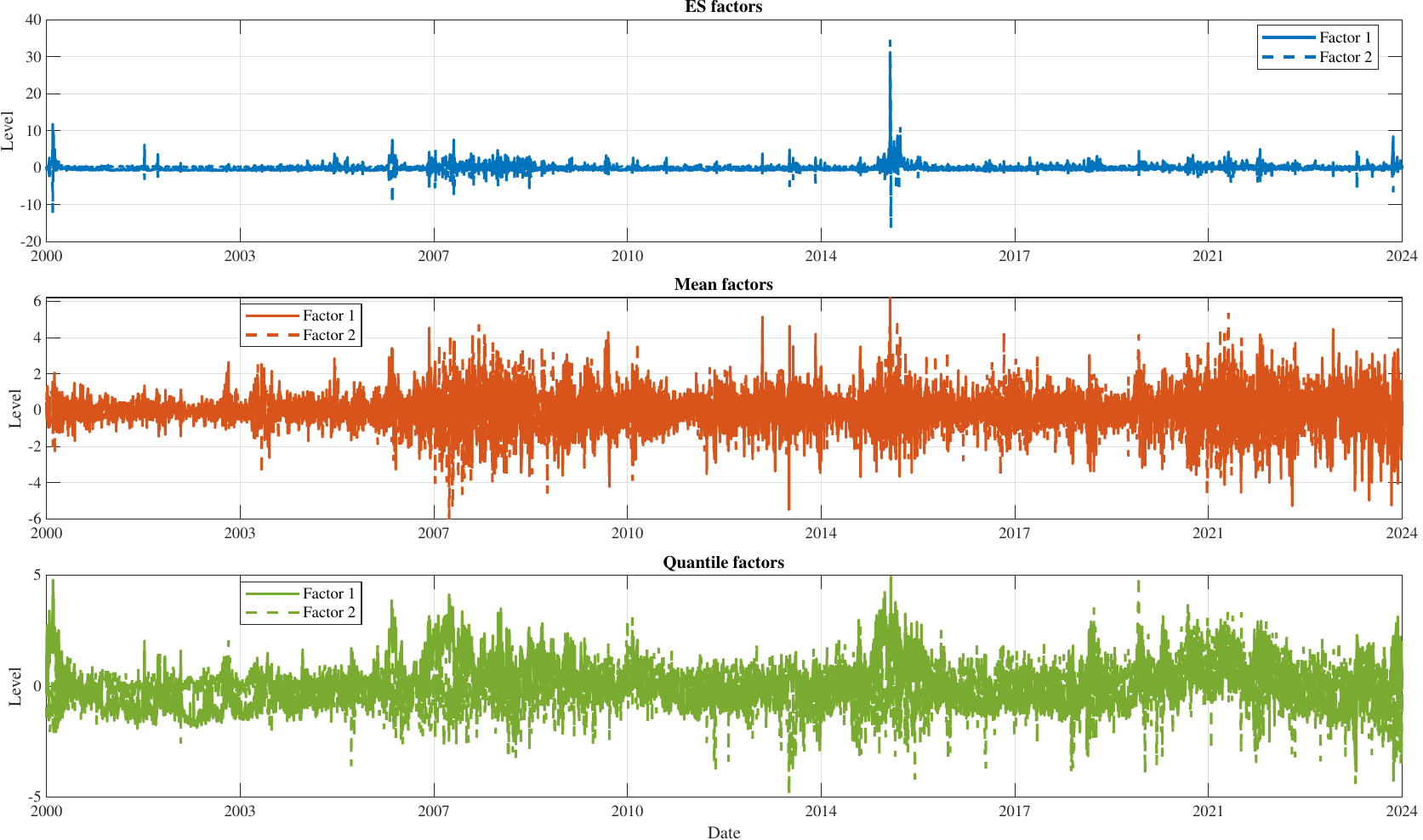}}
{Estimated factor paths across models ($r=2,\tau=0.05$).
\label{fig:factor_VW_0p05}}
{This figure plots the estimated latent factor paths from three models: ESFM (top panel), the mean factor model (middle panel), and the quantile factor model (bottom panel). Each model extracts two factors ($r = 2$) based on value-weighted Fama--French five factors. The quantile and ES factors are constructed using tail information at level $\tau = 0.05$.}
\end{figure}

\begin{figure}
\FIGURE
{\includegraphics[scale=0.59]{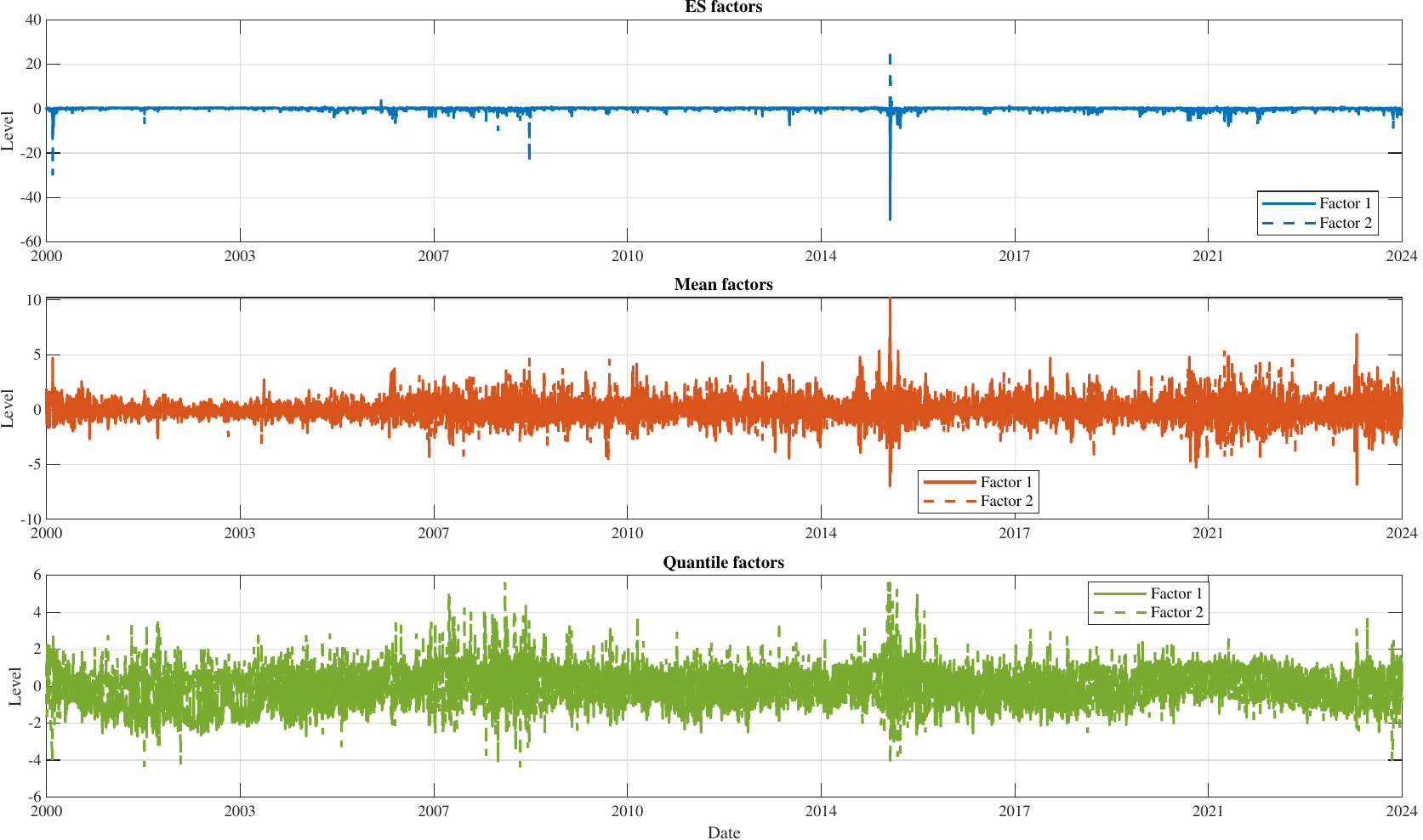}}
{Estimated factor paths across models ($r=2,\tau=0.01$).
\label{fig:factor_EW_0p01}}
{This figure plots the estimated latent factor paths from three models: ESFM (top panel), the mean factor model (middle panel), and the quantile factor model (bottom panel). Each model extracts two factors ($r = 2$) based on equal-weighted Fama--French five factors. The quantile and ES factors are constructed using tail information at level $\tau = 0.01$.}
\end{figure}

\begin{figure}
\FIGURE
{\includegraphics[scale=0.59]{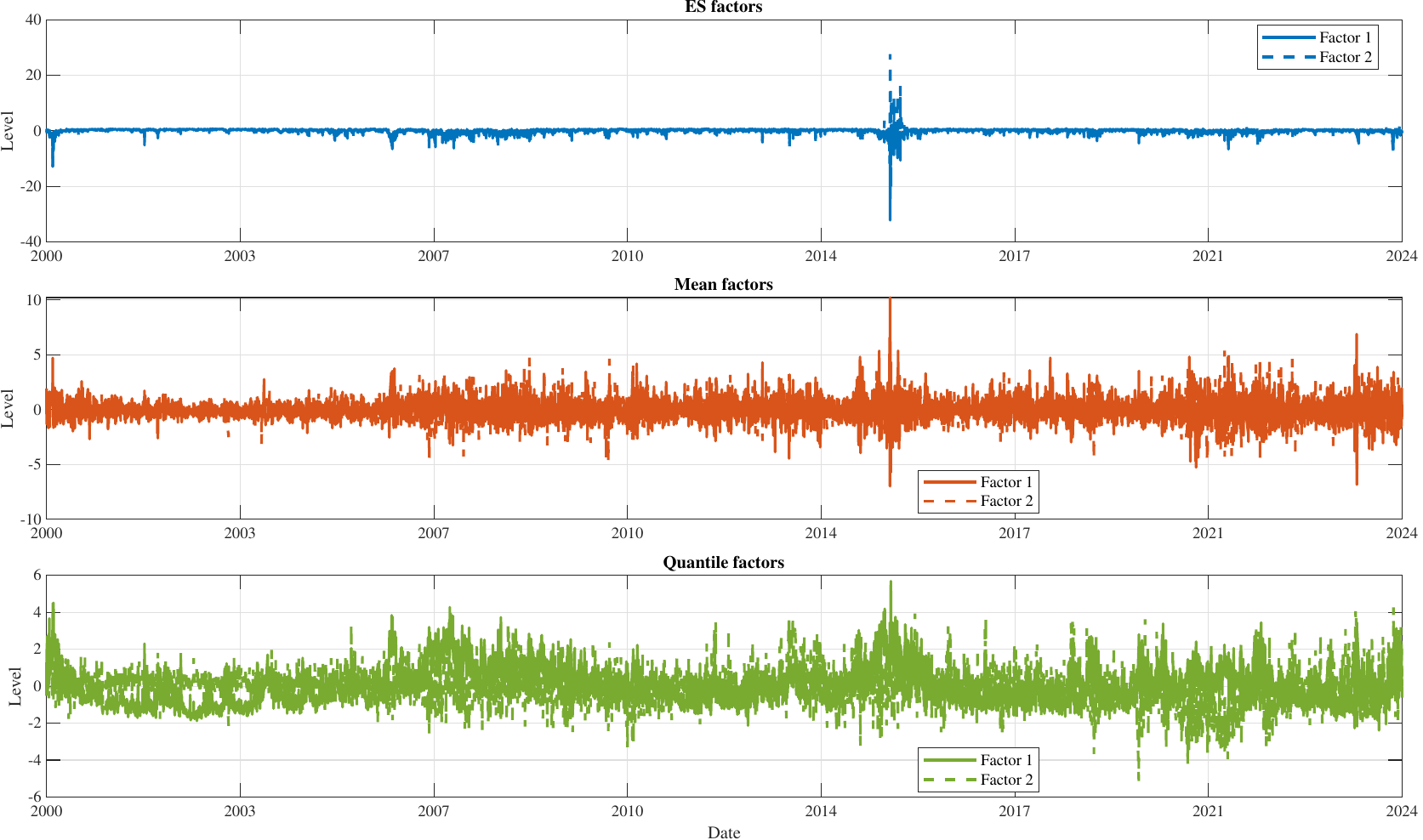}}
{Estimated factor paths across models ($r=2,\tau=0.05$).
\label{fig:factor_EW_0p05}}
{This figure plots the estimated latent factor paths from three models: ESFM (top panel), the mean factor model (middle panel), and the quantile factor model (bottom panel). Each model extracts two factors ($r = 2$) based on equal-weighted Fama--French five factors. The quantile and ES factors are constructed using tail information at level $\tau = 0.05$.}
\end{figure}

\begin{figure}
\FIGURE
{\includegraphics[scale=0.59]{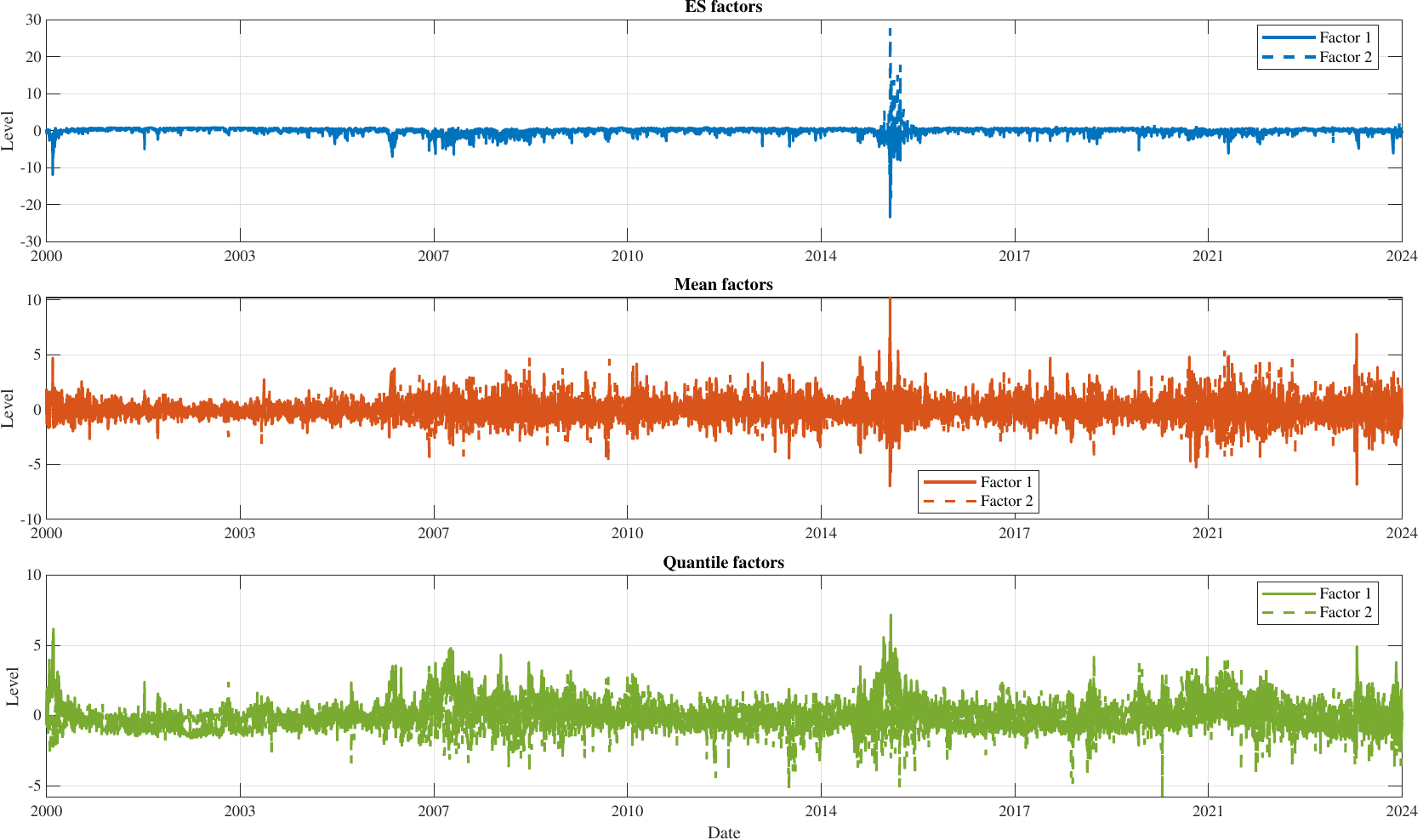}}
{Estimated factor paths across models ($r=2,\tau=0.10$).
\label{fig:factor_EW_0p10}}
{This figure plots the estimated latent factor paths from three models: ESFM (top panel), the mean factor model (middle panel), and the quantile factor model (bottom panel). Each model extracts two factors ($r = 2$) based on equal-weighted Fama--French five factors. The quantile and ES factors are constructed using tail information at level $\tau = 0.10$.}
\end{figure}

Overall, the appendix results confirm that the main empirical findings are not sensitive to alternative specifications. In particular, the ES factors continue to isolate tail risk components that are not captured by mean or quantile factor models, reinforcing the empirical relevance of the proposed framework.

\subsection{Asset Pricing}
As a robustness check, we re-compute the factor-mimicking portfolios using equal-weighted (EW) returns for the underlying Fama–French factors. Figure \ref{fig:cumu_returns_EW} plots the cumulative high-minus-low returns, and Table \ref{tab:Alphas in EW} reports the corresponding average returns and pricing alphas.

\begin{figure}
\FIGURE
{\includegraphics[scale=0.45]{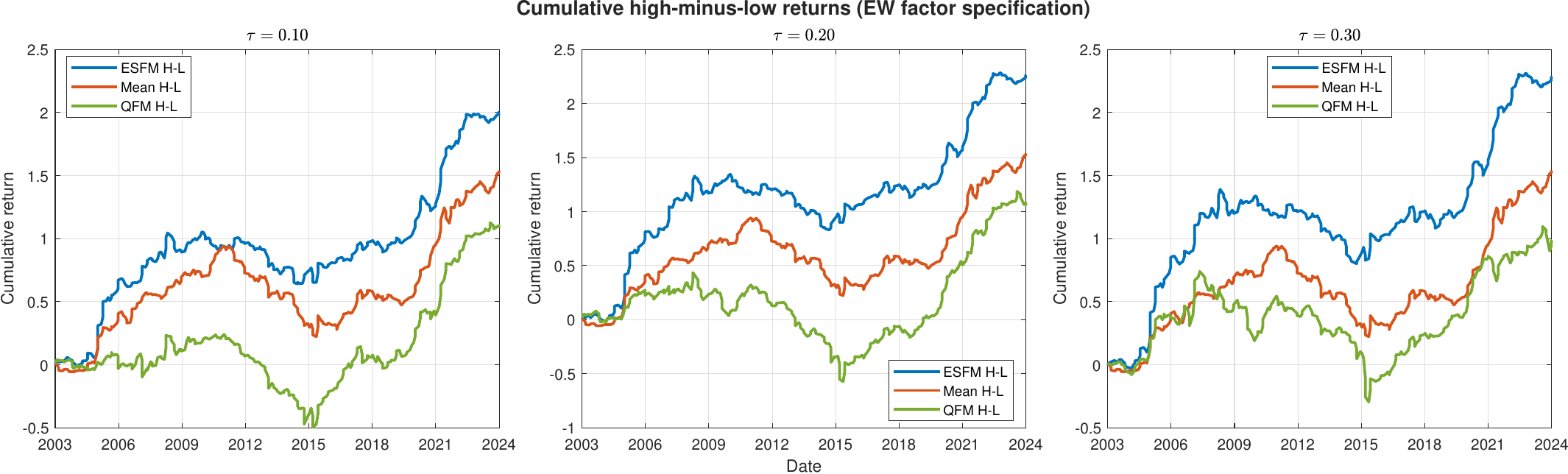}}
{Cumulative high-minus-low returns across models.
\label{fig:cumu_returns_EW}}
{This figure plots cumulative returns of high-minus-low (H--L) portfolios formed on factor exposures estimated from three models: ESFM, the mean factor model (Mean), and the quantile factor model (QFM). Stocks are sorted into five portfolios based on estimated exposures, and the H--L portfolio is constructed as the difference between the highest- and lowest-exposure portfolios. Portfolio returns are computed using equal-weighted (EW) returns, while the underlying Fama--French five factors are constructed using EW returns. Each panel corresponds to a different tail level ($\tau=0.10,0.20,0.30$).}
\end{figure}

The overall patterns remain closely aligned with the results in the main text. In Figure \ref{fig:cumu_returns_EW}, the ESFM-based portfolios continue to deliver the strongest and most persistent H–L performance across all tail levels, while the spreads implied by the mean and quantile factor models are noticeably weaker. This ranking is particularly clear at higher tail levels, where the ESFM spreads widen further relative to the alternatives.

\begin{table}
	\footnotesize
	\TABLE
	{Portfolio returns and alphas sorted on factor exposures across models.
	\label{tab:Alphas in EW}}
	{\footnotesize
	\setlength{\tabcolsep}{2.5pt}
	\renewcommand{\arraystretch}{0.85}
	\begin{tabular*}{\textwidth}
		{@{\extracolsep{\fill}}rrccccrccccrcccc}
		\toprule
		& & \multicolumn{4}{c}{$\tau=0.10$}
		& & \multicolumn{4}{c}{$\tau=0.20$}
		& & \multicolumn{4}{c}{$\tau=0.30$} \\
		\cmidrule{3-6}\cmidrule{8-11}\cmidrule{13-16}
		No.G & & Average & CAPM & FF3 & FF5
		& & Average & CAPM & FF3 & FF5
		& & Average & CAPM & FF3 & FF5 \\
		\midrule
		
		5 & ESFM
		& 9.16\% & 10.18\% & 11.58\% & 11.32\%
		& & 10.33\% & 11.45\% & 13.41\% & 12.17\%
		& & 10.37\% & 11.58\% & 14.02\% & 12.38\% \\
		
		& 
		& (2.86) & (3.32) & (3.96) & (3.62)
		& & (2.93) & (3.35) & (4.16) & (3.68)
		& & (2.80) & (3.22) & (4.27) & (3.74) \\
		
		5 & Mean
		& 7.02\% & 7.41\% & 7.82\% & 8.05\%
		& & 7.02\% & 7.41\% & 7.82\% & 8.05\%
		& & 7.02\% & 7.41\% & 7.82\% & 8.05\% \\
		
		&
		& (2.43) & (2.55) & (2.83) & (2.58)
		& & (2.43) & (2.55) & (2.83) & (2.58)
		& & (2.43) & (2.55) & (2.83) & (2.58) \\
		
		5 & QFM
		& 4.95\% & 5.94\% & 4.94\% & 5.08\%
		& & 4.82\% & 5.73\% & 6.93\% & 6.05\%
		& & 4.42\% & 5.07\% & 8.62\% & 4.66\% \\
		
		&
		& (1.66) & (2.04) & (1.72) & (1.72)
		& & (1.46) & (1.76) & (2.33) & (1.91)
		& & (1.21) & (1.35) & (2.67) & (1.48) \\
		
		\midrule
		
		10 & ESFM
		& 10.75\% & 11.99\% & 13.66\% & 13.48\%
		& & 12.73\% & 14.15\% & 16.68\% & 15.80\%
		& & 12.79\% & 14.34\% & 17.65\% & 16.06\% \\
		
		&
		& (2.71) & (3.18) & (3.85) & (3.45)
		& & (3.04) & (3.56) & (4.51) & (4.03)
		& & (2.95) & (3.46) & (4.65) & (4.07) \\
		
		10 & Mean
		& 9.92\% & 10.43\% & 10.96\% & 11.50\%
		& & 9.92\% & 10.43\% & 10.96\% & 11.50\%
		& & 9.92\% & 10.43\% & 10.96\% & 11.50\% \\
		
		&
		& (2.76) & (2.87) & (2.99) & (2.78)
		& & (2.76) & (2.87) & (2.99) & (2.78)
		& & (2.76) & (2.87) & (2.99) & (2.78) \\
		
		10 & QFM
		& 7.56\% & 8.77\% & 7.58\% & 7.28\%
		& & 6.10\% & 7.22\% & 8.97\% & 7.47\%
		& & 5.20\% & 5.85\% & 10.16\% & 5.92\% \\
		
		&
		& (2.01) & (2.41) & (2.05) & (1.89)
		& & (1.60) & (1.92) & (2.53) & (1.98)
		& & (1.32) & (1.46) & (2.91) & (1.71) \\
		
		\bottomrule
	\end{tabular*}
	}
	{This table reports annualized portfolio returns and alphas (in \%) sorted on factor exposures estimated from different models (ESFM, Mean, and QFM) across three tail levels ($\tau=0.10,0.20,0.30$). We report estimated intercepts (alphas) from regressing the returns on different sets of asset pricing factors: market (CAPM), the three-factor model of \citeEC{fama1993common} (FF3), and the five-factor model of \citeEC{fama2015five} (FF5). Stocks are sorted into No.G $=5$ or $10$ portfolios based on estimated exposures, and portfolio returns are computed using the equal-weighted (EW) scheme. The underlying Fama--French three- and five-factor returns are constructed using EW returns. The Average column reports the annualized mean H--L return. The CAPM, FF3, and FF5 columns report annualized intercepts from the corresponding time-series regressions. Newey--West $t$-statistics (six lags) are reported in parentheses. The ``High--Low'' spread captures the return difference between the highest- and lowest-exposure portfolios (the average is reported in column ``Average'').}
\end{table}

Table \ref{tab:Alphas in EW} provides complementary evidence from the alpha perspective. The ESFM portfolios generate the largest return spreads and corresponding alphas across specifications, and these alphas remain economically large after controlling for CAPM, FF3, and FF5 factors constructed using EW returns. In contrast, the pricing performance of the mean and quantile-based portfolios is substantially weaker and less stable.

Taken together, these results confirm that our main findings are not driven by the weighting scheme used to construct the benchmark factors. The ESFM-based tail factors continue to exhibit strong and robust pricing ability under alternative implementations.

\FloatBarrier
\clearpage

\bibliographystyleEC{informs2014}
\bibliographyEC{ref.bib}

\end{document}